\documentclass[manuscript,screen]{acmart}

\AtBeginDocument{%
  }

\setcopyright{acmlicensed}
\copyrightyear{2018}
\acmYear{2018}
\acmDOI{XXXXXXX.XXXXXXX}

\acmConference[Conference acronym 'XX]{Make sure to enter the correct
  conference title from your rights confirmation emai}{June 03--05,
  2018}{Woodstock, NY}
\acmISBN{978-1-4503-XXXX-X/18/06}

\usepackage{graphicx} 
\usepackage{amsmath} 

\usepackage{proof}
\usepackage{booktabs} 
\usepackage{amsthm} 
\usepackage{enumerate} 

\usepackage{hyperref}

\newtheorem{definition}{Definition}[section]
\newtheorem{theorem}{Theorem}[section]

\newtheorem{lemma}{Lemma}[section]
\newtheorem{prop}{Proposition}[section]

\newtheorem{example}{Example}[section]

\newcommand{\rNum}[1]{\expandafter{\romannumeral #1\relax}}
\newcommand{\rNUM}[1]{\uppercase\expandafter{\romannumeral #1\relax}}

\newcommand{\bff}[1]{\mathbf{#1}}
\newcommand{\mbb}[1]{\mathbb{#1}}
\newcommand{\mcl}[1]{\mathcal{#1}}
\newcommand{\spc}[1]{\begin{spacing}{#1}}
\newcommand{\spce}{\end{spacing}}

\newcommand{\la}[0]{\langle}
\newcommand{\ra}[0]{\rangle}
\newcommand{\mfr}[1]{\mathfrak{#1}}

\newcommand{\dddef}[0]{=_{df}}

    \newcommand{\Conf}[0]{\bff{L}}
    
    \newcommand{\Prog}[0]{\bff{P}}

    \newcommand{\true}[0]{\mathit{true}}
    \newcommand{\false}[0]{\mathit{false}}

    \newcommand{\cnt}[0]{\mathit{C}}

    \newcommand{\psuf}[0]{\prec_s}

    \newcommand{\mult}[0]{\preceq_m}
    \newcommand{\pmult}[0]{\prec_m}

    \newcommand{\PDL}[0]{\mbox{PDL}}
    \DeclareMathOperator{\seq}{;}
    \DeclareMathOperator{\cho}{\cup}

    \newcommand{\Sub}[0]{{\textit{sub}}}
    \newcommand{\lup}[0]{*}
    
    \newcommand{\Var}[0]{\mathit{Var}}

    \newcommand{\lm}[0]{\mfr{m}}
    
    \newcommand{\EX}[0]{\textit{CES}}
    \newcommand{\EXA}[0]{\textit{ES}}
    \newcommand{\ex}[0]{\textit{ex}}
    \newcommand{\mex}[0]{\textit{mex}}
    
    \newcommand{\cat}[0]{\diamond}
    \newcommand{\Wd}[0]{\mcl{S}}
    \newcommand{\I}[0]{\mcl{I}}
    \newcommand{\LM}[0]{\bff{M}}
    \newcommand{\LVar}[0]{\mathit{LVar}}
    \newcommand{\GiiiPPL}[0]{\mbox{G3PPL}}
    \newcommand{\GiiiPPLcyc}[0]{\mbox{G3PPL}}
    \newcommand{\GiiiFOPL}[0]{\mbox{G3FOPL}}
    \newcommand{\GiiiFOPLcyc}[0]{\mbox{G3FOPL}}
    \DeclareMathOperator{\ofst}{\mathbf{f}}
    \DeclareMathOperator{\osuf}{\mathbf{suf}}
    \DeclareMathOperator{\ofus}{\mathbf{fus}}
    \DeclareMathOperator{\op}{\mathit{op}}
    \newcommand{\str}[0]{str}
    
    \newcommand{\Arith}[0]{\mbox{$\mfr{A}$}}
    \newcommand{\fo}[0]{{\textit{fo}}}
    \newcommand{\WVar}[0]{\textit{WVar}}
    \newcommand{\TVar}[0]{\textit{TVar}}
    \newcommand{\Nec}[0]{\textit{nec}}
    \newcommand{\Upd}[0]{\textbf{U}}
    \newcommand{\wm}[0]{\mfr{w}}

    \newcommand{\len}[0]{\textbf{len}}

    \newcommand{\upd}[0]{\mcl{U}}
    \DeclareMathOperator{\R}{R}
    
    \DeclareMathOperator{\nxt}{\mathbf{n}}
    \DeclareMathOperator{\unt}{U}
    
    \DeclareMathOperator{\ev}{\Diamond}
    \DeclareMathOperator{\last}{\mathbf{last}}
    \DeclareMathOperator{\all}{A}
    \DeclareMathOperator{\ext}{E}
    \DeclareMathOperator{\nt}{X}
    \newcommand{\tr}[0]{\mfr{T}}
   \newcommand{\PPL}[0]{\mbox{PPL}}
   
    \newcommand{\FOPL}[0]{\mbox{FOPL}}
    \DeclareMathOperator{\bnxt}{\bar{\nxt}}
    \DeclareMathOperator{\fin}{\mathbf{fin}}
    \newcommand{\FODL}[0]{\mbox{FODL}}
    \newcommand{\DLT}[0]{\mbox{DLT}}

    \newcommand{\diffTL}[0]{\mbox{DLT$_d$}}
    \newcommand{\diffTLii}[0]{\mbox{DLT$^2_d$}}

\newcommand{\A}[1]{\mbox{ \textbf{#1} }}
\newcommand{\MO}{\A{MO}}
\newcommand{\EMO}{\A{EMO}}
\newcommand{\VE}{\A{VE}}
\newcommand{\EVE}{\A{EVE}}

\begin{document}

\title{Labelled Process Logic}

\author{Yuanrui Zhang}
\email{yuanruizhang@nuaa.edu.cn}
\orcid{0000-0002-0685-6905}
\affiliation{%
  \institution{College of Software, Nanjing University of Aeronautics and Astronautics}
  \city{Nanjing}
  \state{Jiangsu}
  \country{China}
}

\renewcommand{\shortauthors}{Y. Zhang}

\begin{abstract}
This paper develops a complete labelled proof-theoretic framework for process logic --- an extension of dynamic logic in which formulas specify properties of execution traces rather than only final states. The main difficulty is that first-order process logic must reason about concrete computations while preserving temporal information along regular-program traces. Existing compositional calculi cover important fragments but do not provide a complete treatment of full first-order process logic over regular programs. We address this difficulty by enriching process-logic formulas with labels that explicitly record trace and update information during derivations. Based on this construction, we define cyclic labelled proof systems for both propositional and first-order process logic, respectively denoted by \GiiiPPLcyc\ and \GiiiFOPLcyc. 
We study their soundness and completeness. 
Their soundness follows from a global progress condition: an invalid cyclic derivation would induce infinitely many strict decreases in a natural-valued counter-example measure.  The completeness of \GiiiPPLcyc\ is established by deriving the axioms and rules of the Hilbert system for propositional process logic, while the completeness of \GiiiFOPLcyc\ is obtained relative to the underlying arithmetic theory through the standard arithmetical encoding used for first-order dynamic logic. 
\ifx
!++@rewrite @inst[make it densed, not too long. ]
The soundness is shown by using the cyclic condition to obtain an infinite descent in a well-founded multiset ordering, and the completeness is proved by showing that the labelled systems can derive the established proof rules of process logic and first-order dynamic logic.
++! 
\fi
The result is a uniform framework for process logic in which, for the first time, trace-based program properties and first-order computations can be handled within the same proof structure. 

\ifx
\MO We prove the soundness by showing that every cyclic derivation satisfying the cyclic condition induces an infinite descent in a well-founded multiset ordering, and prove the completeness separately for each system: the completeness for \GiiiPPLcyc\ is established by deriving every axiom and rule of the Hilbert-style proof system for \PPL; completeness of \GiiiFOPLcyc\ follows by adapting the classical arithmetical completeness argument for \FODL\ to the labelled trace-based setting.  The result is a uniform framework in which both trace-based program properties and first-order computations can be verified within the same labelled cyclic proof structure. \EMO
\fi
\end{abstract}

\ccsdesc[500]{Theory of computation~Modal and temporal logics}
\ccsdesc[500]{Theory of computation~Programming logic}
\ccsdesc[500]{Theory of computation~Logic and verification}
\ccsdesc[300]{Theory of computation~Proof theory}

\keywords{Process Logic, Program Logic, Temporal Properties, Labelled Proof Systems, Cyclic Proofs, Program Verification}


\maketitle

\section{Introduction}
\label{section:Introduction}

Process logic~\cite{Harel82} extends dynamic logic~\cite{Harel00} from reasoning about the final states of program executions to reasoning about the whole execution process. In ordinary dynamic logic, a formula $[\alpha]\phi$ states that every terminating execution of program $\alpha$ ends in a state satisfying $\phi$. This view has supported a wide range of program logics, including logics for programming languages~\cite{Beckert2016,RustyDL26}, process algebra~\cite{DLForCCS,GKAT-PDL25}, synchronous systems~\cite{Zhang21,Zhang22}, hybrid systems~\cite{Platzer18,Platzer15dGL}, and many related verification frameworks (e.g.~\cite{Reif95,Rustan07,Platzer08,Zhang21,Minh24}). Process logic strengthens this setting by allowing $\phi$ to describe trace properties using linear temporal logic (LTL~\cite{TemporalLogic}): a formula $[\alpha]\phi$ states that every execution trace of $\alpha$ satisfies the temporal property $\phi$. $\phi$ can speak about what happens during the execution rather than only after termination~\cite{Harel82}. This makes process logic a natural formalism for specifications in which intermediate behaviour is essential.

The propositional theory of process logic (\PPL) has been thoroughly studied, including decidability and completeness~\cite{Harel82}. The first-order level, however, is more delicate. First-order process logic (\FOPL) interprets atomic programs as concrete computations, such as assignments over an arithmetic structure, and interprets atomic propositions as predicates over program states. This expressiveness is necessary for reasoning about actual computational systems, but it also exposes a structural difficulty: the proof system of \PPL\ does not lift smoothly to the first-order case (cf.~\cite{Harel82,Harel00}). In particular, for a sequential program $\alpha\seq\beta$, a proof of $[\alpha\seq\beta]\phi$ would require a compositional way to split the trace property $\phi$ across the executions of $\alpha$ and $\beta$. Such a split depends heavily on the shape of the LTL formula involving the operators $\nxt$ (next), $\Box$ (always), $\ev$ (eventually), and $\unt$ (until). 

\begin{table}[t]
\centering
\small
\begin{tabular}{p{0.10\linewidth}p{0.12\linewidth}p{0.10\linewidth}p{0.16\linewidth}p{0.24\linewidth}p{0.12\linewidth}}
\toprule
Logic & Propositional level & First-order level & LTL operators & Program model & Complete proof system \\
\midrule
\PPL~\cite{Harel82} & $\surd$ & $\times$ & Full & Regular programs & $\surd$ \\
\DLT~\cite{Beckert01} & $\times$ & $\surd$ & $\Box$ & Deterministic while programs & $\surd$ \\
\diffTL~\cite{PlatzersTDL} & $\times$ & $\surd$ & $\Box$ & Regular programs & $\times$ \\
\diffTLii~\cite{DTL2} & $\times$ & $\surd$ & $\Box,\Diamond$ & Regular programs & Partial~\cite{Ahmad21} \\
DTL~\cite{Beckert13} & $\times$ & $\surd$ & Full & Deterministic while programs & $\surd$ \\
LPL & $\surd$ & $\surd$ & Full & Regular programs & $\surd$ \\
\bottomrule
\end{tabular}
\caption{Comparison of propositional process logic, trace-based dynamic logics, and labelled process logic (LPL).  Here, $\times$ in the final column means that completeness is not established yet. For a first-order logic, its completeness is related to some arithmetic theory. }
\label{table:comparison of related work}
\end{table}

Several lines of work address fragments of this difficulty.  Beckert's sequent calculus for first-order dynamic logic (\FODL) with trace modalities (DLT)~\cite{Beckert01} studies formulas of the form $[|\alpha|]\phi$, which correspond to the special process-logic case $[\alpha]\Box\phi$.
Its proof system is complete for deterministic while programs. 
The crucial compositional rule for sequence programs reduces $[\alpha\seq\beta]\Box\phi$ to obligations about the trace of $\alpha$ and the subsequent behaviour of $\beta$, but this rule is specific to the always modality $\Box$. 
Later, \cite{PlatzersTDL} employs the trace modality and inherits the proof system of \DLT\ in a differential temporal dynamic logic (to avoid ambiguity, we name it \diffTL) for regular programs, but it remains open whether its proof system is complete. 
\cite{Zhang24} shows, by means of a counterexample, that the proof system of DLT/\diffTL\ is insufficient to derive every valid temporal property of general sequential regular programs and thus needs to be improved. 
A variant of \diffTL\ which we name \diffTLii~\cite{DTL2} extends \DLT\ with a richer collection of temporal modalities, so that it can reason about properties such as $[\alpha]\Diamond\phi$. 
Moreover, \cite{Ahmad21} establishes its completeness for the deterministic fragment of regular programs.  
But \diffTL$^2$ still cannot handle the temporal operator until $\unt$. Based on \DLT, dynamic trace logic (DTL)~\cite{Beckert13} first proposed a complete calculus for all LTL operators, but only for deterministic while programs rather than regular programs.

The comparison of these works with ours is illustrated in Table~\ref{table:comparison of related work}. 
Consequently, none of these systems above provides a complete proof-theoretic treatment of full first-order process logic over regular programs.

\ifx
\MO Several lines of work address fragments of this difficulty.  Dynamic logic with trace modalities (\DLT,~\cite{Beckert01}) extends \FODL\ with a trace modality $[|\alpha|]\phi$ that captures the intermediate worlds visited during execution. Its proof system is complete for deterministic while programs.  A subsequent compositional rule $([\seq]-\Box)$ reduces $[\alpha\seq\beta]\Box\phi$ to separate obligations on the $\alpha$- and $\beta$-traces, but this rule is specific to the always modality $\Box$ and does not extend directly to other temporal operators.  Differential temporal dynamic logic (dTL$^2$,~\cite{DTL2}) extends this treatment by adding a dedicated compositional rule for ``eventually'' properties, though completeness for general regular programs with genuine nondeterminism remains open.  Dynamic trace logic (DTL,~\cite{Beckert13}) achieves completeness for arbitrary LTL trace properties, but only for deterministic while programs.  When the program model is extended from while programs to nondeterministic regular programs, the key compositional axiom of \DLT\ is no longer valid in general: a test in a sequential composition can block a trace of the first program from continuing as a trace of the second, which invalidates the direction $[\alpha\seq\beta]\Box\phi\rightarrow [\alpha]\Box\phi\wedge [\alpha](\last [\beta]\Box\phi)$~\cite{Zhang24}.  Consequently, no existing system provides a complete proof-theoretic treatment of full first-order process logic over regular programs. \EMO
\fi

\textbf{This paper} takes a different route. Instead of directly searching for a suitable set of compositional rules for temporal formulas, we enrich process-logic formulas with \emph{labels} that explicitly record the relevant execution information.  A labelled formula has the form $\sigma:\phi$, where $\sigma$ represents a trace-like program configuration and $\phi$ is a process-logic formula. 
Labels allow the intermediate execution results of programs to accumulate at each derivation step, so that the proof of a trace property can be deferred until the entire program execution has been unfolded.  This circumvents the difficulty encountered in previous approaches: the proof no longer needs to split every sequential trace by a special-purpose temporal rule at the point where the sequence is encountered. 
The labelled form is conservative over ordinary process logic: an unlabelled formula can be recovered by using a free trace variable as its label.

With this labelling technique, we manage to develop a complete logical framework for both propositional and first-order process logic: the propositional system, \GiiiPPL, extends labelled sequents for both modal logics~\cite{Negri05} and propositional dynamic logic~\cite{Docherty19}, but adapts them to the trace semantics of process logic; the first-order system, \GiiiFOPL, is then developed based on \GiiiPPL\ by building specific rules for explicit program updates in place of their counterparts for abstract labels. 
Since regular programs include iteration, proof search may generate cyclic derivations.  We therefore equip the labelled systems with cyclic proof conditions in the style of cyclic proof theory~\cite{Stirling91,Brotherston07,Brotherston08}.  The cyclic condition provides the mechanism that distinguishes sound cyclic reasoning from arbitrary finite graphs with unsound back-links. 

\ifx
The proposed labelled calculus itself provides a unified as well as modular framework for proving process logical formulas on both propositional and first-ordered levels. 
It is more efficient, in the sense that compared to the Hilbert-style system for PPL (see Table~\ref{???}), it only contains 4 rules instead of 15 rules for the temporal operators (see Table~\ref{???}). 
It provides convenience for automatically proof search, that when conducting proof search, it is much easier to realize the proof by conducting rule $(\osuf)$ to expand a temporal formula, than to choose among (i) - (xv) rules to construct suitable intermediate formulas. 
\fi

The proposed labelled calculus provides a uniform and modular framework for reasoning about process-logic formulas at both the propositional and first-order levels.  Compared to the Hilbert-style proof system for \PPL\ (see Table~\ref{table:ppl-proof-system}), it is more concise: it requires only 5 rules for the temporal operators (see Table~\ref{table:giiippl-trace-rules}) in place of the 15 rules (i)--(xv) (Table~\ref{table:ppl-proof-system}).  This compactness also benefits automated proof search: rather than choosing among rules (i)--(xv) and constructing suitable intermediate formulas, a proof-search procedure need only apply the rules $(\osuf R)$, $(\osuf L)$ to unfold each temporal formula uniformly.

To the best of our knowledge, our work is the first to provide a complete proof system for first-order process logic over regular programs. 

\ifx
The main challenges are as follows.  First, trace properties of sequential programs such as $[\alpha\seq\beta]\phi$ cannot be split compositionally in general, because the required decomposition of $\phi$ depends on its temporal structure and this structure differs for each operator in $\{$next, always, eventually, until$\}$.  Second, no existing system provides a complete proof-theoretic treatment of full first-order process logic over regular programs: prior work either restricts to deterministic while programs, or handles only special temporal operators such as $\Box$.  Third, extending the framework from the propositional to the first-order level requires a fundamentally different labelling mechanism --- one that records program-state information via explicit program updates rather than abstract world variables --- and the corresponding completeness argument must be rebuilt from the classical completeness proof for first-order dynamic logic.

\MO We address these challenges as follows.  The first challenge --- the absence of a uniform compositional rule for sequential programs --- is resolved by enriching process-logic formulas with labels that explicitly record trace and program-state information during derivations.  Because the label carries this information, the proof system can handle all temporal operators uniformly without requiring a special-purpose rule for each connective.  The second challenge --- the lack of a complete proof system for full process logic over regular programs --- is resolved by defining the cyclic labelled systems \GiiiPPLcyc\ and \GiiiFOPLcyc, and proving their completeness: \GiiiPPLcyc\ is shown complete by deriving every axiom and rule of the Hilbert-style proof system for \PPL, and \GiiiFOPLcyc\ is shown complete through the arithmetical encoding of \FODL.  The third challenge --- the need for a different labelling mechanism at the first-order level --- is resolved by introducing update-based labels for \FOPL, which record program-state information via explicit program updates rather than abstract world variables, and by reconstructing the completeness argument from the classical proof for \FODL. \EMO
\fi

The contributions of this paper are as follows.
\begin{itemize}
\item We introduce labelled versions of process-logic formulas, defining trace labels for \PPL\ and update-based labels for \FOPL.
\item We define cyclic labelled proof systems \GiiiPPLcyc\ and \GiiiFOPLcyc, sharing a common labelled framework while building rules specific to propositional and first-order process logic.
\item We prove the soundness of the cyclic labelled systems by reducing a progressing infinite derivation path to an infinite descent in a well-founded ordering. 
\item We prove the completeness of \GiiiPPLcyc\ by deriving the proof system of \PPL, and establish the corresponding first-order result for \GiiiFOPLcyc\ through the usual arithmetical encoding of \FODL.
\end{itemize}


The rest of this paper is organized as follows.  Section~\ref{section:Brief Introduction to Process Logic} first recalls the syntax, semantics, and proof-theoretic background of process logic.  
Section~\ref{section:Overview} then provides an overview: it motivates our labelled approach by showing in detail how classical and existing trace-based program logics are embedded into process logic and what their limitations are, describes the high-level design of our work, and compares our work with other recent temporal program logics. 
Section~\ref{section:Labelled Propositional Process Logic} introduces labelled \PPL\ formulas and the proof system \GiiiPPLcyc, followed by the analysis and proofs of its soundness and completeness in Section~\ref{section:Soundness of GiiiPPLcyc} and~\ref{section:Completeness of GiiiPPLcyc} respectively.  The first-order labelled system \GiiiFOPLcyc\ is developed afterwards in Section~\ref{section:Labelled First-Order Process Logic - 2}, where we also analyze its soundness and completeness.  
Section~\ref{section:Conclusion} concludes this paper and discusses about future work. 
The appendix contains the detailed derivations and auxiliary proofs used by the main theorems. 

\ifx Old version, not useful anymore
Currently there is no first-ordered process logic that has been fully studied. 
The main reason is that formulas like $[x:=e]\phi$ cannot be properly dealt with as there is no universal rules to deal with the situation of $x := e$. 
(The same for programs like $\alpha^*$?)   
The essential of this arises from that the execution of $x:=e$ should depend on the cases of the forms of LTL formulas (whether it is $\Box\phi$, $\ev \phi$, $\nxt\phi$ or $\phi\unt\psi$). But in $[x:=e]\phi$ however, we cannot tell deal with the complexity of $\phi$. 

This non-compositionality of regular expressions is the essential of the difficulties of proposing a suitable first-order process logic. 
On the other hand, however, most programs interested in the domain of software engineering have operational semantics as their nature. 
Reasoning based on programs' transitional behaviours can circumvent the non-compositionality problem of process logic. 
Also by transitional-behaviour-based reasoning, we can benefit from that our framework

We provide a solution of this by proposing a cyclic labelled proof system for reasoning about process logical formulas based on their transitional behaviours. 
The essential idea is that 
we study a labelled formula like $\tr : [x:=e]\phi$, in which we can temporally store the result of the computation of $x := e$ in a trace structure $\tr$, delaying the analysis of different cases of target temporal formulas. 
\fi

\section{Brief Introduction to Process Logic}
\label{section:Brief Introduction to Process Logic}

This section recalls the parts of process logic that are used in the rest of the paper.  We first review propositional process logic~\cite{Harel82}, including its syntax, trace semantics, and proof system, and then describe the first-order instance over arithmetic that is used when defining labelled \FOPL\ formulas.

\subsection{Propositional Process Logic}

\begin{definition}[Syntax of \PPL\ Formula]
\label{def:Syntax of PPL Formula}
    Let $a$ range over atomic programs and $p$ range over atomic propositions.  The syntax of programs, trace formulas, and state formulas is given by:
$$
\begin{aligned}
    \alpha \dddef &\ a\ |\ \phi?\ |\ \alpha\seq \alpha\ |\ \alpha\cho\alpha\ |\ \alpha^*,\\
    \pi \dddef &\ \phi\ |\ \ofst \pi\ |\ \pi\osuf \pi, \\
    \phi\dddef &\ \true\ |\ p\ |\ \neg\phi\ |\ \phi\wedge \phi\ |\ [\alpha]\pi, 
\end{aligned}
$$
\end{definition}

An atomic program $a$ is also called an \emph{action}.  The program $\phi?$ is a test, $\alpha\seq\beta$ is sequential composition, $\alpha\cho\beta$ is nondeterministic choice, and $\alpha^*$ is finite iteration, also called a \emph{star program}.  The modality $[\alpha]\pi$ states that every trace generated by executing $\alpha$ satisfies the trace formula $\pi$.

$\true$ is the Boolean true.
A trace formula $\pi$ is evaluated over a finite trace rather than only over a single state.  The formula $\ofst\pi$ states that $\pi$ holds at the first state of the current trace, while $\pi_1\osuf\pi_2$ is the process-logic until-like suffix connective (cf.~\cite{Harel82}): along a proper suffix of the current trace, $\pi_2$ eventually holds and $\pi_1$ holds on the intermediate suffixes.  
The standard temporal operators of LTL can be expressed by $\ofst$ and $\osuf$ (cf.~\cite{Harel82}), for instance: 
\begin{center}
$\nxt\phi \dddef (\neg\true)\osuf\phi$, $\phi\unt\psi \dddef \psi\vee(\phi\wedge\phi\osuf\psi)$,  
\end{center}
where $\nxt$ denotes the \emph{next} operator and $\unt$ denotes the \emph{until} operator of LTL.  Intuitively, $\nxt\phi$ holds on a trace when $\phi$ holds from the second state onward, while $\phi\unt\psi$ holds when $\psi$ eventually holds and $\phi$ holds on all intermediate steps.  
Other Boolean connectives, such as $\vee$ and $\to$, are treated as abbreviations in the standard way.  We also use the dual modality $\la\alpha\ra\pi$, defined as $\neg[\alpha]\neg\pi$, to state that there exists an execution trace of $\alpha$ satisfying $\pi$.
We introduce a dual operator $\ofus$ for the operator $\osuf$. 
Formula $\phi\ofus\psi$ is defined as $\neg(\neg\phi\osuf \neg\psi)$. 

The semantics of \PPL\ is based on \emph{worlds}, which abstract program states.  A \emph{trace} $tr$ over a set $\Wd$ of worlds is a finite non-empty sequence of worlds.  
We introduce a concatenation operator $\cdot$ between any two \emph{strings}, i.e., a sequence of some type of elements. It does not apply only to traces, as seen in later sections. 
Given two strings $\str_1$ and $\str_2$ with $\str_1 = w_1\ldots w_n$ $(n\ge 1)$ and $\str_2 = t_1\ldots t_m\ldots$ ($m\ge 1$), $\str_1\cdot \str_2$ is partially defined as $w_1\ldots w_nt_2\ldots t_m\ldots$ provided that $w_n = t_1$. 
Sometimes we omit the ``$\cdot$'' and simply write $\str_1\cdot \str_2$ as $\str_1\str_2$. 
For a string $\str$, let $\str_b$ and $\str_e$ denote its first and last elements. 
We use $|\str|$ to denote the length of a string $\str$. 

To align with the later notion of ``empty label'' (Definition~\ref{def:Labels and Label Mappings}), we also introduce an ``empty string'' $\epsilon$, which satisfies $\str\cdot \epsilon = \epsilon\cdot \str = \str$ for any string $\str$. 

A \emph{Kripke frame} is a pair $(\Wd,\I)$, where $\Wd$ is a set of worlds and $\I$ interprets each atomic program as a set of length-two traces and each atomic proposition as a set of worlds.

\begin{definition}[Semantics of \PPL]
    \label{Semantics of PPL}
Given a Kripke frame $(\Wd, \I)$, the semantics of programs and formulas is defined by extending the function $\I$, using simultaneous induction on the syntactic structures of programs and formulas: 
\begin{enumerate}
\item $\I(\phi?)\dddef \{ss\ |\ s\in \I(\phi)\}$.
\item $\I(\alpha\seq\beta)\dddef \I(\alpha) \cdot \I(\beta)$, where $\I(\alpha) \cdot \I(\beta)\dddef \{tr_1\cdot tr_2\ |\ tr_1\in \I(\alpha), tr_2\in \I(\beta)\}$.
\item $\I(\alpha\cho\beta)\dddef \I(\alpha)\cup \I(\beta)$.
\item $\I(\alpha^*)\dddef \bigcup_{n\ge 0} \I(\alpha^n)$, where $\alpha^0\dddef \true?$, $\alpha^n\dddef \alpha\seq\alpha^{n-1}$ ($n\ge 1$). 
\item $\I(\true)\dddef \Wd$.
\item $\I(\neg\phi)\dddef \Wd \setminus \I(\phi)$.
\item $\I(\phi\wedge \psi)\dddef \I(\phi)\cap \I(\psi)$.
\item $\I(\ofst \phi)\dddef \{tr\ |\ tr_b\in \I(\phi)\}$.
\item $\I(\phi\osuf \psi)\dddef \{tr\ |\ \exists tr'. tr'\psuf tr\wedge tr'\in \I(\psi) \wedge (\forall tr''. tr'\psuf tr''\psuf tr \to tr''\in \I(\phi))\}$, where $tr_1\psuf tr_2$ means that $tr_1$ is a proper suffix of $tr_2$.  
\item $\I([\alpha]\phi)\dddef \{tr\ |\ \mbox{for all $tr'\in \I(\alpha)$, $tr\cdot tr'\in \I(\phi)$}\}$.
\end{enumerate}
\end{definition}

The clauses above follow the intended trace interpretation of regular programs and process formulas. Tests produce stuttering traces when their condition holds; sequence composes compatible traces; choice takes unions; and iteration takes the union over all finite unfoldings. Boolean connectives are interpreted set-theoretically. The trace operators inspect the beginning and suffixes of a trace, and $[\alpha]\phi$ holds of a trace precisely when appending any trace generated by $\alpha$ yields a trace satisfying $\phi$. 

Given a Kripke frame $(\Wd,\I)$, a trace $tr\in\Wd^*$ \emph{satisfies} a formula $\phi$, written $tr\models_{(\Wd,\I)}\phi$ or simply $tr\models\phi$, if $tr\in\I(\phi)$.  State satisfaction is included as the special case of traces of length one.

The proof system of \PPL\ is recalled in Table~\ref{table:ppl-proof-system}.  It consists of the usual propositional dynamic logic (PDL) axioms and rules (cf.~\cite{Harel00}), together with the additional rules for process-logic trace operators.
The non-primitive operators $\bnxt$ and $\fin$ appearing in the rules are defined as follows:
$$
\bnxt \phi\dddef \neg\nxt\neg\phi, \mbox{ }\fin \dddef \bnxt \false \vee (\true\osuf(\true\wedge\bnxt \false)), 
$$
where $\bnxt$ is the dual of $\nxt$; $\bnxt\phi$ holds on a trace when it has a next state from which $\phi$ holds onwards or when it does not have a next state; and $\fin$ intuitively means that the trace has finite length.

For further explanation of these additional rules, see~\cite{Harel82}.

\begin{table}[t]
\centering
\small
\begin{tabular}{p{0.09\linewidth}p{0.84\linewidth}}
\toprule
No. & Rule or axiom \\
\midrule
\multicolumn{2}{l}{\textbf{PDL part}}\\
(1) & $\begin{aligned}[\alpha](\phi\to\psi)\to([\alpha]\phi\to[\alpha]\psi)\end{aligned}$\\
(2) & $\begin{aligned}[\alpha](\phi\wedge\psi)\leftrightarrow[\alpha]\phi\wedge[\alpha]\psi\end{aligned}$\\
(3) & $\begin{aligned}[\alpha\cho\beta]\phi\leftrightarrow[\alpha]\phi\wedge[\beta]\phi\end{aligned}$\\
(4) & $\begin{aligned}[\alpha\seq\beta]\phi\leftrightarrow[\alpha][\beta]\phi\end{aligned}$\\
(5) & $\begin{aligned}[\psi?]\phi\leftrightarrow(\psi\to\phi)\end{aligned}$\\
(6) & $\begin{aligned}[\alpha^*]\phi\leftrightarrow\phi\wedge[\alpha][\alpha^*]\phi\end{aligned}$\\
(7) & $\begin{aligned}\phi\wedge[\alpha^*](\phi\to[\alpha]\phi)\to[\alpha^*]\phi\end{aligned}$\\
(MP) & $\begin{aligned}\infer[^{(\textit{MP})}]{\phi}{\psi & \psi\to\phi}\end{aligned}$\\
(Gen) & $\begin{aligned}\infer[^{(\textit{Gen})}]{[\alpha]\phi}{\phi}\end{aligned}$\\
\midrule
\multicolumn{2}{l}{\textbf{\PPL-specific part}}\\
(i) & $\begin{aligned}\ofst(\phi\vee\psi)\leftrightarrow\ofst\phi\vee\ofst\psi\end{aligned}$\\
(ii) & $\begin{aligned}\ofst\neg\phi\leftrightarrow\neg\ofst\phi\end{aligned}$\\
(iii) & $\begin{aligned}(\phi\osuf\psi)\vee(\phi\osuf\chi)\leftrightarrow\phi\osuf(\psi\vee\chi)\end{aligned}$\\
(iv) & $\begin{aligned}&(\phi\osuf\psi)\wedge(\chi\osuf\omega)\leftrightarrow(\phi\wedge\chi)\osuf(\psi\wedge\omega)\vee(\phi\wedge\chi)\osuf(\phi\wedge\omega\wedge\phi\osuf\psi)\vee (\phi\wedge\chi)\osuf(\psi\wedge\chi\wedge\chi\osuf\omega)\end{aligned}$\\
(v) & $\begin{aligned}\neg(\phi\osuf\psi)\leftrightarrow\neg(\true\osuf\psi)\vee(\neg\psi)\osuf(\neg\phi\wedge\neg\psi)\end{aligned}$\\
(vi) & $\begin{aligned}\phi\osuf\psi\leftrightarrow\nxt\psi\vee\nxt(\phi\wedge(\phi\osuf\psi))\end{aligned}$\\
(vii) & $\begin{aligned}\phi\osuf(\phi\wedge(\phi\osuf\psi))\leftrightarrow\nxt(\phi\wedge(\phi\osuf\psi))\end{aligned}$\\
(viii) & $\begin{aligned}\phi\osuf(\phi\wedge(\phi\osuf\psi))\leftrightarrow\phi\osuf(\phi\wedge\nxt\psi)\end{aligned}$\\
(ix) & $\begin{aligned}\ofst\nxt\true\leftrightarrow\false\end{aligned}$\\
(x) & $\begin{aligned}\neg\phi\wedge\ofst\phi\to\nxt\true\end{aligned}$\\
(xi) & $\begin{aligned}p\leftrightarrow\ofst p\end{aligned}$, where $p$ is an atomic proposition\\
(xii) & $\begin{aligned}\ofst\phi\wedge\la\alpha\ra\psi\leftrightarrow\la\alpha\ra(\ofst\phi\wedge\psi)\end{aligned}$\\
(xiii) & $\begin{aligned}\nxt\la\alpha\ra\phi\leftrightarrow\nxt\true\wedge\la\alpha\ra\bnxt\phi\end{aligned}$\\
(xiv) & $\begin{aligned}\la\alpha\ra\true\to\fin\end{aligned}$\\
(xv) & $\begin{aligned}((\nxt\phi\to\phi)\osuf\phi)\to\nxt\phi\end{aligned}$\\
\bottomrule
\end{tabular}
\caption{Axioms and rules of propositional process logic.}
\label{table:ppl-proof-system}
\end{table}

The following result is directly from~\cite{Harel82}.

\begin{prop}
The proof system of $\PPL$ in Table~\ref{table:ppl-proof-system} is sound and complete. 
\end{prop}

\subsection{First-order Process Logic}
\label{section:First-order Process Logic}

First-order process logic (\FOPL) enriches \PPL\ by interpreting atomic programs as concrete assignment steps over an arithmetic structure and atomic propositions as first-order predicates over program states. 
This expressiveness is sufficient to capture a wide range of computational systems: first-order dynamic logic (\FODL) underlies verification frameworks for Java~\cite{Beckert2016}, hybrid systems~\cite{Platzer18}, and synchronous programs~\cite{Zhang21,Zhang22}, among others. 

Studies of \FODL\ often consider a first-order structure in which first-order quantification is allowed~\cite{Harel00}. 
For simplicity, we consider a basic version of the arithmetic theory over the integers, named \Arith, as the first-order structure for \FOPL. 
\Arith\ is expressive enough to capture most interesting computational structures in computer systems. 
The results of this paper can also be extended easily to other first-order structures.

The setting of \Arith\ is as follows. Let $\Var$ be the set of variables over the integer set $\mbb{Z}$. 
We use letters $u,v$ to range over the variables of $\Var$. 
An \emph{arithmetic expression} (simply ``expression'') $e$ is defined recursively as follows:
$$
e \dddef u\ |\ n\ |\ e \bowtie e,
$$
where $u\in\Var$; $n\in \mbb{Z}$ is a constant; and $\bowtie\in \{+,-,*,\div\}$ represents one of the usual arithmetic operators. 
A world is a mapping from $\Var$ to $\mbb{Z}$.
Given a world $s : \Var\to \mbb{Z}$,  
$s[u\mapsto i]$ returns a mapping that maps $u$ to the value $i\in \mbb{Z}$ and maps any other variable $v$ to the value $s(v)$. 
Given an expression $e$, the \emph{evaluation} of $e$ by $s$, denoted by $s(e)$, is defined in the usual way: $s(n) \dddef n$; $s(e_1\bowtie e_2)\dddef s(e_1)\bowtie s(e_2)$ for any $e_1, e_2$.  

A formula in \FOPL\ is defined just as in Definition~\ref{def:Syntax of PPL Formula}, except that an atomic program is an assignment $u := e$, where $u\in \Var$ is a variable and $e$ is an expression; an atomic formula is a predicate of the form $e_1 \triangle e_2$, where $\triangle$ is one of the usual binary integer relations in $\{<, \le, >, \ge, =\}$; and, given an \FOPL\ formula $\phi$, $\forall u.\phi(u)$ is also a formula quantified by the operator $\forall$ over a variable $u\in \Var$. 
$\forall$ has the usual dual operator $\exists$. 
The evaluation of a predicate by $s$ is defined in the usual way: $s(e_1\triangle e_2)\dddef s(e_1)\triangle s(e_2)$.

The \emph{substitution} in \FOPL\ is defined in the usual way. 
Given a formula $\phi$, $\phi[e/u]$ denotes the formula obtained by replacing each free occurrence of the variable $u$ (not bound by any quantifier) with the term $e$ in $\phi$. 
We always assume that a substitution is \emph{admissible}, in the sense that the substitution does not affect the freeness of the other existing variables in $\phi$ except $u$.

The special theory \Arith\ fixes the Kripke frame in \FOPL, namely $(\Wd_{\fo}, \I_{\fo})$, where $\Wd_{\fo}$ is the set of all worlds in $\Arith$, and $\I_\fo$ interprets an assignment and a predicate in \Arith\ as follows:
\begin{enumerate}
\item $\I_\fo(u := e) \dddef \{ss' \mid s' = s[u\mapsto s(e)]\}$. 
\item $\I_\fo(e_1\triangle e_2) \dddef \{s \ |\ s(e_1\triangle e_2)\mbox{ is true in theory $\Arith$}\}$. 
\end{enumerate} 
The semantics of \FOPL\ is obtained by extending $(\Wd_\fo, \I_\fo)$ to other programs and formulas as in Definition~\ref{Semantics of PPL}, except that, for a quantified formula of the form $\forall u.\phi$, we have 
$$
\I_\fo(\forall u.\phi) \dddef \{tr \mid \forall i\in \mbb{Z}. (tr)[u\mapsto i]\in \I_\fo(\phi)\}, 
$$
where $(tr)[u\mapsto i]$ is the trace obtained by replacing each world $s$ in $tr$ with $s[u\mapsto i]$. 

\ifx
\textbf{Proof System of \FODL}
(I) All rules from \PDL
(II) $[x:=e]\phi\leftrightarrow \phi[e/x]$
(III)
$
\begin{aligned}
    \infer[^{(\textit{Con})}]
    {\phi(n)\to \la\alpha*\ra\phi(0)}
    {\phi(n+1)\to\la\alpha\ra\phi(n)}
\end{aligned}
$

\fi



\section{Overview}
\label{section:Overview}

Having recalled the basic definitions of process logic, we now step back and ask: \emph{why do we need a new proof framework, and how does our approach work?}  This section gives an informal tour of the main ideas before the technical development begins.

Section~\ref{section:Motivations} motivates our work by showing that classical and trace-based program logics can be embedded into process logic, making a complete proof system for it broadly useful. 
\ifx
!++@rewrite @keep
Section~\ref{section:Labelled System and Our Approach} introduces the core idea of enriching process-logic formulas with labels that explicitly record execution-trace information, enabling uniform reasoning without complex compositional rules.  Section~\ref{section:Analysis of Soundness and Completeness} then outlines the soundness and completeness arguments for the proposed systems. 
++!
\fi
Section~\ref{section:Labelled System and Our Approach} introduces the core idea of the cyclic labelled approach and outlines the soundness and completeness arguments for the proposed labelled systems. 
Section~\ref{section:Comparisons to Other Program Logics} compares our approach with other recent program logics for reasoning about temporal properties. 

\subsection{Motivations of Process Logic}
\label{section:Motivations}

Process logic is highly expressive: its trace-based formulas subsume classical Hoare triples, dynamic-logic modalities, and the richer LTL-flavoured properties studied in recent traced program logics.  Providing a \emph{complete} proof system for it is therefore both theoretically significant and broadly applicable, yet non-trivial --- especially at the first-order level.  
In the subsections below we make this motivation concrete by showing how the key logics from the literature each embed into process logic.

\subsubsection{From State Properties to Trace Properties: Hoare Logic and Dynamic Logic}

Because process logic evaluates formulas over execution \emph{traces}, it is strictly more expressive than purely state-based logics.  In particular, it subsumes the classical Hoare-logic specification $\{\phi\}\alpha\{\psi\}$, which asserts that if a program state satisfies the precondition $\phi$ then, after executing program $\alpha$, the resulting state satisfies the postcondition $\psi$.

In process logic, this correctness condition is captured by the formula
$$
\phi\to [\alpha]\Box(\L_0\to \psi), 
$$
where the auxiliary abbreviations are defined as follows~\cite{Harel82}:
$$
\begin{aligned}
    \Box \varphi&\dddef \neg(\Diamond \neg \varphi),\\
    \Diamond \varphi&\dddef \true \unt \varphi,\\
    \L_0&\dddef \neg \nxt \true.
\end{aligned}
$$
Here, $\Box\varphi$ (``always $\varphi$'') means that $\varphi$ holds at \emph{every} state along the current trace; $\Diamond\varphi$ (``eventually $\varphi$'') means that $\varphi$ holds at \emph{some} state along the trace; and $\L_0$ (``last state'') is a formula that holds precisely at the \emph{final} state of a trace, since $\neg\nxt\true$ is satisfied exactly when there is no next step.  Thus $\Box(\L_0\to\psi)$ asserts that at the last state of every execution trace of $\alpha$ the formula $\psi$ holds --- which is exactly the Hoare postcondition. 

Process logic also subsumes classical dynamic logic.  In dynamic logic, every formula is a \emph{state formula} --- it is satisfied by a single world, not by a trace.  The box modality $[\alpha]\phi$ therefore asserts that the world obtained after executing $\alpha$ (in any possible way) satisfies $\phi$.  Because process logic evaluates formulas over traces rather than over single states, the equivalent process-logic statement must single out the \emph{last} state of the execution trace.  This is achieved by the formula
$$
    [\alpha]\Box(\L_0\to \phi),
$$
which, by the same reasoning as above, says that at the final state of every execution trace of $\alpha$ the formula $\phi$ holds.  

\subsubsection{Existing Trace-Based Dynamic Logics and Their Limitations}

Process logic also subsumes the various trace-semantic-based dynamic logics discussed in Section~\ref{section:Introduction}. For each, we briefly explain its core idea and show how its key formulas and rules can be expressed in process logic. This discussion motivates our own labelled framework, introduced in Section~\ref{section:Labelled System and Our Approach}, by highlighting the limitations of each approach. 

\textbf{Dynamic Logic with Trace Modalities (DLT).}  DLT~\cite{Beckert01} extends \FODL\ with a trace modality $[|\cdot|]$ that captures the intermediate worlds visited during program execution.  
In \DLT, the underlying program model consists of while programs, and the proof system is proved to be complete for this class~\cite{Beckert01}. 
The formula $[|\alpha|]\phi$ means that $\phi$ (a state formula) holds at \emph{every} intermediate world along every execution trace of $\alpha$.  In process logic, this is captured precisely by
$$
    [\alpha]\Box \phi,
$$
using the ``always'' operator $\Box$ introduced earlier.

For sequential programs, \cite{Beckert01} proposes the following compositional rule, which we call $([\seq]-\Box)$:
$$
\infer[]
{[|\alpha\seq\beta|]\phi}
{[|\alpha|]\phi\wedge [\alpha][|\beta|]\phi}
$$
Expressed as a process-logic axiom, this rule becomes:
$$
[\alpha\seq\beta]\Box\phi\leftrightarrow [\alpha]\Box\phi\wedge [\alpha](\last [\beta]\Box\phi), 
$$
where $\last \varphi\dddef \Diamond(\varphi\wedge \L_0)$ for any formula $\varphi$.  Intuitively, $\last\varphi$ means that $\varphi$ holds \emph{at the last state of the current trace}: $\Diamond(\varphi\wedge\L_0)$ requires that there is some state on the trace that is simultaneously the final state ($\L_0$) and satisfies $\varphi$.  The axiom then says: every trace of $\alpha\seq\beta$ satisfies $\Box\phi$ if and only if every trace of $\alpha$ already satisfies $\Box\phi$, \emph{and} at the end of every trace of $\alpha$, every continuation trace of $\beta$ also satisfies $\Box\phi$.  In other words, the axiom splits the sequential execution into two parts and ensures that $\phi$ is maintained throughout both segments.  In process logic, this rule is more complex than the simpler rule $(4)$ of Table~\ref{table:ppl-proof-system} for handling sequential programs.

The subsequent work~\cite{PlatzersTDL} proposes the logic \diffTL, which adopts the trace modalities of \DLT\ as an extension to differential dynamic logic~\cite{Platzer07b}, applying them to the more general regular programs.  However, it remains open whether the system fragment of \DLT\ in \diffTL\ is complete.  Moreover, \cite{Zhang24} shows that for regular programs the direction $[\alpha\seq\beta]\Box\phi\rightarrow [\alpha]\Box\phi\wedge [\alpha](\last [\beta]\Box\phi)$ does not hold: a test in $\alpha\seq\beta$ may block a trace of $\alpha$ from continuing as a trace of $\beta$, thereby invalidating the compositional split.  Consequently, rule $([\seq]-\Box)$ is too strong and is thus insufficient to derive all valid trace properties for arbitrary regular programs. 

\textbf{Differential Temporal Dynamic Logic \diffTLii.} It was later observed in~\cite{PlatzersTDL} that rule $([\seq]-\Box)$ cannot be applied to reasoning about the formula $[\alpha]\Diamond\phi$ for the ``eventually'' operator $\Diamond$: the rule breaks down because distributing the $\Diamond$ modality over a sequential composition does not follow the same pattern as $\Box$. 
To handle ``eventually'' formulas, \cite{DTL2} introduces the logic \diffTLii, which enriches the formula $[\alpha]\Diamond\phi$ into the richer form $[\alpha](\psi\sqcup \Diamond\phi)$ and proposes a compositional rule named $([\seq]-\Diamond)$: 
$$
\infer[]
{[\alpha\seq\beta](\psi\sqcup \Diamond\phi)}
{[\alpha]([\beta](\psi\sqcup\Diamond\phi)\sqcup \Diamond\phi)}. 
$$
In process logic, the formula $[\alpha](\psi\sqcup \Diamond\phi)$ is captured as
$$
[\alpha](\psi\vee \Diamond\phi), 
$$
which means that each trace of $\alpha$ satisfies either $\psi$ or $\Diamond\phi$. 
And the rule $([\seq]-\Diamond)$ can be captured as the axiom:
$$
[\alpha\seq\beta](\psi\vee \Diamond\phi) \leftrightarrow [\alpha](\last ([\beta](\psi\vee \Diamond\phi))\vee \Diamond\phi). 
$$
Informally, the axiom decomposes the ``eventually'' property of the sequential program $\alpha\seq\beta$: the combined execution satisfies $\psi\vee\Diamond\phi$ if and only if, at the end of every $\alpha$-trace, either (i) $\Diamond\phi$ has already been witnessed within that $\alpha$-trace, or (ii) the subsequent $\beta$-execution will witness $\psi\vee\Diamond\phi$ from that terminal state. 

Later work~\cite{Ahmad21} proves that \diffTLii\ is complete for deterministic program behaviours, i.e., given a fixed starting world, there is only one possible execution trace. 
In regular programs, this is possible if the choice operator $\cho$ is used with mutually exclusive Boolean test conditions. For example, consider the program $\phi?\seq a\cho \neg\phi?\seq b$ for any $\phi$, $a$, and $b$. 
However, the completeness for general regular programs with genuine nondeterminism remains open; moreover, it is not clear whether any fixed set of compositional rules suffices to handle all LTL operators uniformly.

\textbf{Dynamic Trace Logic (DTL).} Earlier work~\cite{Beckert13} proposes another variant of trace-based dynamic logic called dynamic trace logic (DTL). 
The program models of DTL are restricted to while programs, which are deterministic. 
Under this restriction, the proof system of DTL\ can derive every valid formula $[\alpha]\phi$ for an \emph{arbitrary} LTL formula $\phi$.  The key insight is that compositional rules for a single deterministic trace work uniformly for all types of temporal properties.  However, this argument breaks down for regular programs in general, where nondeterminism allows multiple possible execution traces and the compositional rules must handle branching.

\subsection{Labelled System \& Our Approach}
\label{section:Labelled System and Our Approach}
\subsubsection{A Labelled Approach for Process Logic}
We adopt an alternative approach to provide a complete logical framework for process logic at both propositional and first-order levels. Our approach draws on the rich tradition of \emph{labelled proof systems} for modal and dynamic logics (cf., e.g.,~\cite{Negri05,Docherty19}), which augment formulas with explicit structural information---called \emph{labels}---that records the relational or semantic context in which a formula is evaluated. By making this context part of the formula itself, labelled systems often admit cleaner, more modular inference rules than their unlabelled counterparts. 

Following the main idea from prior work on labelled dynamic-logic systems~\cite{Docherty19,Zhang25-2,zhang2025parameterizeddynamiclogic}, we introduce a label $\sigma$ to capture the current execution-trace information accumulated during a derivation. A labelled process-logic formula takes the form $\sigma : [\alpha]\phi$, meaning that under the trace information recorded in $\sigma$, every execution trace of $\alpha$ satisfies the formula $\phi$. We shall see that this labelled form is strictly more general than the plain formula $[\alpha]\phi$: when $\sigma$ is a free trace variable with no additional constraints, $\sigma : [\alpha]\phi$ and $[\alpha]\phi$ share the same validity. 

To see why labels help, consider the formula $$\{x:=1\} : [x:=2 \seq x := 3]\Box (x\ge 0),$$ which means that given the execution trace information $\{x:=1\}$ (recording that $x$ was set to $1$ at the first step), every continuation via program $x:=2\seq x:=3$ produces a whole trace satisfying $\Box (x\ge 0)$.  By applying a labelled version of the sequential-composition rule, we can unfold this derivation step by step, accumulating execution information in the label:
$$
\infer[]
{\{x:=1\} : [x:=2\seq x:=3]\Box (x\ge 0)}
{
    \infer[]
    {\{x:=1\} : [x:=2][x:=3]\Box (x\ge 0)}
    {
        \infer[]
        {\{x:=1\}\cat \{x:=2\} : [x:=3]\Box (x\ge 0)}
        {
            \{x:=1\}\cat \{x:=2\}\cat \{x:=3\} : \Box (x\ge 0)
        }
    }
}, 
$$
where the label $\{x:=1\}\cat \{x:=2\}\cat \{x:=3\}$ records the full execution history.  Once all programs have been unfolded, the verification goal reduces to checking whether this concrete label satisfies the temporal formula $\Box(x\ge 0)$ --- a much simpler task.  This contrasts with the traditional approach using the rule $([\seq]-\Box)$, which would split $[x:=2\seq x:=3]\Box (x\ge 0)$ into separate sub-goals $[x:=2]\Box (x\ge 0)$ and $[x:=3](\last ([x:=3]\Box (x\ge 0)))$ that must be handled independently.

The formal construction of the labelled systems for \PPL\ and \FOPL\ --- namely \GiiiPPL\ and \GiiiFOPL\ --- is carried out in Sections~\ref{section:Labelled PPL Formulas}, \ref{section:A Proof System for Labelled PPL}, \ref{section:Labelled First-Order Process Logic}, and~\ref{section:System GiiiFOPLcyc}, respectively. 

Similar to~\cite{Docherty19,zhang2025parameterizeddynamiclogic}, our labelled system for process logic does not directly support reasoning about iterative regular programs: when unwrapping the iterative structure, the same program form repeats during a derivation, producing an infinite derivation.
The following derivation illustrates this: 
$$
\begin{aligned}
    \infer[]
    {X : [(x:=1)^*]\true}
    {
        \infer[]
        {X\cat \{x:=1\} : [(x:=1)^*]\true}
        {
            \infer*[]
            {\ldots}
            {
                \infer[]
                {\ldots}
                {
                    \infer[]
                    {X\cat \{x:=1\}\cat\ldots\cat\{x:=n\} : [(x:=1)^*]\true}
                    {\ldots}
                }
            }
        }
    }, 
\end{aligned}
$$
where at each derivation step (backwardly from bottom to top), $x:=1$ is symbolically executed and $(x:=1)^*$ is transformed into itself, starting a new loop.  

We adapt the cyclic proof approach~\cite{Stirling91,Brotherston07,Brotherston08} to the labelled systems of \PPL\ and \FOPL\ in order to handle such infinite derivations. 
A cyclic proof (Definition~\ref{def:Cyclic Proofs of GiiiPPL}) admits a certain type of derivation branch whose leaf nodes can ``back-link'' to identified ancestors along the same branch. 
For instance, the formula $X : [(x:=1)^*]\true$ can be closed by a suitable variable substitution to form a leaf node $Y : [(x:=1)^*]\true$ that back-links to the root $X : [(x:=1)^*]\true$ (they are identical modulo free-variable renaming; here $X$ and $Y$ are free trace variables): 
$$
\infer[]
{X : [(x:=1)^*]\true}
{
    \infer[^{(\Sub)}]
    {\{x:=1\}\cat X : [(x:=1)^*]\true}
    {
        Y : [(x:=1)^*]\true
    }
}.
$$
We discuss the condition --- called ``cyclic condition'' (Definition~\ref{def:Cyclic Proofs of GiiiPPL}) --- under which such a cyclic derivation constitutes a legal proof that leads to a valid conclusion. 

Sections~\ref{section:Cyclic Proof Structure of GiiiPPLcyc} and~\ref{section:System GiiiFOPLcyc} present the details of the cyclic theories for \GiiiPPL\ and \GiiiFOPL, respectively.

\subsubsection{Analysis of Soundness and Completeness}
\label{section:Analysis of Soundness and Completeness}

As the major theoretical results of this paper, we fully prove the soundness and completeness of the proposed labelled systems for \PPL\ and \FOPL.

\textbf{Soundness.}  The soundness is only analyzed for \GiiiPPL\ because the case for \GiiiFOPL\ is similar. We establish the soundness at the two levels.  First, each individual rule of the labelled system is locally sound (Theorem~\ref{theo:Soundness of Each Rule of GiiiPPL}).  Local soundness alone does not suffice, however, because a cyclic derivation may contain infinite derivation branches, and local soundness does not directly guarantee that such a branch leads to a valid conclusion.  We therefore establish a second, global soundness result: every cyclic derivation satisfying the cyclic condition yields a valid conclusion.  This is discussed in detail in Section~\ref{section:Soundness of GiiiPPLcyc}.  For both \GiiiPPL\ and \GiiiFOPL, the soundness arguments follow the main strategy from~\cite{zhang2025parameterizeddynamiclogic}, but adapted with non-trivial technical details specific to the trace semantics of process logic.

\textbf{Completeness.}  We establish the completeness using different proof methods for \GiiiPPL\ and \GiiiFOPL, discussed in Sections~\ref{section:Completeness of GiiiPPLcyc} and~\ref{section:Soundness and Completeness of GiiiFOPLcyc} respectively.  For \GiiiPPLcyc, the key idea is to show that the labelled system can simulate every axiom and rule of the established Hilbert-style proof system for \PPL\ displayed in Table~\ref{table:ppl-proof-system}.  We present the technical details for the critical ``Necessitation'' rule, the PDL axioms (rules (1)--(7), (MP), (Gen)), and one side of the process-logic-specific axiom (iii), leaving the proofs for the other process-logic-specific axioms (rules (i), (ii), (iv) -- (xv), and the other side of rule (iii)) to the appendix.  For \GiiiFOPLcyc, we adapt the classical completeness argument for \FODL~\cite{Harel79,Harel00} to the labelled setting with trace semantics, providing the necessary technical modifications. 

\subsection{Comparisons of Labelled Process Logic to Other Program Logics}
\label{section:Comparisons to Other Program Logics}


In addition to the trace-based dynamic logics discussed in Section~\ref{section:Motivations}, several recent program logics address temporal reasoning about programs. 
These logical theories are accompanied by mechanized verification tools. 
Below, we compare them further in terms of theoretical aspects such as expressiveness, program structures, and proof techniques. This comparison shows that the labelled process logic that we propose is competitive with these works and is worth further development and mechanization, but it also has limitations. 


\cite{Gadi20}~proposed a Hoare-style logic for reasoning about temporal properties for pointer while programs. It can be seen as a separation logic extended by CTL~\cite{ClarkeEmerson81} formulas.  
These CTL formulas can be captured equivalently in labelled process logic by combining its modality $[\cdot]$ and LTL formulas. 
E.g., a specification ``$\varphi\vdash \alpha : \all(\phi \unt \psi)$'' in~\cite{Gadi20} just means $X : \varphi\Rightarrow X : [\alpha]\phi\unt \psi$ with $X\in \TVar$ in labelled process logic,  
where the modality $[\alpha]$ captures the meaning of the for-all-paths symbol $\all$ in CTL formulas. The same applies to other CTL operators such as $\all \nt$ and $\ext \unt$. 
A cyclic proof system was built that adopted compositional rules similar to those in~\cite{Beckert13} but specialized to CTL formulas in order to decompose while programs. 
Compared to~\cite{Gadi20}, labelled process logic cannot handle pointers; it is a more abstract logical framework. 
As discussed in Section~\ref{section:Introduction}, we take a different approach that avoids directly decomposing sequential programs by introducing trace labels. 
We also handle dynamic formulas more complex than Hoare triples and program models more general than while programs.

Temporal Interaction and Choice Logic (Ticl)~\cite{Ticl25} takes a monadic tree structure as its program model. 
Temporal properties in Ticl are expressed by the satisfaction of CTL formulas by monad trees. 
As above, labelled process logic can capture the meanings of all CTL formulas in Ticl using its $[\cdot]$ modality and LTL formulas. 
The regular programs adopted in labelled process logic are closer to the structures of practical programs and powerful enough to express the finite-trace behaviours of standard practical programs. 
By contrast, Ticl's monadic trees can express infinite-trace behaviours and program features not expressible by regular programs, including ``return'', ``visibility'', ``silent events'', and ``bind''. 
Yet the completeness of Ticl's proof system is not discussed in~\cite{Ticl25}. 
 When dealing with loop structures, our work takes a cyclic-proof approach, whereas Ticl uses a well-founded proof system and the coinductive semantics of CTL formulas and monadic trees. 

The Hoare-style program logics LiLi~\cite{LiLi16} and TaDA Live~\cite{TaDALive21} both concentrate on temporal properties of concurrent while programs. 
LiLi offers a rely-guarantee-style proof system, in which a particular set of CTL properties that concern ``progress properties of blocking objects'' are handled in a modular way. 
TaDA Live is a concurrent separation logic for compositionally reasoning about the termination of blocking concurrent while programs. 
To this end, certain safety and liveness properties related to termination are discussed and proved.
Compared to the two works, labelled process logic cannot specify concurrent models with special features. 
However, it is a more general logical framework that can handle a wide range of temporal properties, not just those specific to certain program domains. 

Other recent works, such as~\cite{Trillium24} and~\cite{Lawyer26}, also address the use of program logics to verify temporal properties of program traces. However, their approaches are based mainly on program refinement rather than direct reasoning in a program logic, as in our work. 

\textbf{Limitation of Labelled Process Logic.}
Compared to these works, labelled process logic has the main limitation that it can handle only finite traces because it adopts regular programs as its program model. 
The same limitation applies to the original process logic proposed in~\cite{Harel82}. 
As stated in~\cite{Harel00}, in order to cover infinite traces, we need to introduce the omega loop $\alpha^\omega$, and build suitable rules for that operator.  
More work can be done in the future to address this extension. 

\ifx
Compared to these work, like the work of process logic in~\cite{Harel82}, the main limitation of our work is that in labelled process logic we only manage to handle finite traces, because the regular programs only generate finite traces. To also cover the infinite traces we need to introduce the omega loop $\alpha^\omega$ as stated in~\cite{Harel00}. 
More work can be done in future to address this limitation. 
\fi

\ifx
Temporal Interaction and Choice Logic (Ticl)~\cite{Ticl25} is a program logic extended from CTL~\cite{???}!++@ref++! designed to support modular, mechanized proofs of general temporal properties. 
Ticl takes a monad tree structure as its program models. 
Temporal properties are expressed by that tree structures satisfy CTL formulas. 
Labelled process logic is similar to Ticl in that it also captures CTL properties starting from a world of a tree-like structure of program behaviours by combining the expressiveness of the modality $[\cdot]$ and the temporal operators: E.g., formula $\sigma : [\alpha]\phi$ in Ticl captures that a model $\alpha$ (expressed as a tree structure in Ticl) satisfies the CTL formula $\textit{AG}\ \phi$ starting from $\sigma$. 
However, their logic forms and proof approaches adopted are quite different. 
Labelled process logic adopts a cyclic proof approach to deal with loop programs, while Ticl uses induction and coinduction definitions. adopts a relatively simpler and less-expressive regular-program models than monad trees, but powerful enough to express the usual practical program traces.  
Yet the completeness of Tcil's proof system is not discussed in~\cite{Ticl25}.

\MO
Structural temporal logic (Ticl)~\cite{Ticl25} and the cyclic temporal logic of Tellez and Brotherston~\cite{Gadi20} both organize temporal reasoning around proof principles that account for potentially unbounded executions.
Ticl derives safety and liveness properties through structural rules based on invariants and variants, while the latter logic combines CTL formulas with separation-logic assertions and validates cyclic derivations through a global descent condition.
Our approach shares the use of cyclic reasoning with the latter and the aim of treating temporal properties uniformly with the former, but differs in its logical target: labels explicitly represent traces or program updates, formulas belong to process logic, and the calculi are proved complete at both the propositional and first-order levels.

Trillium~\cite{Trillium24} and Lawyer~\cite{Lawyer26} relate concurrent executions to auxiliary transition systems.
Trillium establishes intensional refinement between program traces and an abstract model, so that temporal properties of the model can be transferred to the program under suitable fairness assumptions.
Lawyer instead represents liveness obligations by a logical labelled transition system and uses refinement to derive fair termination.
Both approaches therefore obtain temporal guarantees through a relation between a program and an auxiliary semantic structure.
In our systems, temporal properties are stated directly in the trace-formula language of process logic, and labels expose the execution information required by the inference rules without introducing a separate abstract transition system.

TaDA Live~\cite{TaDALive21} and LiLi~\cite{LiLi16} concentrate on progress properties of concurrent programs.
TaDA Live reduces compositional termination reasoning to environment liveness assumptions and layered obligations, whereas LiLi combines rely--guarantee reasoning with progress-aware abstraction to establish linearizability, starvation-freedom, and deadlock-freedom under fair scheduling.
These logics provide finer accounts of interference, blocking, fairness, and concurrent abstraction than our sequential regular-program setting, but their temporal guarantees concern designated progress properties.
By contrast, our labelled calculi give a uniform proof theory for the full trace-formula language of process logic and establish soundness and completeness for the resulting propositional and first-order systems.

In summary, the related logics obtain temporal reasoning through structural induction and coinduction, cyclic descent, refinement to auxiliary models, or obligation-based accounts of concurrency.
Our work is complementary to these approaches: it does not address concurrency or fairness, but develops sound and complete labelled cyclic calculi in which general temporal properties of regular-program traces are expressed and derived within process logic itself.
\EMO

\textbf{Structural temporal logic for mechanized program verification.}
The structural temporal logic Ticl~\cite{Ticl25} supports modular, mechanized proofs of general temporal properties for effectful, nondeterministic, and potentially nonterminating programs.
Rather than requiring users to reason directly about traces or automata, Ticl internalizes the necessary inductive and coinductive arguments into structural proof principles based on invariants and variants; its case studies include scheduling, shared memory, and distributed consensus.
This approach is broader than ours with respect to program constructs and mechanization.
In contrast, our work develops an object-level proof theory for process logic: trace information is represented explicitly by labels, iteration is handled by cyclic derivations, and soundness and completeness are established for both the propositional and first-order systems.

\textbf{Cyclic temporal verification of pointer programs.}
Tellez and Brotherston~\cite{Gadi20} give a cyclic proof system whose judgements combine CTL properties with separation-logic assertions while symbolically executing nondeterministic pointer programs.
The system supports fairness constraints and automated proof search in the Cyclist prover, and its global cyclic condition, like ours, validates finite proof graphs through infinite descent.
Their calculus is specialized to heap-manipulating programs and establishes a sound verification method for branching-time properties.
Our calculi instead target the trace-formula language of process logic over regular programs, use trace- and update-based labels rather than heap assertions, and provide completeness results at both the propositional and first-order levels.

\textbf{Trillium.}
Trillium~\cite{Trillium24} is a language-independent separation-logic framework for proving intensional refinement relations between program traces and an abstract model.
Its Fairis instance transfers liveness properties from a model to higher-order concurrent programs under fair scheduling, while its distributed instance connects program executions to TLA+ specifications.
This refinement-based architecture supports concurrency, distribution, and mechanized reasoning beyond the scope of the present paper.
Our framework does not require a separate abstract model: temporal claims are stated directly as process-logic formulas and proved in labelled sequents, with the cyclic condition controlling repeated program unfolding and the completeness theorems characterizing the resulting logical systems.

\textbf{Lawyer.}
Lawyer~\cite{Lawyer26} is a recent mechanized concurrent separation logic for modular proofs of fair termination in higher-order, impredicative programs.
It records liveness obligations in a purely logical labelled transition system and proves that the program refines this system, avoiding instrumentation of the source program; its Rocq development supports separate verification of libraries and clients whose termination depends on fair concurrent services.
Compared with Lawyer, our labels record execution traces or program updates rather than modular obligations, and our formulas express temporal trace properties beyond termination.
On the other hand, our present systems do not address higher-order state, concurrent module abstraction, fair scheduling, or mechanized verification.

\textbf{TaDA Live.}
TaDA Live~\cite{TaDALive21} is a concurrent separation logic for compositional termination proofs of blocking, fine-grained concurrent programs.
\VE Its abstract specifications expose blocking behavior as liveness assumptions on the environment, and layered subjective obligations ensure that clients and implementations can be verified without breaking abstraction boundaries. \EVE
\MO Its abstract specifications expose blocking behaviour as liveness assumptions on the environment, and layered subjective obligations ensure that clients and implementations can be verified without breaking abstraction boundaries. \EMO
TaDA Live therefore provides substantially finer support for interference, blocking, and modular concurrency than our sequential regular-program setting.
Conversely, it concentrates on termination, whereas our labelled calculi provide a uniform proof theory for the full trace-formula language of process logic and establish completeness rather than only a sound, domain-specific verification discipline.

\textbf{Concurrent objects under fair scheduling.}
LiLi~\cite{LiLi16} is a rely--guarantee-style program logic that verifies linearizability together with starvation-freedom and deadlock-freedom for concurrent objects under fair scheduling.
It supports thread-modular proofs of blocking synchronization algorithms and provides progress-aware abstraction for reasoning about the safety and liveness of client code.
These scheduling-sensitive guarantees and object abstractions are not treated by our systems.
Our contribution is complementary: it offers a general temporal specification language for complete regular-program traces, together with labelled cyclic proof systems whose soundness and completeness cover propositional and first-order process logic rather than a fixed collection of concurrent-object progress conditions.
\fi

\ifx
\textbf{Dynamic logic with trace modalities.}
The trace modality of \DLT~\cite{Beckert01} internalizes a safety property of a deterministic while-program execution: $[|\alpha|]\phi$ requires the state formula $\phi$ at every intermediate state of the trace.
Its relatively complete first-order calculus is an important predecessor because it demonstrates that trace information can be incorporated into dynamic-logic reasoning without first constructing a separate transition system.
However, its program model is deterministic and its trace language is restricted to an invariant-like modality.
Our framework instead treats nondeterministic regular programs and the full trace-formula language of process logic; labels accumulate execution information, and the cyclic proof condition handles iteration without requiring a separate compositional rule for each temporal operator.

\textbf{Dynamic trace logic.}
Dynamic trace logic (DTL)~\cite{Beckert13} is the closest prior system with respect to the specification language: it gives dynamic-logic modalities a trace semantics and allows arbitrary LTL formulas to describe executions.
Its sequent calculus is sound and relatively complete for deterministic while programs, where a fixed initial state determines a single execution trace.
This determinism permits temporal reasoning to be decomposed along the program structure.
By contrast, process logic quantifies over the potentially branching traces of regular programs, including nondeterministic choice and iteration; our labelled systems make individual executions explicit in sequents and establish completeness at both the propositional and first-order levels in this nondeterministic setting.

\textbf{Differential temporal dynamic logics.}
Differential temporal dynamic logic and its extension $\mathrm{dTL}^{2}$~\cite{PlatzersTDL,DTL2} combine differential dynamic logic with temporal specifications over hybrid traces, and a later program logic supports Signal Temporal Logic specifications of hybrid systems~\cite{Ahmad21}.
These systems are particularly suited to programs containing continuous evolution and to real-time properties of hybrid-system trajectories.
Their proof rules are consequently specialized to differential equations and selected temporal forms.
The present work does not address continuous dynamics or metric time; it instead provides a uniform proof-theoretic treatment of the temporal connectives of process logic over discrete regular-program traces, with completeness obtained through labelled cyclic derivations rather than reductions to differential invariants or hybrid reachability obligations.

\textbf{Dynamic logic with branching modalities.}
The logic $\mathrm{DL}_{b}$~\cite{Zhang24} extends DLT from deterministic while programs to regular programs with tests by replacing the trace modality with a branching modality that also accounts for executions blocked by tests.
Its propositional and first-order calculi are sound and (relatively) complete, and therefore directly address the failure of the DLT compositional rule under nondeterminism.
The price of this compositional treatment is that $\mathrm{DL}_{b}$ targets safety properties expressible by its branching modality rather than arbitrary temporal trace formulas.
Our systems retain the linear traces generated by each nondeterministic execution and support the full process-logic trace language, using labels and cyclic proofs instead of changing the semantic object from traces to branching structures.

\textbf{Cyclic temporal verification of pointer programs.}
Tellez and Brotherston~\cite{Gadi20} give a cyclic proof system whose judgements combine temporal properties with separation-logic assertions while symbolically executing nondeterministic pointer programs.
Their system supports CTL properties, fairness constraints, and automated proof search in the Cyclist prover; its cyclic soundness condition, like ours, justifies finite proof graphs by infinite descent.
The two frameworks nevertheless have different aims: their calculus is specialized to heap-manipulating programs and establishes soundness for temporal verification, whereas our calculi are proof systems for process logic itself, use trace or update labels rather than heap assertions, and establish both soundness and completeness for the propositional and first-order logical systems.

\textbf{Operational propositional dynamic logic.}
\VE Operational propositional dynamic logic (OPDL)~\cite{Acclavio2024-uv} separates programs from their traces and parameterizes PDL by an operational semantics, thereby accommodating program models such as CCS and choreographic programs whose concurrent behavior is not naturally represented by regular expressions. \EVE
\MO Operational propositional dynamic logic (OPDL)~\cite{Acclavio2024-uv} separates programs from their traces and parameterizes PDL by an operational semantics, thereby accommodating program models such as CCS and choreographic programs whose concurrent behaviour is not naturally represented by regular expressions. \EMO
Its non-wellfounded sequent calculus provides adequacy, completeness, and cut elimination for propositional state properties over such operationally defined programs.
Our framework is less general with respect to program syntax and does not presently treat concurrency, but it is richer with respect to specifications: formulas can state temporal properties of complete traces, and the first-order system reasons about concrete state changes through update-based labels.

\textbf{Parameterized dynamic logic.}
Parameterized dynamic logic $\mathrm{DL}_{p}$~\cite{zhang2025parameterizeddynamiclogic} is methodologically close to the present work because it uses program configurations as labels, symbolically executes programs supplied by a domain-specific transition system, and controls loop-induced infinite derivations by a cyclic preproof condition.
Its main objective is a language-independent dynamic-logic framework for arbitrary program and formula forms, for which soundness is established.
The present work specializes the labelled approach to the trace semantics of process logic: this narrower setting permits dedicated trace connectives, distinct trace- and update-based label structures for \PPL\ and \FOPL, and completeness theorems in addition to soundness.
\fi

\section{Labelled Propositional Process Logic}
\label{section:Labelled Propositional Process Logic}
This section develops the labelled propositional process logic and its proof system.  We first define \emph{labels} and \emph{labelled formulas}, then present the labelled sequent rules and the cyclic proof condition required for regular-program iteration.

\subsection{Labelled \PPL\ Formulas}
\label{section:Labelled PPL Formulas}

\begin{definition}[Labels of \PPL]
\label{def:Labels and Label Mappings}
A label $\sigma$ is a trace defined as follows:
$$
\sigma \dddef \varepsilon\ |\ x\ |\ X\ |\ \sigma\cat \sigma, 
$$
where $x$ is a label variable for worlds; $X$ is a label variable for traces of worlds;
$\cat$ is a concatenation operator between label variables; 
$\varepsilon$ is called the ``empty label'', which satisfies $\sigma \cat \varepsilon = \varepsilon\cat \sigma = \sigma$ for any label $\sigma$. 

\end{definition}

Use $\Conf$ to denote the set of all labels. 
Use $\WVar$ to express the set of all label variables for worlds; use $\TVar$ to express the set of all label variables for traces. 
Define $\LVar\dddef \WVar\cup \TVar$ to be the set of all label variables. 

The \emph{substitution} on labels is defined in a natural way: for a label $\sigma$, $\sigma[\sigma'/U]$ returns the label obtained by replacing each free occurrence of a variable $U\in \LVar$ with the label $\sigma'$. 
Since no quantifiers for label variables exist, all occurrences of label variables are free.

\begin{definition}[Label Mappings of \PPL]
Given a Kripke frame $(\Wd, \I)$, 
a world mapping $\wm$ in \PPL\ is a function from $\WVar$ to $\Wd$. 

A label mapping $\lm$ is a function from $\Conf$ to the set $\Wd^*$ of world traces, defined as follows:
\begin{enumerate}[i]
\item $\lm$ maps a variable $X\in\TVar$ to a trace of worlds in $\Wd^*$.
\item $\lm(\varepsilon)\dddef \epsilon$.
\item $\lm(x)\dddef \wm(x)$.
\item $\lm(\sigma_1\cat \sigma_2)\dddef \lm(\sigma_1)\cdot\lm(\sigma_2)$.
\end{enumerate} 
\end{definition}

Denote the set of all label mappings (w.r.t. $(\Wd, \I)$) by $\LM(\Wd, \I)$ (simply $\LM$). 

\begin{definition}[Label Relations]
Given an action $a$ and label variables $x, y\in \LVar$, $x R_a y$ is called a ``label relation''.
\end{definition}

\begin{definition}[Labelled \PPL\ Formulas]
\label{def:Labelled PPL Formulas}
A ``labelled \PPL\ formula'' is of the form $$\sigma : \phi,$$ where $\sigma\in \Conf$ and $\phi$ is a \PPL\ formula. 
\end{definition}

We sometimes use the symbol $\tau$ to represent a labelled formula. 

\begin{definition}[Semantics of Labelled \PPL\ Formulas \& Label Relations]
    \label{def:semantics of labelled PPL formulas and label relations}
Given a Kripke frame $(\Wd, \I)$ and a set $\LM$ of label mappings w.r.t. $(\Wd, \I)$, the satisfaction relation $\lm\models_{\LM} \sigma:\phi$ of a labelled formula $\sigma:\phi$ by a label mapping $\lm$ (w.r.t. $\LM$) is defined by 
   $$
   \lm\models_{\LM} \sigma:\phi, \mbox{ if }\lm(\sigma) \models_{(\Wd, \I)} \phi. 
   $$
The satisfaction relation $\lm\models_{\LM}xR_a y$ of a label relation $xR_a y$ by a label mapping $\lm$ (w.r.t. $\LM$) is defined by 
   $$
    \lm\models_{\LM}xR_a y, \mbox{ if }\lm(x)\lm(y)\in \I(a). 
   $$
\end{definition}
We write $\models_{\LM}$ simply as $\models$ when $\LM$ and the frame $(\Wd, \I)$ are clear from context. 

\textbf{Relationship to \PPL\ Formulas}. 
A valid \PPL\ formula $\phi$ can be captured by its equivalent labelled form $X : \phi$, with $X\in \TVar$.
Indeed, $\phi$ is valid if and only if every trace satisfies it: $\phi$ is valid iff for every trace $tr$, $tr \models \phi$, iff for every label mapping $\lm$, $\lm(X) \models \phi$, iff for every label mapping $\lm$, $\lm \models X : \phi$.  Thus the validity of $X : \phi$ in labelled process logic implies the validity of $\phi$ in process logic, and vice versa. 

\subsection{A Proof System for Labelled \PPL}
\label{section:A Proof System for Labelled PPL}
A \emph{sequent} is a logical argumentation of the form: 
$$\Gamma\Rightarrow \Delta,$$ where $\Gamma$ and $\Delta$ are finite sets of formulas, called the \emph{left side} and the \emph{right side} of the sequent, respectively. 
We use a dot $\cdot$ to represent $\Gamma$ or $\Delta$ when it is the empty set. 
Intuitively, a sequent $\Gamma\Rightarrow \Delta$ means that if all formulas in $\Gamma$ hold, then one of the formulas in $\Delta$ holds. 
We use $\nu$ to represent a sequent. 
An \emph{inference rule} is of the form
$$
\begin{aligned}
\infer[]
{
    \nu
}
{
    \nu_1
    &
    \ldots 
    &
    \nu_n
}
,\end{aligned}$$
with the premises $\nu_1,\ldots,\nu_n$ and the conclusion $\nu$. 
The semantics of the rule is that the validity of sequents $\nu_1, \ldots, \nu_n$ implies the validity of sequent $\nu$. 
We call a formula pair $(\tau_1, \tau_2)$ a \emph{conclusion-premise} (CP) pair if $\tau_1$ and $\tau_2$ belong to sequents $\nu, \nu_i$ for some $1\le i\le n$ respectively. 
In this paper, we use a double-lined inference form:
$$
\begin{aligned}
   \infer=[]
{\phi}
{\phi_1 & \ldots  & \phi_n} 
\end{aligned}
$$
as a shorthand to represent both rules
\begin{center}
$
\begin{aligned}
    \infer[]
{
    \Gamma\Rightarrow \phi, \Delta
}
{
    \Gamma\Rightarrow \phi_1, \Delta
    &
    \ldots  
    &
    \Gamma\Rightarrow \phi_n, \Delta
}
\mbox{ and }
\infer[]
{
    \Gamma, \phi\Rightarrow \Delta
}
{
    \Gamma, \phi_1\Rightarrow  \Delta
    &
    \ldots
    &
    \Gamma, \phi_n\Rightarrow \Delta
}, 
\end{aligned}
$
\end{center}
for arbitrary contexts $\Gamma$ and $\Delta$. 

\ifx
!++@check @mini @keep
Given a proof system, i.e., a set of inference rules, a \emph{proof} of a sequent $\nu$ is a finite tree constructed by applying these rules backwardly (from the conclusion to the premises). 
In this tree, each node is a sequent. The root node is the sequent $\nu$ itself (the conclusion). 
From each node $N$ to its child nodes $N_1,...,N_m$ $(m\ge 0)$ is an instance of a rule:
$$
\infer[]
{N}
{N_1 & ... & N_m}.
$$
Each leaf node is \emph{terminal} in the sense that it is the conclusion of an axiom. 
We say that a sequent can be proved if there is a proof for it. 
++!
\fi
Given a proof system, i.e., a set of inference rules, a \emph{proof} of a sequent $\nu$ is a finite tree constructed by applying these rules backwardly (from conclusions to premises).
In this tree, each node is a sequent, and the root is $\nu$.
If a node is $N$ and its children are $N_1,\ldots,N_m$ $(m\geq 0)$, then
$$
\infer[]
{N}
{N_1 & \ldots & N_m}
$$
is an instance of an inference rule.
Each leaf node is \emph{terminal}, in the sense that it is the conclusion of an axiom.
We say that a sequent is provable if it has a proof.

\ifx
Given a proof system, i.e., a set of inference rules, a \emph{proof} of a sequent is a finite tree formed by deriving the sequent by applying these rules backwardly (i.e., from the conclusions to the premises). 
In this tree, each node is a sequent; the root node is the sequent itself (the conclusion); each leaf node is \emph{terminal} in the sense that it is the conclusion of an axiom. 
We say that a sequent can be proved if there is a proof for it. 
\fi

The proof system of labelled \PPL\ is displayed in Tables~\ref{table:giiippl-regular-rules}, \ref{table:giiippl-trace-rules}, and~\ref{table:giiippl-propositional-rules} and consists of three parts: the rules for regular programs (Table~\ref{table:giiippl-regular-rules}), which are extended from the G3 sequent calculus for modal logics~\cite{Negri05}; the rules for trace formulas (Table~\ref{table:giiippl-trace-rules}), which are specific to process logic; and the rules for the labelled versions of propositional formulas (Table~\ref{table:giiippl-propositional-rules}), which are obtained by augmenting the classical sequent calculus for propositional logic. 
We call the resulting proof system $\GiiiPPL$. 

\begin{table}[t]
\centering
\small
\begin{tabular}{p{0.18\linewidth}p{0.74\linewidth}}
\toprule
Rule & Schematic form \\
\midrule
$([a]R)$ & $\begin{aligned}\infer[]{\Gamma\Rightarrow \sigma\cat x:[a]\phi,\Delta}{\Gamma,x\R_a y\Rightarrow \sigma\cat x\cat y:\phi,\Delta}\end{aligned}$, where $y\in\WVar$ is fresh\\[1.2em]
$([a]L)$ & $\begin{aligned}\infer[]{\Gamma,\sigma\cat x:[a]\phi,x\R_a y\Rightarrow\Delta}{\Gamma,\sigma\cat x\cat y:\phi\Rightarrow\Delta}\end{aligned}$, for some $y\in\WVar$\\[1.2em]
$([\phi?]R)$ & $\begin{aligned}\infer=[]{\Gamma\Rightarrow\sigma\cat x:[\psi?]\phi,\Delta}{\Gamma,x\cat x:\psi\Rightarrow\sigma\cat x\cat x:\phi,\Delta}\end{aligned}$\\[1.2em]
$([\phi?]L)$ & $\begin{aligned}\infer=[]{\Gamma,\sigma\cat x:[\psi?]\phi,x\cat x:\psi\Rightarrow\Delta}{\Gamma,\sigma\cat x\cat x:\phi\Rightarrow\Delta}\end{aligned}$\\[1.2em]
$([\seq])$ & $\begin{aligned}\infer=[]{\sigma:[\alpha\seq\beta]\phi}{\sigma:[\alpha][\beta]\phi}\end{aligned}$\\[1.2em]
$([\cho])$ & $\begin{aligned}\infer=[]{\sigma:[\alpha\cho\beta]\phi}{\sigma:[\alpha]\phi & \sigma:[\beta]\phi}\end{aligned}$\\[1.2em]
$([\lup])$ & $\begin{aligned}\infer=[]{\sigma:[\alpha^\lup]\phi}{\sigma:[\true?]\phi & \sigma:[\alpha][\alpha^\lup]\phi}\end{aligned}$\\[1.2em]
$(\Sub)$ & $\begin{aligned}\infer[]{\Gamma[\sigma/U]\Rightarrow\Delta[\sigma/U]}{\Gamma\Rightarrow\Delta}\end{aligned}$, where $U\in\LVar$\\
\bottomrule
\end{tabular}
\caption{Rules of \GiiiPPL\ for regular programs and label substitution.}
\label{table:giiippl-regular-rules}
\end{table}

\begin{table}[t]
\centering
\small
\begin{tabular}{p{0.18\linewidth}p{0.74\linewidth}}
\toprule
Rule & Schematic form \\
\midrule
$(p)$ & $\begin{aligned}\infer=[]{x\cat\sigma:p}{x:p}\end{aligned}$\\[1.2em]
$(\ofst)$ & $\begin{aligned}\infer=[]{x\cat\sigma:\ofst\phi}{x:\phi}\end{aligned}$\\[1.2em]
$(\osuf R)$ & $\begin{aligned}\infer[]{\Gamma\Rightarrow x\cat\sigma:\phi\osuf\psi, \Delta}{\Gamma\Rightarrow \sigma:\psi, \sigma : \phi, \Delta & \Gamma\Rightarrow \sigma : \psi, \sigma : \phi\osuf \psi, \Delta}\end{aligned}$\\[1.2em]
$(\osuf L)$ & $\begin{aligned}\infer[]{\Gamma, x\cat\sigma:\phi\osuf\psi\Rightarrow \Delta}{\Gamma, \sigma:\psi\Rightarrow \Delta & \Gamma, \sigma : \phi, \sigma : \phi\osuf\psi\Rightarrow \Delta}\end{aligned}$\\[1.2em]
$(\osuf\ x)$ & $\begin{aligned}\infer[]{\Gamma,x:\phi\osuf\psi\Rightarrow\Delta}{}\end{aligned}$\\
\bottomrule
\end{tabular}
\caption{Rules of \GiiiPPL\ for trace formulas.}
\label{table:giiippl-trace-rules}
\end{table}

In a right rule of the form
$$
\begin{aligned}
\infer[]
{\Gamma\Rightarrow \phi, \Delta}
{\Gamma_1\Rightarrow \phi_1, \Delta & \ldots & \Gamma_n\Rightarrow \phi_n, \Delta}
\end{aligned}, 
$$
we call $\phi$ the ``target formula'' of the rule and call each $\phi_i$ a ``replacement'' of $\phi$.
For example, in Table~\ref{table:giiippl-regular-rules}, $\sigma\cat x : [a]\phi$ is the target of the rule $([a]R)$, and $\sigma\cat x\cat y:\phi$ is the replacement of $\sigma\cat x : [a]\phi$. 

\ifx
In the displayed rules of \GiiiPPL, 
a \emph{target pair} is a CP pair that are explicitly transformed in the rule, apart from the unchanged contexts $\Gamma$ and $\Delta$. 
E.g., $(\sigma\cat x : [a]\phi, \sigma\cat x\cat y:\phi)$ is the target pair of rule $([a]R)$, where $\sigma\cat x\cat y:\phi$ is transformed from $\sigma\cat x : [a]\phi$. 
While $(\sigma\cat x : [a]\phi, x\R_a y)$ is not a target pair, since $x\R_a y$ is new generated, not transformed from $\sigma\cat x : [a]\phi$. 
The structural rules $(\Sub)$, $(ax)$, $(cut)$, $(wk)$ and $(ctr)$ do not have target pairs.
\fi

\ifx
In the displayed rules of \GiiiPPL, a \emph{target pair} is the CP pair explicitly transformed in the rule, apart from the unchanged contexts $\Gamma$ and $\Delta$. E.g., $(\sigma\cat x : [a]\phi, \sigma\cat x\cat y:\phi)$ is a target pair of an instance of rule $([a]R)$, where $\sigma\cat x\cat y:\phi$ is transformed from $\sigma\cat x : [a]\phi$. 
\fi

\begin{table}[t]
\centering
\small
\begin{tabular}{p{0.18\linewidth}p{0.74\linewidth}}
\toprule
Rule & Schematic form \\
\midrule
$(\textit{ax})$ & $\begin{aligned}\infer[]{\Gamma,\sigma:\phi\Rightarrow\sigma:\phi,\Delta}{}\end{aligned}$\\[1.2em]
$(\textit{cut})$ &
         $
         \begin{aligned}
         \infer[^{(\textit{Cut})}]
         {\Gamma, \Sigma\Rightarrow \Delta, \Lambda}
         {\Gamma\Rightarrow \sigma : \phi, \Delta & 
         \Sigma, \sigma : \phi\Rightarrow \Lambda}
         \end{aligned}
         $\\[1.2em]
$(\textit{wk})$ & $\begin{aligned}\infer=[]{\sigma:\phi}{}\end{aligned}$\\[1.2em]
$(\neg R)$ & $\begin{aligned}\infer[]{\Gamma\Rightarrow\sigma:\neg\phi,\Delta}{\Gamma,\sigma:\phi\Rightarrow\Delta}\end{aligned}$\\[1.2em]
$(\neg L)$ & $\begin{aligned}\infer[]{\Gamma,\sigma:\neg\phi\Rightarrow\Delta}{\Gamma\Rightarrow\sigma:\phi,\Delta}\end{aligned}$\\[1.2em]
$(\wedge R)$ & $\begin{aligned}\infer[]{\Gamma\Rightarrow\sigma:\phi\wedge\psi,\Delta}{\Gamma\Rightarrow\sigma:\phi,\Delta & \Gamma\Rightarrow\sigma:\psi,\Delta}\end{aligned}$\\[1.2em]
$(\wedge L)$ & $\begin{aligned}\infer[]{\Gamma,\sigma:\phi\wedge\psi\Rightarrow\Delta}{\Gamma,\sigma:\phi,\sigma:\psi\Rightarrow\Delta}\end{aligned}$\\
\bottomrule
\end{tabular}
\caption{Primitive propositional and structural rules of \GiiiPPL.}
\label{table:giiippl-propositional-rules}
\end{table}

The rules for regular programs in Table~\ref{table:giiippl-regular-rules} follow the standard labelled-sequent treatment of modal operators~\cite{Negri05}, adapted to the trace semantics of \PPL.  Rules $([a]R)$ and $([a]L)$ handle the box modality over atomic programs: to prove $\sigma\cat x : [a]\phi$, rule $([a]R)$ introduces a fresh world variable $y$ representing the successor of $x$ under action $a$, reducing the goal to $\sigma\cat x\cat y : \phi$ with the relational atom $xR_a y$ in the antecedent; rule $([a]L)$ uses an existing relational atom to consume the modality from the left.  Rules $([\phi?]R)$ and $([\phi?]L)$ handle tests: a test $\psi?$ produces a stuttering trace, so the rule appends $x$ to itself and requires the test condition to be satisfied.  Rules $([\seq])$, $([\cho])$, and $([\lup])$ unfold sequential composition, choice, and iteration, respectively, matching the standard PDL axioms (Table~\ref{table:ppl-proof-system}).  Rule $(\Sub)$ substitutes a label variable $U$ with a concrete label $\sigma$ throughout a sequent; it is the key mechanism for constructing back-links in cyclic proofs.

The rules for trace formulas in Table~\ref{table:giiippl-trace-rules} capture the semantics of the temporal operators in process logic.  Rule $(p)$ states that an atomic proposition $p$ evaluated on a trace $x\cat\sigma$ depends only on the first world $x$: if $x : p$ holds, then so does $x\cat\sigma : p$ for any $\sigma$.  Rule $(\ofst)$ similarly reduces the evaluation of $\ofst\phi$ on $x\cat\sigma$ to the evaluation of $\phi$ at $x$ alone.  

\ifx
Rule $(\osuf)$ is the key rule for the suffix connective: on a trace $x\cat\sigma$, the formula $\phi\osuf\psi$ reduces to the disjunction $\psi \vee (\phi\wedge\phi\osuf\psi)$ on $\sigma$, reflecting that either the suffix $\sigma$ already satisfies $\psi$, or it satisfies $\phi$ and the same condition propagates along $\sigma$.  
\fi

Rules $(\osuf R)$ and $(\osuf L)$ are the key rules for the suffix connective, and both are read off the one-step unfolding
$$
x\cat\sigma : \phi\osuf\psi \mbox{ iff }\sigma : \psi\vee (\phi\wedge \phi\osuf\psi), 
$$
which states that a trace $x\cat\sigma$ satisfies $\phi\osuf\psi$ exactly when its suffix $\sigma$ either already satisfies $\psi$, or satisfies $\phi$ while the same condition propagates along $\sigma$. 
Intuitively, rule $(\osuf R)$ is what applying $(\vee R)$ and then $(\wedge R)$ to the right-hand side of this equivalence yields: to prove $x\cat\sigma : \phi\osuf\psi$ on the right of a sequent, it suffices to prove the two premises in which the target formula is replaced by $\sigma : \psi, \sigma : \phi$ and by $\sigma : \psi, \sigma : \phi\osuf\psi$ respectively. 
Dually, rule $(\osuf L)$ is what applying $(\vee L)$ and then $(\wedge L)$ yields: an assumption $x\cat\sigma : \phi\osuf\psi$ on the left splits into the case $\sigma : \psi$, in which the suffix already witnesses $\psi$, and the case $\sigma : \phi, \sigma : \phi\osuf\psi$, in which the witness is postponed to a shorter suffix. 
Note that in both rules the label of the formula $\phi\osuf\psi$ in the premise is strictly shorter than the label of the same formula in the conclusion; this is exactly what makes an application of $(\osuf R)$ or $(\ofus R)$ (Table~\ref{table:giiippl-derived-propositional-rules}) a progressing step in the cyclic condition given in Section~\ref{section:Cyclic Proof Structure of GiiiPPLcyc}. 

Finally, rule $(\osuf\ x)$ is an axiom that closes a branch when the trace reduces to a single world variable $x$: since $x$ has no proper suffix, $x : \phi\osuf\psi$ is vacuously false and the branch terminates. 

The rules for labelled propositional formulas (Table~\ref{table:giiippl-propositional-rules}) have the same meaning as their unlabelled counterparts in classical propositional logic. 

\ifx
Table~\ref{table:giiippl-derived-propositional-rules} gives the rules for non-primitive operators $\la\cdot\ra$, $\vee$ and $\to$. They can be directly derived from the rules in Table~\ref{table:giiippl-propositional-rules}. 
We list them for convenient sake in the rest of the paper, especially in the derivations of the proofs of the purposed theories. 
\fi
Table~\ref{table:giiippl-derived-propositional-rules} collects the derived rules for the non-primitive operators $\la\cdot\ra$, $\vee$, and $\to$; each of these rules follows directly from the primitive rules of Table~\ref{table:giiippl-propositional-rules} and is listed here for convenient reference in subsequent proofs. 
We also give the derived rules for the operator $\ofus$, which transform the negation of a formula $\phi\osuf \psi$ and can be derived from the rules for the operator $\osuf$ in Table~\ref{table:giiippl-propositional-rules}. 

Soundness of the rules in \GiiiPPL\ is stated in the following theorem. 

\begin{theorem}
    \label{theo:Soundness of Each Rule of GiiiPPL}
Each rule in Tables~\ref{table:giiippl-regular-rules},~\ref{table:giiippl-trace-rules}, and~\ref{table:giiippl-propositional-rules} for \GiiiPPL\ is sound.

\end{theorem}

The proof directly follows from the semantics of labelled \PPL\ formulas.  We omit it in this paper. 

\textbf{Note}: 
In the derivations that follow in this paper, we occasionally omit one or more intermediate steps to save space. For example, the annotation $(\osuf R, \osuf L)$ in Figure~\ref{figure:Derivation process for sequent (iii.a)} indicates that instances of the rules $(\osuf R)$ and $(\osuf L)$ are each applied one or more times. Sometimes, we list only the main rule applied. 

\ifx
!++
@word @mini
The rules for $[\cdot]$ (Table~\ref{table:giiippl-regular-rules}) can be applied for formulas with operator $\la\cdot\ra$ with the help of the rules such as $(\neg L)$ and $(\neg R)$ in Table~\ref{table:giiippl-propositional-rules} and~\ref{table:giiippl-derived-propositional-rules}. 
In the rest of this paper, especially in the derivations of the proofs of the purposed theories, we often omit these unimportant steps but only give the main rules for $[\cdot]$.


++!
\fi

\begin{table}[t]
\centering
\small
\begin{tabular}{p{0.18\linewidth}p{0.74\linewidth}}
\toprule
Rule & Schematic form \\
\midrule
$(\la a\ra R)$ & $\begin{aligned}\infer[]{\Gamma, x \R_a y\Rightarrow \sigma\cat x:\la a\ra\phi,\Delta}{\Gamma\Rightarrow \sigma\cat x\cat y:\phi,\Delta}\end{aligned}$, for some $y\in\WVar$\\[1.2em]
$(\la a\ra L)$ & $\begin{aligned}\infer[]{\Gamma,\sigma\cat x:\la a\ra\phi\Rightarrow\Delta}{\Gamma,\sigma\cat x\cat y:\phi, x\R_a y\Rightarrow\Delta}\end{aligned}$, where $y\in\WVar$ is fresh\\[1.2em]
$(\la\phi?\ra R)$ & $\begin{aligned}\infer=[]{\Gamma, x\cat x : \psi\Rightarrow\sigma\cat x:\la\psi?\ra\phi,\Delta}{\Gamma\Rightarrow\sigma\cat x\cat x:\phi,\Delta}\end{aligned}$\\[1.2em]
$(\la\phi?\ra L)$ & $\begin{aligned}\infer=[]{\Gamma,\sigma\cat x:\la\psi?\ra\phi\Rightarrow\Delta}{\Gamma,\sigma\cat x\cat x:\phi, x\cat x : \psi\Rightarrow\Delta}\end{aligned}$\\[1.2em]
$(\la\seq\ra)$ & $\begin{aligned}\infer=[]{\sigma:\la\alpha\seq\beta\ra\phi}{\sigma:\la\alpha\ra\la\beta\ra\phi}\end{aligned}$\\[1.2em]
$(\la\cho\ra R)$ & $\begin{aligned}\infer[]{\Gamma\Rightarrow \sigma:\la\alpha\cho\beta\ra\phi, \Delta}{\Gamma\Rightarrow \sigma:\la\alpha\ra\phi, \sigma:\la\beta\ra\phi, \Delta}\end{aligned}$\\[1.2em]
$(\la\cho\ra L)$ & $\begin{aligned}\infer[]{\Gamma, \sigma:\la\alpha\cho\beta\ra\phi\Rightarrow \Delta}{\Gamma, \sigma:\la\alpha\ra\phi\Rightarrow \Delta & \Gamma, \sigma:\la\beta\ra\phi\rightarrow\Delta}\end{aligned}$\\[1.2em]
$(\la\lup\ra R)$ & $\begin{aligned}\infer[]{\Gamma\Rightarrow \sigma:\la\alpha^\lup\ra\phi, \Delta}{\Gamma\Rightarrow \sigma:\la\true?\ra\phi , \sigma : \la\alpha\ra\la\alpha^\lup\ra\phi, \Delta}\end{aligned}$\\[1.2em]
$(\la\lup\ra L)$ & $\begin{aligned}\infer[]{\Gamma, \sigma:\la\alpha^\lup\ra\phi\Rightarrow \Delta}{\Gamma, \sigma:\la\true?\ra\phi\Rightarrow \Delta & \Gamma, \sigma : \la\alpha\ra\la\alpha^\lup\ra\phi\Rightarrow \Delta}\end{aligned}$\\[1.2em]
$(\vee R)$ & $\begin{aligned}\infer[]{\Gamma\Rightarrow\sigma:\phi\vee\psi,\Delta}{\Gamma\Rightarrow\sigma:\phi,\sigma:\psi,\Delta}\end{aligned}$\\[1.2em]
$(\vee L)$ & $\begin{aligned}\infer[]{\Gamma,\sigma:\phi\vee\psi\Rightarrow\Delta}{\Gamma,\sigma:\phi\Rightarrow\Delta & \Gamma,\sigma:\psi\Rightarrow\Delta}\end{aligned}$\\[1.2em]
$(\to R)$ & $\begin{aligned}\infer[]{\Gamma\Rightarrow\sigma:\phi\to\psi,\Delta}{\Gamma,\sigma:\phi\Rightarrow\sigma:\psi,\Delta}\end{aligned}$\\[1.2em]
$(\to L)$ & $\begin{aligned}\infer[]{\Gamma,\sigma:\phi\to\psi\Rightarrow\Delta}{\Gamma\Rightarrow\sigma:\phi,\Delta & \Gamma,\sigma:\psi\Rightarrow\Delta}\end{aligned}$\\[1.2em]
$(\ofus R)$ & $\begin{aligned}\infer[]{\Gamma\Rightarrow x\cat\sigma : \phi\ofus\psi, \Delta}{\Gamma\Rightarrow \sigma : \psi, \Delta & \Gamma\Rightarrow \sigma : \phi, \sigma : \phi\ofus \psi, \Delta}\end{aligned}$\\[1.2em]
$(\ofus L)$ & $\begin{aligned}\infer[]{\Gamma, x\cat\sigma : \phi\ofus\psi\Rightarrow \Delta}{\Gamma, \sigma : \psi, \sigma : \phi\Rightarrow \Delta & \Gamma, \sigma : \psi, \sigma : \phi\ofus \psi\Rightarrow \Delta}\end{aligned}$\\
\bottomrule
\end{tabular}
\caption{Derived rules used in \GiiiPPL.}
\label{table:giiippl-derived-propositional-rules}
\end{table}

\subsection{Cyclic Proof Structure of \GiiiPPLcyc}
\label{section:Cyclic Proof Structure of GiiiPPLcyc}

Because regular programs include iteration, the derivation of a formula in \GiiiPPL\ does not always terminate, as the same formula forms may appear repeatedly. 

\begin{example}
    \label{ex:invalid cyclic derivation}
In Figure~\ref{figure:example of ill cyclic derivation}, using the rules in \GiiiPPL\ yields the repeated root node $\cdot \Rightarrow x:[a^*]p$ in two branches. 
Continuing to use the same rules leads to an infinite derivation that does not terminate. 

\begin{figure}
$$
\infer[^{([*])}]
{\cdot\Rightarrow \tau_1 :: x : [a^*]p}
{
        \infer[^{(cut)}]
        {\cdot \Rightarrow \tau_{2} :: x:[\true?]p}
        {
            \cdot \Rightarrow \tau_{7} :: x : [a^*]p\mbox{ (bud) }
            &
            \infer[^{([*])}]
            {x : [a^*]p\Rightarrow x : [\true?]p}
            {
                \infer[^{(ax)}]
                {x : [\true?]p, x : [a][a^*]p\Rightarrow x : [\true?]p}
                {}
            }
        }
        &
            \infer[^{([a]R)}]
            {\cdot\Rightarrow \tau_3 :: x:[a][a^*]p}
            {
                \infer[^{(wk)}]
                {x R_a y\Rightarrow \tau_4 :: x\cat y : [a^*]p}
                {
                    \infer[^{(\Sub)}]
                    {\cdot\Rightarrow \tau_5 :: x\cat y:[a^*]p}
                    {\cdot\Rightarrow \tau_6 :: x : [a^*]p \mbox{ (bud) }}
                }
            }
}
$$
\caption{An example of invalid cyclic derivation}
\label{figure:example of ill cyclic derivation}
\end{figure}

\ifx
\begin{figure}
$$
\infer[^{([*])}]
{x : ([a^*]p\to \false)\Rightarrow x : [a^*]p}
{
        \infer[^{(wk)}]
        {x : ([a^*]p\to \false)\Rightarrow x : [a^*]p, x : [\true?]p}
        {
            \cdot\Rightarrow x : [a^*]p \mbox{ (bud) }
        }
        &
        \infer[^{()}]
        {x : ([a^*]p\to \false)\Rightarrow x : [a][a^*]p}
        {
            \cdot\Rightarrow x : [a^*]p \mbox{ (bud) }
        }
}
$$
\caption{An example of invalid cyclic derivation}
\label{figure:example of ill cyclic derivation}
\end{figure}
\fi
\end{example}

In the cyclic-proof approach, a \emph{preproof} is a finite proof tree in which there is a special type of nonterminal leaf node, called a \emph{bud} (cf.~\cite{Brotherston07}). 
A bud is identical to one of its ancestors in its syntactic form in the tree. 
Together, a bud and one of its identical ancestors are called a \emph{back-link}. 
A \emph{derivation path} in a proof tree is expressed as a sequence of nodes $\nu_1\nu_2\ldots\nu_m\ldots$ ($m\ge 1$) starting from the root node $\nu_1$, 
where each node pair $(\nu_i, \nu_{i+1})$ ($i\ge 1$) belongs to an instance of a rule. 
In a preproof, there always exists an infinite derivation path over one or more derivation branches which contain back-links.

A preproof may not be a valid proof. 
As in Example~\ref{ex:invalid cyclic derivation}, the whole derivation is a preproof, with the two leaf nodes back-linked to the root. However, for some choices of $a$, $p$, and the Kripke frame, the formula $x : [a^*]p$ concluded in the preproof may be false. For example, $x : [(\true?)^*]\false$ is false for every $x$. 
As also indicated in works such as~\cite{Brotherston07,Brotherston08}, this indicates that we need to specify a particular \emph{cyclic condition} for \GiiiPPL\ to ensure that a preproof has a sound conclusion.

Below we define the notion of derivation traces in \GiiiPPL, then introduce the progressing derivation traces for \GiiiPPL. Based on these concepts, we finally give the cyclic condition for \GiiiPPL\ in Definition~\ref{def:Cyclic Proofs of GiiiPPL}. 

\begin{definition}[Derivation Traces]
\label{def:Derivation Traces}
A ``derivation trace'' over a derivation path $\mu_1\mu_2\ldots\mu_k\nu_1\nu_2\ldots\nu_m\ldots$ ($k \ge 0$, $m\ge 1$) is a sequence $\tau_1\tau_2\ldots\tau_m\ldots$ of formulas, where each $\tau_i$ ($1\le i\le m$) occurs on the right side of node $\nu_i$ and has one of the forms $\sigma : [\alpha]\phi$, $\sigma : \phi\osuf\psi$, or $\sigma : \phi\ofus\psi$, for some $\sigma$, $\alpha$, $\phi$, and $\psi$. 
Each CP pair $(\tau_i, \tau_{i+1})$ ($i\ge 1$) in the derivation $(\nu_i, \nu_{i+1})$ satisfies the following condition, according to the rule instance containing $(\nu_i, \nu_{i+1})$ in $\GiiiPPLcyc$: 
    \begin{enumerate}
        \item If $(\nu_i, \nu_{i+1})$ belongs to  an instance of the rule $([\alpha]R)$, $([\phi?R])$, $([\seq])$, $([\cho])$, $([\lup])$, then $\tau_i$ is the target formula and $\tau_{i+1}$ is one of its replacements; 
        
        
        \item If $(\nu_i, \nu_{i+1})$ belongs to an instance of the rule $(\osuf R)$, then $\tau_i =  x\cat \sigma : \phi\osuf \psi$ and $\tau_{i+1}=\sigma : \phi\osuf\psi$; 
        
        \item If $(\nu_i, \nu_{i+1})$ belongs to an instance of the rule $(\ofus R)$, then $\tau_i =  x\cat \sigma : \phi\ofus \psi$ and $\tau_{i+1}=\sigma : \phi\ofus\psi$; 
        


        \item If $(\nu_i, \nu_{i+1})$ belongs to an instance of rule $(\Sub)$, 
        then 
        $\tau_i = \sigma[\sigma'/U] : [\alpha]\phi$ and $\tau_{i+1} = \sigma : [\alpha]\phi$ for some $\sigma, \sigma'\in \Conf$ and formula $[\alpha]\phi$;

        \item If $(\nu_i, \nu_{i+1})$ belongs to an instance of the other rules, 
        then $\tau_i = \tau_{i+1}$. 

        \ifx
        \item If $(\nu_i, \nu_{i+1})$ is an instance of rule $([a]R)$,  $([a]L)$, $([\phi?R])$, $([\phi?L])$, $(\neg R)$, $([\seq])$, $[\cho]$, $[\lup]$, $(p)$, $(\ofst)$, $(\osuf)$, $(\osuf\ x)$, $(\neg L)$, $(\wedge R)$ or $(\wedge L)$, 
        then either $(\tau_i, \tau_{i+1})$ is a target pair of the rule, or $\tau_i = \tau_{i+1}$;

        \item If $(\nu_i, \nu_{i+1})$ is an instance of rule $(\Sub)$, 
        then $\tau_i = \Sub(\sigma) : \phi$ and $\tau_{i+1} = \sigma : \phi$ for some $\sigma\in \Conf$ and formula $\phi$;
        
        \item If $(\nu_i, \nu_{i+1})$ is an instance of the other rules (that do not have target pairs), 
        then $\tau_i = \tau_{i+1}$. 
        \fi
    \end{enumerate}
    We call the CP pair $(\tau_i, \tau_{i+1})$ a ``trace pair'' over the derivation $(\nu_i, \nu_{i+1})$, and call $\tau_i$ a ``trace element''. 
\end{definition}

In Definition~\ref{def:Derivation Traces}, the derived rule $(\ofus R)$ needs to be considered because we allow a trace pair to appear only on the right side of a sequent. 

\begin{definition}[Progressing Derivation Traces of \GiiiPPLcyc]
\label{def:progressing Step/progressing Derivation Trace of GiiiPPLcyc}
In a preproof of system $\GiiiPPLcyc$, given a derivation trace $\tau_1\tau_2\ldots\tau_m\ldots$ over an infinite derivation path $\ldots\nu_1\nu_2\ldots\nu_m\ldots$ ($m\ge 1$) starting from $\tau_1$ in node $\nu_1$, 
a trace pair $(\tau_i, \tau_{i+1})$ ($1\le i\le m$) of the derivation $(\nu_i, \nu_{i+1})$ is called a ``progressing step'' if $(\nu_i, \nu_{i+1})$ belongs to an instance of one of the following rules: 
$([a]R)$, $([\phi?]R)$, $(\osuf R)$ and $(\ofus R)$. 

If a derivation trace has an infinite number of progressing steps, we say that the trace is ``progressing''. 
We say that a derivation path is ``progressing'' if it has a progressing derivation trace. 
\end{definition}

\begin{example}
\label{ex:another invalid derivation of the last example}
In Figure~\ref{figure:example of ill cyclic derivation}, we call the two derivation branches with buds the ``left cycle'' and the ``right cycle'', respectively. 
The derivation trace $\tau_1\tau_3\tau_4\tau_5\tau_6(\tau_1)\ldots$ over the infinite derivation path that covers only the right cycle is a progressing one because every pair $(\tau_3, \tau_4)$ appearing in the trace is a progressing step. 
\end{example}

\ifx
\begin{example}
    \label{ex:another invalid derivation of the last example}
Figure~\ref{figure:Another invalid derivation of the last example} shows another derivation of the last example $x : [a^*]p$, where the left and right derivation branches each form a back-link to the root. We call these two derivation branches the ``left cycle'' and the ``right cycle'', respectively. 
The derivation trace $\tau_1\tau_3\tau_4\tau_5\tau_6\tau_1...$ formed by picking the formulas $\{\tau_i\ |\ 1\le i\le 6\}$ over any derivation path that covers only the right cycle is a progressing one, because every pair $(\tau_3, \tau_4)$ appearing in the trace is a progressing step. 
However, the derivation trace $\tau_1\tau_7\tau_8\tau_{9}\tau_{10}\tau_1\tau_7...$ formed by picking the formulas $\{\tau_i\ |\ i = 1, 7\le i\le 10\}$ over any derivation path that covers only the left cycle is not progressing.  
\end{example}
\fi

\ifx
Note that our choice of taking $([a]R)$ and $([\phi?]R)$ as the progressing steps is essentially equivalent to the approach adopted in work such as~\cite{Docherty19,Oda25}, where the rule $([*]R)$ serves as the progressing step in place of $([a]R)$ and $([\phi?]R)$. 
Rule $([*]R)$ in those works is simply the rule $([*])$ restricted to the right side of a sequent.
Intuitively, our soundness and completeness proofs for \GiiiPPLcyc\ confirm this equivalence: after applying $([*])$ on the right side, one must eventually apply $([\phi?]R)$ and $([a]R)$ (on some atomic program) in subsequent steps in order to form a valid cyclic derivation. 
\fi

\begin{definition}[Cyclic Proofs of \GiiiPPL]
\label{def:Cyclic Proofs of GiiiPPL}
A preproof of \GiiiPPL\ is called a ``cyclic proof'' if there is a progressing derivation trace over every infinite derivation path. 
\end{definition}

\begin{example}
In Example~\ref{ex:another invalid derivation of the last example}, any derivation path that only covers the left cycle or covers the left cycle for an infinite number of times does not have a derivation trace, since any pair $(\tau_2, \tau_7(\tau_1))$ over an instance of the rule $(cut)$ is not a trace pair (according to Definition~\ref{def:Derivation Traces}). 
For this reason, the whole derivation (as a preproof) is not a cyclic proof.  
\end{example}

\ifx
\begin{example}
In Example~\ref{ex:another invalid derivation of the last example}, 
for any derivation path that covers the left cycle, we fail to find a progressing derivation trace, although the pair $(\tau_{11}, \tau_{12})$ introduced by rule $([\phi?]R)$ is a progressing step. 
For this reason, the whole derivation (as a preproof) is not a cyclic proof.  
\end{example}
\fi

Section~\ref{section:Soundness of GiiiPPLcyc} shows that a cyclic proof  of \GiiiPPL\ always leads to a valid conclusion. 

In \GiiiPPL, we also say that a sequent is provable if there is a cyclic proof of it. 
For a labelled \PPL\ formula $\tau$, we write $\GiiiPPL\vdash \tau$ (or simply $\vdash \tau$) if the sequent $\cdot \Rightarrow \tau$ is provable in \GiiiPPLcyc. 

\section{Soundness of \GiiiPPLcyc}
\label{section:Soundness of GiiiPPLcyc}

We analyze the soundness of \GiiiPPLcyc\ for any \PPL\ formula. 
By the equivalence between \PPL\ formulas and their labelled counterparts established in Section~\ref{section:Labelled PPL Formulas}, it suffices to analyze soundness (and, in later sections, completeness) for the labelled form $X : \phi$ with $X\in \TVar$, for any \PPL\ formula $\phi$. 

\begin{theorem}[Soundness of \GiiiPPLcyc]
    \label{theo:Soundness of GiiiPPLcyc}
For any labelled \PPL\ formula $X : \phi$ with variable $X\in \TVar$, 
if $\vdash X : \phi$, then $\models X : \phi$. 
\end{theorem}

\textbf{Main Idea of the Proof of Theorem~\ref{theo:Soundness of GiiiPPLcyc}}. 
We mainly follow the idea behind~\cite{Brotherston08,zhang2025parameterizeddynamiclogic} to prove Theorem~\ref{theo:Soundness of GiiiPPLcyc}.
We proceed by contradiction. 
Assume that the conclusion of a cyclic proof is invalid. 
Then, by the soundness of the rules of \GiiiPPL\ and a special selection procedure specified by Lemma~\ref{lemma:infinite descent sequence}, we can find an infinite derivation path all of whose nodes are invalid. 
By Definition~\ref{def:Cyclic Proofs of GiiiPPL} it contains a progressing trace, which would lead to an infinitely descending chain in some \emph{well-founded set} (introduced below). 
But this violates the definition of the well-foundedness itself. 

Our proof follows the proof strategy of~\cite{zhang2025parameterizeddynamiclogic} for a more general parameterized dynamic logic, but relaxes the termination condition used there to accommodate the explicit regular-program model. 
The key technical contribution is the introduction of \emph{counter-example sequences} for \PPL\ formulas (Definition~\ref{def:Minimum Counter-example Traces}) and Lemma~\ref{lemma:infinite descent sequence}. Counter-example sequences offer a semantics-related, well-founded measure, while Lemma~\ref{lemma:infinite descent sequence} describes a mechanism for finding this invalid derivation path and a related progressing trace.  

\ifx
We mainly follow the idea behind~\cite{Brotherston08,zhang2025parameterizeddynamiclogic,} to prove Theorem~\ref{theo:Soundness of GiiiPPLcyc}.
We proceed by contradiction. 
Assume that the conclusion of a cyclic proof is invalid. 
Then by the soundness of the rules of \GiiiPPL\ there must exist an infinite \emph{invalid infinite derivation path} in which each node is invalid. 
By Definition~\ref{def:Cyclic Proofs of GiiiPPL} it must contain a progressing trace, which would lead to an infinitely descending chain in some \emph{well-founded set} (introduced below). 
But this violates the definition of the well-foundedness itself. 
Our proof follows the proof strategy of~\cite{zhang2025parameterizeddynamiclogic} for a more general parameterized dynamic logic, but relaxes the termination condition required there to adapt itself to a more explicit program model---regular programs. 
The key technical contribution is the introduction of \emph{counter-example sequences} for \PPL\ formulas (Definition~\ref{def:Minimum Counter-example Traces}) and Lemma~\ref{lemma:infinite descent sequence}. Counter-example sequences offer a semantics-related, well-founded measure, while Lemma~\ref{lemma:infinite descent sequence} describes a mechanism for finding this invalid derivation path and a related progressing trace.  



\fi

Below, we first introduce the concept of well-foundedness and the well-founded relation on which we rely, and then present the proof skeleton for Theorem~\ref{theo:Soundness of GiiiPPLcyc}. 
Other proof details are given in Appendix~\ref{section:About the Soundness of GiiiPPLcyc}. 

\ifx
Below we firstly introduce the concept of well-foundedness and the well-founded relation we rely on, then we focus on the main skeleton of proving Theorem~\ref{theo:Soundness of GiiiPPLcyc}. 
Other proof details are given in Appendix~\ref{section:About the Soundness of GiiiPPLcyc}. 
\fi

\begin{definition}[Well-foundedness]
Given a set $S$ and a partial-order relation $\preceq$ on $S$, 
$\preceq$ is called a \emph{well-founded relation} over $S$ if, for any element $a$ in $S$, there is no infinite descending sequence $a \succ a_1 \succ a_2 \succ \ldots$ in $S$. 
The set $S$ is called a \emph{well-founded} set w.r.t. $\preceq$. 
\end{definition}

For our soundness proof, we rely on the well-known well-founded set of natural numbers with the usual order $\le$.

\ifx
\begin{definition}[Relation $\mult$]
\label{def:Relation mult}
Given two finite sets $\cnt_1$ and $\cnt_2$ of finite strings (i.e. a sequence of elements),  
    $\cnt_1\mult \cnt_2$ is defined if
    either (1) $\cnt_1 = \cnt_2$; or (2) set $\cnt_1$ can be obtained from $\cnt_2$ by replacing one or more elements of $\cnt_2$ each with a finite number of elements, such that
    for each replaced element $\str$, its replacements $\str_1,...,\str_n$ ($n\ge 1$) in $\cnt_1$ satisfy that $\len(\str_i) < \len(\str)$ for all $1\le i\le n$,  where $\len$ is a self-defined length of strings that maps a string to a natural number. 
\end{definition}


\begin{example}
Consider an example of relation $\mult$: 
Let $C_1=\{tr_1, tr_2, tr_3\}$, where traces $tr_1 \dddef s_0 s_1 s_2 s_3 s_4, tr_2 \dddef s_0 s_1s_5s_6s_7$ and $tr_3 \dddef s_0s_8$; 
$C_2 = \{tr'_1, tr'_2\}$, where 
$tr'_1 \dddef s_1 s_2 s_3 s_4, tr'_2 \dddef s_1 s_5s_6s_7$. 
Choose $\len$ as the length of traces $|\cdot|$.  
We see that $|tr'_1| < |tr_1|$ and $|tr'_2| < |tr_2|$. $tr_3$ can be seen being replaced by an empty trace (with length $0$) not shown in $C_2$.  
Hence $C_2\mult C_1$. 
Since $C_1\neq C_2$, $C_2\pmult C_1$. 
\end{example}
\fi

\ifx
\begin{example}
Consider an example of relation $\mult$: 
Let $C_1=\{tr_1, tr_2, tr_3\}$, where $tr_1 \dddef s_0 s_1 s_2 s_3 s_4, tr_2 \dddef s_0 s_1s_5s_6s_7$ and $tr_3 \dddef s_0s_8$; 
$C_2 = \{tr'_1, tr'_2\}$, where 
$tr'_1 \dddef s_1 s_2 s_3 s_4, tr'_2 \dddef s_1 s_5s_6s_7$. 
We see that $tr'_1$ is a proper suffix of $tr_1$ and $tr'_2$ is a proper suffix of $tr_2$. 
$C_2$ can be obtained from $C_1$ by replacing $tr_1$ and $tr_2$ with $tr'_1$ and $tr'_2$  respectively, and removing $tr_3$ (or equally saying that replacing $tr_3$ with an empty set of elements). 
Hence $C_2\mult C_1$. 
Since $C_1\neq C_2$, $C_2\pmult C_1$. 
\end{example}
\fi

\ifx
\begin{prop}
\label{prop:well-foundedness of relation pmult}
    Relation $\mult$ is a  well-founded relation. 
\end{prop}
We omit the proof of Proposition~\ref{prop:well-foundedness of relation pmult}.  
Relation $\mult$ in fact can be seen as a variant of the ``multi-set ordering'' introduced in~\cite{Dershowitz79}, where it has been proved to be well-founded. 
The only difference of the multi-set ordering from relation $\mult$ is that there it uses suffixes rather than lengths to measure the relationship between replacements and their replaced elements.
Intuitively, we observe that for two sets $\cnt_1$ and $\cnt_2$ such that $\cnt_1\pmult \cnt_2$, 
for each set $D_{\str}$ of the strings in $\cnt_1$ that replace a string $\str$ in $\cnt_2$, 
the maximum length of the strings of $D_{\str}$ is strictly smaller than that of $\str$. 
By that $\cnt_2$ is finite, $\cnt_1$ is finite. 
And the maximum length of a finite set cannot be infinitely decreasing by the well-foundedness of natural numbers.
\fi

\textbf{Proof Skeleton}. 
Below we give the main skeleton of the proof of Theorem~\ref{theo:Soundness of GiiiPPLcyc}, deferring the details of the proof of Lemma~\ref{lemma:infinite descent sequence} to Appendix~\ref{section:About the Soundness of GiiiPPLcyc}.

\ifx
Following the main idea above, we first introduce a class of well-founded sets called \emph{counter-example sequences of \PPL\ formulas} together with some of its relevant definitions. This type of sequences relate to the semantics of \PPL\ formulas along an invalid derivation path.  We then state Lemma~\ref{lemma:infinite descent sequence}, which is the key technical step to the proof of Theorem~\ref{theo:Soundness of GiiiPPLcyc} that follows. 
\fi
Following the main idea above, we introduce two kinds of sequences: \emph{execution sequences} (Definition~\ref{def:Execution Sequences}), which record the full intermediate-trace history generated during the execution of a regular program, and \emph{counter-example sequences} (Definition~\ref{def:Minimum Counter-example Traces}), which collect the execution history that witnesses the failure of a \PPL\ formula at a given trace.  We then state Lemma~\ref{lemma:infinite descent sequence}, the key technical step in the proof of Theorem~\ref{theo:Soundness of GiiiPPLcyc}: along each step of an invalid derivation path, the minimum length of the relevant counter-example sequences does not increase, and it strictly decreases whenever the step is progressing. 

\begin{definition}[Execution Sequences]
\label{def:Execution Sequences}
An ``execution sequence'' (simply ``sequence'') is a sequence of traces. 
    Given a regular program $\alpha$ and a starting trace $tr$, 
the set $\ex(tr, \alpha)$ of execution sequences of $\alpha$ is defined as follows: 
$$\ex(tr, \alpha)\dddef \{(tr\cdot s_1)(tr\cdot s_1s_2)\ldots(tr\cdot s_1\ldots s_n)\ |\ s_1s_2\ldots s_n\in \I(\alpha), n > 1\}. $$ 
Each element of the sequence is enclosed in a pair of parentheses: $(\ldots)$. 

We define the length of a sequence $(tr_1)(tr_2)\ldots(tr_n)$ ($n\ge 0$) as the sum of the lengths of all of its traces: 
$$|(tr_1)(tr_2)\ldots(tr_n)|\dddef |tr_1|+|tr_2|+\cdots+|tr_n|. $$ 
\end{definition}
We often use $seq$ to denote an execution sequence. 

Intuitively, an execution sequence of a program $\alpha$ starting from a trace $tr$ records the full sequence of intermediate traces generated step by step during the execution of $\alpha$.  Rather than collapsing the execution into a single final trace, it retains each intermediate trace as a separate element, so that the sequence $(tr\cdot s_1)(tr\cdot s_1s_2)\cdots(tr\cdot s_1\cdots s_n)$ explicitly captures the history of how the execution progresses from $tr$ through the successive transitions of $\alpha$.  The length of such a sequence is defined as the total number of states accumulated across all its component traces.

\begin{example}
    \label{ex:execution sequences}
Let $a_1$ and $a_2$ be two actions, and suppose that $s_1s_2s_3\in \I(a_1\seq a_2)$, with $s_1s_2\in \I(a_1)$ and $s_2s_3\in \I(a_2)$. Then $\ex(s_1, a_1\seq a_2) = \{(s_1)(s_1s_2)(s_1s_2s_3)\}$. 
\end{example}

\ifx
\begin{definition}[Counter-example Sequences]
    \label{def:Minimum Counter-example Traces}
    Given a Kripke frame $(\Wd, \I)$ and a dynamic formula $\phi$, the set of ``counter-example sequences'' of $\phi$ beginning with a trace $tr$, denoted by $\EX(tr, \phi)$, is the set of sequences inductively defined as follows:
    \begin{enumerate}
        \item $\EX(tr, \psi)\dddef \{(tr)\ |\ \mbox{$\psi$ is an atomic formula and $tr\not\models \psi$}\}$; 
        \item $\EX(tr, \ofst\psi)\dddef \{(tr)\cdot seq\ |\ seq\in \EX(tr_b, \psi)\}$;
        \item $\EX(tr, \phi_1\osuf \phi_2)\dddef \{(tr)\cdot seq\ |\ seq\in \EX(tr', \phi_2\vee (\phi_1\wedge \phi_1\osuf\phi_2)), tr = s\cdot tr'\}$;
        \item $\EX(tr, [\alpha]\psi)\dddef \{seq\cdot seq'\ |\ seq\in \ex(tr, \alpha), seq'\in \EX(seq_e, \psi)\}$.
    \end{enumerate}
\end{definition}
\fi

\begin{definition}[Counter-example Sequences]
    \label{def:Minimum Counter-example Traces}
Given a Kripke frame $(\Wd, \I)$ and a dynamic formula $\phi$, the set of ``counter-example sequences'' of $\phi$ beginning with a trace $tr$, denoted by $\EX(tr, \phi)$, is defined according to the form of $\phi$ as follows: 
    \begin{enumerate}
        \item $\EX(tr, [\alpha]\psi)\dddef \{seq\cdot seq'\ |\ seq\in \ex(tr, \alpha), seq'\in \EX(seq_e, \psi)\}$;
        \item $\EX(tr, \phi)\dddef \{(tr)\ |\ tr\not\models \phi\}$, where $\phi$ is of any other form. 
    \end{enumerate}
\end{definition}

\ifx
\begin{definition}[Counter-example Sequences]
    \label{def:Minimum Counter-example Traces}
    Given a Kripke frame $(\Wd, \I)$ and a dynamic formula $\phi$, the set of ``counter-example sequences'' of $\phi$ beginning with a trace $tr$, denoted by $\EX(tr, \phi)$, is the set of sequences inductively defined as follows:
    \begin{enumerate}
        \item $\EX(tr, \psi)\dddef \{(tr)\ |\ \mbox{$\psi$ is an atomic formula and $tr\not\models \psi$}\}$; 
        \item $\EX(tr, \ofst\psi)\dddef \{(tr)\cdot seq\ |\ seq\in \EX(tr_b, \psi)\}$;
        \item $\EX(tr, \phi_1\osuf \phi_2)\dddef \{(tr)\cdot seq\ |\ seq\in \EX(tr', \phi_2\vee (\phi_1\wedge \phi_1\osuf\phi_2)), tr = s\cdot tr'\}$;
        \item $\EX(tr, \neg\psi) \dddef \EXA(tr, \psi)$;
        \item $\EX(tr, \psi\wedge \varphi)\dddef \EX(tr, \psi)\cup \EX(tr, \varphi)$;
        \item $\EX(tr, [\alpha]\psi)\dddef \{seq\cdot seq'\ |\ seq\in \ex(tr, \alpha), seq'\in \EX(seq_e, \psi)\}$.
    \end{enumerate}
    Where the set $\EXA(tr, \phi)$ is defined inductively simultaneously with $\EX(tr, \phi)$ as follows:
    \begin{enumerate}[i.]
    \item $\EXA(tr, \psi)\dddef \{(tr)\ |\ \mbox{$\psi$ is an atomic formula and $tr\models \psi$}\}$;
    \item $\EXA(tr, \ofst\psi)\dddef \{(tr)\cdot seq\ |\ seq\in \EXA(tr_b, \psi)\}$;
    \item $\EXA(tr, \phi_1\osuf \phi_2)\dddef \{(tr)\cdot seq\ |\ seq\in \EXA(tr', \phi_2\vee (\phi_1\wedge \phi_1\osuf\phi_2)), tr = s\cdot tr'\}$;
    \item $\EXA(tr, \neg\psi) \dddef \EX(tr, \psi)$;
    \item $\EXA(tr, \psi\wedge \varphi)\dddef \EXA(tr, \psi)\cap \EXA(tr, \varphi)$;
    \item $\EXA(tr, [\alpha]\psi)\dddef \{(tr)\ |\ \mbox{for all $seq\in \ex(tr, \alpha)$, $\EXA(seq_e, \psi)\neq \emptyset$}\}$.
    \end{enumerate}
\end{definition}
\fi

\ifx
Intuitively, item~5 of Definition~\ref{def:Minimum Counter-example Traces} captures the following idea.  To construct a counter-example to $[\alpha]\psi$ starting from a trace $tr$, one first selects an execution sequence $seq \in \ex(tr, \alpha)$, which records all intermediate traces generated step by step during the execution of $\alpha$ from $tr$, and then appends a counter-example sequence $seq' \in \EX(seq_e, \psi)$ for $\psi$ beginning from the last trace $seq_e$ of $seq$.  The resulting concatenation $seq \cdot seq'$ constitutes a counter-example sequence for $[\alpha]\psi$ starting from $tr$, and the full set $\EX(tr, [\alpha]\psi)$ consists of all such concatenations. 
The set $\EXA(tr, \phi)$, defined simultaneously with $\EX(tr, \phi)$ via mutual induction, serves as its dual: while $\EX(tr, \phi)$ collects sequences that witness the \emph{failure} of $\phi$ at $tr$, $\EXA(tr, \phi)$ collects sequences that witness the \emph{satisfaction} of $\phi$ at $tr$.  The need for $\EXA$ arises from negation --- item 4 of Definition~\ref{def:Minimum Counter-example Traces}.  
\fi

Intuitively, item~1 captures the following idea.  To construct a counter-example to $[\alpha]\psi$ starting from a trace $tr$, one first selects an execution sequence $seq \in \ex(tr, \alpha)$, which records all intermediate traces generated step by step during the execution of $\alpha$ from $tr$, and then appends a counter-example sequence $seq' \in \EX(seq_e, \psi)$ for $\psi$ beginning from the last trace $seq_e$ of $seq$.  The resulting concatenation $seq \cdot seq'$ constitutes a counter-example sequence for $[\alpha]\psi$ starting from $tr$, and the full set $\EX(tr, [\alpha]\psi)$ consists of all such concatenations. 
Item~2 assigns the singleton sequence $(tr)$ to every formula not of the form $[\alpha]\psi$ that fails at $tr$. 

\ifx
Intuitively, item~5 of Definition~\ref{def:Minimum Counter-example Traces} captures the following idea.  To construct a counter-example to $[\alpha]\psi$ starting from a trace $tr$, one first selects an execution sequence $seq \in \ex(tr, \alpha)$, which records all intermediate traces generated step by step during the execution of $\alpha$ from $tr$, and then appends a counter-example sequence $seq' \in \EX(seq_e, \psi)$ for $\psi$ beginning from the last trace $seq_e$ of $seq$.  The resulting concatenation $seq \cdot seq'$ constitutes a counter-example sequence for $[\alpha]\psi$ starting from $tr$, and the full set $\EX(tr, [\alpha]\psi)$ consists of all such concatenations. 
The set $\EXA(tr, \phi)$, defined simultaneously with $\EX(tr, \phi)$ via mutual induction, serves as its dual: while $\EX(tr, \phi)$ collects sequences that witness the \emph{failure} of $\phi$ at $tr$, $\EXA(tr, \phi)$ collects sequences that witness the \emph{satisfaction} of $\phi$ at $tr$.  The need for $\EXA$ arises from negation --- item 4 of Definition~\ref{def:Minimum Counter-example Traces}.  
\fi

\begin{example}
    \label{ex:CT and MCT}
 Let $\I(\alpha)_s$ be the set of traces in $\I(\alpha)$ that start with the world $s$. 
In Example~\ref{ex:another invalid derivation of the last example} (see Figure~\ref{figure:example of ill cyclic derivation}), for the formula $\tau_3 = x : [a][a^*]p$, let $s_1$ be an arbitrary world. Suppose that, for some worlds $s_2$, $s_3$, and $s_4$, $\I(a)_{s_1} = \{s_1s_2\}$ and $\I(a^*)_{s_2}=\{s_2s_2, s_2s_3\}$. Then, by Definition~\ref{def:Execution Sequences}, $\ex(s_1, a) = \{(s_1)(s_1s_2)\}$ and $\ex(s_1s_2, a^*)=\{(s_1s_2)(s_1s_2s_2), (s_1s_2)(s_1s_2s_3)\}$. 
Based on $\ex(s_1, a)$ and $\ex(s_1s_2, a^*)$, we can compute: $\EX(s_1, [a][a^*]p) = \{(s_1)(s_1s_2)(s_1s_2s_2), (s_1)(s_1s_2)(s_1s_2s_3)\}$ by Definition~\ref{def:Minimum Counter-example Traces}.   

\ifx
Let $\I(\alpha)_s$ be the set of traces of $I(\alpha)$ that starts with world $s$. 
In Example~\ref{ex:another invalid derivation of the last example} (see Figure~\ref{figure:Another invalid derivation of the last example}), for formula $\tau_4 = x : [a][a^*]p$, let $s_1$ be an arbitrary world, if for some worlds $s_2,s_3,s_4$ there are $\I(a)_{s_1} = \{s_1s_2\}$, $\I(a^*)_{s_2}=\{s_2s_2, s_2s_3, s_2s_3s_4\}$, then we have $\mex(s_1, a) = \{(s_1)(s_1s_2)\}$, $\mex(s_2, a^*)=\{(s_2)(s_2s_2), (s_2)(s_2s_3)\}$ according to Definition~\ref{def:Minimum Execution Traces}. Note that $(s_2)(s_2s_3)(s_2s_3s_4)\notin \mex(s_2,a^*)$ because it is not minimal: its action sequence is $aa$, while the action sequence of the sequence $(s_2)(s_2s_3)$ is $a$. 
Based on $\mex(s_1, a)$ and $\mex(s_2, a^*)$, we can compute: $\EX(s_1, [a][a^*]p) = \{(s_1)(s_1s_2), (s_1)(s_1s_2)(s_1s_2s_2), (s_1)(s_1s_2)(s_1s_2s_3)\}$ by Definition~\ref{def:Minimum Counter-example Traces}.   
\fi
\end{example}

\begin{definition}[Minimum Length of Execution Sequences]
For a set $A$ of execution sequences, $\min(A)$ denotes the length of the shortest sequence in $A$. 
\end{definition}

We call $\lm\in \LM$ a \emph{counter-example mapping} of a node $\nu$ if it is one of the mappings that make $\nu$ invalid. 

\begin{lemma}
    \label{lemma:infinite descent sequence}
    \ifx NEW ONE
    In a cyclic proof, 
    let $\nu$ be an invalid conclusion of an instance of a rule with $\lm$ being a counter-example mapping, and let $\tau$ be a formula on the right side of $\nu$. 
    Then there exist a formula $\tau'$ on the right of a premise $\nu'$ of $\nu$ and a counter-example mapping $\lm'$ such that $\lm'\not\models \tau'$ (thus making $\nu'$ invalid). 
    And if $(\tau, \tau')$ is a trace pair of the form $(\sigma : \phi, \sigma' : \phi')$, then $\min(\EX(\lm'(\sigma'), \phi')) \le \min(\EX(\lm(\sigma), \phi))$. 
    Moreover, if $(\tau, \tau')$ is a progressing step, then $\min(\EX(\lm'(\sigma'), \phi')) < \min(\EX(\lm(\sigma), \phi))$. 
    \fi

    In a cyclic proof, 
let $\nu$ be an invalid conclusion of an instance of a rule, where $\lm$ is a counter-example mapping, and let $\tau\dddef \sigma : \phi$ be an invalid trace element on the right side of $\nu$. 
    Then there exists an invalid premise $\nu'$ of the instance, a counter-example mapping $\lm'$ of $\nu'$ and an invalid formula $\tau'\dddef \sigma' : \phi'$ on the right side of $\nu'$, such that 
    $(\tau, \tau')$ is a trace pair over $(\nu, \nu')$, and
    $\min(\EX(\lm'(\sigma'), \phi')) \le \min(\EX(\lm(\sigma), \phi))$. 
    Moreover, if $(\tau, \tau')$ is a progressing step, then $\min(\EX(\lm'(\sigma'), \phi')) < \min(\EX(\lm(\sigma), \phi))$. 
    
\end{lemma}

\ifx
The proof of this lemma is given in Appendix~\ref{section:About the Soundness of GiiiPPLcyc}. 
The proof proceeds by a case analysis over the derivation rules that can give rise to a step $(\nu, \nu')$.  For the progressing rules $([a]R)$ and $([\phi?]R)$, one shows that each execution sequence in $\EX(\lm'(\sigma'), \phi')$ is a proper suffix of the corresponding sequence in $\EX(\lm(\sigma), \phi)$, so the total length strictly decreases; combined with the finiteness of $\EX$ sets (Proposition~\ref{prop:finite execution strings}), this yields $\EX(\lm'(\sigma'), \phi')\pmult \EX(\lm(\sigma), \phi)$.  For the non-progressing rules $([a]L)$, $([\phi?]L)$, $(\Sub)$, $(\osuf)$, and $(\ofst)$, a similar but weaker argument gives the non-strict relation $\mult$.
\fi

\ifx
!++ @check @keep
The proof of this lemma is given in Appendix~\ref{section:About the Soundness of GiiiPPLcyc}.
The proof proceeds by a case analysis over the derivation rules that can give rise to a step $(\nu, \nu')$.  For the progressing rules $([a]R)$, $([\phi?]R)$, $(\osuf R)$ and $(\ofus R)$, one shows that the shortest sequence in $\EX(\lm'(\sigma'), \phi')$ is a proper suffix of the corresponding shortest sequence in $\EX(\lm(\sigma), \phi)$, so the minimum length strictly decreases, yielding $\min(\EX(\lm'(\sigma'), \phi')) < \min(\EX(\lm(\sigma), \phi))$.  For the other non-progressing rules, a similar but weaker argument establishes the non-strict inequality $\min(\EX(\lm'(\sigma'), \phi')) \le \min(\EX(\lm(\sigma), \phi))$. 
++!
\fi

The proof of this lemma is given in Appendix~\ref{section:About the Soundness of GiiiPPLcyc}.
It proceeds by a case analysis on the derivation rule giving rise to $(\nu,\nu')$.  For each progressing instance of $([a]R)$, $([\phi?]R)$, $(\osuf R)$, or $(\ofus R)$, the argument identifies a counter-example sequence for the premise whose length is strictly smaller than the minimum length for the conclusion; hence $\min(\EX(\lm'(\sigma'),\phi'))<\min(\EX(\lm(\sigma),\phi))$.  For a non-progressing instance, it selects an invalid premise for which the corresponding minimum length does not increase, yielding $\min(\EX(\lm'(\sigma'),\phi'))\leq\min(\EX(\lm(\sigma),\phi))$.

\begin{example}
Based on the setting in Example~\ref{ex:CT and MCT}, Definition~\ref{def:Minimum Counter-example Traces} also allows us to compute the set $\EX$ for the formula $\tau_4 = x\cat y : [a^*]p$ in Figure~\ref{figure:example of ill cyclic derivation} w.r.t. the trace $s_1s_2$: $\EX(s_1s_2, [a^*]p) = \{(s_1s_2)(s_1s_2s_2), (s_1s_2)(s_1s_2s_3)\}$, in which 
sequences $(s_1s_2)(s_1s_2s_2)$ and $(s_1s_2)(s_1s_2s_3)$ are proper suffixes of sequences $(s_1)(s_1s_2)(s_1s_2s_2)$ and $(s_1)(s_1s_2)(s_1s_2s_3)$ in $\EX(s_1, [a][a^*]p)$ respectively. 
Let $\lm$ be a mapping such that $\lm(x) = s_1$. 
By $\I(a)_{s_1}$ and Definition~\ref{def:semantics of labelled PPL formulas and label relations}, we have $\lm(y) = s_2$. 
Therefore, for the progressing step $(\tau_3, \tau_4)$ in the derivation of Figure~\ref{figure:example of ill cyclic derivation}, we have 
$\min(\EX(\lm(x\cat y), [a^*]p)) = \min(\EX(s_1s_2, [a^*]p)) = 2 < 3 = \min(\EX(s_1, [a][a^*]p)) = \min(\EX(\lm(x), [a][a^*]p))$. 

\ifx
 Based on the settings in Example~\ref{ex:CT and MCT}, by Definition~\ref{def:Minimum Counter-example Traces}, we can also compute set $\EX$ for formula $\tau_5 = x\cat y : [a^*]p$ w.r.t. trace $s_1s_2$: $\EX(s_1s_2, [a^*]p) = \{(s_1s_2)(s_1s_2s_2), (s_1s_2)(s_1s_2s_3)\}$, in which
sequences $(s_1s_2)(s_1s_2s_2)$ and $(s_1s_2)(s_1s_2s_3)$ are proper suffixes of sequences $(s_1)(s_1s_2)(s_1s_2s_2)$ and $(s_1)(s_1s_2)(s_1s_2s_3)$ in $\EX(s_1, [a][a^*]p)$ respectively. 
Let $\lm$ be a mapping such that $\lm(x) = s_1$. 
By $\I(a)_{s_1}$ and Definition~\ref{def:semantics of labelled PPL formulas and label relations}, we have $\lm(y) = s_2$. 
Therefore, for the progressing step $(\tau_4, \tau_5)$ in the derivation of Figure~\ref{figure:Another invalid derivation of the last example}, we have 
$\EX(\lm(x\cat y), [a^*]p) = \EX(s_1s_2, [a^*]p)\pmult\EX(s_1, [a][a^*]p) = \EX(\lm(x), [a][a^*]p)$. 
\fi
\end{example}

\begin{proof}[Proof of Theorem~\ref{theo:Soundness of GiiiPPLcyc}]
    Let $\nu = (\cdot \Rightarrow \sigma : \phi)$. 
For contradiction, suppose $\not\models \sigma : \phi$; that is, $\nu$ is invalid. 
    By the soundness of the rules in $\GiiiPPLcyc$, 
    using Lemma~\ref{lemma:infinite descent sequence}, from $\nu$, we can either conclude that some axiom is invalid, which contradicts its soundness, or obtain an invalid derivation path $P$. From a certain node of $P$, there exists a derivation trace $T=(\sigma_1 : \phi_1)(\sigma_2 : \phi_2)\ldots(\sigma_i : \phi_i)\ldots$ ($i\ge 1$) over $P$ such that $$\min(\EX(\lm_1(\sigma_1), \phi_1)) \ge \min(\EX(\lm_2(\sigma_2), \phi_2))\ge\cdots\ge \min(\EX(\lm_i(\sigma_i), \phi_i))\ge\cdots.$$ 
Since $P$ is progressing, Definition~\ref{def:progressing Step/progressing Derivation Trace of GiiiPPLcyc} implies that it must contain infinitely many instances of the rules $([a]R)$, $([\phi?]R)$, $(\osuf R)$, or $(\ofus R)$. 
    By Definition~\ref{def:Derivation Traces}, this means that in $T$, those trace pairs over the instances of these rules are progressing steps. 
Therefore, by Lemma~\ref{lemma:infinite descent sequence}, the chain above must contain infinitely many strict inequalities $>$. 
This contradicts the well-foundedness of the natural numbers under the relevant strict order. 
    
\end{proof}

\section{Completeness of \GiiiPPLcyc}
\label{section:Completeness of GiiiPPLcyc}

We analyze the completeness of \GiiiPPLcyc\ for any \PPL\ formula $\phi$, captured by its equivalent labelled form $X : \phi$ with $X\in \TVar$.

\begin{theorem}[Completeness of \GiiiPPLcyc]
    \label{theo:Completeness of GiiiPPLcyc}
    Given a labelled \PPL\ formula $X : \phi$ with $X\in \TVar$, if $\models X : \phi$, then $\vdash X : \phi$. 
\end{theorem}

\textbf{Main Idea of the Proof of Theorem~\ref{theo:Completeness of GiiiPPLcyc}}
The proof mainly follows the standard approach for establishing the completeness of labelled proof systems~(cf.~\cite{Docherty19,Brotherston08}): we show that \GiiiPPLcyc\ is capable of deriving every axiom and rule of the proof system of \PPL\ (Table~\ref{table:ppl-proof-system}).  Since the proof system of \PPL\ is complete~\cite{Harel00}, this suffices to conclude the completeness of \GiiiPPLcyc. 

Below we explain how to derive each rule in Table~\ref{table:ppl-proof-system}.  

\textbf{Necessitation Lemma}. 
We first establish a Necessitation Lemma (Lemma~\ref{lemma:Necessitation}) as a generalization of rule $(\textit{Gen})$ in Table~\ref{table:ppl-proof-system}.
This lemma is required when deriving the other rules. We then derive the PDL rules, following the proof strategy of~\cite{Docherty19} but adapting it to the trace labels of process logic. Finally, we derive the rules specific to \PPL. This final step is the core of the completeness proof for the process-logic part.

For a set $A$ of formulas, we use $X : A$ to denote the set $\{X : \phi\ |\ \phi\in A\}$; we use $\op A$, where $\op\in \{[\alpha],\la\alpha\ra\}$, to denote the set $\{\op\phi\ |\ \phi\in A\}$. 

\begin{lemma}[Necessitation]
    \label{lemma:Necessitation}
For any sequent $X:\op\Gamma\Rightarrow X:\op\Delta$, where $\op\in \{[\alpha],\la\alpha\ra\}$, 
there is an ``almost'' cyclic derivation of it: every infinite derivation path has a progressing derivation trace, except that the derivation may have leaves of the form $X :\Gamma\Rightarrow X:\Delta$ that are neither terminal nodes nor buds. Moreover, there is a progressing trace along each derivation path from the root to such a leaf. 
\end{lemma}

Below, we use $(nec)$ to represent the derivation of Lemma~\ref{lemma:Necessitation}.

\begin{proof}[Proof of Lemma~\ref{lemma:Necessitation}]
\ifx
We only consider the case when $\op$ is of the form: $[\alpha]$. Once it holds, the case for $\la\alpha\ra$ trivially follows by applying $(\neg L)$ and $(\neg R)$ to turn the sequent into the case of $[\alpha]$. 
\fi
We consider only the case when $\op$ is of the form $[\alpha]$; the case when $\op$ is $\la\alpha\ra$ then follows immediately, since $\la\alpha\ra\phi = \neg[\alpha]\neg\phi$ and the rules $(\neg L)$ and $(\neg R)$ reduce any sequent involving $\la\alpha\ra$ to the corresponding sequent involving $[\alpha]$. 
And we only consider the case in which $\Gamma$ and $\Delta$ each contain exactly one formula. The general case follows by the same argument applied independently to each formula in $\Gamma$ and $\Delta$ using the same set of rules.

Let $\Gamma=\{\phi\}$, $\Delta=\{\varphi\}$. 
We prove by induction on the structure of regular programs. 

When $\alpha$ is an atomic program $a$, without loss of generality, let $X = Y\cat x$. 
We have the following derivation:

$$
\infer[^{([a]R)}]
{Y\cat x : [a]\phi\Rightarrow Y\cat x:[a]\varphi }
{
    \infer[^{([a]L)}]
    {x\R_a y, Y\cat x:[a]\phi\Rightarrow Y\cat x\cat y:\varphi}
    {
        \infer[^{(\Sub)}]
        {Y\cat x\cat y:\phi\Rightarrow Y\cat x\cat y:\varphi}
        {
            \infer[^{(\Sub)}]
            {Y\cat y : \phi\Rightarrow Y\cat y :\varphi}
            {
                Y\cat x:\phi\Rightarrow Y\cat x:\varphi
            }
        }
    }
}
$$
The derivation above is progressing since rule $([a]R)$ is applied (see Definition~\ref{def:progressing Step/progressing Derivation Trace of GiiiPPLcyc}). 
In the first application of rule $(\Sub)$, we instantiate the label variable by $Y\cat x\cat y = (Y\cat y)[Y\cat x/Y]$.  In the second application of $(\Sub)$, we instantiate by $y = x[y/x]$. 

The case when $\alpha$ is a test $\psi?$ is similar. We omit it here. 

\ifx
When $\alpha$ is a test $\psi?$. Let $X = Y\cat x$, we have
$$
\infer[^{([\phi?]R)}]
{Y \cat x: [\psi?]\phi\Rightarrow Y \cat x : [\psi?]\varphi }
{
    \infer[^{([\phi?]L)}]
    {Y\cat x: [\psi?]\phi, x\cat x:\psi\Rightarrow Y\cat x\cat x:\varphi}
    {
            \infer[^{(\Sub)}]
            {Y\cat x\cat x: \phi\Rightarrow Y\cat x\cat x:\varphi}
            {
                \infer[^{(\Sub)}]
                {Z\cat x : \phi\Rightarrow Z \cat x: \varphi}
                {Y\cat x : \phi\Rightarrow Y \cat x: \varphi}
            }
    }
}
$$
The derivation is progressing since rule $([\phi?]R)$ is applied. 
In the first application of rule $(\Sub)$, we instantiate the label variable by $Y\cat x\cat x = (Z\cat x)[Y\cat x/Z]$. 
In the second application of $(\Sub)$, we instantiate by $Z\cat x = (Y\cat x)[Z/Y]$. 
\fi

\ifx
In the first (inner) application of rule $(\Sub)$, we instantiate by $Z\cat x = (Y\cat x)[Z/Y]$, yielding $Z\cat x : \phi\Rightarrow Z\cat x : \varphi$ from the leaf $Y\cat x : \phi\Rightarrow Y\cat x : \varphi$.  In the second (outer) application of $(\Sub)$, we instantiate by $Y\cat x\cat x = (Z\cat x)[Y\cat x/Z]$, yielding $Y\cat x\cat x : \phi\Rightarrow Y\cat x\cat x : \varphi$. 
\fi

When $\alpha$ is a sequence program $\alpha_1\seq\alpha_2$, we have
$$
\infer[^{([\seq])}]
{X:[\alpha_1\seq\alpha_2]\phi\Rightarrow X:[\alpha_1\seq\alpha_2]\varphi}
{
    \infer[^{\textit{IH}}]
    {X:[\alpha_1][\alpha_2]\phi\Rightarrow X:[\alpha_1][\alpha_2]\varphi}
    {
        \infer[^{\textit{IH}}]
        {X:[\alpha_2]\phi\Rightarrow X:[\alpha_2]\varphi}
        {
            X:\phi\Rightarrow X:\varphi
        }
    }
}, 
$$
where the steps $\textit{IH}$ use the induction hypothesis on $\alpha_1$ and $\alpha_2$, respectively. 
By the induction hypothesis, the derivation is progressing. 

When $\alpha$ is a choice program $\alpha_1\cho\alpha_2$, we have
$$
\infer[^{([\cup])}]
{X:[\alpha_1\cho\alpha_2]\phi\Rightarrow X:[\alpha_1\cho\alpha_2]\varphi}
{
    \infer[^{(wk)}]
    {X:[\alpha_1]\phi, X:[\alpha_2]\phi\Rightarrow X:[\alpha_1]\varphi}
    {
        \infer[^{\textit{IH}}]
        {X:[\alpha_1]\phi\Rightarrow X:[\alpha_1]\varphi}
        {
            X:\phi\Rightarrow X:\varphi
        }
    }
    &
    \infer[^{(wk)}]
    {X:[\alpha_1]\phi, X:[\alpha_2]\phi\Rightarrow X:[\alpha_2]\varphi}
    {
        \infer[^{\textit{IH}}]
        {X:[\alpha_2]\phi\Rightarrow X:[\alpha_2]\varphi}
        {
            X:\phi\Rightarrow X:\varphi
        }
    }
},
$$
where the steps $\textit{IH}$ use the induction hypothesis on $\alpha_1$ and $\alpha_2$, respectively. 
From the induction hypothesis, we see that the two derivation branches are progressing. 

When $\alpha$ is a star program $\alpha^*_1$, let $X = Y\cat x$. We have
$$
\infer[^{([\lup])}]
{Y\cat x:[\alpha^*_1]\phi\Rightarrow Y\cat x:[\alpha^*_1]\varphi}
{
    \infer[^{([\lup])}]
    {Y\cat x:[\alpha^*_1]\phi\Rightarrow Y\cat x:[\true?]\varphi}
    {   
        \infer[^{([\phi?]L, [\phi?]R)}]
        {Y\cat x:[\true?]\phi, Y\cat x:[\alpha_1][\alpha^*_1]\phi\Rightarrow Y\cat x:[\true?]\varphi}
        {
            \infer[^{(wk)}]
            {Y\cat x\cat x : \phi, Y\cat x : [\alpha_1][\alpha^\lup_1]\phi\Rightarrow Y\cat x\cat x:\varphi}
            {
                \infer[^{(\Sub)}]
                {Y\cat x\cat x : \phi\Rightarrow Y\cat x\cat x : \varphi}
                {
                    \infer[^{(\Sub)}]
                    {Z\cat x : \phi\Rightarrow Z\cat x : \varphi}
                    {Y \cat x : \phi\Rightarrow Y\cat x : \varphi}
                }
            }
        }
    }
    &
    \mbox{Cont.}
}, 
$$
where \mbox{Cont.} is the derivation as follows:
$$
\infer[^{[\lup]}]
{\mbox{Cont.} : Y\cat x : [\alpha^*_1]\phi \Rightarrow Y\cat x : [\alpha_1][\alpha^*_1]\varphi}
{
    \infer[^{(wk)}]
    {Y\cat x:[\true?]\phi, Y\cat x:[\alpha_1][\alpha^*_1]\phi\Rightarrow Y\cat x : [\alpha_1][\alpha^*_1]\varphi}
    {
        \infer[^{\textit{IH}}]
        {Y\cat x:[\alpha_1][\alpha^*_1]\phi\Rightarrow Y\cat x : [\alpha_1][\alpha^*_1]\varphi}
        {
            Y\cat x : [\alpha^*_1]\phi\Rightarrow Y\cat x : [\alpha^*_1]\varphi \mbox{ (bud)}
        }
    }
}. 
$$
The left derivation branch is progressing because the rule $([\phi?]R)$ is applied; on the right derivation branch (from ``Cont.''), the induction-hypothesis step $\textit{IH}$ on $\alpha_1$ shows that the derivation path from the root node to its bud is progressing. 
\end{proof}

\ifx
\textbf{Note}: 
!++
@word @mini
The statement of Necessitation Lemma also works for sequent of the form: $X:\la\alpha\ra\Gamma\Rightarrow X:\la\alpha\ra\Delta$, which we do not consider in Lemma~\ref{lemma:Necessitation} (since $\la \cdot\ra$ is not primitive). 
This can be done by firstly turning the sequent into the case for $[\alpha]$ by applying $(\neg L)$ and $(\neg R)$. 
++!
\fi

\textbf{Proving the Rules for PDL}. 
The Necessitation lemma is in fact a general form of the rule $(\textit{Gen})$. 
Thus from it we can trivially derive $(\textit{Gen})$. 
For the other rules of PDL shown in Table~\ref{table:ppl-proof-system}, the derivations of rules (3)--(6) follow directly by applying the rules $([\cho])$, $([\seq])$, $([\phi?])$, and $([\lup])$ of \GiiiPPLcyc. 
The derivation of rule $(\textit{MP})$ is a special case of rule $(cut)$ of \GiiiPPLcyc. 
Below, we only show the derivations for the rules (1), (2), (7) in \GiiiPPLcyc. 
The key idea is to apply the Necessitation lemma to reduce a goal to either terminal nodes or some buds so that a cyclic proof can be formed. 

\ifx
\MO For rule~(1), the key idea is to apply the Necessitation lemma to reduce $X:[\alpha](\phi\to\varphi), X:[\alpha]\phi\Rightarrow X:[\alpha]\varphi$ to the propositional sequent $X:(\phi\to\varphi), X:\phi\Rightarrow X:\varphi$, which is then closed by $(\to L)$ and $(ax)$.  For rule~(2), the biconditional is split into two directions, each again handled via $\Nec$, reducing the modal goals to simple propositional conjunctions.  For rule~(7), the derivation employs a cyclic construction: the companion node carries the induction invariant $X:\phi, X:[\alpha^*](\phi\to[\alpha]\phi)\Rightarrow X:[\alpha^*]\phi$, and the bud back-links to it after applying $\Nec$ to discharge the inductive step, with progressingness guaranteed by the Necessitation lemma. \EMO
\fi

The derivation for formula $X : [\alpha](\phi\to \varphi)\to ([\alpha]\phi\to [\alpha]\varphi)$  with $X\in \TVar$ (corresponding to rule (1)) is shown as follows:
$$
\infer[^{(\to R)}]
{\cdot \Rightarrow X : [\alpha](\phi\to \varphi)\to ([\alpha]\phi\to [\alpha]\varphi)}
{
    \infer[^{(\to R)}]
    {X : [\alpha](\phi\to \varphi)\Rightarrow X : ([\alpha]\phi\to [\alpha]\varphi)}
    {
        \infer[^{\Nec}]
        {X: [\alpha](\phi\to \varphi), X:[\alpha]\phi\Rightarrow X : [\alpha]\varphi}
        {
            \infer[^{(\to L)}]
            {X:(\phi\to\varphi), X : \phi\Rightarrow X : \varphi}
            {
                \infer[^{(ax)}]
                {X:\phi\Rightarrow X:\phi, X:\varphi}
                {}
                &
                \infer[^{(ax)}]
                {X:\varphi, X:\phi\Rightarrow X:\varphi}
                {}
            }
        }
    }
}
$$

The derivation for formula $X : [\alpha](\phi\wedge \varphi) \leftrightarrow ([\alpha]\phi \wedge [\alpha]\varphi)$ with $X\in \TVar$ (corresponding to rule (2)) is shown as follows for both directions separately:

$$
\infer[^{(\to R)}]
{\cdot \Rightarrow X : [\alpha](\phi\wedge \varphi) \to ([\alpha]\phi \wedge [\alpha]\varphi)}
{
    \infer[^{(\wedge R)}]
    {X:[\alpha](\phi\wedge \varphi)\Rightarrow X : ([\alpha]\phi \wedge [\alpha]\varphi)}
    {
        \infer[^{\Nec}]
        {X:[\alpha](\phi\wedge \varphi)\Rightarrow X : [\alpha]\phi}
        {
            \infer[^{(\wedge L)}]
            {X:(\phi\wedge \varphi)\Rightarrow X : \phi}
            {
                \infer[^{(ax)}]
                {X:\phi, X:\varphi\Rightarrow X:\phi}
                {}
            }
        }
        &
        \infer[^{\Nec}]
        {X:[\alpha](\phi\wedge \varphi)\Rightarrow X : [\alpha]\varphi}
        {
            \infer[^{(\wedge L)}]
            {X:(\phi\wedge \varphi)\Rightarrow X : \varphi}
            {
                \infer[^{(ax)}]
                {X:\phi, X:\varphi\Rightarrow X:\varphi}
                {}
            }
        }
    }
}, 
$$

$$
\infer[^{(\to R)}]
{\cdot \Rightarrow X : ([\alpha]\phi \wedge [\alpha]\varphi) \to [\alpha](\phi\wedge \varphi)}
{
    \infer[^{(\wedge L)}]
    {X : ([\alpha]\phi \wedge [\alpha]\varphi)\Rightarrow X : [\alpha](\phi\wedge \varphi)}
    {
        \infer[^{\Nec}]
        {X:[\alpha]\phi, X:[\alpha]\varphi\Rightarrow X:[\alpha](\phi\wedge\varphi)}
        {
            \infer[^{(\wedge R)}]
            {X:\phi, X:\varphi\Rightarrow X:(\phi\wedge \varphi)}
            {
                \infer[^{(ax)}]
                {X:\phi, X:\varphi\Rightarrow X:\phi}
                {}
                &
                \infer[^{(ax)}]
                {X:\phi, X:\varphi\Rightarrow X:\varphi}
                {}
            }
        }
    }
}
$$

The derivation for the formula $X : \phi\wedge [\alpha^*](\phi\to [\alpha]\phi)\to [\alpha^*]\phi$ with $X\in \TVar$ (corresponding to rule (7)) is shown as follows: 
$$
\infer[^{(\to R)}]
{\cdot\Rightarrow X : \phi\wedge [\alpha^*](\phi\to [\alpha]\phi)\to [\alpha^*]\phi}
{
    \infer[^{(\wedge L)}]
    {X : \phi\wedge [\alpha^*](\phi\to [\alpha]\phi)\Rightarrow X: [\alpha^*]\phi}
    {
        \infer[^{([\lup])}]
        {X : \phi, X : [\alpha^*](\phi\to [\alpha]\phi)\Rightarrow X: [\alpha^*]\phi \mbox{ (companion)}}
        {
            \infer[]
            {X:\phi, X:[\alpha^*](\phi\to [\alpha]\phi)\Rightarrow X: [\true?]\phi}
            {\mbox{Lemma~\ref{lemma:lemma for proving rule (inv)}}}
            &
            \mbox{Cont.}
        }
    }
},
$$
where Cont. is the derivation shown below: 
$$
\infer[^{[\lup]}]
{\mbox{Cont. : } X:\phi, X:[\alpha^*](\phi\to [\alpha]\phi)\Rightarrow X: [\alpha][\alpha^*]\phi}
{
    \infer[^{(cut, \textit{lem})}]
    {X:\phi, X: [\true?](\phi\to [\alpha]\phi), X:[\alpha][\alpha^*](\phi\to [\alpha]\phi)\Rightarrow X: [\alpha][\alpha^*]\phi}
    {
        \infer[^{(\to L)}]
        {X:\phi, X: (\phi\to [\alpha]\phi), X:[\alpha][\alpha^*](\phi\to [\alpha]\phi)\Rightarrow X: [\alpha][\alpha^*]\phi}
        {
            \infer[^{(ax)}]
            {X : \phi, X : [\alpha][\alpha^*](\phi\to [\alpha]\phi)\Rightarrow X : [\alpha][\alpha^*]\phi, X : \phi}
            {}
            &
            \mbox{Cont. 2}
        }
    }
},
$$
Cont. 2 is the derivation shown below: 
$$
\infer[^{(wk)}]
{\mbox{Cont. 2 : } X : \phi, X : [\alpha]\phi, X : [\alpha][\alpha^*](\phi\to [\alpha]\phi)\Rightarrow X : [\alpha][\alpha^*]\phi}
{
    \infer[^{\Nec}]
    {X : [\alpha]\phi, X : [\alpha][\alpha^*](\phi\to [\alpha]\phi)\Rightarrow X : [\alpha][\alpha^*]\phi}
    {
        X : \phi, X : [\alpha^*](\phi\to [\alpha]\phi)\Rightarrow X : [\alpha^*]\phi \mbox{ (bud)}
    }
}. 
$$
The ``bud'' back-links to the ``companion''. 
The derivation above also relies on Lemma~\ref{lemma:lemma for proving rule (inv)} below, tagged as the derivation step ``lem''. Its proof is given in Appendix~\ref{section:Completeness of GiiiPPLcyc}. 
\begin{lemma}
\label{lemma:lemma for proving rule (inv)}
For any $X\in \TVar$ and formula $\phi$, sequents $X: \phi\Rightarrow X : [\true?]\phi$ and $X : [\true?]\phi\Rightarrow X : \phi$ are provable in \GiiiPPLcyc. 
\end{lemma}
The derivation step ``(cut, \textit{lem})'' above is shorthand for first applying $(cut)$ (with the additional formula $X : (\phi\to [\alpha]\phi)$) and then applying Lemma~\ref{lemma:lemma for proving rule (inv)} to prove $\ldots X : [\true?](\phi\to[\alpha]\phi)\ldots\Rightarrow \ldots X:(\phi\to[\alpha]\phi)\ldots$. 
The derivation path from the node ``companion'' to the node ``bud'' is progressing because the Necessitation lemma ($\Nec$) is applied. 

\textbf{Proving the Rules Specific to \PPL}. 
It remains to prove the special rules (i)--(xv) for \PPL\ in Table~\ref{table:ppl-proof-system}. 
The derivations for rules~(i)--(xv) follow a uniform pattern: each axiom is lifted to its labelled form, the structural rules of \GiiiPPLcyc\ together with the trace rules $(\osuf R)$, $(\osuf L)$, $(\ofus R)$, $(\ofus L)$ and $(\ofst)$ are used to unfold the temporal connectives and program structures on the label, and the resulting propositional residuals are closed by $(ax)$.  
Each rule is applied backwardly (from the conclusion to the premises). 
Cyclic constructions arise for rules involving the suffix operator and programs, where a companion--bud pair is needed; progressiveness is guaranteed at each back-link by a prior application of $(\osuf R)$, $(\ofus R)$, or the Necessitation lemma. Full derivations, except for one direction of rule~(iii) given below, appear in Appendix~\ref{section:Completeness Proofs of GiiiPPLcyc}. 

\ifx
!++ @word @mini
\textbf{Main Idea of the Proofs of Propositions~\ref{prop:Rule (i)} ---~\ref{prop:Rule (xv)}}:
We start from the root we wish to prove. At each step we proceed a derivation backwardly (from the conclusion to the premises) by applying a rule. We end this process until each branch terminates or back-links to an ancestor such that the whole proof is a cyclic one. 
++!
\fi

As a representative example, below we give the derivation for one direction of rule $(iii)$. 
\begin{figure}
$$
\infer[^{(\vee L)}]
{(iii.a)}
{
    \infer[^{(\osuf R, \osuf L)}]
    {(iii.a1)\mbox{ (comp.)}}
    {
        \infer[^{(\vee R, ax)}]
        {iii.a11 - a13}
        {\surd}
        &
        \infer[^{(wk)}]
        {(iii.a14)}
        {
            (iii.a141)\mbox{ (bud)}
        }
    }
    &
    \infer[^{(\osuf R, \osuf L)}]
    {(iii.a2)\mbox{ (comp.)}}
    {
        \infer[^{(\vee R, ax)}]
        {iii.a21 - a23}{
            \surd
        }
        &
        \infer[^{(wk)}]
        {(iii.a24)}
        {
            (iii.a241)\mbox{ (bud)}
        }
    }
}
$$
\caption{Derivation process of critical nodes for sequent (iii.a)}
\label{figure:Derivation process for sequent (iii.a)}
\end{figure}

\begin{proof}[Proof of One Direction of Proposition~\ref{prop:Rule (iii)}]
Let $X\in \TVar$. We need to prove 
$$
(iii.a)\ X : (\phi\osuf \psi)\vee (\phi\osuf \chi)\Rightarrow X : \phi\osuf (\psi\vee \chi). 
$$
Without loss of generality, let $X = x\cat Y$. By $(\vee L)$, we need
\begin{enumerate}
\item[$(iii.a1)$] $x\cat Y : (\phi\osuf \psi)\Rightarrow x\cat Y : \phi\osuf (\psi\vee \chi)$.
\item[$(iii.a2)$]  $x\cat Y  : (\phi\osuf \chi)\Rightarrow x\cat Y : \phi\osuf (\psi\vee \chi)$. 
\end{enumerate}
Note that, by $(\osuf x)$, the case in which $Y=\varepsilon$ is trivial. Thus, we consider only the case in which $Y\neq \varepsilon$. 

For (iii.a1), by applying $(\osuf R)$ and $(\osuf L)$ on both sides, we have
\begin{enumerate}
\item[$(iii.a11)$] $Y : \psi \Rightarrow Y : \psi\vee \chi, Y : \phi$, which can be proved by $(\vee R)$ and $(ax)$. 
\item[$(iii.a12)$] $Y : \phi, Y : \phi\osuf \psi\Rightarrow Y : \psi\vee \chi, Y : \phi$, which can be proved by $(ax)$. 
\item[$(iii.a13)$] $Y : \psi \Rightarrow Y : \psi\vee \chi, Y : \phi\osuf (\psi\vee \chi)$, which can be proved by $(\vee R)$ and $(ax)$. 
\item[$(iii.a14)$] $Y : \phi, Y : \phi\osuf\psi \Rightarrow Y : \psi\vee \chi, Y : \phi\osuf (\psi\vee \chi)$. 
\end{enumerate}
From $(iii.a14)$, by $(wk)$, we have 
$$
(iii.a141)\ \ Y : \phi\osuf \psi\Rightarrow Y : \phi\osuf (\psi\vee \chi), 
$$
which is exactly (iii.a1) (by replacing $Y$ with $X$). 
The derivation path from (iii.a1) to (iii.a141) is progressing since during the process $(\osuf R)$ is applied. 
\ifx
For (iii.a1), by applying $(\osuf R)$ and $(\osuf L)$ on both sides, we have 
$$
Y : \psi \vee (\phi\wedge \phi\osuf \psi)\Rightarrow Y : \psi\vee \chi, Y : \phi\wedge \phi\osuf (\psi\vee \chi).
$$
Applying rule $(\vee L)$, we have
\begin{enumerate}
\item[$(iii.a11)$] $Y : \psi \Rightarrow Y : \psi\vee \chi, Y : \phi\wedge \phi\osuf (\psi\vee \chi)$, which can be proved by $(\vee R)$ and $(ax)$. 
\item[$(iii.a12)$] $Y : \phi, Y : \phi\osuf \psi\Rightarrow Y : \psi\vee \chi, Y : \phi\wedge \phi\osuf (\psi\vee \chi)$, from which, by splitting $Y : \phi\wedge \phi\osuf (\psi\vee \chi)$ with rule $(\wedge R)$, we need to prove $Y : (\phi\wedge \phi\osuf \psi)\Rightarrow Y : \psi\vee \chi, Y : \phi$, which can be closed by $(ax)$, and 
$$
(iii.a121)\ \ Y : (\phi\wedge \phi\osuf \psi)\Rightarrow Y : \psi\vee \chi, Y : \phi\osuf (\psi\vee \chi). 
$$
By $(wk)$, we have
$$
(iii.a1211)\ \ Y : \phi\osuf \psi\Rightarrow Y : \phi\osuf (\psi\vee \chi), 
$$
which is exactly (iii.a1) (by replacing $Y$ with $X$). 
The derivation path from (iii.a1) to (iii.a1211) is progressing since during the process $(\osuf)$ is applied. 
\end{enumerate}
\fi

The case of (iii.a2) is similar. 
Applying $(\osuf R)$ and $(\osuf L)$ yields 
\begin{enumerate}
\item[$(iii.a21)$] $Y : \chi \Rightarrow Y : \psi\vee \chi, Y : \phi$, which can be proved by $(\vee R)$ and $(ax)$. 
\item[$(iii.a22)$] $Y : \phi, Y : \phi\osuf \chi\Rightarrow Y : \psi\vee \chi, Y : \phi$, which can be proved by $(ax)$. 
\item[$(iii.a23)$] $Y : \chi \Rightarrow Y : \psi\vee \chi, Y : \phi\osuf (\psi\vee \chi)$, which can be proved by $(\vee R)$ and $(ax)$. 
\item[$(iii.a24)$] $Y : \phi, Y : \phi\osuf\chi \Rightarrow Y : \psi\vee \chi, Y : \phi\osuf (\psi\vee \chi)$. 
\end{enumerate}
From $(iii.a24)$, by $(wk)$, we have 
$$
(iii.a241)\ \ Y : \phi\osuf \chi\Rightarrow Y : \phi\osuf (\psi\vee \chi), 
$$
which is exactly (iii.a2) (by replacing $Y$ with $X$). 
The derivation path from (iii.a2) to (iii.a241) is progressing because $(\osuf R)$ is applied during the derivation from (iii.a2) to (iii.a24), with $(x\cat Y : \phi\osuf (\psi\vee \chi), Y : \phi\osuf(\psi\vee \chi))$ as the progressing step. 

\ifx
\begin{enumerate}
\item[$(iii.a21)$] $Y : \chi \Rightarrow Y : \psi\vee \chi, Y : \phi\wedge \phi\osuf (\psi\vee \chi)$, which can be proved by $(\vee R)$ and $(ax)$. 
\item[$(iii.a22)$]   $Y : (\phi\wedge \phi\osuf \chi)\Rightarrow Y : \psi\vee \chi, Y : \phi\wedge \phi\osuf (\psi\vee \chi)$, from which, by splitting $Y : \phi\wedge \phi\osuf (\psi\vee \chi)$ with rule $(\wedge R)$, we need to prove $Y : (\phi\wedge \phi\osuf \chi)\Rightarrow Y : \psi\vee \chi, Y : \phi$, which can be closed by $(ax)$, and 
$$
(iii.a221)\ \ Y : (\phi\wedge \phi\osuf \chi)\Rightarrow Y : \psi\vee \chi, Y : \phi\osuf (\psi\vee \chi). 
$$
By $(wk)$, we have
$$
(iii.a2211)\ \ Y : \phi\osuf \chi\Rightarrow Y : \phi\osuf (\psi\vee \chi), 
$$
which is exactly (iii.a2) (by replacing $Y$ with $X$). 
The derivation path from (iii.a2) to (iii.a2211) is progressing since during the process $(\osuf)$ is applied. 
\end{enumerate}
\fi
\end{proof}

Figure~\ref{figure:Derivation process for sequent (iii.a)} shows the derivation process for the critical nodes appearing in the above proof, from which we can clearly see how the nodes terminate and how the back-links are formed. 
All the proofs for the rest of the rules in Appendix~\ref{section:Completeness Proofs of GiiiPPLcyc}  follow similar derivation processes. 

\textbf{Note}:
In the proofs in Appendices~\ref{section:Completeness Proofs of GiiiPPLcyc} and~\ref{section:About the Completeness of GiiiFOPLcyc}, unless explicitly stated otherwise, we omit the case in which a label variable is the empty label $\varepsilon$ because these cases are trivial. 

\section{Labelled First-Order Process Logic}
\label{section:Labelled First-Order Process Logic - 2}

At the first-order level of process logic, assignments generate program-state information that must be tracked explicitly throughout derivations. We address this by replacing the ``trace-world variables'' ($\WVar$) of \PPL\ labels with \emph{program updates}, which record the cumulative effect of assignments as structured label components. 
 The following subsections define update-based labels and labelled \FOPL\ formulas, present the proof system \GiiiFOPLcyc, and analyze its soundness and relative completeness.

\subsection{Labelled \FOPL\ Formulas}
\label{section:Labelled First-Order Process Logic}
Program updates, introduced in~\cite{Beckert13} as a first-order label structure for capturing program configurations during derivations, provide the foundation for the labelled treatment of \FOPL.  An update $\upd$ records the accumulated effect of assignment steps as an explicit substitution applied to the current world variable, making the program configuration available inside the proof system without requiring rule-specific axioms for each assignment form.

\begin{definition}[Program Updates]
\label{def:Program Updates}
A program update $\upd$ in \FOPL\ is defined as a first-order configuration as follows: 
$$
\upd\dddef x\ |\ \{u := e\}\upd,
$$
where $x\in \WVar$, $u\in \Var$ and $e$ is an expression in $\Arith$. 
\end{definition}

We use \Upd\ to represent the set of all updates in \FOPL. 
$\{u:=e\}$ captures the effect of assigning the value of $e$ to the variable $u$ in the current context. 

The labels of \FOPL\ enrich those of \PPL\ with the explicit structures of program updates $\upd$. 
We give the definitions of labels and label mappings of \FOPL\ in the following definitions.

\begin{definition}[Labels of \FOPL]
A label $\sigma$ in \FOPL\ is a trace defined as follows:
$$
\sigma \dddef \varepsilon\ |\ \upd\ |\ X\ |\ \sigma\cat\sigma,
$$
where $X\in \TVar$, $\varepsilon$ is the empty label, $\cat$ is the concatenation operator, and $\upd$ is a program update. 

We use $\Conf_\fo$ to denote the set of all labels of \FOPL.
\end{definition}

\begin{definition}[Label mappings of \FOPL]
A world mapping $\wm$ in \FOPL\ is a function from $\Upd$ to $\Wd_{\fo}$, defined inductively as follows:
\begin{enumerate}
\item $\wm$ maps a variable $x\in \WVar$ to a world in $\Wd_{\fo}$. 
\item $\wm(\{u := e\}\upd)\dddef \wm(\upd)[u\mapsto \wm(\upd)(e)]$. 
\end{enumerate}
A label mapping $\lm$ is a function from $\Conf_\fo$ to $\Wd^*_{\fo}$, defined as follows:
\begin{enumerate}[i]
\item $\lm$ maps a variable $X\in\TVar$ to a trace of worlds in $\Wd^*_{\fo}$.
\item $\lm(\varepsilon)\dddef \epsilon$.
\item $\lm(\upd)\dddef \wm(\upd)$.
\item $\lm(\sigma_1\cat \sigma_2)\dddef \lm(\sigma_1)\cdot\lm(\sigma_2)$.
\end{enumerate} 
We use $\LM_\fo$ to denote the set of all label mappings from $\Conf_\fo$ to $\Wd^*_{\fo}$.
\end{definition}

In \FOPL, label relations are no longer needed since actions are explicit assignments. 
A \emph{labelled \FOPL\ formula} is defined exactly as for \PPL\ (Definition~\ref{def:Labelled PPL Formulas}), but with $\Conf$ replaced by $\Conf_\fo$.

The semantics of labelled \FOPL\ formulas is introduced as follows.

\begin{definition}[Semantics of Labelled \FOPL\ Formulas]
Given the Kripke frame $(\Wd_\fo, \I_\fo)$ of \FOPL\ and the set $\LM_\fo$ of label mappings, the satisfaction relation $\lm\models_{\LM_\fo}\sigma : \phi$ of a labelled formula $\sigma : \phi$ by a label mapping $\lm$ (w.r.t. $\LM_\fo$) is defined such that
$$
\lm\models_{\LM_\fo}\sigma : \phi, \mbox{ if } \lm(\sigma) \models_{(\Wd_\fo, \I_\fo)} \phi. 
$$
\end{definition}

\subsection{System \GiiiFOPLcyc}
\label{section:System GiiiFOPLcyc}

The proof system \GiiiFOPL\ for labelled \FOPL\ formulas extends \GiiiPPL\ by specializing the rules for atomic programs and first-order arithmetic formulas, while keeping all other rules unchanged.  As shown in Table~\ref{table:giiifopl-rules}, the rule $([u:=e])$ replaces the pair $([a]R)$ and $([a]L)$ in \GiiiPPL: it appends the update $\{u:=e\}$ to the label, recording the assignment effect in the label structure rather than substituting into the formula.  The rule $(\textit{at})$ replaces rule $(p)$ for the evaluation of arithmetic predicates at an update: it strips the trailing label $\sigma$ and evaluates $e_1\triangle e_2$ at the update $\upd$ alone.  The rule $(\textit{ter})$ closes a derivation branch when all remaining formulas are arithmetic sequents of the form $\upd : e_1\triangle e_2$; such sequents can be discharged by an external arithmetic decision procedure or SMT solver.  
The rules $(\forall R)$ and $(\forall L)$ handle universal quantification in the standard way as in \FODL.  

All other rules of \GiiiFOPL\ are inherited from \GiiiPPL\ with labels interpreted in $\Conf_\fo$, and with each world variable $x$ in those rules replaced by a program update $\upd\in \Upd$.  For example, rule $(\osuf R)$ in \GiiiFOPL\ takes the form
$$
  \begin{aligned}
 \infer[^{(\osuf R)}]
 {\Gamma\Rightarrow \upd\cat \sigma : \phi\osuf\psi, \Delta}
 {\Gamma\Rightarrow \sigma : \psi, \sigma : \phi, \Delta & \Gamma\Rightarrow  \sigma : \psi, \sigma : \phi\osuf \psi, \Delta}
 \end{aligned},
$$
with $\sigma\in \Conf_\fo$, following the same trace-reduction principle as in \GiiiPPL.

\begin{table}[t]
\centering
\small
\begin{tabular}{p{0.18\linewidth}p{0.74\linewidth}}
\toprule
Rule & Schematic form \\
\midrule
Inherited rules & All rules of \GiiiPPL\ except $([a]R)$ and $([a]L)$, with labels interpreted in $\Conf_\fo$\\[0.6em]
$([u:=e])$ & $\begin{aligned}\infer=[]{\sigma\cat\upd:[u:=e]\phi}{\sigma\cat\upd\cat\{u:=e\}\upd:\phi}\end{aligned}$\\[1.2em]
$(\textit{at})$ & $\begin{aligned}\infer=[]{\upd\cat\sigma:e_1\triangle e_2}{\upd:e_1\triangle e_2}\end{aligned}$, where $\triangle\in \{<, \le, >, \ge, =\}$\\[1.2em]
$(\textit{ter})$ & $\begin{aligned}\infer[]{\Gamma\Rightarrow\Delta}{}\end{aligned}$,  where each formula of $\Gamma\Rightarrow\Delta$ is of the form $\upd:e_1\triangle e_2$, and the formula $\bigwedge_{\phi\in \Gamma} \phi \to \bigvee_{\psi\in \Delta}\psi$ is valid \\[1.2em]
$(\forall R)$ & $\begin{aligned}\infer[]{\Gamma\Rightarrow\sigma : \forall u.\phi,\Delta}{\Gamma\Rightarrow\sigma : \phi[v/u],\Delta}\end{aligned}$, where $v\in\Var$ is fresh w.r.t. $\Gamma,\Delta,\phi$\\[1.2em]
$(\forall L)$ & $\begin{aligned}\infer[]{\Gamma,\sigma : \forall u.\phi\Rightarrow\Delta}{\Gamma,\sigma : \phi[e/u]\Rightarrow\Delta}\end{aligned}$, where $e$ is a term\\
\bottomrule
\end{tabular}
\caption{Rules specific to labelled first-order process logic.}
\label{table:giiifopl-rules}
\end{table}

\ifx
\ \ 
 $
  \begin{aligned}
 \infer=[^{(\phi)}]
 {xT : \phi}
 {x : \phi}
 \end{aligned}
$
, where $\phi$ a non-dynamical state formula. 
\\
\fi

Table~\ref{table:giiiFOPL-derived-propositional-rules} lists the derived rules for the non-primitive operators $\la\cdot\ra$ and $\exists$ in \GiiiFOPL. Each follows directly from the rules in Table~\ref{table:giiifopl-rules}. 

\begin{table}[t]
\centering
\small
\begin{tabular}{p{0.18\linewidth}p{0.74\linewidth}}
\toprule
Rule & Schematic form \\
\midrule
$(\la u:=e\ra)$ & $\begin{aligned}\infer=[]{\sigma\cat\upd:\la u:=e\ra\phi}{\sigma\cat\upd\cat\{u:=e\}\upd:\phi}\end{aligned}$\\[1.2em]
$(\exists R)$ & $\begin{aligned}\infer[]{\Gamma\Rightarrow\exists u.\phi,\Delta}{\Gamma\Rightarrow\phi[e/u],\Delta}\end{aligned}$, where $e$ is a term\\[1.2em]
$(\exists L)$ & $\begin{aligned}\infer[]{\Gamma,\exists u.\phi\Rightarrow\Delta}{\Gamma,\phi[v/u]\Rightarrow\Delta}\end{aligned}$, where $v\in\Var$ is fresh w.r.t. $\Gamma,\Delta,\phi$\\
\bottomrule
\end{tabular}
\caption{Derived rules used in \GiiiFOPL.}
\label{table:giiiFOPL-derived-propositional-rules}
\end{table}

The soundness of the \GiiiFOPL\ rules follows directly from the semantics of labelled \FOPL\ formulas.

\begin{theorem}
Each rule of \GiiiFOPL\ is sound.
\end{theorem}

The cyclic proof structure of \GiiiFOPL\ is defined in the same way as that of \GiiiPPL, except that, for its derivation traces and progressing derivation traces (Definitions~\ref{def:Derivation Traces} and~\ref{def:progressing Step/progressing Derivation Trace of GiiiPPLcyc}), trace pairs and progressing steps no longer correspond to instances of the rule $([a]R)$ but to instances of its counterpart $([u:=e])$ in \GiiiFOPL. 

Similarly, in \GiiiFOPL, we say that a sequent is also provable when there is a cyclic proof of it. 
For a labelled \FOPL\ formula $\tau$, we write $\GiiiFOPLcyc\vdash \tau$ (or simply $\vdash \tau$) if the sequent $\cdot \Rightarrow \tau$ is provable in $\GiiiFOPLcyc$.

\subsection{Soundness and Relative Completeness of \GiiiFOPLcyc}
\label{section:Soundness and Completeness of GiiiFOPLcyc}

We discuss the soundness and completeness of the proof system \GiiiFOPLcyc\ with respect to all \FOPL\ formulas. 
The soundness proof for \GiiiFOPLcyc\ follows the same structure as that of \GiiiPPLcyc, adapting the cases for the rules $([u:=e])$, $(\forall R)$, and $(\forall L)$ that are specific to the first-order setting.  The completeness of \GiiiFOPLcyc\ is obtained by adapting the standard completeness argument for \FODL~\cite{Harel00} to the trace semantics of \FOPL. 


\begin{theorem}[Soundness of \GiiiFOPLcyc]
    \label{theorem:Soundness of GiiiFOPLcyc}
    For any labelled \FOPL\ formula $X :\phi$, where $X\in \TVar$, if $\vdash X : \phi$, then $\models X : \phi$.
\end{theorem}

The proof of Theorem~\ref{theorem:Soundness of GiiiFOPLcyc} proceeds along the same lines as the proof of soundness for \GiiiPPLcyc. The argument is modified in two respects. First, the case for progressing steps corresponding to the rule $([a]R)$ in \GiiiPPLcyc\ is replaced by the case for the rule $([u:=e])$ in \GiiiFOPLcyc; the proof carries through essentially unchanged because the update rule records the effect of the assignment in the label and the same multiset-ordering argument applies. Second, the cases for the rules $(\forall R)$ and $(\forall L)$ must be handled; these are straightforward and proceed similarly to the case for the rule $(\Sub)$ in \GiiiPPLcyc. 
We omit the proof of Theorem~\ref{theorem:Soundness of GiiiFOPLcyc} in this paper.

 \begin{theorem}[Completeness of \GiiiFOPLcyc]
\label{theorem:Completeness of GiiiFOPLcyc}
    For a labelled \FOPL\ formula $X : \phi$, where $X\in \TVar$ is a label variable, if $\models X : \phi$, then $\vdash X : \phi$. 
\end{theorem}

The completeness of \GiiiFOPLcyc\ adapts the main approach for the completeness of \FODL~\cite{Harel79,Harel00} to trace formulas and trace semantics of \FOPL.  The argument proceeds in two steps.  We first establish two auxiliary lemmas---Lemma~\ref{lemma:completeness lemma 1} and Lemma~\ref{lemma:completeness lemma 2}---that are required for the main completeness argument.  The full proof of Theorem~\ref{theorem:Completeness of GiiiFOPLcyc} then follows from these two lemmas and the structure of \FOPL\ formulas. 

\begin{lemma}
\label{lemma:completeness lemma 1}
    For any FOPL formula $\phi$, there exists a pure arithmetical first-order logic (FOL) formula $\psi^\flat$ such that $\models \psi^\flat \leftrightarrow \phi$.
\end{lemma}

We usually write a formula $\phi$ as $\phi^\flat$ to indicate that it is a pure arithmetical FOL formula. 

We omit the proof of Lemma~\ref{lemma:completeness lemma 1}. 
Here is an informal explanation. 
It is well known that temporal properties can be encoded in first-order form over suitable trace positions (cf.~\cite{TemporalLogic}). 
On the other hand, any \FODL\ formula can be expressed by a pure arithmetical FOL formula~\cite{Harel79,Harel00}. 
The result of Lemma~\ref{lemma:completeness lemma 1} follows directly by combining these two facts.

\begin{lemma}
\label{lemma:completeness lemma 2}
For any \FOPL\ formulas $\phi$ and $\psi$, label variable $X$, and $\op\in \{[\alpha], \la\alpha\ra\}$ for any program $\alpha$, if $\models X : \phi^\flat\Rightarrow X : \op\psi^\flat$, then $\vdash X : \phi^\flat\Rightarrow X : \op\psi^\flat$. 
\end{lemma}

The proof of this lemma is given in Appendix~\ref{section:About the Completeness of GiiiFOPLcyc}. 


The proof of Theorem~\ref{theorem:Completeness of GiiiFOPLcyc} follows the main skeleton of the proof for \FODL~\cite{Harel79}, but uses a refined induction that also considers the lengths of the labels of \FOPL\ formulas.

\begin{proof}[Proof of Theorem~\ref{theorem:Completeness of GiiiFOPLcyc}]
    For a labelled formula $X : \phi$, 
$\phi$ is semantically equivalent to a formula in conjunctive normal form, $C_1\wedge \cdots\wedge C_n$ $(n\ge 1)$. 
    Each clause $C_i$ ($1\le i\le n$) is a disjunction of literals: 
$C_i = l_{i,1}\vee \cdots\vee l_{i,m_i}$, where $l_{i,j}$ ($1\le i\le n$, $1\le j\le m_i$) is an atomic \FODL\ formula or its negation. 
    By the rule for FOL formulas (Table~\ref{table:giiippl-propositional-rules} plus the rules $(\forall R)$ and $(\forall L)$), 
    to prove formula $X : \phi$, 
    it is enough to show that for each clause $C_i$, $\models X : C_i$ implies $\vdash X : C_i$.  
We prove the result by simultaneous induction on both the number $n$ of occurrences of the modalities $[\cdot]$ and $\la\cdot\ra$ in $C_i$ and the length of the label $X$. 

If $n=0$ and the length of $X$ is 1, there are no occurrences of $[\alpha]$ or $\la\alpha\ra$ in $C_i$, so $C_i$ is an FOL formula. By the rule $(\textit{ter})$, we immediately have $\vdash X : C_i$. 

Suppose $n>0$ and the length of $X$ is $N>0$. Without loss of generality, assume that $C_i$ has one of the following forms: $\ofst \psi$, $\psi_1\osuf \psi_2$, or $\psi_1\vee \op \psi_2$. 

    If $C_i = \psi_1\osuf \psi_2$ (the case for $C_i = \ofst \psi$ is similar), without loss of generality, let $X = \upd\cat Y$ for some $\upd\in \Upd$. 
Since $\models \upd\cat Y : \psi_1\osuf \psi_2$ entails $\models Y : \psi_2\vee (\psi_1\wedge \psi_1\osuf \psi_2)$ by the semantics of $\osuf$, we have (a) $\models Y : \psi_2$ or $\models Y : \psi_1$ and (b) $\models Y : \psi_2$ or $\models Y : \psi_1\osuf\psi_2$. 
By the induction hypothesis (the length of $Y$ is $N-1$), 
    we obtain that (1) $\vdash Y : \psi_2$ or $\vdash Y : \psi_1$, and (2) $\vdash Y : \psi_2$ or $\vdash Y : \psi_1\osuf\psi_2$. 
    Therefore, by $(\osuf R)$, $\vdash \upd\cat Y : \psi_1\osuf \psi_2$. 

    If $C_i = \psi_1\vee \op \psi_2$, where $\op\in \{[\alpha], \la\alpha\ra\}$,  
proving $\vdash X : (\psi_1\vee \op \psi_2)$ is equivalent to proving $$\vdash X : \neg \psi_1\Rightarrow X : \op \psi_2$$ by $(\vee R)$ and $(\neg R)$. 
By Lemma~\ref{lemma:completeness lemma 1}, there exist pure arithmetical FOL formulas 
    $\phi^\flat_1$ and $\phi^\flat_2$ such that $\models \phi^\flat_1 \leftrightarrow \neg\psi_1$ and $\models \phi^\flat_2\leftrightarrow \psi_2$. 
    Therefore, we have both $\models X : \phi^\flat_1\Rightarrow X : \op\phi^\flat_2$ and $\models (X: \phi^\flat_2\Rightarrow X : \psi_2)$. 
    From $\models X : \phi^\flat_1\Rightarrow X : \op\phi^\flat_2$, by Lemma~\ref{lemma:completeness lemma 2}, we have 
    \begin{equation}
    \label{equ1}
    \vdash X : \phi^\flat_1\Rightarrow X : \op\phi^\flat_2.
    \end{equation}
By the induction hypothesis on $n$, from $\models X : \neg\psi_1\Rightarrow X : \phi^\flat_1$ (because $\models \phi^\flat_1 \leftrightarrow \neg\psi_1$), we have 
    \begin{equation}
        \label{equ2}
        \vdash X : \neg\psi_1\Rightarrow X : \phi^\flat_1.  
    \end{equation}
    Therefore, by (\ref{equ1}) and (\ref{equ2}), using $(cut)$, we get
    \begin{equation}
    \label{equ4}
    \vdash X : \neg\psi_1\Rightarrow X : \op\phi^\flat_2. 
    \end{equation}
    On the other hand, 
    from $\models (X: \phi^\flat_2\Rightarrow X : \psi_2)$, 
by the induction hypothesis on $n$, we have $\vdash X: \phi^\flat_2\Rightarrow X : \psi_2$. 
By the Necessitation lemma, we have 
    \begin{equation}
        \label{equ3}
    \vdash  X : \op\phi^\flat_2\Rightarrow X : \op\psi_2.
    \end{equation}
    The result is directly obtained by (\ref{equ4}) and (\ref{equ3}).  
    
\end{proof}

\section{Conclusion}
\label{section:Conclusion}

\ifx
This paper develops a complete labelled proof-theoretic framework for process logic. 
The core idea different from previous work is that the labelling discipline makes trace-execution information explicit in derivations, so that to prove a temporal property over a program, we do not have to distribute it into subprograms of the program, but to delay the reasoning of a temporal property after the execution of a program. 
The central contribution is that we build a uniform, modular and cyclic logical framework for process logic on both propositional level (\GiiiPPLcyc) and first-order level (\GiiiFOPLcyc), and make a full analysis and prove their soundness and completeness respectively. 
Together, these results show that labelled process logic is a promising framework for reasoning about trace-based program properties. 

As mentioned in this paper, one of future work is to try to extend the theories of labelled process logic with the omega operator $\alpha^\omega$ to cover temporal properties over infinite program traces. 
We will also explore the applications of \GiiiFOPLcyc\ to practical programs or languages, developing automated proof-search strategies and mechanizing the labelled frameworks in Rocq~\cite{???}. 
\fi

This paper firstly develops a complete labelled proof-theoretic framework for process logic. 
The distinguishing idea is that the labelling discipline makes trace-execution information explicit in derivations. Thus, rather than distributing a temporal property over the subprograms of a program, the framework can defer reasoning about that property until after the program has been executed.
Our central contribution is a uniform, modular, and cyclic logical framework at both the propositional level (\GiiiPPLcyc) and the first-order level (\GiiiFOPLcyc). We make a full analysis and prove the soundness and completeness/relative completeness for \GiiiPPLcyc\ and \GiiiFOPLcyc\ respectively. 
Together, these results show that labelled process logic provides a promising framework for reasoning about trace-based program properties.

One direction for future work is to extend labelled process logic with the omega-iteration operator $\alpha^\omega$ to cover temporal properties over infinite program traces~\cite{Harel00}.
We will also explore applications of \GiiiFOPLcyc\ to practical programs and languages, develop automated proof-search strategies, and mechanize the labelled frameworks in Rocq (formerly Coq)~\cite{Bertot04}.

\ifx
The central contribution is a labelling discipline that makes trace-execution information explicit in derivations: propositional labels record traces over atomic actions, while first-order labels record sequences of program updates. By making this information part of the label structure rather than encoding it into the formula, both systems admit uniform, modular inference rules for the trace connectives of process logic and the modalities of first-order dynamic logic.

The theoretical results establish that the labelled framework is both sound and complete. Soundness follows from combining local soundness of the individual rules with a global cyclic condition: every infinite derivation path that satisfies the progress condition gives rise to an infinite descent in the well-founded multiset ordering on counter-example sequences, which is a contradiction. Completeness is established separately for each system: \GiiiPPLcyc\ is proved complete by showing that it can derive every axiom and rule of the Hilbert-style proof system for \PPL, while the completeness of \GiiiFOPLcyc\ is proved by adapting the classical completeness argument for \FODL\ to the trace-based setting of process logic. 

Together, these results show that labelling provides a principled and uniform mechanism for cyclic reasoning about trace-based program properties.  

 Directions for future work include the application of \GiiiFOPLcyc\ to practical programs or languages, the development of automated proof-search strategies and the mechanization of \GiiiPPLcyc\ and \GiiiFOPLcyc, and the extension of the framework to richer program models such as concurrent or probabilistic processes. 
\fi

\bibliographystyle{ACM-Reference-Format}
\bibliography{logic20260609}

\clearpage
\appendix

\section{Other Proofs of the Main Theorems}

\subsection{Soundness Proofs for \GiiiPPLcyc}
\label{section:About the Soundness of GiiiPPLcyc}

\ifx
\textbf{Content of Lemma~\ref{lemma:infinite descent sequence}}:
In a cyclic proof, 
let $\nu$ be an invalid conclusion of an instance of a rule, where $\lm$ is a counter-example mapping, and let $\tau\dddef \sigma : \phi$ be an invalid trace element on the right side of $\nu$. 
    Then there exists an invalid premise $\nu'$ of the instance, a counter-example mapping $\lm'$ of $\nu'$ and an invalid formula $\tau'\dddef \sigma' : \phi'$ on the right side of $\nu'$, such that 
    $(\tau, \tau')$ is a trace pair over $(\nu, \nu')$, and
    $\min(\EX(\lm'(\sigma'), \phi')) \le \min(\EX(\lm(\sigma), \phi))$. 
    Moreover, if $(\tau, \tau')$ is a progressing step, then $\min(\EX(\lm'(\sigma'), \phi')) < \min(\EX(\lm(\sigma), \phi))$. 

    \ifx
    In a cyclic proof, 
    let $(\sigma:\phi, \sigma':\phi')$ be a step of a derivation trace over a derivation $(\nu, \nu')$ of an invalid derivation path.  
    For any set $\EX(\lm(\sigma), \phi)$ of $\sigma : \phi$ w.r.t a counter-example mappiong $\lm$, 
    there exists a counter-example mapping $\lm'$ and a set $\EX(\lm'(\sigma'), \phi')$ of $\sigma':\phi'$ such that 
    $$\EX(\lm'(\sigma'), \phi')\mult \EX(\lm(\sigma), \phi).$$ Moreover, if $(\sigma:\phi, \sigma':\phi')$ is a progressing step, then 
    $\EX(\lm'(\sigma'), \phi')\pmult \EX(\lm(\sigma), \phi)$. 
\fi
\fi

\begin{proof}[Proof of Lemma~\ref{lemma:infinite descent sequence}]
We consider only the cases in which the derivation step is an instance of one of the rules $([a]R)$, $([\cho])$, $(\Sub)$, and $(\osuf R)$. The cases for instances of the rules $([\phi?]R)$, $([*])$, and $(\ofus R)$ are similar to those for $([a]R)$, $([\cho])$, and $(\osuf R)$, respectively. The other cases are trivial. 

\ifx
\textbf{Case for rule $([a]R)$: }
    If from node $\nu$ rule $([a]R)$ is applied with $\tau\dddef \sigma\cat x : [a]\phi$ the target formula, 
    let $\tau' = (\sigma\cat x\cat y : \phi)$ for some $y$, so $\nu=(\Gamma\Rightarrow \tau, \Delta)$ and $\nu' = (\Gamma, xR_a y\Rightarrow \tau', \Delta)$. 
    In this case, $(\nu, \nu')$ is a progressing step (Definition~\ref{def:progressing Step/progressing Derivation Trace of GiiiPPLcyc}). 
    Since $\lm$ is a counter-example of $\nu$, $\lm\not\models \tau$, so $\EX(\lm(\sigma\cat x), [a]\phi)\neq \emptyset$. 
    Let $\lm(x) = s$, $\lm(y) = s_1$. 
    By the soundness of rule $([a]R)$ and Definition~\ref{def:semantics of labelled PPL formulas and label relations}, for each trace $s_1...s_n\in \EX(\lm(\sigma\cat x\cat y), \phi)$ ($n\ge 0$), 
    trace $ss_1...s_n\in \EX(\lm(\sigma\cat x), [a]\phi)$ with $ss_1\in \I(a)$ has $s_1...s_n$ as its proper suffix. 
    So clearly $|s_1...s_n| < |ss_1...s_n|$. 
    On the other hand, by Proposition~\ref{prop:finite execution strings}, 
    $\EX(\lm(\sigma\cat x), [a]\phi)$ and $\EX(\lm(\sigma\cat x\cat y), \phi)$
    are also finite. 
    Therefore $\EX(\lm(\sigma\cat x\cat y), \phi)\pmult \EX(\lm(\sigma\cat x), [a]\phi)$ by Definition~\ref{def:Relation mult}. 
\fi
\textbf{Case for rule $([a]R)$: }
If the rule $([a]R)$ is applied to $\nu$, 
let $\nu = (\Gamma\Rightarrow \tau, \Delta)$ with $\tau = \sigma\cat x : [a]\phi$, 
then $\nu'$ is the only premise $\Gamma, xR_a y\Rightarrow \sigma\cat x\cat y : \phi, \Delta$ of $\nu$. 
By Definition~\ref{def:Derivation Traces},  $\tau' = \sigma\cat x\cat y : \phi$. 
    By the invalidity of $\nu$, we get that $\lm\not\models \tau$ and $\nu'$ is invalid. 
    From $\lm\not\models \tau$, we have that $\EX(\lm(\sigma\cat x), [a]\phi)\neq \emptyset$. 
Because $\nu'$ is invalid, we know that $\lm\not\models \tau'$ and $\lm\models xR_a y$. 
    By Definition~\ref{def:Minimum Counter-example Traces}, 
for any sequence $(tr_1)(tr_2)\ldots(tr_n)\in\EX(\lm(\sigma\cat x), [a]\phi)$ with $tr_1 = \lm(\sigma\cat x)$, $(tr_2)\ldots(tr_n)\in \EX(\lm(\sigma\cat x\cat y), \phi)$ with $tr_2 = \lm(\sigma\cat x\cat y)$ and $(tr_1)(tr_2)\in \ex(tr_1, a)$ (since $\lm\models xR_a y$). 
Moreover, since $\EX(\lm(\sigma\cat x), [a]\phi)\neq \emptyset$, 
    we thus have $\min(\EX(\lm(\sigma\cat x), [a]\phi)) > \min(\EX(\lm(\sigma\cat x\cat y, \phi))$. 

    \ifx
    From $\lm\models xR_a y$, we have that for any sequence $(tr_1)(tr_2)...(tr_n)\in \EX(\lm(\sigma\cat x\cat y, \phi))$ ($n\ge 1$) starting from $tr_1 = \lm(\sigma\cat x\cat y) = \lm(\sigma)\lm(x)\lm(y)$, 
    we have $(tr_0)(tr_1)(tr_2)...(tr_n)\in \EX(\lm(\sigma\cat x), [a]\phi)$ with $tr_0 = \lm(\sigma)\lm(x)$.  
    This means, every sequence of $\EX(\lm(\sigma\cat x\cat y, \phi))$ is a suffix of some sequence of $\EX(\lm(\sigma\cat x), [a]\phi)$. 
    By that $\EX(\lm(\sigma\cat x\cat y), \phi)\neq \emptyset$, we have $\min(\EX(\lm(\sigma\cat x), [a]\phi)) > \min(\EX(\lm(\sigma\cat x\cat y, \phi)))$. 
    \fi
    
    \ifx
    Let $w_1\xrightarrow{\alpha/\alpha'}w_2$ for some $\alpha'\in \Prog$. 
    By Coincidence property of the operational semantics (Definition~\ref{def:Matching Proof System}), 
    there is a label $\sigma'\in \Conf$ such that $\lm(\sigma') = w_2$ and $(\alpha', \sigma')\in Br(\alpha, \sigma)$. 
    So from $\nu$, we can take the premise of the instance of rule $([\alpha]R)$ that has $\sigma' : [\alpha']\phi$ as the replacement of $\tau$. 
    Let $\tau' = \sigma' : [\alpha']\phi$, and let $\nu'$ be the premise node. 
    $(\tau, \tau')$ is then the trace pair of $(\nu, \nu')$. 
    Since $w_2...w_n\in \ex(\lm(\sigma'), \alpha')$ and $w_n\not\models \phi$, $\nu'$ is invalid. 
    And since $w_1w_2...w_n$ is a shortest path of $\ex(\lm(\sigma), \alpha)$, $w_2...w_n$ must be a shortest path of $\ex(\lm(\sigma'), \alpha')$. 
    Therefore, clearly $\min(\ex(\lm(\sigma'), \alpha')) < \min(\ex(\lm(\sigma), \alpha))$. 

    If from node $\nu$ rule $([a]R)$ is applied with $\tau\dddef \sigma\cat x : [a]\phi$ the target formula, 
    let $\tau' = (\sigma\cat x\cat y : \phi)$ for some $y$, so $\nu=(\Gamma\Rightarrow \tau, \Delta)$ and $\nu' = (\Gamma, xR_a y\Rightarrow \tau', \Delta)$. 
    In this case, $(\nu, \nu')$ is a progressing step (Definition~\ref{def:progressing Step/progressing Derivation Trace of GiiiPPLcyc}). 
    Since $\lm$ is a counter-example of $\nu$, $\lm\not\models \tau$, so $\EX(\lm(\sigma\cat x), [a]\phi)\neq \emptyset$. 
    By the soundness of rule $([a]R)$ and efinition~\ref{def:semantics of labelled PPL formulas and label relations}, we have $\lm(x\cat y) = \lm(x)\lm(y)\in \I(a)$. 
So for each sequence $(tr_1)...(tr_n)\in \EX(\lm(\sigma\cat x\cat y), \phi)$ ($n\ge 1$) with $tr_1 = \lm(\sigma\cat x\cat y)$, 
    there is a sequence $(tr)(tr_1)...(tr_n)\in \EX(\lm(\sigma\cat x), [a]\phi)$ with $tr = \lm(\sigma\cat x)$ that has $(tr_1)...(tr_n)$ as its proper suffix. 
    Hence clearly $|(tr_1)...(tr_n)| < |(tr)(tr_1)...(tr_n)|$. 
    On the other hand, by Proposition~\ref{prop:finite execution strings}, 
    $\EX(\lm(\sigma\cat x), [a]\phi)$ and $\EX(\lm(\sigma\cat x\cat y), \phi)$
    are also finite. 
    Therefore $\EX(\lm(\sigma\cat x\cat y), \phi)\pmult \EX(\lm(\sigma\cat x), [a]\phi)$ by Definition~\ref{def:Relation mult}. 
    \fi
\ifx
\textbf{Case for rule $([a]L)$: }
    Similar to the case $([a]R)$, except that it is possible that $\EX(\lm(\sigma\cat x), [a]\phi) = \emptyset$ (when there is no sequences of $[a]\phi$. That does not affect $\lm$ being a counter-example mapping for $\nu$).
    So, we have  $\EX(\lm(\sigma\cat x\cat y), \phi)\mult \EX(\lm(\sigma\cat x), [a]\phi)$. 
    (Note that in this case, $(\nu, \nu')$ is not a progressing step. )

    \textbf{Case for rule $([\phi?]R)$: } Similar to the case of $([a]R)$ and we omit it. 
    
    \textbf{Case for rule $([\phi?]L)$: } Similar to the case of $([a]L)$ and we omit it.
\fi

    \textbf{Case for rule $([\cho])$: }
If the rule $([\cho])$ is applied to $\nu$, let $\tau = \sigma : [\alpha\cho\beta]\phi$. 
    By Definition~\ref{def:Execution Sequences} and~\ref{def:Minimum Counter-example Traces}, $\EX(\lm(\sigma), [\alpha\cho\beta]\phi)=\EX(\lm(\sigma), [\alpha]\phi)\cup\EX(\lm(\sigma),[\beta]\phi)$. 
    By the invalidity of $\nu$, $\lm\not\models \tau$, so $\EX(\lm(\sigma), [\alpha\cho\beta]\phi)\neq\emptyset$. 
Assume that the shortest sequence in $\EX(\lm(\sigma), [\alpha\cho\beta]\phi)$ belongs to $\EX(\lm(\sigma), [\alpha]\phi)$, and choose $\tau' = \sigma : [\alpha]\phi$ as the replacement of the target formula $\tau$; let $\nu'$ be the premise containing $\tau'$. Clearly, $\min(\EX(\lm(\sigma), [\alpha\cho\beta]\phi)) = \min(\EX(\lm(\sigma), [\alpha]\phi))$. 
Since $\EX(\lm(\sigma), [\alpha\cho\beta]\phi)\neq\emptyset$, we also have $\EX(\lm(\sigma), [\alpha]\phi)\neq\emptyset$; therefore, $\lm\not\models \sigma : [\alpha]\phi$, which means that $\nu'$ is invalid. 

    \textbf{Case for rule $(\Sub)$: }
If the rule $(\Sub)$ is applied to $\nu$, then $\nu$ has the unique premise $\nu'$. 
    Thus the invalidity of $\nu$ implies that of $\nu'$. 
Let $\tau = \sigma[\sigma'/U] : [\alpha]\phi$, where $\sigma'$ is a label and $U\in \LVar$ is a label variable. Choose $\tau' = \sigma : [\alpha]\phi$. 
    $(\tau, \tau')$ is then a trace pair over $(\nu, \nu')$ by Definition~\ref{def:Derivation Traces}. 
    Let $\lm'$ be the label mapping such that $\lm'(U)\dddef \lm(\sigma')$; $\lm'(V)\dddef \lm(V)$ for any other $V\in \LVar$. 
    By the definition of substitution on labels (see Section~\ref{section:Labelled PPL Formulas}), 
    we have $\lm'(\sigma)=\lm(\sigma[\sigma'/U])$. 
    Hence $\EX(\lm(\sigma[\sigma'/U]), \phi) = \EX(\lm'(\sigma), \phi)$. 
    So $\min(\EX(\lm(\sigma[\sigma'/U]), \phi)) = \min(\EX(\lm'(\sigma), \phi))$. 
    
    \ifx
    \textbf{Case for rule $(p)$: }
    If from node $\nu$ rule $(p)$ is applied, $\nu$ has the only premise $\nu'$. 
    Let $\tau = x\cat \sigma : p$. 
    $\EX(\lm(x\cat \sigma), p) = \{(\lm(x)\lm(\sigma))\}$ by Definition~\ref{def:Minimum Counter-example Traces} and that $\nu$ is invalid.
    So $\lm(x)\not\models p$, which means that $\EX(\lm(x), p) = \{(\lm(x))\}$ and $\nu'$ is invalid. 
    Clearly, $\min(\EX(\lm(x\cat \sigma), p)) \ge \min(\EX(\lm(x), p))$ ($=$ holds when $\sigma = \varepsilon$). 
    \fi


    \textbf{Case for rule $(\osuf R)$: }
If the rule $(\osuf R)$ is applied to the node $\nu$, let $\tau = x\cat \sigma : (\phi\osuf \psi)$ be the target formula of $\nu$. 
    By the invalidity of $\nu$, $\lm\not\models \tau$, so $\EX(\lm(x\cat \sigma), \phi\osuf \psi) = \{(\lm(x)\lm(\sigma))\}$. 
According to the rule $(\osuf R)$, since $\lm\not\models \tau$, if $\lm\models \sigma : \phi\osuf \psi$, then $\lm$ must not satisfy $\sigma : \phi$. In this case, by Definition~\ref{def:Derivation Traces}, $(\tau, \sigma : \phi)$ is not a trace pair over the derivation (violating that $\tau$ is a trace element). 
If $\lm\not\models \sigma : \phi\osuf \psi$, let $\tau' = \sigma : \phi\osuf \psi$, and let $\nu'$ be the premise containing $\tau'$. 
    $(\tau, \tau')$ is then a trace pair over $(\nu, \nu')$. 
    By Definition~\ref{def:Minimum Counter-example Traces}, clearly $\EX(\lm(\sigma), \phi\osuf \psi) = \{(\lm(\sigma))\}$. 
Thus, $\min(\EX(\lm(x\cat \sigma), \phi\osuf \psi)) = |(\lm(x)\lm(\sigma))| > |(\lm(\sigma))| = \min(\EX(\lm(\sigma), \phi\osuf \psi))$, according to the definition of the length of an execution sequence in Definition~\ref{def:Execution Sequences}. 
Moreover, from $\lm\not\models \sigma : \phi\osuf \psi$, we immediately obtain the invalidity of $\nu'$. 

    \ifx
    and let $\tau' = \sigma : \phi\osuf \psi$ and $\nu'$ the premise containing $\tau'$. 
    By the invalidity of $\nu$, $\lm\not\models \tau$
    By Definition~\ref{def:Derivation Traces}, $(\tau, \tau')$ is the trace pair over $(\nu, \nu')$. 
    By Definition~\ref{def:Minimum Counter-example Traces}, clearly $\EX(\lm(x\cat \sigma), \phi\osuf \psi) = \{(\lm(x)\lm(\sigma))\}$ and $\EX(\lm(\sigma), \phi\osuf \psi) = \{(\lm(\sigma))\}$.
    So $\min(\EX(\lm(x\cat \sigma), \phi\osuf \psi)) = |(\lm(x)\lm(\sigma))| > |(\lm(\sigma))| = \min(\EX(\lm(\sigma), \phi\osuf \psi))$, according to the definition of lengths of execution sequences in Definition~\ref{def:Execution Sequences}. 
    And by $\EX(\lm(\sigma), \phi\osuf \psi)\neq\emptyset$ we trivially obtain the invalidity of $\nu'$.

    By the invalidity of $\nu$, $\nu'$ is invalid, so $\EX(\lm(\sigma), \psi\vee (\phi\wedge \phi\osuf\psi))\neq \emptyset$. 
    By Definition~\ref{def:Minimum Counter-example Traces}, 
    for any sequence $(tr_1)(tr_2)...(tr_n)\in \EX(\lm(x\cat \sigma), \phi\osuf\psi)$ with $tr_1 = \lm(x\cat \sigma) = \lm(x)\lm(\sigma)$, 
    we have $(tr_2)...(tr_n)\in \EX(\lm(\sigma), \psi\vee (\phi\wedge \phi\osuf\psi))$ with $tr_2 = \lm(\sigma)$ and $tr_1 =\lm(x)\cdot tr_2$. 
    By that $\EX(\lm(x\cat \sigma), \phi\osuf\psi)\neq \emptyset$, $\min(\EX(\lm(x\cat \sigma), \phi\osuf\psi)) > \min(\EX(\lm(\sigma), \psi\vee (\phi\wedge \phi\osuf\psi)))$. 
    \fi

    \ifx
    we have $\EX(\lm(x\cat \sigma), \phi\osuf \psi) = \{(\lm(x)\lm(\sigma))\}$, and $\EX(\lm(\sigma), \psi\vee (\phi\wedge \phi\osuf \psi)) \supseteq \{\lm(\sigma)\}$, since it must includes $\EX(\lm(\sigma), \phi\osuf \psi)$, whose value is $\{(\lm(\sigma))\}$.  
    Since $|(\lm(\sigma))| < |(\lm(x)\lm(\sigma))|$ (by Definition~\ref{def:Execution Sequences}), 
    $\min(\EX(\lm(x\cat \sigma), \phi\osuf \psi)) > \min(\EX(\lm(\sigma), \psi\vee (\phi\wedge \phi\osuf \psi)))$. 
    \fi
    
\ifx
    \textbf{Case for rule $(\wedge R)$: }
    If from node $\nu$ rule $(\wedge R)$ is applied, let $\tau=\sigma : \phi\wedge \psi$. 
    By Definition~\ref{def:Minimum Counter-example Traces}, $\EX(\lm(\sigma), \phi\wedge \psi) = \EX(\lm(\sigma), \phi)\cup \EX(\lm(\sigma), \psi)$.  
    Assume the shortest sequence of $\EX(\lm(\sigma), \phi\wedge \psi)$ belongs to $\EX(\lm(\sigma), \phi)$, choose $\tau' = \sigma : \phi$ in $\nu'$ as the replacement of the target formula $\tau$. Clearly $\min(\EX(\lm(\sigma), \phi\wedge \psi)) = \min(\EX(\lm(\sigma), \phi))$. 
    And by that $\EX(\lm(\sigma), \phi)\neq\emptyset$, $\lm\not\models \sigma : \phi$, which means that $\nu'$ is invalid. 

    \textbf{Case for rule $(\vee R)$: }
    If from node $\nu$ rule $(\vee R)$ is applied, $\nu$ has the only premise $\nu'$. 
    Let $\tau=\sigma : \phi\vee \psi$.
    Then $\tau'$ can be either $\sigma : \phi$ or $\sigma : \psi$. 
    Assume $\tau' = \sigma : \phi$ (the other case is similar). 
    By Definition~\ref{def:Minimum Counter-example Traces}, we can obtain that $\EX(\lm(\sigma), \phi\vee \psi) = \EX(\lm(\sigma), \phi)\cap \EX(\lm(\sigma), \psi)$. 
    So the shortest sequence of $\EX(\lm(\sigma), \phi\vee \psi)$ is also in $\EX(\lm(\sigma), \phi)$. Considering that $\EX(\lm(\sigma), \phi\vee \psi)\subseteq \EX(\lm(\sigma), \phi)$ and $\EX(\lm(\sigma), \phi\vee \psi)\neq\emptyset$ (the validity of $\nu$), we have $\min(\EX(\lm(\sigma), \phi\vee \psi)) \ge \min(\EX(\lm(\sigma), \phi))$. 
    And by that $\EX(\lm(\sigma), \phi)\neq\emptyset$, $\lm\not\models \sigma : \phi$, which means that $\nu'$ is invalid. 

    \textbf{Case for rule $(\to R)$: }
    If from node $\nu$ rule $(\to R)$ is applied, $\nu$ has the only premise $\nu'$. 
    Let $\tau=\sigma : \phi\to \psi$.
    Let $\tau' = \sigma : \psi$. 
    By Definition~\ref{def:Minimum Counter-example Traces}, we obtain that $\EX(\lm(\sigma), \phi\to\psi) = \EXA(\lm(\sigma), \phi)\cap\EX(\lm(\sigma), \psi)$. 
    So the shortest sequence of $\EX(\lm(\sigma), \phi\to\psi)$ is also in $\EX(\lm(\sigma), \psi)$. 
    By that $\EX(\lm(\sigma), \phi\to\psi)\subseteq \EX(\lm(\sigma), \psi)$ and $\EX(\lm(\sigma), \phi\to\psi)\neq \emptyset$, $\min(\EX(\lm(\sigma), \phi\to\psi))\ge \min(\EX(\lm(\sigma))$. 
    Moreover, by that $\EX(\lm(\sigma), \psi)\neq\emptyset$, $\lm\not\models \sigma : \psi$. So $\nu'$ is invalid. 
\fi
 \end{proof}

\subsection{Completeness Proofs of \GiiiPPLcyc}
\label{section:Completeness Proofs of GiiiPPLcyc}

\ifx
\begin{lemma}
\label{lemma:lemma for proving rule (inv)}
For any $X\in \LVar$ and formula $\phi$, sequents $X: \phi\Rightarrow X : [\true?]\phi$ and $X : [\true?]\phi\Rightarrow X : \phi$ are provable in \GiiiPPLcyc. 
\end{lemma}
\fi

\ifx
\textbf{Content of Lemma~\ref{lemma:lemma for proving rule (inv)}}:
For any $X\in \LVar$ and formula $\phi$, sequents $X: \phi\Rightarrow X : [\true?]\phi$ and $X : [\true?]\phi\Rightarrow X : \phi$ are provable in \GiiiPPLcyc. 
\fi

Lemma~\ref{lemma:lemma for proving rule (inv)} is subsequently used in the proof of Lemma~\ref{lemma:inv and con rules} in Section~\ref{section:About the Completeness of GiiiFOPLcyc}. 

\begin{proof}[Proof of Lemma~\ref{lemma:lemma for proving rule (inv)}]
    We only prove $X: \phi\Rightarrow X : [\true?]\phi$. 
    The other direction is similar. 

Without loss of generality, let $X = Y\cat x$. 
By $([\phi?]R)$ and $(wk)$, we need
$$
(A2)\ \ Y\cat x : \phi\Rightarrow Y \cat x\cat x : \phi
$$
We prove $(A2)$ by induction on the structure of formula $\phi$. 

The base step is when $\phi$ is an atomic formula. Then $(A2)$ can be trivially derived by applying $(p)$ and $(ax)$. 

For the inductive step, consider $\phi = \ofst \psi$ and $\phi = \phi_1\osuf \phi_2$. 
The case for $\phi = \ofst \psi$ is similar to the base step, and we omit it here. 

If $\phi = \phi_1\osuf\phi_2$, let $Y = y\cat Z$. Then from (A2), by applying $(\osuf R)$ and $(\osuf L)$, 
we need to prove: 
\begin{enumerate}
\item[(A2.1)] $Z\cat x : \phi_2\Rightarrow Z\cat x\cat x : \phi_2, Z\cat x\cat x : \phi_1$.
\item[(A2.2)] $Z\cat x : \phi_2\Rightarrow Z\cat x\cat x : \phi_2, Z\cat x\cat x : \phi_1\osuf \phi_2$.
\item[(A2.3)] $Z\cat x : \phi_1, Z\cat x : \phi_1\osuf \phi_2\Rightarrow Z\cat x\cat x : \phi_2, Z\cat x\cat x : \phi_1$
\item[(A2.4)] $Z\cat x : \phi_1, Z\cat x : \phi_1\osuf \phi_2\Rightarrow Z\cat x\cat x : \phi_2, Z\cat x\cat x : \phi_1\osuf \phi_2$
\end{enumerate}
For $(A2.1)$ and $(A2.2)$, by $(wk)$ on the right, we obtain
    $$
    (A2.21)\ \ Z\cat x : \phi_2\Rightarrow Z\cat x\cat x : \phi_2.
    $$
By first replacing $Z$ with $Y$ (using $(\Sub)$) and then applying the induction hypothesis to $\phi_2$, (A2.21) is provable. 

$(A2.3)$ is similar. By $(wk)$ we obtain 
$$
(A2.31)\ \ Z\cat x : \phi_1\Rightarrow Z\cat x\cat x : \phi_1, 
$$
which can be proved by first replacing $Z$ with $Y$ and then applying the induction hypothesis to $\phi_1$. 

For $(A2.4)$, after $(wk)$ on both sides, we obtain
$$
(A2.41)\ \ Z\cat x : \phi_1\osuf\phi_2\Rightarrow Z\cat x\cat x : \phi_1\osuf\phi_2. 
$$
(A2.41) is exactly (A2) after variable replacement. Moreover, from (A2) to (A2.41) the derivation path is progressing since from A2 to A2.4, rule $(\osuf R)$ is applied with $(Y \cat x\cat x : \phi, Z\cat x\cat x : \phi_1\osuf \phi_2)$ the progressing step. 

\ifx
\begin{enumerate}
\item[(A2.1)] $Z\cat x : \phi_2\Rightarrow Z\cat x\cat x : \phi_2, Z\cat x\cat x : \phi_1\wedge \phi_1\osuf \phi_2$. 

\item[(A2.2)] $Z\cat x : \phi_1\wedge \phi_1\osuf \phi_2\Rightarrow Z\cat x\cat x : \phi_2, Z\cat x\cat x : \phi_1\wedge \phi_1\osuf \phi_2$.  
\end{enumerate}

For $(A2.2)$, by $(\wedge R)$ on $Z\cat x\cat x : \phi_1\wedge \phi_1\osuf \phi_2$, (A2.2) can be split into 
\begin{enumerate}
\item[(A2.21)] $Z\cat x : \phi_1\wedge \phi_1\osuf \phi_2\Rightarrow Z\cat x\cat x : \phi_2, Z\cat x\cat x : \phi_1$, which can be easily proved by induction hypothesis on $\phi_1$ after $(wk)$ and $(\wedge L)$. 
\item[(A2.22)] $Z\cat x : \phi_1\wedge \phi_1\osuf \phi_2\Rightarrow Z\cat x\cat x : \phi_1, Z\cat x\cat x : \phi_1\osuf\phi_2$, from which, after $(\wedge L)$ and $(wk)$ on both sides, we obtain
$$
(A2.221)\ \ Z\cat x : \phi_1\osuf\phi_2\Rightarrow Z\cat x\cat x : \phi_1\osuf\phi_2. 
$$
(A2.221) is exactly (A2) after variable replacement. Moreover, from (A2) to (A2.221) the derivation path is progressing since rule $(\osuf)$ is used. 

\end{enumerate}
\fi

\end{proof}

Below, we show that each rule specific to process logic in \PPL\ (Table~\ref{table:ppl-proof-system}) is provable in \GiiiPPLcyc. We first describe each rule as a proposition and then prove it (Propositions~\ref{prop:Rule (i)}--\ref{prop:Rule (xv)}).

\begin{prop}[Rule (i)]
    \label{prop:Rule (i)}
Rule (i):
 $\ofst(\phi \vee \psi) \leftrightarrow \ofst \phi \vee \ofst \psi$
 is derived in \GiiiPPLcyc. 
\end{prop}

\begin{proof}
Let $X\in \TVar$. 
We only prove 
$$
(i.a)\ \ X : \ofst(\phi\vee\psi)\Rightarrow X : \ofst\phi\vee\ofst \psi.
$$
The other direction is similar. 

Without loss of generality, let $X=x\cat Y$. 
Then by $(\ofst)$, we need 
$$
(i.a1)\ \ x : \phi\vee\psi \Rightarrow x : \phi, x : \psi, 
$$
which is proved trivially by $(\vee L)$ and $(ax)$. 
\end{proof}

\begin{prop}[Rule (ii)]
Rule (ii):
$\ofst \neg \phi\leftrightarrow \neg\ofst \phi$
 is derived in \GiiiPPLcyc. 
\end{prop}

\begin{proof}
Let $X\in \TVar$. 
We only prove 
$$
(ii.a)\ \ X : \ofst \neg \phi\Rightarrow X : \neg\ofst \phi. 
$$
The other direction is similar. 

Without loss of generality, let $X=x\cat Y$. 
By $(\ofst)$ and $(\neg R)$ we obtain
$$
x : \neg\phi, x : \phi\Rightarrow \cdot, 
$$
which is trivially proved by $(ax)$. 
\end{proof}

\begin{prop}[Rule (iii)]
\label{prop:Rule (iii)}
Rule (iii):
$(\phi\osuf \psi)\vee (\phi\osuf \chi)\leftrightarrow \phi\osuf (\psi\vee \chi)$
 is derived in \GiiiPPLcyc. 
\end{prop}

\begin{proof}
Let $X\in \TVar$. 
We need to prove 
\begin{enumerate}
\item[(iii.a)] $X : (\phi\osuf \psi)\vee (\phi\osuf \chi)\Rightarrow X : \phi\osuf (\psi\vee \chi)$.
\item[(iii.b)]  $X : \phi\osuf (\psi\vee \chi)\Rightarrow X : (\phi\osuf \psi)\vee (\phi\osuf \chi)$.
\end{enumerate}
Part (iii.a) is proved in Section~\ref{section:Completeness of GiiiPPLcyc}; here, we prove only (iii.b). 
That is, it remains to prove 
$$
(iii.b1)\ \ X : \phi\osuf (\psi\vee \chi)\Rightarrow X : (\phi\osuf \psi), X : (\phi\osuf \chi)
$$
by $(\vee R)$. 
Without loss of generality, let $X = x\cat Y$. 
By applying $(\osuf L)$ on the left, we split (iii.b1) into
\begin{enumerate}
\item[$(iii.b11)$] $Y : \psi\vee \chi\Rightarrow x\cat Y : (\phi\osuf \psi), x\cat Y : (\phi\osuf \chi)$.
\item[$(iii.b12)$]  $Y : \phi, Y : \phi\osuf (\psi\vee \chi)\Rightarrow x\cat Y : (\phi\osuf \psi), x\cat Y : (\phi\osuf \chi)$.
\end{enumerate}
For (iii.b11), by applying $(\osuf R)$ on both $x\cat Y : (\phi\osuf \psi)$ and $x\cat Y : (\phi\osuf \chi)$, we obtain
\begin{enumerate}
\item[(iii.b111)] $Y : \psi\vee \chi\Rightarrow Y : \psi, Y : \phi, Y : \chi, Y : \phi$.
\item[(iii.b112)]  $Y : \psi\vee \chi\Rightarrow Y : \psi, Y : \phi, Y : \chi, Y : \phi\osuf \chi$.
\item[(iii.b113)] $Y : \psi\vee \chi\Rightarrow Y : \psi, Y : \phi\osuf\psi, Y : \chi, Y : \phi$.
\item[(iii.b114)]  $Y : \psi\vee \chi\Rightarrow Y : \psi, Y : \phi\osuf\psi, Y : \chi, Y : \phi\osuf \chi$. 
\end{enumerate}

Since the right side of each of (iii.b111)--(iii.b114) contains both $Y : \psi$ and $Y:\chi$, these sequents can be proved by $(ax)$ for the different cases on the left side. 

For (iii.b12), by  applying $(\osuf R)$ on both $x\cat Y : (\phi\osuf \psi)$ and $x\cat Y : (\phi\osuf \chi)$, we obtain
\begin{enumerate}
\item[(iii.b121)] $Y : \phi, Y : \phi\osuf (\psi\vee \chi)\Rightarrow Y : \psi, Y : \phi, Y : \chi, Y : \phi$.
\item[(iii.b122)]  $Y : \phi, Y : \phi\osuf (\psi\vee \chi)\Rightarrow Y : \psi, Y : \phi, Y : \chi, Y : \phi\osuf \chi$.
\item[(iii.b123)] $Y : \phi, Y : \phi\osuf (\psi\vee \chi)\Rightarrow Y : \psi, Y : \phi\osuf\psi, Y : \chi, Y : \phi$.
\item[(iii.b124)]  $Y : \phi, Y : \phi\osuf (\psi\vee \chi)\Rightarrow Y : \psi, Y : \phi\osuf\psi, Y : \chi, Y : \phi\osuf \chi$. 
\end{enumerate}
Each of (iii.b121)--(iii.b123) can be proved directly by applying $(ax)$ because both sides contain $Y : \phi$. 

For (iii.b124), by (wk), we have 
$$
(iii.b1241)\ \ Y : \phi\osuf (\psi\vee \chi)\Rightarrow Y : \phi\osuf\psi, Y : \phi\osuf \chi, 
$$
which is exactly (iii.b1) (replacing $Y$ with $X$). 
The derivation from (iii.b1) to (iii.b1241) is progressing since from (iii.b12) to (iii.b124) rule $(\osuf R)$ is applied with $(x\cat Y : (\phi\osuf \chi), Y : \phi\osuf \chi)$ the progressing step. 
This makes the whole preproof a cyclic proof.

\ifx
$--------------------------------$

$$
(iii.b111)\ \ Y : \psi\vee \chi\Rightarrow Y : \psi, Y : \phi\wedge \phi\osuf\psi, Y : \chi, Y : \phi\wedge \phi\osuf \chi. 
$$
Observing that since the right side contains both $Y : \psi$ and $Y:\chi$, (iii.b111) can be proved by (ax) for different cases on the left side. 

For (iii.b12), by  applying $(\osuf)$ on the right, we obtain
$$
(iii.b121)\ \ Y : \phi\wedge \phi\osuf (\psi\vee \chi)\Rightarrow Y : \psi, Y : \phi\wedge \phi\osuf\psi, Y : \chi, Y : \phi\wedge \phi\osuf \chi. 
$$
By $(\wedge R)$ on $Y : \phi\wedge \phi\osuf \psi$ and $Y : \phi\wedge \phi\osuf\chi$, we split the right side into several branches. 
Since both sides contain $Y : \phi$, all branches but the following one can be proved by directly applying $(ax)$: 
$$
(iii.b1211)\ \ Y : \phi\wedge \phi\osuf (\psi\vee \chi)\Rightarrow Y : \psi, Y : \phi\osuf\psi, Y : \chi, Y : \phi\osuf \chi.
$$
From (iii.b1211), by $(wk)$, we have
$$
(iii.b12111)\ \ Y : \phi\osuf (\psi\vee \chi)\Rightarrow Y : \phi\osuf\psi, Y : \phi\osuf \chi, 
$$
which is exactly (iii.b1) (replacing $Y$ with $X$). 
The derivation from (iii.b1) to (iii.b12111) is progressing since rule $(\osuf)$ is applied. 
This makes the whole preproof a cyclic proof. 
\fi

\end{proof}

\begin{prop}[Rule (iv)]
Rule (iv): 
$
(\phi\osuf \psi)\wedge (\chi\osuf \omega)\leftrightarrow (\phi\wedge \chi)\osuf (\psi\wedge \omega)\vee (\phi\wedge \chi)\osuf(\phi\wedge \omega\wedge \phi\osuf \psi)\vee (\phi\wedge \chi)\osuf (\psi\wedge \chi\wedge \chi\osuf\omega)
$
is derived in \GiiiPPLcyc. 
\end{prop}

\begin{proof}
Let $X\in \TVar$. We prove that 
\begin{enumerate}
\item[$(iv.a)$] $X : (\phi\osuf \psi)\wedge (\chi\osuf \omega)\Rightarrow X : (\phi\wedge \chi)\osuf (\psi\wedge \omega)\vee (\phi\wedge \chi)\osuf(\phi\wedge \omega\wedge \phi\osuf \psi)\vee (\phi\wedge \chi)\osuf (\psi\wedge \chi\wedge \chi\osuf\omega)$.
\item[$(iv.b)$] $X : (\phi\wedge \chi)\osuf (\psi\wedge \omega)\vee (\phi\wedge \chi)\osuf(\phi\wedge \omega\wedge \phi\osuf \psi)\vee (\phi\wedge \chi)\osuf (\psi\wedge \chi\wedge \chi\osuf\omega)\Rightarrow X : (\phi\osuf \psi)\wedge (\chi\osuf \omega)$. 
\end{enumerate}

\textbf{Case for proving $(iv.a)$}: 
Let 
$$
\begin{aligned}
A \dddef&\ (\phi\wedge \chi)\osuf (\psi\wedge \omega),\\
B \dddef&\ (\phi\wedge \chi)\osuf(\phi\wedge \omega\wedge \phi\osuf \psi),\\
C \dddef&\ (\phi\wedge \chi)\osuf (\psi\wedge \chi\wedge \chi\osuf\omega). 
\end{aligned}
$$
By applying the rules $(\wedge L)$ and $(\vee R)$ on all formulas on both sides respectively, we obtain the sequent:
$$
(iv.a1)\ \ X : \phi\osuf \psi, X : \chi\osuf\omega\Rightarrow X : A , X : B , X : C. 
$$
Let $X = y\cat Y$. By applying the rules $(\osuf R)$ and $(\osuf L)$ on both sides, we obtain the following 32 cases: 
\begin{center}
\begin{tabular}{l l l l}
$(iv.a1.1)\ W_1\Rightarrow S_1$
& 
$(iv.a1.9)\ W_2\Rightarrow S_1$
& 
$(iv.a1.17)\ W_3\Rightarrow S_1$
& 
$(iv.a1.25)\ W_4\Rightarrow S_1$
\\
$(iv.a1.2)\ W_1\Rightarrow S_2$
& 
$(iv.a1.10)\ W_2\Rightarrow S_2$
& 
$(iv.a1.18)\ W_3\Rightarrow S_2$
& 
$(iv.a1.26)\ W_4\Rightarrow S_2$
\\
$(iv.a1.3)\ W_1\Rightarrow S_3$
& 
$(iv.a1.11)\ W_2\Rightarrow S_3$
& 
$(iv.a1.19)\ W_3\Rightarrow S_3$
& 
$(iv.a1.27)\ W_4\Rightarrow S_3$
\\
$(iv.a1.4)\ W_1\Rightarrow S_4$
& 
$(iv.a1.12)\ W_2\Rightarrow S_4$
& 
$(iv.a1.20)\ W_3\Rightarrow S_4$
& 
$(iv.a1.28)\ W_4\Rightarrow S_4$
\\
$(iv.a1.5)\ W_1\Rightarrow S_5$
& 
$(iv.a1.13)\ W_2\Rightarrow S_5$
& 
$(iv.a1.21)\ W_3\Rightarrow S_5$
& 
$(iv.a1.29)\ W_4\Rightarrow S_5$
\\
$(iv.a1.6)\ W_1\Rightarrow S_6$
& 
$(iv.a1.14)\ W_2\Rightarrow S_6$
& 
$(iv.a1.22)\ W_3\Rightarrow S_6$
& 
$(iv.a1.30)\ W_4\Rightarrow S_6$
\\
$(iv.a1.7)\ W_1\Rightarrow S_7$
& 
$(iv.a1.15)\ W_2\Rightarrow S_7$
& 
$(iv.a1.23)\ W_3\Rightarrow S_7$
& 
$(iv.a1.31)\ W_4\Rightarrow S_7$
\\
$(iv.a1.8)\ W_1\Rightarrow S_8$
& 
$(iv.a1.16)\ W_2\Rightarrow S_8$
& 
$(iv.a1.24)\ W_3\Rightarrow S_8$
& 
$(iv.a1.32)\ W_4\Rightarrow S_8$
\end{tabular}
\end{center}
where 
$$
\begin{aligned}
W_1 \dddef&\ Y : \psi, Y : \omega, \\
W_2 \dddef&\ Y : \phi, Y : \phi\osuf \psi, Y : \omega, \\
W_3 \dddef&\ Y : \psi, Y : \chi, Y : \chi\osuf \omega, \\
W_4 \dddef&\ Y : \phi, Y : \phi\osuf \psi, Y : \chi, Y : \chi\osuf \omega.
\end{aligned}
$$
$S_1,\ldots,S_8$ are defined as follows: 
$$
\begin{aligned}
S_1\dddef&\ A', A_1, B', B_1, C', C_1, \\
S_2\dddef&\ A', A_1, B', B_1, C', C_2, \\
S_3\dddef&\ A', A_1, B', B_2, C', C_1, \\
S_4\dddef&\ A', A_1, B', B_2, C', C_2, \\
S_5\dddef&\ A', A_2, B', B_1, C', C_1, \\
S_6\dddef&\ A', A_2, B', B_1, C', C_2, \\
S_7\dddef&\ A', A_2, B', B_2, C', C_1, \\
S_8\dddef&\ A', A_2, B', B_2, C', C_2, \\
\end{aligned}
$$
where $X'$ and $X_i$ are obtained from $X$ ($i\in \{1,2\}$, $X\in \{A, B, C\}$). The symbols are defined as follows: 
\begin{center}
\begin{tabular}{l l l}
    $A' \dddef Y : \psi\wedge \omega$
    &
    $B' \dddef Y : \phi\wedge \omega\wedge \phi\osuf \psi$
    &
    $C' \dddef Y : \psi\wedge \chi\wedge \chi\osuf\omega$
    \\
    $A_1\dddef Y : \phi\wedge \chi$
    &
    $B_1\dddef Y : \phi\wedge \chi$
    &
    $C_1\dddef Y : \phi\wedge \chi$
    \\
    $A_2\dddef Y : (\phi\wedge \chi)\osuf (\psi\wedge \omega)$
    &
    $B_2\dddef Y : (\phi\wedge \chi)\osuf(\phi\wedge \omega\wedge \phi\osuf \psi)$
    &
    $C_2\dddef Y : (\phi\wedge \chi)\osuf (\psi\wedge \chi\wedge \chi\osuf\omega)$. 
\end{tabular}
\end{center}

The proofs of the cases from (iv.a1.1) to (iv.a1.24) are trivial and are all terminal branches (without buds), since in each of these cases we can find a formula on the left side that directly matches a subformula of a right-side disjunct among $A', B', C'$, allowing the branch to be closed by $(\wedge R)$ and $(ax)$ without any application of the rules $(\osuf R)$ and $(\osuf L)$.
For example, for 
$$
(iv.a1.16)\ \ Y : \phi, Y : \phi\osuf \psi, Y : \omega\Rightarrow A', A_2, B', B_2, C', C_2, 
$$
by $(\wedge R)$ on $B'$, we can obtain:
\begin{enumerate}
\item[(iv.a1.16.1)] $Y : \phi, Y : \phi\osuf \psi, Y : \omega\Rightarrow A', A_2, Y : \phi, C', C_2$.
\item[(iv.a1.16.2)] $Y : \phi, Y : \phi\osuf \psi, Y : \omega\Rightarrow A', A_2, Y : \omega, C', C_2$.
\item[(iv.a1.16.3)] $Y : \phi, Y : \phi\osuf \psi, Y : \omega\Rightarrow A', A_2, Y : \phi\osuf \psi, C', C_2$.
\end{enumerate}
Each of (iv.a1.16.1) - (iv.a1.16.3) can be proved directly by $(ax)$. 

In the cases from (iv.a1.25) to (iv.a1.31), the left side is $W_4$, which contains $Y : \phi$ and $Y : \chi$, and the right side contains one of $A_1$, $B_1$, or $C_1$, each equal to $Y : \phi\wedge\chi$. Therefore, these cases can be proved by $(\wedge R)$ and $(ax)$. 
For example, for (iv.a.31), by $(\wedge R)$ on $C_1$ on the right side, we obtain:
\begin{enumerate}
\item[(iv.a.31.1)] $Y : \phi, Y : \phi\osuf \psi, Y : \chi, Y : \chi\osuf \omega\Rightarrow A', A_2, B', B_2, C', Y : \phi$,
\item[(iv.a.31.2)] $Y : \phi, Y : \phi\osuf \psi, Y : \chi, Y : \chi\osuf \omega\Rightarrow A', A_2, B', B_2, C', Y : \chi$. 
\end{enumerate}
Both cases can be proved directly by $(ax)$. 

Lastly, for the case (iv.a1.32), since it has $A_2, B_2$ and $C_2$, by applying (wk) we obtain:
$$
(iv.a1.32.1)\ Y : \phi\osuf\psi, Y : \chi\osuf \omega\Rightarrow A_2, B_2, C_2. 
$$
After replacing $Y$ with $X$ (i.e., by applying $(\Sub)$), (iv.a1.32.1) is exactly (iv.a1). 
Note that along the derivation path from (iv.a1) to (iv.a1.32.1), rule $(\osuf R)$ is applied, which guarantees the whole preproof is a cyclic proof. 

\ifx
$----------------------------------$
$$
\begin{aligned}
(iv.a11)\ Y : \psi \vee (\phi\wedge \phi\osuf \psi), Y : \omega\vee (\chi\wedge \chi\osuf \omega)\Rightarrow 
Y : A_1, Y : A_2, Y : B_1, Y : B_2, Y : C_1, Y : C_2,
\end{aligned}
$$
where let 
$$
\begin{aligned}
A_1 \dddef&\ \psi\wedge \omega, &&\ \mbox{obtained from $A$}\\
A_2 \dddef&\ \phi\wedge \chi\wedge (\phi\wedge \chi)\osuf (\psi\wedge \omega), &&\ \mbox{obtained from $A$}\\
B_1 \dddef&\ \phi\wedge \omega\wedge \phi\osuf \psi, &&\ \mbox{obtained from $B$}\\
B_2 \dddef&\ (\phi\wedge \chi)\wedge (\phi\wedge \chi)\osuf(\phi\wedge \omega\wedge \phi\osuf \psi), &&\ \mbox{obtained from $B$}\\
C_1 \dddef&\ \psi\wedge \chi\wedge \chi\osuf\omega, &&\ \mbox{obtained from $C$}\\
C_2 \dddef&\ (\phi\wedge \chi)\wedge (\phi\wedge \chi)\osuf (\psi\wedge \chi\wedge \chi\osuf\omega), &&\ \mbox{obtained from $C$}\\
\end{aligned}
$$

\ifx
\begin{itemize}
\item $A_1 \dddef Y : \psi\wedge \omega$, $A_2 \dddef Y : \phi\wedge \chi\wedge (\phi\wedge \chi)\osuf (\psi\wedge \omega)$, obtained from $A$.
\item $B_1 = Y : \phi\wedge \omega\wedge \phi\osuf \psi$, $B_2 = Y : (\phi\wedge \chi)\wedge (\phi\wedge \chi)\osuf(\phi\wedge \omega\wedge \phi\osuf \psi)$, obtained from $B$.
\item $C_1=Y : \psi\wedge \chi\wedge \chi\osuf\omega$, $C_2 = Y : (\phi\wedge \chi)\wedge (\phi\wedge \chi)\osuf (\psi\wedge \chi\wedge \chi\osuf\omega)$, obtained from $C$.
\end{itemize}
\fi

By the rule $(\vee L)$, we split (iv.a11) into 4 cases: 
\begin{enumerate}
    \item[(iv.a111)] $Y : \psi, Y : \omega\Rightarrow Y : A_1, Y : A_2, Y : B_1, Y : B_2, Y : C_1, Y : C_2$.
    \item[(iv.a112)] $Y : (\phi\wedge \phi\osuf \psi), Y : \omega\Rightarrow Y : A_1, Y : A_2, Y : B_1, Y : B_2, Y : C_1, Y : C_2$.
    \item[(iv.a113)] $Y : \psi, Y : (\chi\wedge \chi\osuf \omega)\Rightarrow Y : A_1, Y : A_2, Y : B_1, Y : B_2, Y : C_1, Y : C_2$.
    \item[(iv.a114)] $Y : (\phi\wedge \phi\osuf \psi), Y : (\chi\wedge \chi\osuf \omega)\Rightarrow Y : A_1, Y : A_2, Y : B_1, Y : B_2, Y : C_1, Y : C_2$.
\end{enumerate}
The proofs of cases (iv.a111), (iv.a112), (iv.a113) are trivial and are all terminal branches (without buds), since in each of these cases a formula on the left side directly matches one of the right-side disjuncts $Y : A_1$, $Y : A_2$, $Y : B_1$, $Y : B_2$, $Y : C_1$, or $Y : C_2$ up to logical connectives, allowing the branch to be closed by $(\vee R)$, $(\wedge L)$, and $(ax)$ without any application of rule $(\osuf)$.

For case (iv.a114), by applying the rule $(\wedge R)$ on $Y : A_2$, $Y : B_2$, $Y : C_2$ respectively, we obtain two parts of proof branches, named (iv.a1141) and (iv.a1142) respectively. 
Part (iv.a1141) contains 3 sequents of the form: 
$$
Y : (\phi\wedge \phi\osuf \psi), Y : (\chi\wedge \chi\osuf \omega)\Rightarrow ..., Y : (\phi\wedge \chi),..., 
$$
which can be easily proved by applying the rule $(\wedge L)$ on the left and the rule $(\textit{ax})$. 
Part (iv.a1142) is the sequent of the form:  
$$
Y : (\phi\wedge \phi\osuf \psi), Y : (\chi\wedge \chi\osuf \omega)\Rightarrow ..., Y : A, Y : B, Y : C, ....
$$
By applying the rule $(\wedge L)$, the rule $(\textit{wk})$ on both sides for several times, we can get 
$$
(iv.a11421)\ \ Y : \phi\osuf \psi, Y : \chi\osuf\omega\Rightarrow Y : A, Y : B, Y : C, 
$$
which is exactly (iv.a1). 
Note that along the derivation path from (iv.a1) to (iv.a11421), rule $(\osuf)$ is applied, which guarantees the whole preproof is a cyclic proof. 
\fi

\textbf{Case for proving $(iv.b)$}: 
By applying the rules $(\vee L)$ and $(\wedge R)$ on the left and right sides of the sequent respectively, we split (iv.b) into 6 cases as follows:
\begin{enumerate}
\item[(iv.b1)] $X : (\phi\wedge \chi)\osuf (\psi\wedge \omega)\Rightarrow X : (\phi\osuf \psi)$.
\item[(iv.b2)] $X : (\phi\wedge \chi)\osuf(\phi\wedge \omega\wedge \phi\osuf \psi)\Rightarrow X : (\phi\osuf \psi)$.
\item[(iv.b3)] $X : (\phi\wedge \chi)\osuf (\psi\wedge \chi\wedge \chi\osuf\omega)\Rightarrow X : (\phi\osuf \psi)$.
\item[(iv.b4)] $X : (\phi\wedge \chi)\osuf (\psi\wedge \omega)\Rightarrow X : (\chi\osuf \omega)$.
\item[(iv.b5)] $X : (\phi\wedge \chi)\osuf(\phi\wedge \omega\wedge \phi\osuf \psi)\Rightarrow X : (\chi\osuf \omega)$.
\item[(iv.b6)] $X : (\phi\wedge \chi)\osuf (\psi\wedge \chi\wedge \chi\osuf\omega)\Rightarrow X : (\chi\osuf \omega)$.
\end{enumerate}
Without loss of generality, let $X = y\cat Y$. 

Case (iv.b1): 
By applying the rules $(\osuf R)$ and $(\osuf L)$ on both sides, we can have
\begin{enumerate}
\item[(iv.b11)] $Y : \psi\wedge \omega\Rightarrow Y : \psi, Y : \phi$. 
\item[(iv.b12)] $Y : \psi\wedge \omega\Rightarrow Y : \psi, Y : \phi\osuf \psi$.
\item[(iv.b13)] $Y : \phi\wedge \chi, Y : (\phi\wedge \chi)\osuf(\psi\wedge \omega)\Rightarrow Y : \psi, Y : \phi$
\item[(iv.b14)] $Y : \phi\wedge \chi, Y : (\phi\wedge \chi)\osuf(\psi\wedge \omega)\Rightarrow Y : \psi, Y : \phi\osuf\psi$
\end{enumerate}
The cases (iv.b11)--(iv.b13) can be proved trivially by $(\wedge L)$ and $(ax)$. 

For the case (iv.b14), by $(wk)$, we have
$$
(iv.b141)\ Y : (\phi\wedge\chi)\osuf (\psi\wedge \omega) \Rightarrow Y : \phi\osuf \psi, 
$$
which is exactly (iv.b1), forming a back-link.  
The derivation path from (iv.b1) to (iv.b141) is progressing since rule $(\osuf R)$ is applied from (iv.b1) to (iv.b14) and $(X : (\phi\osuf\psi), Y : \phi\osuf\psi)$ is the progressing step.  

Case (iv.b2):
By applying the rules $(\osuf R)$ and $(\osuf L)$ on both sides, we can have
\begin{enumerate}
\item[(iv.b21)] $Y : \phi\wedge \omega\wedge \phi\osuf\psi\Rightarrow Y : \psi, Y : \phi$.
\item[(iv.b22)] $Y : \phi\wedge \omega\wedge \phi\osuf\psi\Rightarrow Y : \psi, Y : \phi\osuf \psi$.
\item[(iv.b23)] $Y : \phi\wedge \chi, Y : (\phi\wedge \chi)\osuf(\phi\wedge\omega\wedge\phi\osuf\psi)\Rightarrow Y : \psi, Y : \phi$.
\item[(iv.b24)] $Y : \phi\wedge \chi, Y : (\phi\wedge \chi)\osuf(\phi\wedge\omega\wedge\phi\osuf\psi)\Rightarrow Y : \psi, Y : \phi\osuf\psi$.
\end{enumerate}
The cases for (iv.b21) - (iv.b23) can be proved by $(\wedge L)$ and $(ax)$ easily as both sides contain the same formula. 

For the case (iv.b24), by (wk) we have

$$
(iv.b241)\ \ Y : (\phi\wedge\chi)\osuf (\phi\wedge \omega\wedge \phi\osuf \psi) \Rightarrow Y : \phi\osuf \psi, 
$$
which is exactly (iv.b2), forming a back-link.
The derivation path from (iv.b2) to (iv.b241) is progressing since rule $(\osuf R)$ is applied from (iv.b2) to (iv.b24), with $(X : (\phi\osuf\psi), Y : \phi\osuf \psi)$ the progressing step.  

Case (iv.b3): 
By applying the rules $(\osuf R)$ and $(\osuf L)$, we obtain
\begin{enumerate}
\item[(iv.b31)] $Y : \psi\wedge\chi\wedge \chi\osuf\omega\Rightarrow Y : \psi, Y : \phi$.
\item[(iv.b32)] $Y : \psi\wedge\chi\wedge \chi\osuf\omega\Rightarrow Y : \psi, Y : \phi\osuf \psi$.
\item[(iv.b33)] $Y : \phi\wedge\chi, Y : (\phi\wedge \chi)\osuf (\psi\wedge\chi\wedge \chi\osuf\omega)\Rightarrow Y : \psi, Y : \phi$.
\item[(iv.b34)] $Y : \phi\wedge\chi, Y : (\phi\wedge \chi)\osuf (\psi\wedge\chi\wedge \chi\osuf\omega)\Rightarrow Y : \psi, Y : \phi\osuf \psi$.
\end{enumerate}
The cases (iv.b31)--(iv.b33) can be proved by $(\wedge L)$ and $(ax)$. 

For the case (iv.b34), by (wk) we get
$$
(iv.b341)\ \ Y : (\phi\wedge\chi)\osuf (\psi\wedge \chi\wedge \chi\osuf\omega) \Rightarrow Y : \phi\osuf \psi, 
$$
which is exactly (iv.b3), forming a back-link.
The derivation path from (iv.b3) to (iv.b341) is progressing since rule $(\osuf R)$ is applied from (iv.b3) to (iv.b34), where $(X : (\phi\osuf\psi), Y : \phi\osuf \psi)$ is the progressing step.  

Case (iv.b4): By applying the rules  $(\osuf R)$ and $(\osuf L)$, we have
\begin{enumerate}
\item[(iv.b41)] $Y : \psi\wedge\omega\Rightarrow Y : \omega, Y : \chi$.
\item[(iv.b42)] $Y : \psi\wedge\omega\Rightarrow Y : \omega, Y : \chi\osuf\omega$.
\item[(iv.b43)] $Y : \phi\wedge \chi, Y : (\phi\wedge\chi)\osuf(\psi\wedge\omega)\Rightarrow Y : \omega, Y : \chi$.
\item[(iv.b44)] $Y : \phi\wedge \chi, Y : (\phi\wedge\chi)\osuf(\psi\wedge\omega)\Rightarrow Y : \omega, Y : \chi\osuf\omega$.
\end{enumerate}
The cases (iv.b41)--(iv.b43) can be proved trivially using $(\wedge L)$ and $(ax)$. 

For the case (iv.b44), by (wk) we obtain
$$
(iv.b441)\ \ Y : (\phi\wedge\chi)\osuf (\psi\wedge \omega) \Rightarrow Y : \chi\osuf \omega, 
$$
which is exactly (iv.b4), forming a back-link.
The derivation path from (iv.b4) to (iv.b441) is progressing since rule $(\osuf R)$ is applied from (iv.b4) to (iv.b44), with $(X : (\chi\osuf\omega), Y : \chi\osuf\omega)$ the progressing step.

Case (iv.b5): By applying the rules $(\osuf R)$ and $(\osuf L)$, we have
\begin{enumerate}
\item[(iv.b51)] $Y : \phi\wedge\omega\wedge \phi\osuf\psi\Rightarrow Y : \omega, Y : \chi$. 
\item[(iv.b52)] $Y : \phi\wedge\omega\wedge \phi\osuf\psi\Rightarrow Y : \omega, Y : \chi\osuf\omega$. 
\item[(iv.b53)] $Y : (\phi\wedge\chi), Y : (\phi\wedge\chi)\osuf(\phi\wedge \omega\wedge \phi\osuf \psi)\Rightarrow Y : \omega, Y : \chi$. 
\item[(iv.b54)] $Y : (\phi\wedge\chi), Y : (\phi\wedge\chi)\osuf(\phi\wedge \omega\wedge \phi\osuf \psi)\Rightarrow Y : \omega, Y : \chi\osuf\omega$. 
\end{enumerate}
The cases (iv.b51)--(iv.b53) can be proved trivially using $(\wedge L)$ and $(ax)$. 

For the case (iv.b54), we obtain by applying $(wk)$:
$$
(iv.b541)\ \ Y : (\phi\wedge\chi)\osuf (\phi\wedge \omega\wedge \phi\osuf \psi) \Rightarrow Y : \chi\osuf \omega, 
$$
which is exactly (iv.b5), forming a back-link.
The derivation path from (iv.b5) to (iv.b541) is progressing because, from (iv.b5) to (iv.b54), the rule $(\osuf R)$ is applied, with $(X : (\chi\osuf\omega), Y : \chi\osuf\omega)$ as the progressing step. 

Case (iv.b6): By applying the rules $(\osuf R)$ and $(\osuf L)$, we can have
\begin{enumerate}
\item[(iv.b61)] $Y : \psi\wedge \chi\wedge \chi\osuf\omega\Rightarrow Y : \omega, Y : \chi$. 
\item[(iv.b62)] $Y : \psi\wedge \chi\wedge \chi\osuf\omega\Rightarrow Y : \omega, Y : \chi\osuf\omega$. 
\item[(iv.b63)] $Y : \phi\wedge \chi, Y : (\phi\wedge\chi)\osuf(\psi\wedge\chi\wedge\chi\osuf\omega)\Rightarrow Y : \omega, Y : \chi$. 
\item[(iv.b64)] $Y : \phi\wedge \chi, Y : (\phi\wedge\chi)\osuf(\psi\wedge\chi\wedge\chi\osuf\omega)\Rightarrow Y : \omega, Y : \chi\osuf\omega$. 
\end{enumerate}
The cases (iv.b61)--(iv.b63) can be proved trivially by $(\wedge L)$ and $(ax)$. 

For the case (iv.b64), by (wk), we get
$$
(iv.b641)\ \ Y : (\phi\wedge\chi)\osuf (\psi\wedge \chi\wedge \chi\osuf\omega) \Rightarrow Y : \chi\osuf \omega. 
$$
This is exactly (iv.b6), forming a back-link.
The derivation path from (iv.b6) to (iv.b641) is progressing, since rule $(\osuf R)$ is applied from (iv.b6) to (iv.b64), where $(X : (\chi\osuf \omega), Y : \chi\osuf\omega)$ is the progressing step. 

These 6 cases show that the whole derivation from (iv.b) is a cyclic proof. 
Thus we conclude the proof of (iv.b). 

\ifx
By applying the rules $(\osuf R)$ and $(\osuf L)$ on both sides, we can have 
$$
(iv.b11)\ \  Y : (\psi\wedge \omega) \vee (\phi\wedge\chi)\wedge (\phi\wedge\chi)\osuf (\psi\wedge \omega)\Rightarrow y\cat Y : (\phi\osuf \psi)
$$
From (iv.b11), we apply the rule $(\vee L)$, and obtain:
$$
\begin{aligned}
(iv.b111)\ \ & Y : (\psi\wedge \omega) \Rightarrow y\cat Y : (\phi\osuf \psi), \\
(iv.b112)\ \ & Y : (\phi\wedge\chi)\wedge (\phi\wedge\chi)\osuf (\psi\wedge \omega) \Rightarrow y\cat Y : (\phi\osuf \psi).
\end{aligned}
$$
The case for (iv.b111) is trivial (just extend $y\cat Y : (\phi\osuf \psi)$ on the right side by the rule $(\osuf)$ and can obtain the same item $Y : \psi$ as on the left side). 
For case (iv.b112), we extend $y\cat Y : (\phi\osuf\psi)$ with the rule $(\osuf)$, we obtain:
$$
\begin{aligned}
(iv.b1121)\ \ & Y : (\phi\wedge\chi), Y : (\phi\wedge\chi)\osuf (\psi\wedge \omega) \Rightarrow Y : \phi,\\
(iv.b1122)\ \ & Y : (\phi\wedge\chi), Y : (\phi\wedge\chi)\osuf (\psi\wedge \omega) \Rightarrow Y : \phi\osuf \psi, 
\end{aligned}
$$
where (iv.b1121) is obvious since both sides contain $Y:\phi$. 
From (iv.b1122), we obtain by applying the rule $(wk)$:
$$
(iv.b11221)\ \ Y : (\phi\wedge\chi)\osuf (\psi\wedge \omega) \Rightarrow Y : \phi\osuf \psi
$$
which is exactly (iv.b1), forming a back-link.  
The derivation path from (iv.b1) to (iv.b11221) is progressing since rule $(\osuf)$ is applied.  
\fi

\ifx
Case (iv.b2): Without loss of generality, let $X = y\cat Y$.
From (iv.b2), by applying the rule $(\osuf)$, we can have
$$
(iv.b21)\ \  Y : (\phi\wedge \omega\wedge \phi\osuf \psi) \vee (\phi\wedge\chi)\wedge (\phi\wedge\chi)\osuf (\phi\wedge \omega\wedge \phi\osuf \psi)\Rightarrow y\cat Y : (\phi\osuf \psi). 
$$
From (iv.b21), by rule $(\vee L)$, we obtain:
$$
\begin{aligned}
(iv.b211)\ \ & Y : (\phi\wedge \omega\wedge \phi\osuf \psi) \Rightarrow y\cat Y : (\phi\osuf \psi), \\
(iv.b212)\ \ & Y : (\phi\wedge\chi)\wedge (\phi\wedge\chi)\osuf (\phi\wedge \omega\wedge \phi\osuf \psi) \Rightarrow y\cat Y : (\phi\osuf \psi). 
\end{aligned}
$$
For case (iv.b211), by extending $y\cat Y : (\phi\osuf\psi)$ on the right side by $(\osuf)$ and $(\wedge R)$, we need to show $Y : \phi\wedge \phi\osuf\psi$, which is obtainable from the left side by $(\wedge L)$ and $(\textit{ax})$. 
For case (iv.b212), by extending $y\cat Y : (\phi\osuf\psi)$ with the rule $(\osuf)$, we obtain:
$$
\begin{aligned}
(iv.b2121)\ \ & Y : (\phi\wedge\chi), Y : (\phi\wedge\chi)\osuf (\phi\wedge \omega\wedge \phi\osuf \psi) \Rightarrow Y : \phi,\\
(iv.b2122)\ \ & Y : (\phi\wedge\chi), Y : (\phi\wedge\chi)\osuf (\phi\wedge \omega\wedge \phi\osuf \psi) \Rightarrow Y : \phi\osuf \psi, 
\end{aligned}
$$
where (iv.b2121) is obvious since both sides contain $Y:\phi$.
From (iv.b2122), we obtain by applying the rule $(wk)$:
$$
(iv.b21221)\ \ Y : (\phi\wedge\chi)\osuf (\phi\wedge \omega\wedge \phi\osuf \psi) \Rightarrow Y : \phi\osuf \psi, 
$$
which is exactly (iv.b2), forming a back-link.
The derivation path from (iv.b2) to (iv.b21221) is progressing since rule $(\osuf)$ is applied.
\fi

\ifx
Case (iv.b3): Without loss of generality, let $X = y\cat Y$.
From (iv.b3), by applying the rule $(\osuf)$, we can have
$$
(iv.b31)\ \  Y : (\psi\wedge \chi\wedge \chi\osuf\omega) \vee (\phi\wedge\chi)\wedge (\phi\wedge\chi)\osuf (\psi\wedge \chi\wedge \chi\osuf\omega)\Rightarrow y\cat Y : (\phi\osuf \psi). 
$$
From (iv.b31), we apply $(\vee L)$ and obtain:
$$
\begin{aligned}
(iv.b311)\ \ & Y : (\psi\wedge \chi\wedge \chi\osuf\omega) \Rightarrow y\cat Y : (\phi\osuf \psi), \\
(iv.b312)\ \ & Y : (\phi\wedge\chi)\wedge (\phi\wedge\chi)\osuf (\psi\wedge \chi\wedge \chi\osuf\omega) \Rightarrow y\cat Y : (\phi\osuf \psi). 
\end{aligned}
$$
The case for (iv.b311) is trivial (just extend $y\cat Y : (\phi\osuf \psi)$ on the right side by $(\osuf)$ and can obtain the same item $Y : \psi$ as on the left side).
For case (iv.b312), by extending $y\cat Y : (\phi\osuf\psi)$ with $(\osuf)$, we obtain:
$$
\begin{aligned}
(iv.b3121)\ \ & Y : (\phi\wedge\chi), Y : (\phi\wedge\chi)\osuf (\psi\wedge \chi\wedge \chi\osuf\omega) \Rightarrow Y : \phi,\\
(iv.b3122)\ \ & Y : (\phi\wedge\chi), Y : (\phi\wedge\chi)\osuf (\psi\wedge \chi\wedge \chi\osuf\omega) \Rightarrow Y : \phi\osuf \psi
\end{aligned}
$$
where (iv.b3121) is obvious since both sides contain $Y:\phi$.
From (iv.b3122), by $(wk)$:
$$
(iv.b31221)\ \ Y : (\phi\wedge\chi)\osuf (\psi\wedge \chi\wedge \chi\osuf\omega) \Rightarrow Y : \phi\osuf \psi, 
$$
which is exactly (iv.b3), forming a back-link.
The derivation path from (iv.b3) to (iv.b31221) is progressing since rule $(\osuf)$ is applied.
\fi

\ifx
Case (iv.b4): Without loss of generality, let $X = y\cat Y$.
From (iv.b4), by applying the rule $(\osuf)$, we can have
$$
(iv.b41)\ \  Y : (\psi\wedge \omega) \vee (\phi\wedge\chi)\wedge (\phi\wedge\chi)\osuf (\psi\wedge \omega)\Rightarrow y\cat Y : (\chi\osuf \omega). 
$$
From (iv.b41), by $(\vee L)$, we obtain:
$$
\begin{aligned}
(iv.b411)\ \ & Y : (\psi\wedge \omega) \Rightarrow y\cat Y : (\chi\osuf \omega).\\
(iv.b412)\ \ & Y : (\phi\wedge\chi)\wedge (\phi\wedge\chi)\osuf (\psi\wedge \omega) \Rightarrow y\cat Y : (\chi\osuf \omega).
\end{aligned}
$$
The case for (iv.b411) is trivial (just extend $y\cat Y : (\chi\osuf \omega)$ on the right side by $(\osuf)$ and can obtain the same item $Y : \omega$ as on the left side).
For case (iv.b412), by extending $y\cat Y : (\chi\osuf\omega)$ with the rule $(\osuf)$, we obtain:
$$
\begin{aligned}
(iv.b4121)\ \ & Y : (\phi\wedge\chi), Y : (\phi\wedge\chi)\osuf (\psi\wedge \omega) \Rightarrow Y : \chi,\\
(iv.b4122)\ \ & Y : (\phi\wedge\chi), Y : (\phi\wedge\chi)\osuf (\psi\wedge \omega) \Rightarrow Y : \chi\osuf \omega, 
\end{aligned}
$$
where (iv.b4121) is obviously closed.
From (iv.b4122), we obtain by applying $(wk)$:
$$
(iv.b41221)\ \ Y : (\phi\wedge\chi)\osuf (\psi\wedge \omega) \Rightarrow Y : \chi\osuf \omega, 
$$
which is exactly (iv.b4), forming a back-link.
The derivation path from (iv.b4) to (iv.b41221) is progressing since rule $(\osuf)$ is applied.
\fi

\ifx
Case (iv.b5): Without loss of generality, let $X = y\cat Y$.
From (iv.b5), by applying the rule $(\osuf)$, we have
$$
(iv.b51)\ \  Y : (\phi\wedge \omega\wedge \phi\osuf \psi) \vee (\phi\wedge\chi)\wedge (\phi\wedge\chi)\osuf (\phi\wedge \omega\wedge \phi\osuf \psi)\Rightarrow y\cat Y : (\chi\osuf \omega).
$$
From (iv.b51), we apply the rule $(\vee L)$, and obtain:
$$
\begin{aligned}
(iv.b511)\ \ & Y : (\phi\wedge \omega\wedge \phi\osuf \psi) \Rightarrow y\cat Y : (\chi\osuf \omega).\\
(iv.b512)\ \ & Y : (\phi\wedge\chi)\wedge (\phi\wedge\chi)\osuf (\phi\wedge \omega\wedge \phi\osuf \psi) \Rightarrow y\cat Y : (\chi\osuf \omega). 
\end{aligned}
$$
The case for (iv.b511) is trivial (just extend $y\cat Y : (\chi\osuf \omega)$ on the right side by the rule $(\osuf)$ and obtain the same item $Y : \omega$ as on the left side).
For case (iv.b512), by extending $y\cat Y : (\chi\osuf\omega)$ with the rule $(\osuf)$, we obtain:
$$
\begin{aligned}
(iv.b5121)\ \ & Y : (\phi\wedge\chi), Y : (\phi\wedge\chi)\osuf (\phi\wedge \omega\wedge \phi\osuf \psi) \Rightarrow Y : \chi,\\
(iv.b5122)\ \ & Y : (\phi\wedge\chi), Y : (\phi\wedge\chi)\osuf (\phi\wedge \omega\wedge \phi\osuf \psi) \Rightarrow Y : \chi\osuf \omega
\end{aligned}
$$
where (iv.b5121) is obviously closed.
From (iv.b5122), we obtain by applying $(wk)$:
$$
(iv.b51221)\ \ Y : (\phi\wedge\chi)\osuf (\phi\wedge \omega\wedge \phi\osuf \psi) \Rightarrow Y : \chi\osuf \omega, 
$$
which is exactly (iv.b5), forming a back-link.
The derivation path from (iv.b5) to (iv.b51221) is progressing since rule $(\osuf)$ is applied.
\fi
\ifx
Case (iv.b6): Without loss of generality, let $X = y\cat Y$.
From (iv.b6), by applying the rule $(\osuf)$, we can have
$$
(iv.b61)\ \  Y : (\psi\wedge \chi\wedge \chi\osuf\omega) \vee (\phi\wedge\chi)\wedge (\phi\wedge\chi)\osuf (\psi\wedge \chi\wedge \chi\osuf\omega)\Rightarrow y\cat Y : (\chi\osuf \omega). 
$$
From (iv.b61), we apply $(\vee L)$, and obtain:
$$
\begin{aligned}
(iv.b611)\ \ & Y : (\psi\wedge \chi\wedge \chi\osuf\omega) \Rightarrow y\cat Y : (\chi\osuf \omega).\\
(iv.b612)\ \ & Y : (\phi\wedge\chi)\wedge (\phi\wedge\chi)\osuf (\psi\wedge \chi\wedge \chi\osuf\omega) \Rightarrow y\cat Y : (\chi\osuf \omega).
\end{aligned}
$$
For case (iv.b611), by extending $y\cat Y : (\chi\osuf\omega)$ on the right side by the rule $(\osuf)$ and $(\wedge R)$, we need to show $Y : \chi\wedge\chi\osuf\omega$, which is closed since the left side contains exactly the same item. 
For case (iv.b612), by extending $y\cat Y : (\chi\osuf\omega)$ with $(\osuf)$, we obtain:
$$
\begin{aligned}
(iv.b6121)\ \ & Y : (\phi\wedge\chi), Y : (\phi\wedge\chi)\osuf (\psi\wedge \chi\wedge \chi\osuf\omega) \Rightarrow Y : \chi,\\
(iv.b6122)\ \ & Y : (\phi\wedge\chi), Y : (\phi\wedge\chi)\osuf (\psi\wedge \chi\wedge \chi\osuf\omega) \Rightarrow Y : \chi\osuf \omega, 
\end{aligned}
$$
where (iv.b6121) is obviously closed (both sides containing item $Y : \chi$).
From (iv.b6122), we obtain by applying the rule $(wk)$:
$$
(iv.b61221)\ \ Y : (\phi\wedge\chi)\osuf (\psi\wedge \chi\wedge \chi\osuf\omega) \Rightarrow Y : \chi\osuf \omega. 
$$
This is exactly (iv.b6), forming a back-link.
The derivation path from (iv.b6) to (iv.b61221) is progressing since rule $(\osuf)$ is applied.

From these 6 cases, 
it is not hard to see that the whole derivation from (iv.b) is a cyclic proof. 
Thus we conclude the proof of (iv.b). 
\fi

\end{proof}

\begin{prop}[Rule (v)]
Rule (v): 
$\neg(\phi\osuf \psi) \leftrightarrow \neg(\true \osuf \psi)\vee (\neg \psi)\osuf (\neg\phi\wedge \neg\psi)$
is derived in \GiiiPPLcyc. 
\end{prop}

\begin{proof}
Let $X\in \TVar$. We need to prove
\begin{enumerate}
\item[$(v.a)$] $X : \neg(\phi\osuf \psi) \Rightarrow X : \neg(\true \osuf \psi)\vee (\neg \psi)\osuf (\neg\phi\wedge \neg\psi)$.
\item[$(v.b)$] $X : \neg(\true \osuf \psi)\vee (\neg \psi)\osuf (\neg\phi\wedge \neg\psi)\Rightarrow X : \neg(\phi\osuf \psi)$.
\end{enumerate}

\textbf{Case for proving $(v.a)$}: 
By applying $(\neg L)$ on the left and $(\vee R)$, $(\neg R)$ on the right, we obtain the sequent:
$$
(v.a1)\ \ X : (\true\osuf \psi)\Rightarrow X : (\phi\osuf \psi), X : (\neg\psi)\osuf(\neg\phi\wedge \neg\psi).
$$
Without loss of generality, let $X = y\cat Y$. By applying the rules $(\osuf L)$ on $X : (\true\osuf\psi)$ and $(\osuf R)$ on $X : \phi\osuf \psi$, we have:
\begin{enumerate}
\item[(v.a11)] $Y : \psi\Rightarrow Y : \psi, Y : \phi, X : (\neg\psi)\osuf(\neg\phi\wedge \neg\psi)$. 
\item[(v.a12)]  $Y : \psi\Rightarrow Y : \psi, Y : \phi\osuf\psi, X : (\neg\psi)\osuf(\neg\phi\wedge \neg\psi)$. 
\item[(v.a13)]  $Y : \true, Y : \true\osuf\psi \Rightarrow Y : \psi, Y : \phi, X : (\neg\psi)\osuf(\neg\phi\wedge \neg\psi)$.
\item[(v.a14)] $Y : \true, Y : \true\osuf\psi \Rightarrow Y : \psi, Y : \phi\osuf\psi, X : (\neg\psi)\osuf(\neg\phi\wedge \neg\psi)$.  
\end{enumerate}
(v.a11) and (v.a12) can be directly proved by (ax). 

For (v.a13) and (v.a14), by $(\osuf R)$ on $X : (\neg\psi)\osuf(\neg\phi\wedge \neg\psi)$, we further need the following sequents:
\begin{enumerate}
\item[v.a131] $Y : \true, Y : \true\osuf\psi \Rightarrow Y : \psi, Y : \phi, Y : (\neg\phi\wedge \neg\psi), Y : (\neg\psi)$.
\item[v.a132] $Y : \true, Y : \true\osuf\psi \Rightarrow Y : \psi, Y : \phi, Y : (\neg\phi\wedge \neg\psi), Y : (\neg\psi)\osuf (\neg\phi\wedge \neg\psi)$.
\item[v.a141] $Y : \true, Y : \true\osuf\psi \Rightarrow Y : \psi, Y : \phi\osuf\psi, Y : (\neg\phi\wedge \neg\psi), Y : (\neg\psi)$. 
\item[v.a142] $Y : \true, Y : \true\osuf\psi \Rightarrow Y : \psi, Y : \phi\osuf\psi, Y : (\neg\phi\wedge \neg\psi), Y : (\neg\psi)\osuf (\neg\phi\wedge \neg\psi)$. 
\end{enumerate}
(v.a131) and (v.a141) can be trivially closed by $(ax)$ since they both have $Y:\psi$ and $Y:(\neg\psi)$ on their right sides. 
Sequent (v.a132) can be closed because it contains $Y : \psi$, $Y : \phi$, and $Y: (\neg\phi\wedge \neg\psi)$ on the right at the same time. 

For (v.a142), by (wk), we need: 
$$
(v.a1421)\ Y : \true\osuf\psi \Rightarrow  Y : \phi\osuf\psi, Y : (\neg\psi)\osuf (\neg\phi\wedge \neg\psi). 
$$
It has the same form as (v.a1), thus forming a back-link. 
The derivation path from (v.a1) to (v.a1421) is progressing since rule $(\osuf R)$ is applied.

\textbf{Case for proving $(v.b)$}:
By applying $(\vee L)$ on the left, we split it into 2 cases:
\begin{enumerate}
    \item[(v.b1)] $X : \neg(\true\osuf\psi)\Rightarrow X : \neg(\phi\osuf \psi)$.
    \item[(v.b2)] $X : (\neg\psi)\osuf(\neg\phi\wedge \neg\psi)\Rightarrow X : \neg(\phi\osuf \psi)$. 
\end{enumerate}

Case (v.b1): By $(\neg L)$ and $(\neg R)$:
$$
(v.b11)\ \ X : \phi\osuf\psi\Rightarrow X : \true\osuf\psi. 
$$
Without loss of generality, let $X = y\cat Y$. 
By applying $(\osuf R)$ and $(\osuf L)$ on both sides, we have:
\begin{enumerate}
\item[(v.b111)] $Y : \psi\Rightarrow Y : \psi, Y : \true$.
\item[(v.b112)] $Y : \psi\Rightarrow Y : \psi, Y : \true\osuf \psi$.
\item[(v.b113)] $Y : \phi, Y : \phi\osuf\psi\Rightarrow Y : \psi, Y : \true$.
\item[(v.b114)]  $Y : \phi, Y : \phi\osuf\psi\Rightarrow Y : \psi, Y : \true\osuf \psi$.
\end{enumerate}
(v.b111) and (v.b112) can be trivially proved by (ax). 
(v.b113) can be proved since $Y : \true$ is on the right side. 

For (v.b114), by (wk) we need:
$$
(v.b1141)\ \ Y : \phi\osuf\psi\Rightarrow Y : \true\osuf \psi, 
$$
which is exactly (v.b11), forming a back-link. The derivation path is progressing since rule $(\osuf R)$ is applied from (v.b11) to (v.b114), with $(X : \true\osuf\psi, Y : \true\osuf\psi)$ the progressing step. 

Case (v.b2): By the introduction of the operator $\ofus$, it is equivalent to prove:
$$
(v.b21)\ \ X : (\neg\psi)\osuf(\neg\phi\wedge\neg\psi)\Rightarrow X : \neg\phi\ofus \neg\psi.
$$
By applying $(\osuf L)$ on the left and $(\ofus R)$ on the right, we have 
\begin{enumerate}
    \item[(v.b211)] $Y : \neg\phi\wedge\neg\psi\Rightarrow Y : \neg\psi$.
    \item[(v.b212)] $Y : \neg\phi\wedge\neg\psi\Rightarrow Y : \neg \phi, Y : \neg\phi\ofus\neg\psi$.
    \item[(v.b213)] $Y : \neg\psi, Y : (\neg\psi)\osuf(\neg\phi\wedge\neg\psi)\Rightarrow Y : \neg \psi$.
    \item[(v.b214)] $Y : \neg\psi, Y : (\neg\psi)\osuf(\neg\phi\wedge\neg\psi)\Rightarrow Y : \neg \phi, Y : \neg\phi\ofus\neg\psi$.
\end{enumerate}

(v.b211), (v.b212) and (v.b213) are trivially proved by $(\wedge L)$ and $(ax)$. 

For (v.b214), by $(wk)$, we obtain
$$
(v.b2141)\ \ Y : (\neg\psi)\osuf(\neg\phi\wedge\neg\psi), Y : \phi\Rightarrow Y: \neg \phi\ofus\neg\psi, 
$$
which is exactly (v.b21), forming a back-link. The derivation path is progressing since rule $(\ofus R)$ is applied, with $(X : \neg\phi\ofus\neg\psi, Y : \neg\phi\ofus \neg\psi)$ the progressing step.

\ifx
By $(\neg R)$:
$$
(v.b21)\ \ X : (\neg\psi)\osuf(\neg\phi\wedge \neg\psi), X : \phi\osuf\psi\Rightarrow \cdot
$$
Without loss of generality, let $X = y\cat Y$. By applying $(\osuf)$ on both formulas on the left, we have:
$$
Y : (\neg\phi\wedge\neg\psi)\vee (\neg\psi\wedge (\neg\psi)\osuf(\neg\phi\wedge\neg\psi)),\ Y : \psi\vee (\phi\wedge \phi\osuf\psi)\Rightarrow \cdot
$$
By $(\vee L)$, we split into 4 sub-cases:
\begin{enumerate}
    \item[(v.b2111)] $Y : \neg\phi\wedge\neg\psi,\ Y : \psi\Rightarrow \cdot$.
    \item[(v.b2112)] $Y : \neg\phi\wedge\neg\psi,\ Y : \phi\wedge\phi\osuf\psi\Rightarrow \cdot$.
    \item[(v.b2113)] $Y : \neg\psi\wedge (\neg\psi)\osuf(\neg\phi\wedge\neg\psi),\ Y : \psi\Rightarrow \cdot$.
    \item[(v.b2114)] $Y : \neg\psi\wedge (\neg\psi)\osuf(\neg\phi\wedge\neg\psi),\ Y : \phi\wedge\phi\osuf\psi\Rightarrow \cdot$.
\end{enumerate}
Cases (v.b2111) - (v.b2113) can be easily closed by $(\wedge L)$ and then $(\neg L)$ and $(\textit{ax})$. 

\ifx
Case (v.b2111): By $(\wedge L)$, we have $Y : \neg\psi, Y : \psi\Rightarrow \cdot$, which closes by $(\neg L)$ and $(\textit{ax})$.

Case (v.b2112): By $(\wedge L)$, we have $Y : \neg\phi, Y : \phi\Rightarrow \cdot$, which closes by $(\neg L)$ and $(\textit{ax})$.

Case (v.b2113): By $(\wedge L)$, we have $Y : \neg\psi, Y : \psi\Rightarrow \cdot$, which closes by $(\neg L)$ and $(\textit{ax})$.
\fi

For case (v.b2114), by $(\wedge L)$ on both conjunctions:
$$
Y : \neg\psi, Y : (\neg\psi)\osuf(\neg\phi\wedge\neg\psi), Y : \phi, Y : \phi\osuf\psi\Rightarrow \cdot.
$$
By weakening:
$$
(v.b21141)\ \ Y : (\neg\psi)\osuf(\neg\phi\wedge\neg\psi), Y : \phi\osuf\psi\Rightarrow \cdot.
$$
which is exactly (v.b21), forming a back-link. The derivation path is progressing since rule $(\osuf)$ is applied.
\fi
\end{proof}

\begin{prop}[Rule (vi)]
Rule (vi):
$\phi\osuf \psi\leftrightarrow \nxt \psi \vee \nxt (\phi\wedge (\phi\osuf \psi))$
is derived in \GiiiPPLcyc. 
\end{prop}

\begin{proof}
Let $X\in \TVar$. 
Recall that $\nxt \theta \dddef \false\osuf \theta$.
We need to prove
\begin{enumerate}
\item[$(vi.a)$] $X : \phi\osuf \psi\Rightarrow X : \nxt \psi \vee \nxt (\phi\wedge (\phi\osuf \psi))$.
\item[$(vi.b)$] $X : \nxt \psi \vee \nxt (\phi\wedge (\phi\osuf \psi))\Rightarrow X : \phi\osuf \psi$.
\end{enumerate}

\textbf{Case for proving $(vi.a)$}: 
By applying $(\vee R)$, we need:
$$
(vi.a1)\ \ X : \phi\osuf\psi\Rightarrow X : \false\osuf\psi, X : \false\osuf(\phi\wedge(\phi\osuf\psi))
$$
Without loss of generality, let $X = y\cat Y$. 
By applying $(\osuf L)$ and $(\osuf R)$ on $X : \false\osuf \psi$, we have
\begin{enumerate}
\item[(vi.a11)] $Y : \psi\Rightarrow Y : \psi, Y : \false, X : \false\osuf(\phi\wedge(\phi\osuf\psi))$. 
\item[(vi.a12)] $Y : \psi\Rightarrow Y : \psi, Y : \false\osuf \psi, X : \false\osuf(\phi\wedge(\phi\osuf\psi))$. 
\item[(vi.a13)] $Y : \phi, Y : \phi\osuf \psi\Rightarrow Y : \psi, Y : \false, X : \false\osuf(\phi\wedge(\phi\osuf\psi))$. 
\item[(vi.a14)] $Y : \phi, Y : \phi\osuf \psi\Rightarrow Y : \psi, Y : \false\osuf \psi, X : \false\osuf(\phi\wedge(\phi\osuf\psi))$. 
\end{enumerate}
(vi.a11) and (vi.a12) can be proved directly by (ax). 

For (vi.a13), by $(\osuf R)$ on $X : \false\osuf(\phi\wedge(\phi\osuf\psi))$, we need:
\begin{enumerate}
\item[(vi.a131)] $Y : \phi, Y : \phi\osuf \psi\Rightarrow Y : \psi, Y : \false, Y : (\phi\wedge \phi\osuf\psi), Y : \false$. 
\item[(vi.a132)] $Y : \phi, Y : \phi\osuf \psi\Rightarrow Y : \psi, Y : \false, Y : (\phi\wedge \phi\osuf\psi), Y : \false\osuf (\phi\wedge \phi\osuf \psi)$.
\end{enumerate}
Both (vi.a131) and (vi.a132) can be proved by $(\wedge R)$ and $(ax)$ easily since they contain $Y : \phi$, $Y : \phi\osuf \psi$ on the left side and $Y : (\phi\wedge \phi\osuf \psi)$ on the right side. 

For (vi.a14), by $(\osuf R)$ on $X : \false\osuf(\phi\wedge(\phi\osuf\psi))$, we need:
\begin{enumerate}
\item[(vi.a141)] $Y : \phi, Y : \phi\osuf \psi\Rightarrow Y : \psi, Y : \false\osuf \psi, Y : (\phi\wedge \phi\osuf\psi), Y : \false$. 
\item[(vi.a142)] $Y : \phi, Y : \phi\osuf \psi\Rightarrow Y : \psi, Y : \false\osuf \psi, Y : (\phi\wedge \phi\osuf\psi), Y : \false\osuf (\phi\wedge \phi\osuf \psi)$.
\end{enumerate}
Similar to the above, both (vi.a141) and (vi.a142) can be proved by $(\wedge R)$ and $(ax)$ easily since they contain $Y : \phi$, $Y : \phi\osuf \psi$ on the left side and $Y : (\phi\wedge \phi\osuf \psi)$ on the right side.

\textbf{Case for proving $(vi.b)$}: 
By $(\vee L)$, we split into 2 cases:
\begin{enumerate}
    \item[(vi.b1)] $X : \false\osuf \psi\Rightarrow X : \phi\osuf \psi$.
    \item[(vi.b2)] $X : \false\osuf(\phi\wedge(\phi\osuf \psi))\Rightarrow X : \phi\osuf \psi$.
\end{enumerate}

Without loss of generality, let $X = y\cat Y$. 

Case (vi.b1): By $(\osuf R)$ and $(\osuf L)$ on both sides:
\begin{enumerate}
\item[(vi.b11)] $Y : \psi\Rightarrow Y : \psi, Y : \phi$. 
\item[(vi.b12)] $Y : \psi\Rightarrow Y : \psi, Y : \phi\osuf \psi$. 
\item[(vi.b13)] $Y : \false, Y : \false \osuf \psi\Rightarrow Y : \psi, Y : \phi$. 
\item[(vi.b14)] $Y : \false, Y : \false \osuf \psi\Rightarrow Y : \psi, Y : \phi\osuf\psi$. 
\end{enumerate}
(vi.b11) and (vi.b12) are trivially proved by (ax). 
Sequents (vi.b13) and (vi.b14) can be proved easily because they both have $Y : \false$ on the left side. 

Case (vi.b2): By $(\osuf R)$ and $(\osuf L)$ on both sides:
\begin{enumerate}
\item[(vi.b21)] $Y : \phi\wedge \phi\osuf\psi\Rightarrow Y : \psi, Y : \phi$. 
\item[(vi.b22)] $Y : \phi\wedge \phi\osuf\psi\Rightarrow Y : \psi, Y : \phi\osuf\psi$. 
\item[(vi.b23)] $Y : \false, Y : \false \osuf (\phi\wedge \phi\osuf\psi)\Rightarrow Y : \psi, Y : \phi$. 
\item[(vi.b24)] $Y : \false, Y : \false \osuf (\phi\wedge \phi\osuf\psi)\Rightarrow Y : \psi, Y : \phi\osuf\psi$. 
\end{enumerate}
(vi.b21) and (vi.b22) are trivially proved by $(\wedge L)$ and (ax). 
Sequents (vi.b23) and (vi.b24) can be proved easily because they both have $Y : \false$ on the left side.

\end{proof}

\begin{prop}[Rule (vii)]
Rule (vii): 
$\phi\osuf(\phi\wedge (\phi\osuf \psi))\leftrightarrow \nxt (\phi\wedge (\phi\osuf \psi))$
is derived in \GiiiPPLcyc. 
\end{prop}

\begin{proof}
Let $X\in \TVar$. 
Recall that $\nxt(\phi\wedge(\phi\osuf\psi)) \dddef \false\osuf(\phi\wedge(\phi\osuf\psi))$. 
We need to prove
\begin{enumerate}
\item[$(vii.a)$] $X : \phi\osuf(\phi\wedge(\phi\osuf\psi))\Rightarrow X : \false\osuf(\phi\wedge(\phi\osuf\psi))$.
\item[$(vii.b)$] $X : \false\osuf(\phi\wedge(\phi\osuf\psi))\Rightarrow X : \phi\osuf(\phi\wedge(\phi\osuf\psi))$.
\end{enumerate}
Without loss of generality, let $X = y\cat Y$.

\textbf{Case for proving $(vii.a)$}: By $(\osuf R)$ and $(\osuf L)$ on both sides: 
\begin{enumerate}
\item[(vii.a1)] $Y : \phi\wedge \phi\osuf \psi\Rightarrow Y : \phi\wedge \phi\osuf \psi, Y : \false$. 
\item[(vii.a2)] $Y : \phi\wedge \phi\osuf \psi\Rightarrow Y : \phi\wedge \phi\osuf \psi, Y : \false\osuf\phi\wedge\phi\osuf \psi$. 
\item[(vii.a3)] $Y : \phi, Y : \phi\osuf (\phi\wedge \phi\osuf \psi)\Rightarrow Y : \phi\wedge \phi\osuf \psi, Y:\false$. 
\item[(vii.a4)] $Y : \phi, Y : \phi\osuf (\phi\wedge \phi\osuf \psi)\Rightarrow Y : \phi\wedge \phi\osuf \psi, Y:\false\osuf \phi\wedge\phi\osuf \psi$. 
\end{enumerate}
Sequents (vii.a1) and (vii.a2) can be proved trivially by $(ax)$ because they contain $Y : \phi\wedge \phi\osuf\psi$ on both sides. 

For (vii.a3), since $Y : \phi$ is on the left, by $(\wedge R)$ on $Y : \phi\wedge\phi\osuf\psi$, the branch that has $Y : \phi$ on the right can be trivially closed by (ax). 
So after (wk) to eliminate $Y : \phi$ on the left and $Y : \false$ on the right, we only need:
$$
(vii.a31)\ \ Y : \phi\osuf (\phi\wedge \phi\osuf \psi)\Rightarrow Y : \phi\osuf \psi. 
$$
From (vii.a31), let $Y = z\cat Z$. By $(\osuf R)$ and $(\osuf L)$, we have 
\begin{enumerate}
\item[(vii.a311)] $Z : \phi\wedge \phi\osuf \psi\Rightarrow Z : \psi, Z : \phi$. 
\item[(vii.a312)] $Z : \phi\wedge \phi\osuf \psi\Rightarrow Z : \psi, Z : \phi\osuf\psi$. 
\item[(vii.a313)] $Z : \phi, Z : \phi\osuf(\phi\wedge \phi\osuf\psi)\Rightarrow Z : \psi, Z : \phi$. 
\item[(vii.a314)] $Z : \phi, Z : \phi\osuf(\phi\wedge \phi\osuf\psi)\Rightarrow Z : \psi, Z : \phi\osuf\psi$. 
\end{enumerate}
The cases (vii.a311) to (vii.a314) can be trivially proved by $(\wedge L)$ and $(ax)$. 

For (vii.a314), by (wk):
$$
(vii.a3141)\ Z : \phi\osuf(\phi\wedge \phi\osuf\psi)\Rightarrow Z : \phi\osuf\psi, 
$$
which matches (vii.a31), forming a back-link. 
The derivation path from (vii.a31) to (vii.a3141) is progressing, as the rule $(\osuf R)$ is applied from (vii.a31) to (vii.a314), where $(Y : \phi\osuf\psi, Z : \phi\osuf\psi)$ is the progressing step. 

For (vii.a4), by (wk) we obtain:
$$
(vii.a41)\ \ Y : \phi\osuf (\phi\wedge \phi\osuf \psi)\Rightarrow Y:\false\osuf \phi\wedge\phi\osuf \psi, 
$$
which is exactly (vii.a), forming a back-link. The derivation path from (vii.a) to (vii.a41) is progressing, as the rule $(\osuf R)$ is applied from (vii.a) to (vii.a4), where $(X : \false \osuf(\phi\wedge \phi\osuf\psi), Y : \false \osuf(\phi\wedge \phi\osuf\psi))$ is the progressing step. 

\textbf{Case for proving $(vii.b)$}:
By $(\osuf R)$ and $(\osuf L)$ on both sides: 
\begin{enumerate}
\item[(vii.b1)] $Y : \phi\wedge \phi\osuf \psi\Rightarrow Y : \phi\wedge \phi\osuf \psi, Y : \phi$. 
\item[(vii.b2)] $Y : \phi\wedge \phi\osuf \psi\Rightarrow Y : \phi\wedge \phi\osuf \psi, Y : \phi\osuf (\phi\wedge \phi\osuf \psi)$. 
\item[(vii.b3)] $Y : \false, Y : \false\osuf (\phi\wedge\phi\osuf\psi)\Rightarrow Y : \phi\wedge \phi\osuf \psi, Y : \phi$. 
\item[(vii.b4)] $Y : \false, Y : \false\osuf (\phi\wedge\phi\osuf\psi)\Rightarrow Y : \phi\wedge \phi\osuf \psi, Y : \phi\osuf(\phi\wedge \phi\osuf\psi)$. 
\end{enumerate}
The cases for (vii.b1) and (vii.b2) can be trivially closed by (ax). 
The cases (vii.b3) and (vii.b4) can be proved easily because they contain $Y : \false$ on the right side.

\end{proof}

\begin{prop}[Rule (viii)]
Rule (viii): 
 $\phi\osuf(\phi\wedge (\phi\osuf \psi))\leftrightarrow \phi\osuf(\phi\wedge \nxt \psi)$
is derived in \GiiiPPLcyc. 
\end{prop}

\begin{proof}

Let $X\in \TVar$. 
Recall that $\nxt \psi \dddef \false\osuf \psi$.
We need to prove
\begin{enumerate}
\item[$(viii.a)$] $X : \phi\osuf(\phi\wedge(\phi\osuf\psi))\Rightarrow X : \phi\osuf(\phi\wedge \false\osuf \psi)$.
\item[$(viii.b)$] $X : \phi\osuf(\phi\wedge \false\osuf \psi)\Rightarrow X : \phi\osuf(\phi\wedge(\phi\osuf\psi))$.
\end{enumerate}
Without loss of generality, let $X = x\cat Y$.

\textbf{Case for proving $(viii.a)$}: 
By applying the rules $(\osuf R)$ and $(\osuf L)$ on both sides:
\begin{enumerate}
\item[(viii.a1)] $Y : \phi\wedge \phi\osuf\psi\Rightarrow Y : \phi\wedge \false\osuf\psi, Y : \phi$.
\item[(viii.a2)] $Y : \phi\wedge \phi\osuf\psi\Rightarrow Y : \phi\wedge \false\osuf\psi, Y : \phi\osuf(\phi\wedge \false\osuf\psi)$.
\item[(viii.a3)] $Y : \phi, Y : \phi\osuf(\phi\wedge \phi\osuf \psi)\Rightarrow Y : \phi\wedge \false\osuf\psi, Y : \phi$. 
\item[(viii.a4)] $Y : \phi, Y : \phi\osuf(\phi\wedge \phi\osuf \psi)\Rightarrow Y : \phi\wedge \false\osuf\psi, Y : \phi\osuf(\phi\wedge \false\osuf\psi)$. 
\end{enumerate}
The cases (viii.a1) and (viii.a3) can be trivially proved by (ax) since they contain $Y : \phi$ on both sides. 

For (viii.a2), by $(\wedge R)$ on $Y : \phi\wedge \false\osuf\psi$, the branch that contains $Y:\phi$ on the right side can be easily closed since there is a $Y:\phi$ on the left side. 
Thus, the only branch that needs to be proved is: 
$$
Y : \phi\wedge \phi\osuf \psi\Rightarrow Y : \false\osuf\psi, Y : \phi\osuf(\phi\wedge \false\osuf\psi). 
$$
By $(\wedge L)$ and $(wk)$, we need:
$$
(viii.a21)\ \ Y : \phi\osuf\psi\Rightarrow Y : \false\osuf\psi, Y : \phi\osuf(\phi\wedge \false\osuf\psi). 
$$
Without loss of generality, let $Y = z\cat Z$. 
Then, by $(\osuf R)$ and $(\osuf L)$ on $Y : \false\osuf\psi$ and $Y : \phi\osuf\psi$ respectively, we obtain
\begin{enumerate}
\item[viii.a211] $Z : \psi\Rightarrow Z : \psi, Z : \false, Y : \phi\osuf(\phi\wedge \false\osuf\psi)$. 
\item[viii.a212] $Z : \psi\Rightarrow Z : \psi, Z : \false\osuf \psi, Y : \phi\osuf(\phi\wedge \false\osuf\psi)$. 
\item[viii.a213]$Z : \phi, Z : \phi\osuf \psi\Rightarrow Z : \psi, Z : \false, Y : \phi\osuf(\phi\wedge \false\osuf\psi)$. 
\item[viii.a214]$Z : \phi, Z : \phi\osuf \psi\Rightarrow Z : \psi, Z : \false\osuf\psi, Y : \phi\osuf(\phi\wedge \false\osuf\psi)$. 
\end{enumerate}
Where (viii.a211) and (viii.a212) can be closed directly by (ax). 

For (viii.a213), by applying the rule $(\osuf R)$ on the right, we need
\begin{enumerate}
\item[viii.a2131] $Z : \phi, Z : \phi\osuf \psi\Rightarrow Z : \psi, Z : \false, Z : \phi\wedge\false\osuf \psi, Z : \phi$, which is closed by $(ax)$. 
\item[viii.a2132] $Z : \phi, Z : \phi\osuf \psi\Rightarrow Z : \psi, Z : \false, Z : \phi\wedge\false\osuf \psi, Z : \phi\osuf (\phi\wedge\false\osuf\psi)$, from which by (wk) and $(\wedge R)$ on $Z : \phi\wedge \false\osuf\psi$, we still need to prove:
$$
(viii.a21321)\  Z : \phi\osuf \psi\Rightarrow Z : \false\osuf \psi, Z : \phi\osuf (\phi\wedge\false\osuf\psi),
$$
which is exactly (viii.a21), forming a back-link. The derivation path from (viii.a21) to (viii.a21321) is progressing as the rule $(\osuf R)$ is applied with suitable progressing step. 
\end{enumerate}

(viii.a214) is similar to the case (viii.a213). The only exception is that on the right side $Z : \false\osuf\psi$ is weakened instead of $Z : \false$. We omit its proof here. 

For (viii.a4), by (wk):
$$
(viii.a41)\ Y : \phi\osuf(\phi\wedge \phi\osuf \psi)\Rightarrow Y : \phi\osuf(\phi\wedge \false\osuf\psi), 
$$
which is exactly (viii.a), forming a back-link. 
The derivation path from (viii.a) to (viii.a41) is progressing as the rule $(\osuf R)$ is applied.  

\textbf{Case for proving $(viii.b)$}:
By applying $(\osuf R)$ and $(\osuf L)$ on both sides, we get
\begin{enumerate}
\item[(viii.b1)] $Y : \phi\wedge \false\osuf\psi\Rightarrow Y : \phi\wedge \phi\osuf\psi, Y : \phi$. 
\item[(viii.b2)]  $Y : \phi\wedge \false\osuf\psi\Rightarrow Y : \phi\wedge \phi\osuf\psi, Y : \phi\osuf(\phi\wedge \phi\osuf\psi)$. 
\item[(viii.b3)] $Y : \phi, Y : \phi\osuf(\phi\wedge\false\osuf\psi)\Rightarrow Y : \phi\wedge \phi\osuf\psi, Y : \phi$. 
\item[(viii.b4)] $Y : \phi, Y : \phi\osuf(\phi\wedge\false\osuf\psi)\Rightarrow Y : \phi\wedge \phi\osuf\psi, Y : \phi\osuf (\phi\wedge\phi\osuf\psi)$. 
\end{enumerate} 

The cases (viii.b1) and (viii.b3) can be proved trivially by $(ax)$. 

For case (viii.b2), by $(\wedge R)$ on $Y : \phi\wedge \phi\osuf\psi$, the branch that contains $Y:\phi$ on the right side can be easily closed since there is a $Y:\phi$ on the left side. 
Thus, the only branch that needs to be proved is: 
$$
Y : \phi\wedge \false\osuf\psi\Rightarrow Y :  \phi\osuf\psi, Y : \phi\osuf(\phi\wedge \phi\osuf\psi). 
$$
By $(\wedge L)$ and $(wk)$, we need:
$$
(viii.b21)\ Y : \false\osuf\psi\Rightarrow Y : \phi\osuf\psi, Y : \phi\osuf(\phi\wedge \phi\osuf\psi)
$$
From (viii.b21), by $(wk)$ on the right side we drop $Y : \phi\osuf(\phi\wedge \phi\osuf\psi)$, so that it suffices to prove
$$
(viii.b211)\ \ Y : \false\osuf\psi\Rightarrow Y : \phi\osuf\psi. 
$$
Without loss of generality, let $Y = z\cat Z$. 
Then, by $(\osuf L)$ on $Y : \false\osuf\psi$ and $(\osuf R)$ on $Y : \phi\osuf\psi$, we obtain
\begin{enumerate}
\item[(viii.b2111)] $Z : \psi\Rightarrow Z : \psi, Z : \phi$. 
\item[(viii.b2112)] $Z : \psi\Rightarrow Z : \psi, Z : \phi\osuf\psi$. 
\item[(viii.b2113)] $Z : \false, Z : \false\osuf\psi\Rightarrow Z : \psi, Z : \phi$. 
\item[(viii.b2114)] $Z : \false, Z : \false\osuf\psi\Rightarrow Z : \psi, Z : \phi\osuf\psi$. 
\end{enumerate}
The cases (viii.b2111) and (viii.b2112) are closed directly by $(ax)$, since they carry $Z : \psi$ on both sides. 
The cases (viii.b2113) and (viii.b2114) are closed as well, since they both carry $Z : \false$ on the left side. 

For the case (viii.b4), by $(wk)$ on both sides, we need:
$$
(viii.b41)\ \ Y : \phi\osuf(\phi\wedge\false\osuf\psi)\Rightarrow Y : \phi\osuf(\phi\wedge \phi\osuf\psi), 
$$
which is exactly (viii.b), forming a back-link. 
The derivation path from (viii.b) to (viii.b41) is progressing, as the rule $(\osuf R)$ is applied from (viii.b) to (viii.b4), where $(X : \phi\osuf(\phi\wedge \phi\osuf\psi), Y : \phi\osuf(\phi\wedge \phi\osuf\psi))$ is the progressing step. 
This makes the whole derivation from (viii.b) a cyclic proof. 

\ifx

\textbf{Case for proving $(viii.b)$}:
That is, $X : \phi\osuf(\phi\wedge \false\osuf\psi)\Rightarrow X : \phi\osuf(\phi\wedge(\phi\osuf\psi))$.

Without loss of generality, let $X = x\cat Y$. By applying $(\osuf)$ on the left:
$$
Y : (\phi\wedge \false\osuf\psi)\vee (\phi\wedge \phi\osuf(\phi\wedge \false\osuf\psi))\Rightarrow x\cat Y : \phi\osuf(\phi\wedge(\phi\osuf\psi)). 
$$
By $(\vee L)$ on the left, we split into 2 cases:

\begin{enumerate}
    \item[(viii.b1)] $Y : \phi\wedge \false\osuf\psi\Rightarrow x\cat Y : \phi\osuf(\phi\wedge(\phi\osuf\psi))$.
    \item[(viii.b2)] $Y : \phi\wedge \phi\osuf(\phi\wedge \false\osuf\psi)\Rightarrow x\cat Y : \phi\osuf(\phi\wedge(\phi\osuf\psi))$.
\end{enumerate}

Case (viii.b1): By applying $(\osuf)$ on $x\cat Y : \phi\osuf(\phi\wedge(\phi\osuf\psi))$ on the right, we obtain:
$$
Y : \phi\wedge \false\osuf\psi\Rightarrow Y : (\phi\wedge(\phi\osuf\psi)), Y : \phi\wedge \phi\osuf(\phi\wedge(\phi\osuf\psi)). 
$$
By applying $(\wedge R)$ on $Y : (\phi\wedge(\phi\osuf\psi))$ on the right:
\begin{enumerate}
    \item[(viii.b11)] $Y : \phi\wedge \false\osuf\psi\Rightarrow Y : \phi, Y : \phi\wedge \phi\osuf(\phi\wedge(\phi\osuf\psi))$, which is trivially closed (by $(\wedge L)$ and $(\textit{ax})$ on $Y:\phi$).
    \item[(viii.b12)] $Y : \phi\wedge \false\osuf\psi\Rightarrow Y : (\phi\osuf\psi), Y : \phi\wedge \phi\osuf(\phi\wedge(\phi\osuf\psi))$. 
\end{enumerate}
For case (viii.b12), rename $Y = y\cat Z$. By $(\wedge L)$ on the left, we have $y\cat Z : \phi, y\cat Z : \false\osuf\psi$. By $(\osuf)$ on $y\cat Z : \false\osuf\psi$ on the left, since $\false\wedge(...)$ is trivially false, we effectively get $Z : \psi$. By $(\osuf)$ on $y\cat Z : (\phi\osuf\psi)$ on the right, we get $Z : \psi$ and $Z : \phi\wedge(\phi\osuf\psi)$. Together we need:
$$
y\cat Z : \phi, Z : \psi\Rightarrow Z : \psi, Z : \phi\wedge(\phi\osuf\psi), y\cat Z : \phi\wedge \phi\osuf(\phi\wedge(\phi\osuf\psi)), 
$$
which is trivially closed by $(\textit{ax})$ on $Z : \psi$.

Case (viii.b2): By applying $(\osuf)$ on $x\cat Y : \phi\osuf(\phi\wedge(\phi\osuf\psi))$ on the right, we obtain:
$$
Y : \phi\wedge \phi\osuf(\phi\wedge \false\osuf\psi)\Rightarrow Y : (\phi\wedge(\phi\osuf\psi)), Y : \phi\wedge \phi\osuf(\phi\wedge(\phi\osuf\psi)). 
$$
By weakening on $Y : (\phi\wedge(\phi\osuf\psi))$, it suffices to show:
$$
Y : \phi\wedge \phi\osuf(\phi\wedge \false\osuf\psi)\Rightarrow Y : \phi\wedge \phi\osuf(\phi\wedge(\phi\osuf\psi)). 
$$
By applying $(\wedge R)$ on the right:
\begin{enumerate}
    \item[(viii.b21)] $Y : \phi\wedge \phi\osuf(\phi\wedge \false\osuf\psi)\Rightarrow Y : \phi$, which is trivially closed (by $(\wedge L)$ and $(\textit{ax})$ on $Y:\phi$).
    \item[(viii.b22)] $Y : \phi\wedge \phi\osuf(\phi\wedge \false\osuf\psi)\Rightarrow Y : \phi\osuf(\phi\wedge(\phi\osuf\psi))$. 
\end{enumerate}
For case (viii.b22), by $(\wedge L)$ and weakening on the left, we obtain:
$$
(viii.b221)\ \ Y : \phi\osuf(\phi\wedge \false\osuf\psi)\Rightarrow Y : \phi\osuf(\phi\wedge(\phi\osuf\psi)), 
$$
which is exactly (viii.b), forming a back-link. The derivation path from (viii.b) to (viii.b221) is progressing since rule $(\osuf)$ is applied.
\fi

\end{proof}

\begin{prop}[Rule (ix)]
Rule (ix): 
$\ofst\nxt \true \leftrightarrow \false$
is derived in \GiiiPPLcyc. 
\end{prop}

\begin{proof}

Let $X\in \TVar$. 
Recall that $\nxt \true \dddef \false\osuf \true$.
We need to prove
\begin{enumerate}
\item[$(ix.a)$] $X : \ofst(\false\osuf\true)\Rightarrow X : \false$.
\item[$(ix.b)$]  $X : \false\Rightarrow X : \ofst(\false\osuf\true)$.
\end{enumerate}

\textbf{Case for proving $(ix.a)$}: 
Without loss of generality, let $X = x\cat Y$. By the rule $(\ofst)$ on $x\cat Y : \ofst(\false\osuf\true)$:
$$
x : \false\osuf\true\Rightarrow x\cat Y : \false. 
$$
By $(\osuf\ x)$, the sequent is closed. 

\ifx
Since $x$ is a single label variable, the rule $(\osuf)$ requires a decomposition $x = y\cat Z$, but $x$ admits no such decomposition. Hence $x : \false\osuf\true$ cannot be unfolded and is stuck on the left side. Semantically, $\false\osuf\true$ at a single-element trace means there exists a proper suffix where $\true$ holds with all preceding elements satisfying $\false$; but a single-element trace has no proper non-empty suffix, so $x : \false\osuf\true$ is unsatisfiable. The sequent is therefore closed.
\fi

\textbf{Case for proving $(ix.b)$}: 
This is trivially closed since $X : \false$ is on the left side. 

\end{proof}

\begin{prop}[Rule (x)]
Rule (x): 
$\neg\phi\wedge \ofst \phi \rightarrow \nxt \true$
is derived in \GiiiPPLcyc. 
\end{prop}

\begin{proof}

Let $X\in\TVar$. 
We need to prove
$$X : \neg\phi\wedge \ofst \phi\Rightarrow X : \nxt\true.$$
That is, $X : \neg\phi\wedge \ofst \phi\Rightarrow X : \false\osuf\true$.

Without loss of generality, let $X = x\cat Y$. 
By $(\wedge L)$ on the left:
$$
x\cat Y : \neg\phi, x\cat Y : \ofst \phi\Rightarrow x\cat Y : \false\osuf\true.
$$
By the rule $(\ofst)$ on $x\cat Y : \ofst\phi$:
$$
x\cat Y : \neg\phi, x : \phi\Rightarrow x\cat Y : \false\osuf\true.
$$
If $Y=\varepsilon$, from the left we have both $x : \neg\phi$ and $x : \phi$, which can be closed by $(ax)$. 

If $Y\neq \varepsilon$, 
by $(\osuf R)$ on $x\cat Y : \false\osuf\true$ on the right side, we obtain the two branches:
\begin{enumerate}
\item[(x.1)] $x\cat Y : \neg\phi, x : \phi\Rightarrow Y : \true, Y : \false$. 
\item[(x.2)] $x\cat Y : \neg\phi, x : \phi\Rightarrow Y : \true, Y : \false\osuf\true$. 
\end{enumerate}
Both branches are terminal, since each of them carries $Y : \true$ on the right side, which is always provable. 

\end{proof}

\begin{prop}[Rule (xi)]
Rule (xi): 
$p\leftrightarrow \ofst p$, where $p$ is an atomic proposition, 
is derived in \GiiiPPLcyc. 
\end{prop}

\begin{proof}
Let $X\in \TVar$. 
We need to prove
\begin{enumerate}
\item[$(xi.a)$] $X : p\Rightarrow X : \ofst p$.
\item[$(xi.b)$] $X : \ofst p\Rightarrow X : p$.
\end{enumerate}
We prove only (xi.a). The other direction is similar. 

\textbf{Case for proving $(xi.a)$}: 
Without loss of generality, let $X = x\cat Y$. By the rule $(\ofst)$ on $x\cat Y : \ofst p$ on the right:
$$
x\cat Y : p\Rightarrow x : p.
$$
By the rule $(p)$ on $x\cat Y : p$ on the left:
$$
x : p\Rightarrow x : p,
$$
which is immediately closed by $(\textit{ax})$.
\ifx
\textbf{Case for proving $(xi.b)$}: 
Without loss of generality, let $X = x\cat Y$. By the rule $(\ofst)$ on $x\cat Y : \ofst p$ on the left:
$$
x : p\Rightarrow x\cat Y : p
$$
By the rule $(p)$ on $x\cat Y : p$ on the right:
$$
x : p\Rightarrow x : p
$$
which is immediately closed by $(\textit{ax})$.
\fi
\end{proof}

\begin{prop}[Rule (xii)]
Rule (xii): 
$\ofst \phi\wedge \la\alpha\ra\psi\leftrightarrow \la\alpha\ra(\ofst\phi\wedge \psi)$
is derived in \GiiiPPLcyc. 
\end{prop}

\begin{proof}

Let $X\in\TVar$.
We prove the following two directions by structural induction on $\alpha$. Throughout, recall that $\la\alpha\ra\theta\dddef \neg[\alpha]\neg\theta$. 
\begin{enumerate}
\item[$(xii.a)$] $X : \ofst \phi\wedge \la\alpha\ra\psi\Rightarrow X : \la\alpha\ra(\ofst\phi\wedge \psi)$.
\item[$(xii.b)$] $X : \la\alpha\ra(\ofst\phi\wedge \psi)\Rightarrow X : \ofst \phi\wedge \la\alpha\ra\psi$.
\end{enumerate}

\textbf{Case for proving $(xii.a)$}: 
By $(\wedge L)$ on the left, we need to show:
$$
X : \ofst\phi, X : \la\alpha\ra\psi\Rightarrow X : \la\alpha\ra(\ofst\phi\wedge\psi). 
$$

\textbf{Case $\alpha = a$ (atomic action).} Let $X = Y\cat x$. We need:
$$
(xii.a.a)\ \ Y\cat x : \ofst\phi, Y\cat x : \la a\ra\psi\Rightarrow Y\cat x : \la a\ra(\ofst\phi\wedge\psi)
$$
By $(\la a\ra R)$ and $(\la a\ra L)$, it suffices to show, for a new $y$:
$$
Y\cat x : \ofst\phi, Y\cat x\cat y : \psi\Rightarrow Y\cat x\cat y : \ofst\phi\wedge\psi
$$
By $(\wedge R)$ on the right, we need $Y\cat x\cat y : \ofst\phi$ and $Y\cat x\cat y : \psi$. The latter is on the left. For the former, let $Y=z\cat Z$. By $(\ofst)$ on $Y\cat x : \ofst\phi$ on the left, we obtain $z : \phi$; by $(\ofst)$ on $Y\cat x\cat y : \ofst\phi$ on the right, we need $z : \phi$, which is available. This case is closed. 

\textbf{Case $\alpha = \chi?$ (test).} Let $X = Y\cat x$. We need:
$$
(xii.a.\chi?)\ \ Y\cat x : \ofst\phi, Y\cat x : \la\chi?\ra\psi\Rightarrow Y\cat x : \la\chi?\ra(\ofst\phi\wedge\psi)
$$
By $(\la\phi?\ra L)$ on the left, we get $x\cat x : \chi$ and $Y\cat x\cat x : \psi$, then by $(\la\phi?\ra R)$ on the right, we obtain:
$$
Y\cat x : \ofst \phi, Y\cat x\cat x : \psi\Rightarrow Y\cat x\cat x : (\ofst \phi\wedge \psi). 
$$
As in the case $\alpha=a$ above, we can close the proof. 

\ifx
we need:
\begin{enumerate}
    \item[(xii.a.$\chi?$.1)] $Y\cat x : \ofst\phi, x\cat x : \chi, Y\cat x\cat x : \psi\Rightarrow x\cat x : \chi$, which is closed by $(\textit{ax})$.
    \item[(xii.a.$\chi?$.2)] $Y\cat x : \ofst\phi, x\cat x : \chi, Y\cat x\cat x : \psi\Rightarrow Y\cat x\cat x : (\ofst\phi\wedge\psi)$. By $(\wedge R)$, we need $Y\cat x\cat x : \ofst\phi$ and $Y\cat x\cat x : \psi$. The latter is on the left. For the former, let $Y=z\cat Z$, by $(\ofst)$ on $Y\cat x : \ofst\phi$ on the left we get $z : \phi$, and by $(\ofst)$ on $Y\cat x\cat x : \ofst\phi$ on the right we need $z : \phi$, which is available. This case is closed.
\end{enumerate}
\fi

\textbf{Case $\alpha = \alpha_1\seq\alpha_2$ (sequential composition).} 
We need to prove:
$$
(xii.a.;)\ \ X : \ofst\phi, X : \la\alpha_1\seq \alpha_2\ra\psi\Rightarrow X : \la\alpha_1\seq \alpha_2\ra(\ofst\phi\wedge\psi)
$$
By $(\la\seq\ra)$, we need:
$$
(xii.a.;.1)\ \ X : \ofst\phi, X : \la\alpha_1\ra\la\alpha_2\ra\psi\Rightarrow X : \la\alpha_1\ra\la\alpha_2\ra(\ofst\phi\wedge\psi)
$$
From (xii.a.;.1), by applying (cut), we need:
\begin{enumerate}
\item[(xii.a.;.11)] $X : \ofst\phi, X : \la\alpha_1\ra\la\alpha_2\ra\psi\Rightarrow X : \la\alpha_1\ra(\ofst\phi\wedge\la\alpha_2\ra\psi)$.
\item[(xii.a.;.12)]  $X : \la\alpha_1\ra(\ofst\phi\wedge\la\alpha_2\ra\psi)\Rightarrow X : \la\alpha_1\ra\la\alpha_2\ra(\ofst\phi\wedge\psi)$.
\end{enumerate}
By the induction hypothesis on $\alpha_1$, we close (xii.a.;.11). 
For (xii.a.;.12), by Necessitation Lemma we then need:
$$
X : (\ofst\phi\wedge\la\alpha_2\ra\psi)\Rightarrow X : \la\alpha_2\ra(\ofst \phi\wedge \psi)
$$
But this is straightforward by induction hypothesis on $\alpha_2$.

\ifx 
By the induction hypothesis for $\alpha_1$:
$$
(xii.;.11)\ \ X : \ofst\phi, X : \la\alpha_1\ra\la\alpha_2\ra\psi\Rightarrow X : \la\alpha_1\ra(\ofst\phi\wedge\la\alpha_2\ra\psi)
$$
is provable. 
So from (xii.;.1), by applying $(\textit{cut})$, we need to prove (xii.;.11) and :
$$
(xii.;.12)\ \ X : \la\alpha_1\ra(\ofst\phi\wedge\la\alpha_2\ra\psi)\Rightarrow X : \la\alpha_1\ra\la\alpha_2\ra(\ofst \phi\wedge \psi)
$$
By Necessitation Lemma we then need: 
$$
X : (\ofst\phi\wedge\la\alpha_2\ra\psi)\Rightarrow X : \la\alpha_2\ra(\ofst \phi\wedge \psi)
$$
But this is straightforward by induction hypothesis. 
\fi

\textbf{Case $\alpha = \alpha_1\cho\alpha_2$ (choice).} 
 We need to prove 
$$
(xii.a.\cup)\ \ X : \ofst\phi, X : \la\alpha_1\cho \alpha_2\ra\psi\Rightarrow X : \la\alpha_1\cho \alpha_2\ra(\ofst\phi\wedge\psi)
$$
By $(\la\cho\ra L)$ and $(\la\cho\ra R)$, we need to prove: 
\begin{enumerate}
\item[(xii.a.$\cup$.1)] $X : \ofst \phi, X : \la\alpha_1\ra\psi \Rightarrow X : \la\alpha_1\ra(\ofst\phi\wedge\psi), X : \la\alpha_2\ra(\ofst\phi\wedge\psi)$
\item[(xii.a.$\cup$.2)] $X : \ofst \phi, X : \la\alpha_2\ra\psi \Rightarrow X : \la\alpha_1\ra(\ofst\phi\wedge\psi), X : \la\alpha_2\ra(\ofst\phi\wedge\psi)$
\end{enumerate}

Both follow directly by weakening and the induction hypothesis. 
For example, from (xii.a.$\cup$.1), by $(wk)$, we need:
$$
X : \ofst \phi, X : \la\alpha_1\ra\psi \Rightarrow X : \la\alpha_1\ra(\ofst\phi\wedge\psi), 
$$
which is provable by induction hypothesis on $\alpha_1$. 

\textbf{Case $\alpha = \alpha_1^*$ (iteration).} 
We need to prove:
$$
(xii.a.*)\ \ X : \ofst\phi, X : \la\alpha_1^*\ra\psi\Rightarrow X : \la\alpha_1^*\ra(\ofst\phi\wedge\psi)
$$

By $(\la\lup\ra L)$, 
we split (xii.a.$*$) into:
\begin{enumerate}
\item[(xii.a.$*$.1)] $X : \ofst\phi, X : \la\true?\ra\psi\Rightarrow X:\la\alpha^*_1\ra(\ofst\phi\wedge \psi)$
\item[(xii.a.$*$.2)] $X : \ofst\phi, X : \la\alpha_1\ra\la\alpha_1^*\ra\psi\Rightarrow X : \la\alpha^*_1\ra(\ofst\phi\wedge \psi)$
\end{enumerate}

Case (xii.a.$*$.1): By $(\la\lup\ra R)$ on the right, 
it suffices to show 
$$
X : \ofst\phi, X : \la\true?\ra\psi\Rightarrow X:\la\true?\ra(\ofst\phi\wedge \psi). 
$$
Let $X = Y\cat x$. By $(\la\phi?\ra L)$ and $(\la\phi?\ra R)$, we have: 
$$
X : \ofst \phi, Y\cat x\cat x : \psi\Rightarrow Y\cat x\cat x : (\ofst \phi\wedge \psi). 
$$
This again is the case proved above when $\alpha = a$. 

\ifx
This is straightforward. Let $X = Y\cat x$, then by rule $(\la\phi?\ra R)$ on the right we have 
$$
X : \la\true?\ra\psi\Rightarrow x \cat x : \true, Y\cat x\cat x:(\ofst\phi\wedge \psi). 
$$
This closes since we have $x\cat x : \true$ on the right. 
\fi

Case (xii.a.$*$.2): 
By $(\la\lup\ra R)$ on the right, it suffices to show 
$$(xii.a.*.21)\ \ X : \ofst\phi, X : \la\alpha_1\ra\la\alpha^*_1\ra\psi\Rightarrow X : \la\alpha_1\ra\la\alpha_1^*\ra(\ofst\phi\wedge\psi). $$
From (xii.a.$*$.21), by (cut), we need to prove:
\begin{enumerate}
\item[(xii.a.$*$.211)] $X : \ofst\phi, X : \la\alpha_1\ra\la\alpha_1^*\ra\psi\Rightarrow X : \la\alpha_1\ra(\ofst\phi\wedge\la\alpha_1^*\ra\psi)$.
It can be closed directly by induction hypothesis on $\alpha_1$. 
\item[(xii.a.$*$.212)]  $X : \la\alpha_1\ra(\ofst\phi\wedge\la\alpha_1^*\ra\psi)\Rightarrow X : \la\alpha_1\ra\la\alpha_1^*\ra(\ofst\phi\wedge\psi)$. 
\ifLuaTeX
By Necessitation Lemma, we obtain:
$$
(xii.a.*.2121)\ \ X : (\ofst\phi\wedge\la\alpha_1^*\ra\psi)\Rightarrow X : \la\alpha_1^*\ra(\ofst\phi\wedge\psi),
$$
which is exactly (xii.a.$*$) (after applying rule $(\wedge L)$).
\VE And it is a back-link since from (xii.a.$*$) to (xii.a.*.2121), Necessitation Lemma is applied. \EVE
\MO It forms a back-link because the Necessitation lemma is applied on the path from (xii.a.$*$) to (xii.a.*.2121). \EMO
\fi
\end{enumerate}

For (xii.a.$*$.212), by Necessitation Lemma, we obtain:
$$
(xii.a.*.2121)\ \ X : (\ofst\phi\wedge\la\alpha_1^*\ra\psi)\Rightarrow X : \la\alpha_1^*\ra(\ofst\phi\wedge\psi),
$$
which is exactly (xii.a.$*$) (after applying rule $(\wedge L)$).
The derivation path from (xii.a.$*$) to (xii.a.*.2121) is progressing because the Necessitation lemma is applied on this path. 

\ifx
By the induction hypothesis for $\alpha_1$ (with $\la\alpha_1^*\ra\psi$ in place of $\psi$), we know that
$$
(xii.a.*.211) \ \ X : \ofst\phi, X : \la\alpha_1\ra\la\alpha_1^*\ra\psi\Rightarrow X : \la\alpha_1\ra(\ofst\phi\wedge\la\alpha_1^*\ra\psi)
$$ is provable. 
So, from (xii.a.$*$.21), by $(cut)$, 
we need to prove (xii.a.*.211) and:
$$
(xii.a.*.212)\ \ X : \la\alpha_1\ra(\ofst\phi\wedge\la\alpha_1^*\ra\psi)\Rightarrow X : \la\alpha_1\ra\la\alpha_1^*\ra(\ofst\phi\wedge\psi)
$$
From (xii.a.$*$.212), by Necessitation Lemma, we obtain:
$$
(xii.a.*.2121)\ \ X : (\ofst\phi\wedge\la\alpha_1^*\ra\psi)\Rightarrow X : \la\alpha_1^*\ra(\ofst\phi\wedge\psi)
$$
which is exactly (xii.a.$*$) (after applying rule $(\wedge L)$).
\VE And it is a back-link since from (xii.a.$*$) to (xii.a.*.2121), Necessitation Lemma is applied. \EVE
\MO It forms a back-link because the Necessitation lemma is applied on the path from (xii.a.$*$) to (xii.a.*.2121). \EMO
\fi

\textbf{Case for proving $(xii.b)$}: 
That is, $X : \la\alpha\ra(\ofst\phi\wedge \psi)\Rightarrow X : \ofst \phi\wedge \la\alpha\ra\psi$. 

\textbf{Case $\alpha = a$ (atomic action).} Let $X = Y\cat x$. We need:
$$
(xii.b.a)\ \ Y\cat x : \la a\ra(\ofst\phi\wedge\psi)\Rightarrow Y\cat x : \ofst\phi\wedge \la a\ra\psi
$$
By $(\wedge R)$ on the right, we split into:
\begin{enumerate}
    \item[(xii.b.a.1)] $Y\cat x : \la a\ra(\ofst\phi\wedge\psi)\Rightarrow Y\cat x : \ofst\phi$. 
    
    By $(\la a\ra L)$, it suffices to show $$x\R_a y, Y\cat x\cat y : (\ofst\phi\wedge\psi)\Rightarrow Y\cat x : \ofst\phi$$ with a new $y$. 
    By $(\wedge L)$, we have $Y\cat x\cat y : \ofst\phi$. Let $Y = z\cat Z$. By $(\ofst)$ on $Y\cat x\cat y : \ofst\phi$ on the left, we get $z : \phi$. By $(\ofst)$ on $Y\cat x : \ofst\phi$ on the right, we need $z : \phi$, which is available. This case is closed.
    
    \item[(xii.b.a.2)] $Y\cat x : \la a\ra(\ofst\phi\wedge\psi)\Rightarrow Y\cat x : \la a\ra\psi$. 
    
By the Necessitation lemma, we obtain $X : \ofst\phi\wedge\psi\Rightarrow X : \psi$, which is trivially derivable by $(\wedge L)$ and $(\textit{ax})$. 
\end{enumerate}

\textbf{Case $\alpha = \chi?$ (test).} Let $X = Y\cat x$. We need:
$$
(xii.b.\chi?)\ \ Y\cat x : \la\chi?\ra(\ofst\phi\wedge\psi)\Rightarrow Y\cat x : \ofst\phi\wedge \la\chi?\ra\psi
$$
By $(\la\phi?\ra L)$ on the left, we get $x\cat x : \chi$ and $Y\cat x\cat x : \ofst\phi\wedge\psi$. By $(\wedge R)$ on the right:
\begin{enumerate}
    \item[(xii.b.$\chi?$.1)] $x\cat x : \chi, Y\cat x\cat x : \ofst\phi\wedge\psi\Rightarrow Y\cat x : \ofst\phi$. This case is similar to the proof of (xii.b.a.1) above. 
    
    
    \item[(xii.b.$\chi?$.2)] $x\cat x : \chi, Y\cat x\cat x : \ofst\phi\wedge\psi\Rightarrow Y\cat x : \la\chi?\ra\psi$.
    
    By $(\la\phi?\ra R)$ on the right, we need:
    $$
    Y\cat x\cat x : \ofst\phi\wedge\psi\Rightarrow Y\cat x \cat x: \psi. 
    $$
    This is straightforward by $(\wedge L)$ and $(ax)$. 
\end{enumerate}

\textbf{Case $\alpha = \alpha_1\seq\alpha_2$ (sequential composition).} We need:
$$
(xii.b.;)\ \ X : \la\alpha_1\seq\alpha_2\ra(\ofst\phi\wedge\psi)\Rightarrow X : \ofst\phi\wedge \la\alpha_1\seq\alpha_2\ra\psi
$$
By $(\la\seq\ra)$, this becomes:
$$
X : \la\alpha_1\ra\la\alpha_2\ra(\ofst\phi\wedge\psi)\Rightarrow X : \ofst\phi\wedge \la\alpha_1\ra\la\alpha_2\ra\psi
$$
By (cut), we need to prove:
\begin{enumerate}
\item[(xii.b.;.1)] $X : \la\alpha_1\ra\la\alpha_2\ra(\ofst\phi\wedge\psi)\Rightarrow X : \la\alpha_1\ra(\ofst\phi\wedge\la\alpha_2\ra\psi)$.
    This is closed by the induction hypothesis for $\alpha_1$ (with $\la\alpha_2\ra\psi$ in place of $\psi$). 
\item[(xii.b.;.2)] $X : \la\alpha_1\ra(\ofst\phi\wedge\la\alpha_2\ra\psi)\Rightarrow X : \ofst\phi\wedge\la\alpha_1\ra\la\alpha_2\ra\psi$. By Necessitation Lemma we have
$$X : \la\alpha_2\ra(\ofst\phi\wedge\psi)\Rightarrow X : \ofst\phi\wedge\la\alpha_2\ra\psi, $$
which is also closed by the induction hypothesis on $\alpha_2$. 
\end{enumerate}

\ifx
By the induction hypothesis for $\alpha_2$, $X : \la\alpha_2\ra(\ofst\phi\wedge\psi)\Rightarrow X : \ofst\phi\wedge\la\alpha_2\ra\psi$ is provable. By Necessitation Lemma:
$$
(xii.b.;.1)\ \ X : \la\alpha_1\ra\la\alpha_2\ra(\ofst\phi\wedge\psi)\Rightarrow X : \la\alpha_1\ra(\ofst\phi\wedge\la\alpha_2\ra\psi)
$$
is provable.
By the induction hypothesis for $\alpha_1$ (with $\la\alpha_2\ra\psi$ in place of $\psi$):
$$
(xii.b.;.2)\ \ X : \la\alpha_1\ra(\ofst\phi\wedge\la\alpha_2\ra\psi)\Rightarrow X : \ofst\phi\wedge\la\alpha_1\ra\la\alpha_2\ra\psi
$$
is provable. 
The result follows by $(\textit{cut})$ on (xii.b.;.1) and (xii.b.;.2).
\fi

\textbf{Case $\alpha = \alpha_1\cho\alpha_2$ (choice).} We need:
$$
(xii.b.\cup)\ \ X : \la\alpha_1\cho\alpha_2\ra(\ofst\phi\wedge\psi)\Rightarrow X : \ofst\phi\wedge\la\alpha_1\cho\alpha_2\ra\psi
$$
By $(\la\cho\ra L)$, we need to prove:
\begin{enumerate}
\item[(xii.b.$\cho$.1)] $X : \la\alpha_1\ra(\ofst\phi\wedge\psi)\Rightarrow X : \ofst\phi\wedge\la\alpha_1\cho\alpha_2\ra\psi$.
\item[(xii.b.$\cho$.2)]  $X : \la\alpha_2\ra(\ofst\phi\wedge\psi)\Rightarrow X : \ofst\phi\wedge\la\alpha_1\cho\alpha_2\ra\psi$.
\end{enumerate}
Take (xii.b.$\cho$.1) as an example. The other case is similar. 
by (cut) we need:
\begin{enumerate}
\item[(xii.b.$\cho$.11)]  $X : \la\alpha_1\ra(\ofst\phi\wedge\psi)\Rightarrow X : \ofst \phi\wedge \la\alpha_1\ra\psi$. It is closed by induction hypothesis on $\alpha_1$. 
\item[(xii.b.$\cho$.12)]  $X : \ofst \phi\wedge \la\alpha_1\ra\psi\Rightarrow X : \ofst\phi\wedge\la\alpha_1\cho\alpha_2\ra\psi$. This can be easily proved by $(\wedge L)$, $(\wedge R)$, $(\la\cho\ra R)$ and $(ax)$. 
\end{enumerate}
\ifx
$X : \la\alpha_1\cho\alpha_2\ra(\ofst\phi\wedge\psi)$ gives $X : \la\alpha_1\ra(\ofst\phi\wedge\psi)$ or $X : \la\alpha_2\ra(\ofst\phi\wedge\psi)$. Consider the sub-case: 
$$
X : \la\alpha_1\ra(\ofst\phi\wedge\psi)\Rightarrow X : \ofst\phi\wedge\la\alpha_1\ra\psi. 
$$
By the induction hypothesis for $\alpha_1$, it is closed. 
Since $\la\alpha_1\ra\psi$ implies $\la\alpha_1\cho\alpha_2\ra\psi$ (by $([\cho])$), we obtain $X : \ofst\phi\wedge\la\alpha_1\cho\alpha_2\ra\psi$ by $(\textit{cut})$. The sub-case for $\alpha_2$ is symmetric.
\fi

\textbf{Case $\alpha = \alpha_1^*$ (iteration).} We need:
$$
(xii.b.*)\ \ X : \la\alpha_1^*\ra(\ofst\phi\wedge\psi)\Rightarrow X : \ofst\phi\wedge\la\alpha_1^*\ra\psi
$$
By $(\la\lup\ra L)$, 
we need to consider two cases:
\begin{enumerate}
\item[ (xii.b.$*$.1) ] $X : \la\true?\ra(\ofst\phi\wedge\psi)\Rightarrow X : \ofst \phi\wedge \la\alpha^*_1\ra\psi$.
\item[ (xii.b.$*$.2) ] $X : \la\alpha_1\ra\la\alpha_1^*\ra(\ofst\phi\wedge\psi)\Rightarrow X : \ofst \phi\wedge \la\alpha^*_1\ra\psi$.
\end{enumerate}
For (xii.b.$*$.1), 
let $X = Y\cat x$. Then, by $(\la\phi?\ra L)$, $Y : \la\true?\ra(\ofst\phi\wedge\psi)$ gives 
$x \cat x : \true, Y\cat x\cat x : (\ofst \phi\wedge \psi)$. 
By $(\wedge R)$, we need to prove: 
\begin{enumerate}
\item[(xii.b.$*$.11)] $x \cat x : \true, Y\cat x\cat x : (\ofst \phi\wedge \psi)\Rightarrow Y\cat x : \ofst \phi$.  This can be obtained by applying $(\ofst)$ on both sides since they both contain $\ofst \phi$ (after letting $Y = z\cat Z$). 
\item[(xii.b.$*$.12)] $x \cat x : \true, Y\cat x\cat x : (\ofst \phi\wedge \psi)\Rightarrow Y\cat x : \la\alpha^*_1\ra\psi$. 
By $(\la\lup\ra R)$ on $Y\cat x : \la\alpha^*_1\ra\psi$ and then $(\la\phi?\ra R)$, we obtain 
$$
Y\cat x\cat x : (\ofst \phi\wedge \psi)\Rightarrow Y\cat x \cat x: \psi
$$
which is clearly closed since both sides contain $Y\cat x\cat x : \psi$. 
\end{enumerate}

\ifx
This can be obtained by applying $(\la\lup\ra)$ on the right side to match $x\cat x : \true$.

(vii.b.$*$.11) can be obtained by applying $(\ofst)$ on both sides since they both contain $\ofst \phi$. 
(vii.b.$*$.12) can be obtained by applying $(\la\lup\ra)$ on the right side to match $x\cat x : \true$.
\fi

For (xii.b.$*$.2),  by (cut) we need to prove:
\begin{enumerate}
\item[(xii.b.*.21)] $X : \la\alpha_1\ra(\ofst\phi\wedge\la\alpha_1^*\ra\psi)\Rightarrow X : \ofst\phi\wedge\la\alpha_1\ra\la\alpha_1^*\ra\psi$. This can be directly closed by the induction hypothesis for $\alpha_1$.  
\item[(xii.b.*.22)] $X : \la\alpha_1\ra\la\alpha_1^*\ra(\ofst\phi\wedge\psi)\Rightarrow X : \la\alpha_1\ra(\ofst\phi\wedge\la\alpha_1^*\ra\psi)$. 
By the Necessitation lemma, this gives 
$$
(xii.b.*.221)\ \ X : \la\alpha_1^*\ra(\ofst\phi\wedge\psi)\Rightarrow X : (\ofst\phi\wedge\la\alpha_1^*\ra\psi), 
$$ 
which is exactly (xii.b.$*$). 
The derivaiton path from (xii.b.$*$) to (xii.b.$*$.221) is thus progressing because the Necessitation lemma is applied along this path. 
\end{enumerate}
\ifx
For (xii.b.$*$.2),  by the induction hypothesis for $\alpha_1$, we know that
$$
(xii.b.*.21)\ \ X : \la\alpha_1\ra(\ofst\phi\wedge\la\alpha_1^*\ra\psi)\Rightarrow X : \ofst\phi\wedge\la\alpha_1\ra\la\alpha_1^*\ra\psi
$$ is provable. 
So by $(cut)$ and that fact that $\la\alpha_1\ra\la\alpha^*_1\ra\psi$ actually implies $\la\alpha^*_1\ra\psi$, to prove (vii.b.$*$.2) it is sufficient to prove both (xii.b.$*$.21) (which has been proved by the induction hypothesis) and:
$$
(xii.b.*.22)\ \  X : \la\alpha_1\ra\la\alpha_1^*\ra(\ofst\phi\wedge\psi)\Rightarrow X : \la\alpha_1\ra(\ofst\phi\wedge\la\alpha_1^*\ra\psi). 
$$
By Necessitation Lemma, (xii.b.*.22) gives 
$$
(xii.b.*.221)\ \ X : \la\alpha_1^*\ra(\ofst\phi\wedge\psi)\Rightarrow X : (\ofst\phi\wedge\la\alpha_1^*\ra\psi), 
$$ 
which is exactly (xii.b.*). 
We thus obtain a cyclic derivation from (xii.b.$*$) to (xii.b.$*$.211) since during which Necessitation Lemma is applied. 
\fi
\end{proof}

\begin{prop}[Rule (xiii)]
Rule (xiii): 
$\nxt \la\alpha\ra \phi\leftrightarrow \nxt \true \wedge \la\alpha\ra \bnxt \phi$ 
is derived in \GiiiPPLcyc. 
\end{prop}

\begin{proof}

Let $X\in \TVar$. 
Recall that $\nxt\theta\dddef \false\osuf\theta$ and $\bnxt\theta\dddef \neg\nxt\neg\theta = \neg(\false\osuf\neg\theta)$. We prove the following two directions by structural induction on $\alpha$. 
\begin{enumerate}
\item[$(xiii.a)$]  $X : \nxt\la\alpha\ra\phi\Rightarrow X : \nxt\true\wedge \la\alpha\ra\bnxt\phi$.
\item[$(xiii.b)$] $X : \nxt\true\wedge \la\alpha\ra\bnxt\phi\Rightarrow X : \nxt\la\alpha\ra\phi$.
\end{enumerate}

Every trace formula occurring in this proof is of the form $\false\osuf\theta$. 
Below we therefore record, once and for all, how the rules $(\osuf L)$ and $(\osuf R)$ act on such a formula. 
Both derived steps are used repeatedly in the rest of this proof. 
\begin{enumerate}
\item[$(\nxt L)$] Applying $(\osuf L)$ to $x\cat \sigma : \false\osuf\theta$ produces the two premises $\Gamma, \sigma : \theta\Rightarrow \Delta$ and $\Gamma, \sigma : \false, \sigma : \false\osuf\theta\Rightarrow \Delta$. 
The latter is terminal, since it carries $\sigma : \false$ on the left side. 
Hence it suffices to replace $x\cat \sigma : \false\osuf\theta$ on the left side by $\sigma : \theta$. 
\item[$(\nxt R)$] Applying $(\osuf R)$ to $x\cat \sigma : \false\osuf\theta$ produces the two premises $\Gamma\Rightarrow \sigma : \theta, \sigma : \false, \Delta$ and $\Gamma\Rightarrow \sigma : \theta, \sigma : \false\osuf\theta, \Delta$. 
Both of them carry $\sigma : \theta$ on the right side, so after $(wk)$ it suffices to replace $x\cat \sigma : \false\osuf\theta$ on the right side by $\sigma : \theta$. 
\end{enumerate}
Intuitively, $(\nxt L)$ and $(\nxt R)$ express that on a trace $x\cat\sigma$ the formula $\nxt\theta$ simply asserts $\theta$ on the immediate suffix $\sigma$, the second disjunct of the unfolding being blocked by $\false$. 

\textbf{Case for proving $(xiii.a)$}:
That is, $X : \false\osuf\la\alpha\ra\phi\Rightarrow X : (\false\osuf\true)\wedge \la\alpha\ra\neg(\false\osuf\neg\phi)$.

Without loss of generality, let $X = x\cat Y$. 
By applying $(\nxt L)$ to $x\cat Y : \false\osuf\la\alpha\ra\phi$ on the left, we obtain: 
$$
Y : \la\alpha\ra\phi\Rightarrow x\cat Y : (\false\osuf\true)\wedge \la\alpha\ra\neg(\false\osuf\neg\phi)
$$
By $(\wedge R)$ on the right, we split into:
\begin{enumerate}
    \item[(xiii.a.1)] $Y : \la\alpha\ra\phi\Rightarrow x\cat Y : \false\osuf\true$. 
    By $(\nxt R)$ on $x\cat Y : \false\osuf\true$ on the right, we need $Y : \true$, which is always provable. 
    
    \item[(xiii.a.2)] $Y : \la\alpha\ra\phi\Rightarrow x\cat Y : \la\alpha\ra\neg(\false\osuf\neg\phi)$. 
\end{enumerate}

For the case (xiii.a.2), we proceed by structural induction on $\alpha$. 

\textbf{Case $\alpha = a$ (atomic action).} Let $Y = Z\cat y$. We need:
$$
Z\cat y : \la a\ra\phi\Rightarrow x\cat Z\cat y : \la a\ra\neg(\false\osuf\neg\phi). 
$$
By $(\la a\ra L)$ on the left, we get $y\R_a z'$ and $Z\cat y\cat z' : \phi$, with $z'$ fresh. By $(\la a\ra R)$ on the right (with the same $z'$), we need $x\cat Z\cat y\cat z' : \neg(\false\osuf\neg\phi)$. By $(\neg R)$:
$$
Z\cat y\cat z' : \phi, x\cat Z\cat y\cat z' : \false\osuf\neg\phi\Rightarrow \cdot
$$
By $(\nxt L)$ on $x\cat Z\cat y\cat z' : \false\osuf\neg\phi$ on the left, we obtain $Z\cat y\cat z' : \neg\phi$. Then $Z\cat y\cat z' : \phi$ and $Z\cat y\cat z' : \neg\phi$ give a contradiction, closing the branch by $(\neg L)$ and $(\textit{ax})$. 

\textbf{Case $\alpha = \chi?$ (test).} Let $Y = Z\cat y$. We need:
$$
Z\cat y : \la\chi?\ra\phi\Rightarrow x\cat Z\cat y : \la\chi?\ra\neg(\false\osuf\neg\phi)
$$
By $(\la\phi?\ra L)$ on the left, we get $y\cat y : \chi$ and $Z\cat y\cat y : \phi$. By $(\la\phi?\ra R)$ on the right and then $(\neg R)$, we need:
$$
Z\cat y\cat y : \phi, x\cat Z\cat y\cat y : \false\osuf\neg\phi\Rightarrow \cdot. 
$$
By $(\nxt L)$, we obtain $Z\cat y\cat y : \neg\phi$ on the left side, which contradicts $Z\cat y\cat y : \phi$ and closes the branch. 

\ifx
\begin{enumerate}
    \item[(xiii.a.$\chi?$.1)] $Z\cat y\cat y : \phi\Rightarrow x\cat Z\cat y\cat y : \neg(\false\osuf\neg\phi)$, closed by $(\textit{ax})$.
    \item[(xiii.a.$\chi?$.2)] $Z\cat y\cat y : \phi\Rightarrow x\cat Z\cat y\cat y : \neg(\false\osuf\neg\phi)$. By $(\neg R)$, we need $Z\cat y\cat y : \phi, x\cat Z\cat y\cat y : \false\osuf\neg\phi\Rightarrow \cdot$. By $(\osuf)$, since $\false\wedge(...)$ is trivially false, we get $Z\cat y\cat y : \neg\phi$. Contradiction with $Z\cat y\cat y : \phi$. Closed.
\end{enumerate}
\fi

\textbf{Case $\alpha = \alpha_1\seq\alpha_2$ (sequential composition).} 
We need to prove:
$$
Y : \la\alpha_1\seq\alpha_2\ra\phi\Rightarrow x\cat Y : \la\alpha_1\seq\alpha_2\ra\neg(\false\osuf\neg\phi). 
$$
By $(\la\seq\ra)$:
$$
Y : \la\alpha_1\ra\la\alpha_2\ra\phi\Rightarrow x\cat Y : \la\alpha_1\ra\la\alpha_2\ra\neg(\false\osuf\neg\phi). 
$$
By (cut), we obtain:
\begin{enumerate}
\item[(xiii.a.2.;.1)] $Y : \la\alpha_1\ra\la\alpha_2\ra\phi\Rightarrow x\cat Y : \la\alpha_1\ra\neg(\false\osuf\neg\la\alpha_2\ra\phi)$, which is closed by the induction hypothesis on $\alpha_1$ (with $\la\alpha_2\ra\phi$ in place of $\phi$). 
\item[(xiii.a.2.;.2)] $x\cat Y : \la\alpha_1\ra\neg(\false\osuf\neg\la\alpha_2\ra\phi)\Rightarrow x\cat Y : \la\alpha_1\ra\la\alpha_2\ra\neg(\false\osuf\neg\phi)$. 
\end{enumerate}
For (xiii.a.2.;.2), by Necessitation Lemma, it suffices to show:
$$
(xiii.a.2.;.21)\ \ x\cat Y : \neg(\false\osuf\neg\la\alpha_2\ra\phi)\Rightarrow x\cat Y : \la\alpha_2\ra\neg(\false\osuf\neg\phi). 
$$
By (cut) again, we need: 
\begin{enumerate}
\item[(xiii.a.2.;.211)] $x\cat Y : \neg(\false\osuf\neg\la\alpha_2\ra\phi)\Rightarrow Y : \la\alpha_2\ra\phi$. By $(\neg L)$ the left-side formula is moved to the right, and then $(\nxt R)$ replaces it by $Y : \neg\la\alpha_2\ra\phi$. The right side then carries both $Y : \neg\la\alpha_2\ra\phi$ and $Y : \la\alpha_2\ra\phi$, so the branch is closed by $(\neg R)$ and $(ax)$. 
\item[(xiii.a.2.;.212)] $Y : \la\alpha_2\ra\phi\Rightarrow x\cat Y : \la\alpha_2\ra\neg(\false\osuf\neg\phi)$. This is exactly provable by induction hypothesis on $\alpha_2$. 
\end{enumerate}

\ifx
By the induction hypothesis for $\alpha_1$ (with $\la\alpha_2\ra\phi$ in place of $\phi$):
$$
(xiii.a.;.1)\ \ Y : \la\alpha_1\ra\la\alpha_2\ra\phi\Rightarrow x\cat Y : \la\alpha_1\ra\neg(\false\osuf\neg\la\alpha_2\ra\phi)
$$
is provable. We also need:
$$
(xiii.a.;.2)\ \ x\cat Y : \la\alpha_1\ra\neg(\false\osuf\neg\la\alpha_2\ra\phi)\Rightarrow x\cat Y : \la\alpha_1\ra\la\alpha_2\ra\neg(\false\osuf\neg\phi)
$$
so that the result can follow by $(\textit{cut})$. 
By Necessitation Lemma, it suffices to show:
$$
(xiii.a.;.21)\ \ x\cat Y : \neg(\false\osuf\neg\la\alpha_2\ra\phi)\Rightarrow x\cat Y : \la\alpha_2\ra\neg(\false\osuf\neg\phi)
$$
By $(cut)$, it is equivalent to show
\begin{enumerate}
\item[(xiii.a.;.211)] $x\cat Y : \neg(\false\osuf\neg\la\alpha_2\ra\phi)\Rightarrow Y : \la\alpha_2\ra\phi$.
\item[(xiii.a.;.212)] $Y : \la\alpha_2\ra\phi\Rightarrow x\cat Y : \la\alpha_2\ra\neg(\false\osuf\neg\phi)$.
\end{enumerate}
(xiii.a.;.211) is not hard to archieve by applying $(\neg R)$ and $(\osuf)$ on the left to get both $Y : \neg\la\alpha_2\ra\phi$ and $Y : \la\alpha_2\ra\phi$ on the left, which closes the branch. 

(xiii.a.;.212) is exactly provable by induction hypothesis on $\alpha_2$. 
\fi

\textbf{Case $\alpha = \alpha_1\cho\alpha_2$ (choice).} 
We need 
$$
(xiii.a.2.\cho)\ \ Y : \la\alpha_1\cup \alpha_2\ra\phi\Rightarrow x\cat Y : \la\alpha_1\cup \alpha_2\ra\neg(\false\osuf\neg\phi).
$$
By $(\la\cho\ra L)$, we obtain 
\begin{enumerate}
\item[(xiii.a.2.$\cho$.1)] $Y : \la\alpha_1\ra\phi\Rightarrow x\cat Y : \la\alpha_1\cup \alpha_2\ra\neg(\false\osuf\neg\phi)$.
\item[(xiii.a.2.$\cho$.2)] $Y : \la\alpha_2\ra\phi\Rightarrow x\cat Y : \la\alpha_1\cup \alpha_2\ra\neg(\false\osuf\neg\phi)$. 
\end{enumerate}
Take (xiii.a.$\cho$.1) as an example. The other case is similar. 
By $(\la\cho\ra R)$ on the right we obtain
$$
Y : \la\alpha_1\ra\phi\Rightarrow x\cat Y : \la\alpha_1\ra\neg(\false\osuf\neg \phi). 
$$
 This is exactly the statement (xiii.a.2) with $\alpha_1$ in place of $\alpha$, hence it is closed by the induction hypothesis on $\alpha_1$. 
\ifx
By $(\la\cho\ra)$, $Y : \la\alpha_1\cho\alpha_2\ra\phi$ gives $Y : \la\alpha_1\ra\phi$ or $Y : \la\alpha_2\ra\phi$. In the sub-case $Y : \la\alpha_1\ra\phi$: by the induction hypothesis for $\alpha_1$, we get 
$x\cat Y : \la\alpha_1\ra\neg(\false\osuf\neg\phi)$.  
Since $Y : \la\alpha_1\ra\theta$ implies $Y : \la\alpha_1\cho\alpha_2\ra\theta$ (by $(\la\cho\ra)$), we obtain $x\cat Y : \la\alpha_1\cho\alpha_2\ra\neg(\false\osuf\neg\phi)$. The sub-case for $\alpha_2$ is symmetric.
\fi

\textbf{Case $\alpha = \alpha_1^*$ (iteration).} 
We need to prove
$$
(xiii.a.2.*)\ \ Y : \la\alpha^*_1\ra\phi\Rightarrow x\cat Y : \la\alpha^*_1\ra\neg(\false\osuf\neg\phi).
$$
By $(\la\lup\ra L)$, 
we need
\begin{enumerate}
\item[(xiii.a.2.*.1)] $Y : \la\true?\ra\phi\Rightarrow x\cat Y : \la\alpha^*_1\ra\neg(\false\osuf\neg\phi)$. 
\item[(xiii.a.2.*.2)]  $Y : \la\alpha_1\ra\la\alpha_1^*\ra\phi\Rightarrow x\cat Y : \la\alpha^*_1\ra\neg(\false\osuf\neg\phi)$. 
\end{enumerate}

For (xiii.a.2.*.1), we obtain $Z\cat y\cat y : \phi$ by letting $Y = Z\cat y$ and applying $(\la\phi?\ra L)$. By $(\la\lup\ra R)$ on the right, it suffices to show $x\cat Y : \la\true?\ra\neg(\false\osuf\neg\phi)$, i.e., $x\cat Z\cat y\cat y : \neg(\false\osuf\neg\phi)$. By $(\neg R)$: 
$$Z\cat y\cat y : \phi, x\cat Z\cat y\cat y : \false\osuf\neg\phi\Rightarrow \cdot.$$ 
By $(\nxt L)$, we get $Z\cat y\cat y : \neg\phi$ on the left side. Contradiction with $Z\cat y\cat y : \phi$. Closed.

For (xiii.a.2.*.2): By $(\la\lup\ra R)$ on the right, it suffices to show 
$$Y : \la\alpha_1\ra\la\alpha^*_1\ra\phi\Rightarrow x\cat Y : \la\alpha_1\ra\la\alpha_1^*\ra\neg(\false\osuf\neg\phi).$$
By (cut), we need:
\begin{enumerate}
\item[(xiii.a.2.*.21)] $Y : \la\alpha_1\ra\la\alpha_1^*\ra\phi\Rightarrow x\cat Y : \la\alpha_1\ra\neg(\false\osuf\neg\la\alpha_1^*\ra\phi)$, which is closed by the induction hypothesis on $\alpha_1$. 
\item[(xiii.a.2.*.22)] $x\cat Y : \la\alpha_1\ra\neg(\false\osuf\neg\la\alpha_1^*\ra\phi)\Rightarrow x\cat Y : \la\alpha_1\ra\la\alpha_1^*\ra\neg(\false\osuf\neg\phi)$. 
\end{enumerate}
From (xiii.a.2.*.22), by Necessitation Lemma, it suffices to show:
$$
x\cat Y : \neg(\false\osuf\neg\la\alpha_1^*\ra\phi)\Rightarrow x\cat Y : \la\alpha_1^*\ra\neg(\false\osuf\neg\phi). 
$$
By $(\neg L)$: $$\cdot\Rightarrow x\cat Y : \false\osuf\neg\la\alpha_1^*\ra\phi, x\cat Y : \la\alpha_1^*\ra\neg(\false\osuf\neg\phi).$$ 
By $(\nxt R)$ on $x\cat Y : \false\osuf\neg\la\alpha_1^*\ra\phi$, this formula is replaced by $Y : \neg\la\alpha_1^*\ra\phi$ on the right side, i.e., by $(\neg R)$: $$Y : \la\alpha_1^*\ra\phi\Rightarrow x\cat Y : \la\alpha_1^*\ra\neg(\false\osuf\neg\phi),$$ which is exactly (xiii.a.2.*), forming a back-link. The derivation path is progressing since Necessitation Lemma is applied. 

\ifx
By the induction hypothesis for $\alpha_1$:
$$
(xiii.a.*.21)\ \ Y : \la\alpha_1\ra\la\alpha_1^*\ra\phi\Rightarrow x\cat Y : \la\alpha_1\ra\neg(\false\osuf\neg\la\alpha_1^*\ra\phi)
$$
is provable. We also need:
$$
(xiii.a.*.22)\ \ x\cat Y : \la\alpha_1\ra\neg(\false\osuf\neg\la\alpha_1^*\ra\phi)\Rightarrow x\cat Y : \la\alpha_1\ra\la\alpha_1^*\ra\neg(\false\osuf\neg\phi)
$$
in order to prove (xiii.a.*.2) by (cut). 
By Necessitation Lemma, it suffices to show:
$$
x\cat Y : \neg(\false\osuf\neg\la\alpha_1^*\ra\phi)\Rightarrow x\cat Y : \la\alpha_1^*\ra\neg(\false\osuf\neg\phi)
$$
By $(\neg R)$ and $(\neg L)$: $\cdot\Rightarrow x\cat Y : \false\osuf\neg\la\alpha_1^*\ra\phi, x\cat Y : \la\alpha_1^*\ra\neg(\false\osuf\neg\phi)$. By $(\osuf)$ on $x\cat Y : \false\osuf\neg\la\alpha_1^*\ra\phi$, since $\false\wedge(...)$ is trivially false, we need $Y : \neg\la\alpha_1^*\ra\phi$, i.e., by $(\neg R)$: $Y : \la\alpha_1^*\ra\phi\Rightarrow x\cat Y : \la\alpha_1^*\ra\neg(\false\osuf\neg\phi)$, which is exactly (xiii.a.*), forming a back-link. The derivation path is progressing since Necessitation Lemma is applied. 
\fi

\textbf{Case for proving $(xiii.b)$}: 
That is, $X : (\false\osuf\true)\wedge \la\alpha\ra\neg(\false\osuf\neg\phi)\Rightarrow X : \false\osuf\la\alpha\ra\phi$.

By $(\wedge L)$ on the left:
$$
X : \false\osuf\true, X : \la\alpha\ra\neg(\false\osuf\neg\phi)\Rightarrow X : \false\osuf\la\alpha\ra\phi. 
$$
Without loss of generality, let $X = x\cat Y$. By $(\nxt L)$ on $x\cat Y : \false\osuf\true$ on the left, this formula is replaced by $Y : \true$, which is trivially satisfied and can be discarded by $(wk)$. By $(\nxt R)$ on $x\cat Y : \false\osuf\la\alpha\ra\phi$ on the right, we need $Y : \la\alpha\ra\phi$. So it suffices to show:
$$
x\cat Y : \la\alpha\ra\neg(\false\osuf\neg\phi)\Rightarrow Y : \la\alpha\ra\phi
$$

We prove this by structural induction on $\alpha$.

\textbf{Case $\alpha = a$ (atomic action).} Let $Y = Z\cat y$. We need:
$$
x\cat Z\cat y : \la a\ra\neg(\false\osuf\neg\phi)\Rightarrow Z\cat y : \la a\ra\phi
$$
By $(\la a\ra L)$ on the left, we get $y\R_a z'$ and $x\cat Z\cat y\cat z' : \neg(\false\osuf\neg\phi)$. By $(\la a\ra R)$ on the right (with the same $z'$), we need $Z\cat y\cat z' : \phi$. By $(\neg L)$ on $x\cat Z\cat y\cat z' : \neg(\false\osuf\neg\phi)$, we need to show $$\cdot\Rightarrow x\cat Z\cat y\cat z' : \false\osuf\neg\phi, Z\cat y\cat z' : \phi.$$ By $(\nxt R)$ on $x\cat Z\cat y\cat z' : \false\osuf\neg\phi$, this formula is replaced by $Z\cat y\cat z' : \neg\phi$. So we need $\Rightarrow Z\cat y\cat z' : \neg\phi, Z\cat y\cat z' : \phi$, which closes by $(\neg R)$ and $(\textit{ax})$.

\textbf{Case $\alpha = \chi?$ (test).} Let $Y = Z\cat y$. We need:
$$
x\cat Z\cat y : \la\chi?\ra\neg(\false\osuf\neg\phi)\Rightarrow Z\cat y : \la\chi?\ra\phi
$$
By $(\la\phi?\ra L)$ on the left, we get $y\cat y : \chi$ and $x\cat Z\cat y\cat y : \neg(\false\osuf\neg\phi)$. By $(\la\phi?\ra R)$ on the right, we need:
$$
x\cat Z\cat y\cat y : \neg(\false\osuf\neg\phi)\Rightarrow Z\cat y\cat y : \phi. 
$$
By $(\neg L)$ on $x\cat Z\cat y\cat y : \neg(\false\osuf\neg\phi)$, then $(\nxt R)$, this reduces to $\cdot\Rightarrow Z\cat y\cat y : \neg\phi, Z\cat y\cat y : \phi$, which closes by $(\neg R)$ and $(\textit{ax})$.

\ifx
\begin{enumerate}
    \item[(xiii.b.$\chi?$.1)] $y\cat y : \chi, x\cat Z\cat y\cat y : \neg(\false\osuf\neg\phi)\Rightarrow y\cat y : \chi$, closed by $(\textit{ax})$.
    \item[(xiii.b.$\chi?$.2)] $y\cat y : \chi, x\cat Z\cat y\cat y : \neg(\false\osuf\neg\phi)\Rightarrow Z\cat y\cat y : \phi$. By $(\neg L)$ on $x\cat Z\cat y\cat y : \neg(\false\osuf\neg\phi)$, then $(\osuf)$, this reduces to $\Rightarrow Z\cat y\cat y : \neg\phi, Z\cat y\cat y : \phi$, which closes by $(\neg R)$ and $(\textit{ax})$.
\end{enumerate}
\fi

\textbf{Case $\alpha = \alpha_1\seq\alpha_2$ (sequential composition).} 
We need to prove:
$$
(xiii.b.;)\ \ x\cat Y : \la\alpha_1\seq\alpha_2\ra\neg(\false\osuf\neg\phi)\Rightarrow Y : \la\alpha_1\seq\alpha_2\ra\phi. 
$$
By $(\la\seq\ra)$:
$$
(xiii.b.;.1)\ \ x\cat Y : \la\alpha_1\ra\la\alpha_2\ra\neg(\false\osuf\neg\phi)\Rightarrow Y : \la\alpha_1\ra\la\alpha_2\ra\phi. 
$$
By (cut), we obtain:
\begin{enumerate}
\item[(xiii.b.;.11)] $x\cat Y : \la\alpha_1\ra\neg(\false\osuf\neg\la\alpha_2\ra\phi)\Rightarrow Y : \la\alpha_1\ra\la\alpha_2\ra\phi$. It is closed by induction hypothesis on $\alpha_1$. 
\item[(xiii.b.;.12)] $x\cat Y : \la\alpha_1\ra\la\alpha_2\ra\neg(\false\osuf\neg\phi)\Rightarrow x\cat Y : \la\alpha_1\ra\neg(\false\osuf\neg\la\alpha_2\ra\phi)$. 
\end{enumerate}
For (xiii.b.;.12), by the Necessitation lemma, we need 
$$
x\cat Y : \la\alpha_2\ra\neg(\false\osuf\neg\phi)\Rightarrow x\cat Y : \neg(\false\osuf\neg\la\alpha_2\ra\phi).
$$
By (cut) again, this yields
\begin{enumerate}
\item[(xiii.b.;.121)] $x\cat Y : \la\alpha_2\ra\neg(\false\osuf\neg\phi)\Rightarrow Y : \la\alpha_2\ra\phi$, which is provable by induction hypothesis on $\alpha_2$. 
\item[(xiii.b.;.122)] $Y : \la\alpha_2\ra\phi\Rightarrow x\cat Y : \neg(\false\osuf\neg\la\alpha_2\ra\phi)$, which can be proved by further extending $x\cat Y : \neg(\false\osuf\neg\la\alpha_2\ra\phi)$ using $(\ofus R)$ and obtaining both $Y : \true$ and $Y : \la\alpha_2\ra\phi$ on the right. 
\end{enumerate}

\ifx
By induction hypothesis, we have that 
$$
(xiii.b.;.11)\ \ x\cat Y : \la\alpha_1\ra\neg(\false\osuf\neg\la\alpha_2\ra\phi)\Rightarrow Y : \la\alpha_1\ra\la\alpha_2\ra\phi
$$
is provable. 
So by (cut), to prove (xiii.b.;.1), we still need
$$
x\cat Y : \la\alpha_1\ra\la\alpha_2\ra\neg(\false\osuf\neg\phi)\Rightarrow x\cat Y : \la\alpha_1\ra\neg(\false\osuf\neg\la\alpha_2\ra\phi),  
$$
which is
$$
(xiii.b.;.12)\ \ x\cat Y : \la\alpha_2\ra\neg(\false\osuf\neg\phi)\Rightarrow x\cat Y : \neg(\false\osuf\neg\la\alpha_2\ra\phi)
$$
by Necessitation Lemma. 
By (cut), (xiii.b.;.12) yields 
\begin{enumerate}
\item[(xiii.b.;.121)] $x\cat Y : \la\alpha_2\ra\neg(\false\osuf\neg\phi)\Rightarrow Y : \la\alpha_2\ra\phi$, which is provable by induction hypothesis. 
\item[(xiii.b.;.122)] $Y : \la\alpha_2\ra\phi\Rightarrow x\cat Y : \neg(\false\osuf\neg\la\alpha_2\ra\phi)$, which can be proved by further extend $x\cat Y : \neg(\false\osuf\neg\la\alpha_2\ra\phi)$ and obtain both $Y : \la\alpha_2\ra\phi$ and $Y : \neg\la\alpha_2\ra\phi$ on the left. 
\end{enumerate}
\fi

\textbf{Case $\alpha = \alpha_1\cho\alpha_2$ (choice).} 
We need to prove
$$
(xiii.b.\cho)\ \ x\cat Y : \la\alpha_1\cup\alpha_2\ra\neg(\false\osuf\neg\phi)\Rightarrow Y : \la\alpha_1\cup\alpha_2\ra\phi. 
$$
By $(\la\cho\ra L)$, it gives
\begin{enumerate}
\item[(xiii.b.$\cho$.1)] $x\cat Y : \la\alpha_1\ra\neg(\false\osuf\neg\phi)\Rightarrow Y : \la\alpha_1\cup\alpha_2\ra\phi$.
\item[(xiii.b.$\cho$.2)] $x\cat Y : \la\alpha_2\ra\neg(\false\osuf\neg\phi)\Rightarrow Y : \la\alpha_1\cup\alpha_2\ra\phi$.
\end{enumerate}
Take (xiii.b.$\cho$.1) as an example, the other case is similar. 
By $(\la\cho\ra R)$ on the right, it suffices to show $$x\cat Y : \la\alpha_1\ra\neg(\false\osuf\neg\phi)\Rightarrow Y : \la\alpha_1\ra\phi.$$ 
This follows directly from the induction hypothesis on $\alpha_1$, closing the branch.

\ifx
By $(\la\cho\ra)$, $x\cat Y : \la\alpha_1\cho\alpha_2\ra\neg(\false\osuf\neg\phi)$ gives $x\cat Y : \la\alpha_1\ra\neg(\false\osuf\neg\phi)$ or $x\cat Y : \la\alpha_2\ra\neg(\false\osuf\neg\phi)$. In the sub-case for $\alpha_1$: by the induction hypothesis, we get $Y : \la\alpha_1\ra\phi$. Since $Y : \la\alpha_1\ra\phi$ implies $Y : \la\alpha_1\cho\alpha_2\ra\phi$ (by $(\la\cho\ra)$), we can obtain $Y : \la\alpha_1\cho\alpha_2\ra\phi$ by $(cut)$. The sub-case for $\alpha_2$ is symmetric.
\fi

\textbf{Case $\alpha = \alpha_1^*$ (iteration).} 
We need to prove
$$
(xiii.b.*)\ \ x\cat Y : \la\alpha^*_1\ra\neg(\false\osuf\neg\phi)\Rightarrow Y : \la\alpha^*_1\ra\phi.
$$
By $(\la\lup\ra L)$, 
we need to prove:
\begin{enumerate}
\item[(xiii.b.$*$.1)] $x\cat Y : \la\true?\ra\neg(\false\osuf\neg\phi)\Rightarrow Y:\la\alpha^*_1\ra\phi$.
\item[(xiii.b.$*$.2)] $x\cat Y : \la\alpha_1\ra\la\alpha_1^*\ra\neg(\false\osuf\neg\phi)\Rightarrow Y:\la\alpha^*_1\ra\phi$.
\end{enumerate}

Subcase (xiii.b.$*$.1): Let $Y = Z\cat y$. By $(\la\phi?\ra L)$, we can replace $x\cat Y : \la\true?\ra\neg(\false\osuf\neg\phi)$ with $y\cat y : \true$ and $x\cat Z\cat y\cat y : \neg(\false\osuf\neg\phi)$. 
On the right side, by $(\la \lup\ra R)$, we obtain $Z\cat y : \la\true?\ra\phi$ and $Z\cat y: \la\alpha_1\ra\la\alpha^*_1\ra\phi$. 
So it is sufficient to show
$$
y\cat y : \true, x\cat Z\cat y\cat y : \neg(\false\osuf\neg\phi)\Rightarrow Z\cat y : \la\true?\ra\phi.
$$
By $(\la\phi?\ra R)$ we can replace $Z\cat y : \la\true?\ra\phi$ with $Z\cat y\cat y : \phi$. 
The branch is then closed by extending $x\cat Z\cat y\cat y : \neg(\false\osuf\neg\phi)$ (using $(\neg L)$ and then $(\nxt R)$) and obtaining $Z\cat y\cat y: \neg\phi$ on the right. 

Subcase (xiii.b.$*$.2): 
By $(\la\lup\ra R)$ on the right, it suffices to show 
$$x\cat Y : \la\alpha_1\ra\la\alpha_1^*\ra\neg(\false\osuf\neg\phi)\Rightarrow Y : \la\alpha_1\ra\la\alpha_1^*\ra\phi.$$ 
By $(cut)$, this is equivalent to proving: 
\begin{enumerate}
\item[(xiii.b.$*$.21)] $x\cat Y : \la\alpha_1\ra\neg(\false\osuf\neg\la\alpha_1^*\ra\phi)\Rightarrow Y : \la\alpha_1\ra\la\alpha_1^*\ra\phi$. It is closed by the induction hypothesis for $\alpha_1$. 
\item[(xiii.b.$*$.22)] $x\cat Y : \la\alpha_1\ra\la\alpha_1^*\ra\neg(\false\osuf\neg\phi)\Rightarrow x\cat Y : \la\alpha_1\ra\neg(\false\osuf\neg\la\alpha_1^*\ra\phi)$.
\end{enumerate}
From (xiii.b.$*$.22), by the Necessitation lemma, we need 
$$
x\cat Y : \la\alpha_1^*\ra\neg(\false\osuf\neg\phi)\Rightarrow x\cat Y : \neg(\false\osuf\neg\la\alpha_1^*\ra\phi)
$$
By $(\neg R)$: $x\cat Y : \la\alpha_1^*\ra\neg(\false\osuf\neg\phi), x\cat Y : \false\osuf\neg\la\alpha_1^*\ra\phi\Rightarrow \cdot$. By $(\nxt L)$, we get 
$Y : \neg\la\alpha_1^*\ra\phi$ on the left side. By $(\neg L)$: 
$$(xiii.b.*.221) \ \ x\cat Y : \la\alpha_1^*\ra\neg(\false\osuf\neg\phi)\Rightarrow Y : \la\alpha_1^*\ra\phi,$$ 
which is exactly (xiii.b.$*$), forming a back-link. The derivation path from (xiii.b.$*$) to (xiii.b.$*$.221) is progressing because the Necessitation lemma is applied. 

\ifx

note that $Y : \la\alpha_1\ra\la\alpha_1^*\ra\phi$ implies $Y : \la\alpha^*_1\ra\phi$ (by $(\la *\ra)$). 

(viii.b.$*$.22) is provable by inductive hypothesis for $\alpha_1$. 
From (viii.b.$*$.21), by Necessitation Lemma, we have
$$
x\cat Y : \la\alpha_1^*\ra\neg(\false\osuf\neg\phi)\Rightarrow x\cat Y : \neg(\false\osuf\neg\la\alpha_1^*\ra\phi)
$$
By $(\neg R)$: $x\cat Y : \la\alpha_1^*\ra\neg(\false\osuf\neg\phi), x\cat Y : \false\osuf\neg\la\alpha_1^*\ra\phi\Rightarrow \cdot$. By $(\osuf)$, since $\false\wedge(...)$ is trivially false, we get 
$Y : \neg\la\alpha_1^*\ra\phi$. By $(\neg L)$: 
$$(viii.b.*.221) \ \ x\cat Y : \la\alpha_1^*\ra\neg(\false\osuf\neg\phi)\Rightarrow Y : \la\alpha_1^*\ra\phi,$$ 
which is exactly (viii.b.$*$.22), forming a back-link. The derivation path from (viii.b.$*$.22) to (viii.b.$*$.221) is progressing since Necessitation Lemma is applied.
\fi
\end{proof}

\begin{prop}[Rule (xiv)]
Rule (xiv): 
$\la\alpha\ra\true \rightarrow \fin$
is derived in \GiiiPPLcyc. Recall that $\fin\dddef \bnxt \false \vee (\true\osuf(\true\wedge\bnxt \false))$ and $\bnxt \false = \neg \nxt \neg \false = \neg (\false\osuf \true)$. 
\end{prop}

\begin{proof}
Let $X\in \TVar$. We prove:
$$
(xiv.1)\ \ X : \la\alpha\ra\true \Rightarrow X : \bnxt \false \vee (\true\osuf(\true\wedge\bnxt \false)).
$$
Without loss of generality, let $X = x\cat Y$. 

If $Y = \varepsilon$, it suffices to prove 
$$
(xiv.11)\ \ x : \la\alpha\ra\true\Rightarrow x : \bnxt\false
$$
since $x : \true\osuf(\true\wedge\bnxt \false)$ is false and can be safely removed. 
By $(\neg R)$ and $(\osuf x)$, (xiv.11) is closed. 

If $Y\neq \varepsilon$, it suffices to prove
$$
(xiv.12)\ \ X : \la\alpha\ra\true \Rightarrow X : (\true\osuf(\true\wedge\bnxt \false))
$$
by weakening. 
By $(\osuf R)$ on $x\cat Y : \true\osuf(\true\wedge\neg(\false\osuf\true))$ on the right, we split into two sub-cases:
\begin{enumerate}
    \item[(xiv.121)] $X : \la\alpha\ra\true\Rightarrow Y : (\true\wedge\neg(\false\osuf\true)), Y : \true$, which is closed since we have $Y : \true$ on the right.
    \item[(xiv.122)] $X : \la\alpha\ra\true\Rightarrow Y : (\true\wedge\neg(\false\osuf\true)), Y : \true\osuf(\true\wedge\neg(\false\osuf\true))$.
\end{enumerate}

For (xiv.122), by $(wk)$ we obtain:
$$
(xiv.1221)\ \ X : \la\alpha\ra\true\Rightarrow Y : \true\osuf(\true\wedge\neg(\false\osuf\true)). 
$$
Note that $X : \la\alpha\ra\true\Rightarrow Y : \la\alpha\ra \true$ is provable because, by the Necessitation lemma, this is equivalent to proving $X : \true\Rightarrow Y : \true$, which holds since $Y : \true$ is on the right. 

Based on $X : \la\alpha\ra\true\Rightarrow Y : \la\alpha\ra \true$, by (cut) on (xiv.1221), we only need to prove
$$
(xiv.12111)\ \ Y : \la\alpha\ra\true\Rightarrow Y : \true\osuf(\true\wedge\neg(\false\osuf\true)), 
$$
which is exactly (xiv.12) with $Y$ in place of $X$, forming a back-link. The derivation path from (xiv.12) to (xiv.12111) is progressing, since the rule $(\osuf R)$ is applied from (xiv.12) to (xiv.122), where $(X : \true\osuf(\true\wedge\bnxt\false), Y : \true\osuf(\true\wedge\bnxt\false))$ is the progressing step. 

\ifx
Let $X\in \TVar$. 
We prove:
$$
(xiv.1)\ \ X : \la\alpha\ra\true \Rightarrow X : \true\osuf(\true\wedge\neg(\false\osuf\true)).
$$
\ifx
Since $X$ is finite, 
the premise $X : \la\alpha\ra\true$ is in fact not used. 
By $(wk)$ on the left, it suffices to prove the premise-free sequent
$$
(xiv.1)\ \ \cdot\Rightarrow X : \true\osuf(\true\wedge\neg(\false\osuf\true)).
$$
We give a cyclic derivation of (xiv.1).
\fi
Without loss of generality, let $X = x\cat Y$. By $(\osuf)$ on $x\cat Y : \true\osuf(\true\wedge\neg(\false\osuf\true))$ on the right, we obtain:
$$
X : \la\alpha\ra\true\Rightarrow Y : (\true\wedge\neg(\false\osuf\true)) \vee (\true\wedge \true\osuf(\true\wedge\neg(\false\osuf\true))).
$$
By $(\vee R)$, this becomes:
$$
X : \la\alpha\ra\true\Rightarrow Y : (\true\wedge\neg(\false\osuf\true)), Y : \true\wedge \true\osuf(\true\wedge\neg(\false\osuf\true)).
$$
By $(\wedge R)$ on $Y : \true\wedge \true\osuf(\true\wedge\neg(\false\osuf\true))$ on the right, we split into two sub-cases:
\begin{enumerate}
    \item[(xiv.11)] $X : \la\alpha\ra\true\Rightarrow Y : (\true\wedge\neg(\false\osuf\true)), Y : \true$, which is closed since we have $Y : \true$ on the right.
    \item[(xiv.12)] $X : \la\alpha\ra\true\Rightarrow Y : (\true\wedge\neg(\false\osuf\true)), Y : \true\osuf(\true\wedge\neg(\false\osuf\true))$.
\end{enumerate}

For (xiv.12), by $(wk)$ we obtain:
$$
(xiv.121)\ \ X : \la\alpha\ra\true\Rightarrow Y : \true\osuf(\true\wedge\neg(\false\osuf\true)). 
$$
Note that $X : \la\alpha\ra\true\Rightarrow Y : \la\alpha\ra \true$ is in fact provable. 
Because by Necessitation Lemma it is equivalent to prove $X : \true\Rightarrow Y : \true$. Let $Y = y\cat Z$, then by $(p)$ on the right we obtain $y : \true$, so close the derivation. 

Based on $X : \la\alpha\ra\true\Rightarrow Y : \la\alpha\ra \true$, by (cut) on (xiv.121), we have
$$
(xiv.1211)\ \ Y : \la\alpha\ra\true\Rightarrow Y : \true\osuf(\true\wedge\neg(\false\osuf\true)), 
$$
which is exactly (xiv.1) with $Y$ in place of $X$, forming a back-link. The derivation path from (xiv.1) to (xiv.1211) is progressing since $(\osuf)$ is applied.
\fi


\ifx
\textbf{Situation A: $Y$ is a single label variable $w$.} We prove the first formula on the right, i.e., $w : (\true\wedge\neg(\false\osuf\true))$. By $(\wedge R)$:
\begin{enumerate}
    \item[(xiv.a'.2.A.1)] $\cdot\Rightarrow w : \true, \ldots$, trivially closed.
    \item[(xiv.a'.2.A.2)] $\cdot\Rightarrow w : \neg(\false\osuf\true), \ldots$. By $(\neg R)$, it suffices to show $w : \false\osuf\true\Rightarrow \cdot$. Since $w$ is a single label variable, by $(\osuf\ x)$ (as in the proof of (ix.a)), the rule $(\osuf)$ cannot further decompose $w$ on the left, and the sequent is closed.
\end{enumerate}

\textbf{Situation B: $Y$ has more than one element.} By $(wk)$ we drop $Y : (\true\wedge\neg(\false\osuf\true))$ from the right, reducing (xiv.a'.2) to
$$
\cdot\Rightarrow Y : \true\osuf(\true\wedge\neg(\false\osuf\true)),
$$
which is exactly (xiv.a') with $Y$ in place of $X$, forming a back-link. The derivation path from (xiv.a') to this sequent is progressing since $(\osuf)$ is applied.
\fi

\end{proof}

\ifx
\begin{prop}[Rule (xv)]
    \label{prop:Rule (xv)}
Rule (xv): 
$((\nxt \phi\to \phi)\osuf \phi) \rightarrow \nxt \phi$
is derived in \GiiiPPLcyc. 
\end{prop}

\begin{proof}

Recall $\nxt\phi\dddef \false\osuf\phi$.
So we need to prove:
\begin{center}
$(xv.1)$ \ \ $X : ((\false\osuf\phi)\to \phi)\osuf\phi\Rightarrow X : \false\osuf\phi$.
\end{center}

Without loss of generality, let $X = x\cat Y$. 
By applying $(\osuf L)$ and $(\osuf R)$ on both sides, we obtain the following 4 cases:
\begin{enumerate}
    \item[(xv.11)] $Y : \phi\Rightarrow Y : \phi, Y : \false$. 
    \item[(xv.12)] $Y : \phi\Rightarrow Y : \phi, Y : \false\osuf\phi$. 
    \item[(xv.13)] $Y : ((\false\osuf\phi)\to \phi), Y : ((\false\osuf\phi)\to \phi)\osuf\phi\Rightarrow Y : \phi, Y : \false$. 
    \item[(xv.14)] $Y : ((\false\osuf\phi)\to \phi), Y : ((\false\osuf\phi)\to \phi)\osuf\phi\Rightarrow Y : \phi, Y : \false\osuf\phi$. 
\end{enumerate}
The cases (xv.11) and (xv.12) are closed directly by $(ax)$, since they carry $Y : \phi$ on both sides. 

For the case (xv.13), by $(\to L)$ on $Y : ((\false\osuf\phi)\to \phi)$, we split into:
\begin{enumerate}
    \item[(xv.131)] $Y : ((\false\osuf\phi)\to \phi)\osuf\phi\Rightarrow Y : \false\osuf\phi, Y : \phi, Y : \false$. 
    \item[(xv.132)] $Y : \phi, Y : ((\false\osuf\phi)\to \phi)\osuf\phi\Rightarrow Y : \phi, Y : \false$, which is closed by $(ax)$. 
\end{enumerate}
For (xv.131), by $(wk)$ on the right side we obtain
$$
(xv.1311)\ \ Y : ((\false\osuf\phi)\to \phi)\osuf\phi\Rightarrow Y : \false\osuf\phi, 
$$
which is exactly (xv.1) with $Y$ in place of $X$, forming a back-link. 
The derivation path is progressing since the rule $(\osuf R)$ is applied from (xv.1) to (xv.13), with $(X : \false\osuf\phi, Y : \false)$ the progressing step. 

For the case (xv.14), by $(\to L)$ on $Y : ((\false\osuf\phi)\to \phi)$, we split into:
\begin{enumerate}
    \item[(xv.141)] $Y : ((\false\osuf\phi)\to \phi)\osuf\phi\Rightarrow Y : \false\osuf\phi, Y : \phi$. 
    \item[(xv.142)] $Y : \phi, Y : ((\false\osuf\phi)\to \phi)\osuf\phi\Rightarrow Y : \phi, Y : \false\osuf\phi$, which is closed by $(ax)$. 
\end{enumerate}
For (xv.141), by $(wk)$, we obtain
$$
(xv.1411)\ \ Y : ((\false\osuf\phi)\to \phi)\osuf\phi\Rightarrow Y : \false\osuf\phi, 
$$
which again is exactly (xv.1) with $Y$ in place of $X$, forming a back-link. 
Along the derivation path from (xv.1) to (xv.1411), the rule $(\osuf R)$ is applied from (xv.1) to (xv.14), where $(X : \false\osuf\phi, Y : \false\osuf\phi)$ is the progressing step. 

\ifx
\VE On the derivation path from (xv.1) to (xv.1311) the situation is different, and it deserves the author's attention. The rule $(\osuf R)$ is applied from (xv.1) to (xv.13), but (xv.13) is the premise in which the target formula $X : \false\osuf\phi$ is replaced by $Y : \false$ rather than by $Y : \false\osuf\phi$. By Definition~\ref{def:Derivation Traces}(2), no derivation trace of $X : \false\osuf\phi$ continues into (xv.13); the formula $Y : \false\osuf\phi$ appearing in the bud (xv.1311) is newly introduced by $(\to L)$ on $Y : \Theta$ and is therefore not the trace descendant of $X : \false\osuf\phi$. We have not been able to find a derivation of (xv.13) whose cycle carries a derivation trace with infinitely many progressing steps: the branch (xv.13) forces the sequent $Y : \Theta, Y : \Theta\osuf\phi\Rightarrow Y : \phi$, whose right side contains no formula with a top-level $\osuf$ or $\ofus$ operator, so no progressing trace in the sense of Definition~\ref{def:progressing Step/progressing Derivation Trace of GiiiPPLcyc} can survive the corresponding cycle. Either the progress condition should be checked against this case, or an alternative derivation of (xv.13) should be supplied. \EVE
\fi

\end{proof}
\fi

\begin{prop}[Rule (xv)]
    \label{prop:Rule (xv)}
Rule (xv): 
$((\nxt \phi\to \phi)\osuf \phi) \rightarrow \nxt \phi$
is derived in \GiiiPPLcyc. 
\end{prop}

\begin{proof}

Recall that $\nxt\phi\dddef \false\osuf\phi$. 
So we need to prove:
\begin{center}
$(xv.1)$ \ \ $X : ((\false\osuf\phi)\to \phi)\osuf\phi\Rightarrow X : \false\osuf\phi$.
\end{center}

By $(\neg L)$ and $(\neg R)$, we turn (xv.1) into 
$$
(xv.12)\ \ X : \true\ofus \neg\phi\Rightarrow X : ((\false\osuf\phi)\wedge \neg\phi)\ofus \neg\phi. 
$$
Without loss of generality, let $X = x\cat Y$. 
By applying $(\ofus L)$ and $(\ofus R)$ on both sides, 
we obtain the following 4 cases: 
\begin{enumerate}
    \item[(xv.121)] $Y : \neg\phi, Y : \true\Rightarrow Y : \neg\phi$. 
    \item[(xv.122)] $Y : \neg\phi, Y : \true\ofus \neg\phi\Rightarrow Y : \neg\phi$. 
    \item[(xv.123)] $Y : \neg\phi, Y : \true\Rightarrow Y : ((\false\osuf\phi)\wedge \neg\phi), Y : ((\false\osuf\phi)\wedge \neg\phi)\ofus \neg\phi$. 
    \item[(xv.124)] $Y : \neg\phi, Y : \true\ofus \neg\phi\Rightarrow Y : ((\false\osuf\phi)\wedge \neg\phi), Y : ((\false\osuf\phi)\wedge \neg\phi)\ofus \neg\phi$. 
\end{enumerate}
The cases (xv.121) and (xv.122) are closed directly by $(ax)$, since they carry $Y : \neg\phi$ on both sides. 

For the case (xv.123), by $(\wedge R)$ on $Y : ((\false\osuf\phi)\wedge \neg\phi)$ we split into: 
\begin{enumerate}
    \item[(xv.1231)] $Y : \neg\phi, Y : \true\Rightarrow Y : \false\osuf\phi, Y : ((\false\osuf\phi)\wedge \neg\phi)\ofus \neg\phi$. 
    \item[(xv.1232)] $Y : \neg\phi, Y : \true\Rightarrow Y : \neg\phi, Y : ((\false\osuf\phi)\wedge \neg\phi)\ofus \neg\phi$, which is closed by $(ax)$. 
\end{enumerate}
In (xv.1231), $Y : \false\osuf\phi$ is in fact $Y : \neg(\true\ofus\neg\phi)$. So by $(\neg R)$, we need
$$
(xv.12311)\ Y : \neg\phi, Y : \true, Y : \true\ofus\neg\phi\Rightarrow Y : ((\false\osuf\phi)\wedge \neg\phi)\ofus \neg\phi. 
$$
From (xv.12311), by (wk), we need to prove:
$$
(xv.123111)\ Y : \true\ofus\neg\phi\Rightarrow Y : ((\false\osuf\phi)\wedge \neg\phi)\ofus \neg\phi, 
$$
which is exactly (xv.12), forming a back-link. The derivation path is progressing since the rule $(\ofus R)$ is applied from (xv.12) to (xv.123), with $(X : ((\false\osuf\phi)\wedge \neg\phi)\ofus \neg\phi, Y : ((\false\osuf\phi)\wedge \neg\phi)\ofus \neg\phi)$ the progressing step. 

For the case (xv.124), by $(wk)$ on both sides we obtain
$$
(xv.1241)\ \ Y : \true\ofus \neg\phi\Rightarrow Y : ((\false\osuf\phi)\wedge \neg\phi)\ofus \neg\phi, 
$$
which is exactly (xv.12) with $Y$ in place of $X$, forming a back-link. 
The derivation path is progressing since the rule $(\ofus R)$ is applied from (xv.12) to (xv.124), with $(X : ((\false\osuf\phi)\wedge \neg\phi)\ofus \neg\phi, Y : ((\false\osuf\phi)\wedge \neg\phi)\ofus \neg\phi)$ the progressing step.

\ifx
Unlike the corresponding case (v.b212) in the proof of rule (v), the branch (xv.1221) is not closed by $(ax)$, since the left argument $\true$ of the suffix formula on the left side carries no information about $Y : \false\osuf\phi$. 
It is discharged instead by the following auxiliary claim, whose conclusion has an empty left side. 

\textbf{Claim (xv.T)}. 
For any $X\in \TVar$ and \PPL\ formula $\phi$, the sequent
$$
(xv.T)\ \ \cdot\Rightarrow X : \false\osuf\phi, X : \neg\Theta\ofus \neg\phi
$$
is provable in \GiiiPPLcyc. 

\emph{Proof of Claim (xv.T)}. 
Without loss of generality, let $X = x\cat Y$. 
By applying $(\osuf R)$ on $X : \false\osuf\phi$ and $(\ofus R)$ on $X : \neg\Theta\ofus \neg\phi$, we obtain the following 4 cases: 
\begin{enumerate}
    \item[(xv.T1)] $\cdot\Rightarrow Y : \phi, Y : \false, Y : \neg\phi$. 
    \item[(xv.T2)] $\cdot\Rightarrow Y : \phi, Y : \false\osuf\phi, Y : \neg\phi$. 
    \item[(xv.T3)] $\cdot\Rightarrow Y : \phi, Y : \false, Y : \neg\Theta, Y : \neg\Theta\ofus \neg\phi$. 
    \item[(xv.T4)] $\cdot\Rightarrow Y : \phi, Y : \false\osuf\phi, Y : \neg\Theta, Y : \neg\Theta\ofus \neg\phi$. 
\end{enumerate}
The cases (xv.T1) and (xv.T2) are closed by applying $(\neg R)$ on $Y : \neg\phi$ and then $(ax)$, since $Y : \phi$ is on the right side. 

For the case (xv.T3), by $(\wedge R)$ on $Y : \neg\Theta$ we split into: 
\begin{enumerate}
    \item[(xv.T31)] $\cdot\Rightarrow Y : \phi, Y : \false, Y : \false\osuf\phi, Y : \neg\Theta\ofus \neg\phi$. 
    \item[(xv.T32)] $\cdot\Rightarrow Y : \phi, Y : \false, Y : \neg\phi, Y : \neg\Theta\ofus \neg\phi$, which is closed by $(\neg R)$ and $(ax)$. 
\end{enumerate}
For (xv.T31), by $(wk)$ we obtain
$$
(xv.T311)\ \ \cdot\Rightarrow Y : \false\osuf\phi, Y : \neg\Theta\ofus \neg\phi, 
$$
which is exactly (xv.T) with $Y$ in place of $X$, forming a back-link. 
For the case (xv.T4), by $(wk)$ we likewise obtain (xv.T311), which again forms a back-link to (xv.T). 
Both derivation paths are progressing, since the rule $(\ofus R)$ is applied from (xv.T) to (xv.T3) and to (xv.T4) respectively, with $(X : \neg\Theta\ofus \neg\phi, Y : \neg\Theta\ofus \neg\phi)$ the progressing step; 
by Definition~\ref{def:Derivation Traces}(3), this trace element is preserved by the subsequent applications of $(\wedge R)$ and $(wk)$, and therefore reaches each bud. 
On the path to (xv.T4), the rule $(\osuf R)$ contributes $(X : \false\osuf\phi, Y : \false\osuf\phi)$ as a second progressing step. 
Since every other leaf is terminal, Claim (xv.T) is proved. 

Now (xv.1221) is closed by $(wk)$ on the left side, which discharges $Y : \neg\phi$ and $Y : \true$ and yields exactly (xv.T) with $Y$ in place of $X$. 
This completes the alternative derivation of rule (xv). 

Intuitively, Claim (xv.T) is the proof-theoretic counterpart of the following observation: for any trace, either its immediate suffix satisfies $\phi$, in which case $\false\osuf\phi$ holds, or else the first suffix at which $\phi$ holds is preceded by a suffix satisfying $(\false\osuf\phi)\wedge \neg\phi$, in which case $\neg\Theta\ofus \neg\phi$ holds. 
Compared to the derivation given in the proof of the previous proposition, this alternative derivation keeps the progressing step on the suffix formula of the right side throughout, so that every back-link is justified by an application of $(\ofus R)$ rather than of $(\osuf R)$. 

\VE The four cases (xv.121)--(xv.124) and the proof of Claim (xv.T) depend on the exact premises of the derived rule $(\ofus L)$. Its first premise is printed as ``$\Gamma\Rightarrow \sigma : \psi, \sigma : \phi\Rightarrow \Delta$'' in Table~\ref{table:giiippl-derived-propositional-rules}, which is malformed. We have read it as $\Gamma, \sigma : \psi, \sigma : \phi\Rightarrow \Delta$, since this is what applying $(\wedge L)$ and then $(\vee L)$ to the unfolding $x\cat\sigma : \phi\ofus\psi\Leftrightarrow \sigma : \psi\wedge(\phi\vee \phi\ofus\psi)$ yields. Please confirm the intended form of $(\ofus L)$ in that table. \EVE
\EMO
\fi
\end{proof}

\subsection{Completeness Proofs of \GiiiFOPLcyc}
\label{section:About the Completeness of GiiiFOPLcyc}

The proof of Lemma~\ref{lemma:completeness lemma 2} requires two derived rules $([\lup'])$ and $(\la\lup'\ra)$, whose unlabelled versions can be derived from the rules of \FODL~\cite{Harel00}.
Below we show that they are derivable in \GiiiFOPLcyc. 

\begin{lemma}
    \label{lemma:derived rules related to loops}
The following rules can be derived from \GiiiFOPLcyc: 
$$
\infer[^{([\lup'])}]
{X : \Gamma\Rightarrow X : [\alpha^\lup]\phi, X : \Delta}
{X : \Gamma\Rightarrow X : \varphi, X : \Delta
&
X : \varphi\Rightarrow X : [\alpha]\varphi
&
X : \varphi\Rightarrow X : \phi
}
$$

$$
\infer[^{(\la\lup'\ra)}]
{X : \Gamma\Rightarrow X : \la\alpha^\lup\ra\phi, X : \Delta}
{X : \Gamma\Rightarrow X : \exists n\ge 0. \varphi(n), X : \Delta
&
X : \varphi(n+1)\Rightarrow X : \la\alpha_1\ra\varphi(n)
&
X : \varphi(0)\Rightarrow X : \phi
}
$$
\end{lemma}

The proof of Lemma~\ref{lemma:derived rules related to loops} depends on Lemma~\ref{lemma:inv and con rules} below, which is itself established using Lemma~\ref{lemma:lemma for proving rule (inv)} from Section~\ref{section:Completeness of GiiiPPLcyc}.  We therefore state and prove Lemma~\ref{lemma:inv and con rules} first.

\begin{lemma}
\label{lemma:inv and con rules}
The following rules can be derived from \GiiiFOPLcyc:
$$
\infer[^{(\textit{inv})}]
{X : \phi\Rightarrow X : [\alpha^\lup]\phi}
{X : \phi\Rightarrow X : [\alpha]\phi}
$$

$$
\infer[^{(\textit{con})}]
{X : \phi(n)\Rightarrow X : \la\alpha^\lup\ra\phi(0)}
{X : \phi(n+1)\Rightarrow X : \la\alpha\ra\phi(n)}
$$
\end{lemma}

In Lemma~\ref{lemma:inv and con rules}, rule $(\textit{inv})$ corresponds to the \emph{loop invariance rule} of~\cite{Harel00}, while rule $(\textit{con})$ corresponds to the \emph{convergence rule} of~\cite{Harel00}. 

\begin{proof}[Proof of Lemma~\ref{lemma:inv and con rules}]
The derivations for the rules $(\textit{inv})$ and $(\textit{con})$ are given below. Both rely on the rule $([\lup])$ to unfold the iteration and on the Necessitation lemma to handle modalities. 

The derivation for rule $(inv)$:
$$
\infer[^{([\lup])}]
{X : \phi \Rightarrow X : [\alpha^*]\phi}
{
        \infer[]
        {X : \phi\Rightarrow X : [\true?]\phi}
        {
            \mbox{Lemma~\ref{lemma:lemma for proving rule (inv)}}
        }
        &
        \infer[^{(cut)}]
        {X : \phi\Rightarrow X : [\alpha][\alpha^*]\phi}
        {
            X : \phi\Rightarrow X : [\alpha]\phi
            &
            \infer[^{\Nec}]
            {X : [\alpha]\phi\Rightarrow X : [\alpha][\alpha^*]\phi}
            {
                X : \phi\Rightarrow X : [\alpha^*]\phi\mbox{ (bud)}
            }
        }
},
$$
where from $X : \phi\Rightarrow X : [\true?]\phi$ we apply Lemma~\ref{lemma:lemma for proving rule (inv)}. 

The derivation for rule $(\textit{con})$:
$$
    \infer[^{(\la\lup\ra R)}]
    {X : \phi(n)\Rightarrow X : \la\alpha^*\ra\phi(0)}
    {
        \infer[^{(1)}]
        {X : \phi(n)\Rightarrow X : \la\alpha\ra\la\alpha^*\ra\phi(0), X : \la\true?\ra\phi(0)}
        {
            \infer[^{(\vee L)}]
            {X : (n = 0\wedge \phi(0))\vee (n \ge 1\wedge \phi(n))\Rightarrow X : \la\alpha\ra\la\alpha^*\ra\phi(0), X : \la\true?\ra\phi(0)}
            {
                \infer[^{(\wedge L)}]
                {X : n=0\wedge \phi(0)\Rightarrow X : \la\true?\ra\phi(0)}
                {
                    \infer[]
                    {X : n=0, X : \phi(0)\Rightarrow X : \la\true?\ra\phi(0)}
                    {\mbox{Lemma~\ref{lemma:lemma for proving rule (inv)}}}
                }
                &
                \infer[^{(2)}]
                {X : n\ge 1\wedge\phi(n)\Rightarrow X : \la\alpha\ra\la\alpha^\lup\ra\phi(0)}
                {\mbox{Cont.}}
            }
        }
    }, 
$$

$$
\infer[^{(cut)}]
{\mbox{Cont. : }X : \phi(n+1)\Rightarrow X : \la\alpha\ra\la\alpha^*\ra\phi(0)}
{
    X : \phi(n+1)\Rightarrow X : \la\alpha\ra\phi(n)
    &
    \infer[^{\Nec}]
    {X : \la\alpha\ra\phi(n)\Rightarrow X : \la\alpha\ra\la\alpha^*\ra\phi(0)}
    {
        X : \phi(n)\Rightarrow X : \la\alpha^*\ra\phi(0)\mbox{ (bud)}
    }
}, 
$$
where the intuition of derivation (1) is that $\phi(n)$ is logically equivalent to $(n=0\wedge \phi(0))\vee (n\ge 1\wedge \phi(n))$ by splitting $n$ into two cases: $n = 0$ and $n \ge 1$. 
The intuition of derivation (2) is that $n\ge 1\wedge \phi(n)$ is logically equivalent to $\phi(n+1)$ (where we assume $n\ge 0$). From $X : \phi(0)\Rightarrow X : \la\true?\ra\phi(0)$, one can use Lemma~\ref{lemma:lemma for proving rule (inv)} by first applying the rules $(\neg R)$ and $(\neg L)$ to transform the sequent into the required form. 

The derivation $(1)$ is obtained by $(cut)$ and the fact that 
$$(con.1)\ X : \phi(n)\Rightarrow X : (n = 0\wedge \phi(0))\vee (n \ge 1\wedge \phi(n))$$
 is provable. 
 In fact, from $(con.1)$, by $(\vee R)$ and $(\wedge R)$, we need to prove:
 \begin{enumerate}
 \item[(con.11)] $X : \phi(n)\Rightarrow X : n=0, X : n\ge 1$, which can be closed by $(ter)$.
 \item[(con.12)] $X : \phi(n)\Rightarrow X : n=0, \phi(n)$, which is directly closed by $(ax)$.
 \item[(con.13)] $X : \phi(n)\Rightarrow X : \phi(0), X : n\ge 1$.
 \item[(con.14)] $X : \phi(n)\Rightarrow X : \phi(0), X : \phi(n)$, which is directly proved by $(ax)$.    
 \end{enumerate}
 For (con.13), consider two cases of $n$. If $n = 0$, then we need to prove
 $$
(con.131)\ X : \phi(0)\Rightarrow X : \phi(0), X : 0\ge 1, 
 $$
 trivially closed by $(ax)$. 
 If $n\ge 1$, then by (wk), we need to prove
 $$
 (con.132)\ \cdot \Rightarrow n\ge 1, 
 $$
 which is closed by $(ter)$. 

The details of derivation (2) are as follows: 
 $$
\infer[^{(cut)}]
{X : n\ge 1\wedge\phi(n)\Rightarrow X : \la\alpha\ra\la\alpha^\lup\ra\phi(0)}
{
    \infer[^{(\forall L)}]
    {X : n\ge 1\wedge\phi(n), X : \forall k.(\phi(k + 1)\to \Theta)\Rightarrow X : \Theta}
    {
        \infer[^{(\to L)}]
        {X : n\ge 1\wedge\phi(n), X : (\phi(n - 1 + 1)\to \Theta)\Rightarrow X : \Theta}
        {
            \infer[^{(\wedge L)}]
            {X : n\ge 1\wedge\phi(n)\Rightarrow X : \phi(n), X : \Theta}
            {
                \infer[^{(ax)}]
                {X : n\ge 1, X : \phi(n)\Rightarrow X : \phi(n), X : \Theta}
                {
                }
            }
            &
            \infer[^{(ax)}]
            {X : n\ge 1\wedge\phi(n), X : \Theta\Rightarrow X : \Theta}
            {
            }
        }
    }
    &
    \infer[^{(\forall R)}]
    {\cdot \Rightarrow X : \forall k.(\phi(k + 1)\to \Theta)}
    {
        \infer[^{(\to R)}]
        {\cdot\Rightarrow X : \phi(n+1) \to \Theta}
        {
            X : \phi(n+1)\Rightarrow X : \Theta
        }
    }
}, 
$$
where let $\Theta \dddef \la\alpha\ra\la\alpha^\lup\ra\phi(0)$.


\end{proof}

 We now give the derivations of the two rules stated in Lemma~\ref{lemma:derived rules related to loops}. 

\begin{proof}[Proof of Lemma~\ref{lemma:derived rules related to loops}]
    \textbf{Derivation of rule \mbox{$([\lup'])$} in Lemma~\ref{lemma:derived rules related to loops}}. 
    Starting from 
    $$
    (*'.a)\ \ X : \Gamma\Rightarrow X : [\alpha^\lup]\phi, X : \Delta,
    $$
    by applying $(cut)$, we need:
    \begin{enumerate}
    \item[($*'$.a.1)] $X : \Gamma\Rightarrow X : \varphi, X : [\alpha^*]\phi, X : \Delta$, from which by $(wk)$ we obtain
    $$
    (*'.a.11)\ \ X : \Gamma\Rightarrow X : \varphi, X : \Delta
    $$
    \item[($*'$.a.2)] $X : \Gamma,  X : \varphi\Rightarrow X : [\alpha^\lup]\phi, X : \Delta$. 
    \end{enumerate}
    From ($*'$.a.2), by $(cut)$ and $(wk)$, we need:
    \begin{enumerate}
    \item[($*'$.a.21)] $X : \varphi\Rightarrow X : [\alpha^\lup]\varphi$, from which by rule $(inv)$ in Lemma~\ref{lemma:inv and con rules}, we have
    $$
        (*'.a.211)\ \ X : \varphi\Rightarrow X : [\alpha]\varphi.
    $$
    \item[($*'$.a.22)] $X : [\alpha^\lup]\varphi\Rightarrow X : [\alpha^\lup]\phi$, from which by Necessitation Lemma (Lemma~\ref{lemma:Necessitation}), we have
    $$
    (*'.a.221)\ \ X : \varphi\Rightarrow X : \phi
    $$
    \end{enumerate}
We have obtained the premises of the rule $([*'])$: ($*'$.a.11), ($*'$.a.211), and ($*'$.a.221). 

    \textbf{Derivation of rule $(\la\lup'\ra)$ in Lemma~\ref{lemma:derived rules related to loops}}. 
    Starting from 
    $$
    (*'.b)\ \ X : \Gamma\Rightarrow X : \la\alpha^\lup\ra\phi, X : \Delta,
    $$
by applying $(cut)$, we need: 
    \begin{enumerate}
    \item[($*'$.b.1)] $X : \Gamma\Rightarrow X : \exists n. \varphi(n)\wedge n\ge 0, X : \la\alpha^*\ra\phi, X : \Delta$ (with $n$ fresh), from which by $(wk)$ we obtain
    $$
    (*'.b.11)\ \ X : \Gamma\Rightarrow X : \exists n. \varphi(n)\wedge n\ge 0, X : \Delta
    $$
    \item[($*'$.b.2)] $X : \Gamma,  X : \exists n. \varphi(n)\wedge n\ge 0\Rightarrow X : \la\alpha^*\ra\phi, X : \Delta$
    \end{enumerate}
    From ($*'$.b.2), by $(\exists L)$, $(\wedge L)$ and $(wk)$, we have
    $$
    (*'.b.21)\ \ X : \varphi(n)\Rightarrow X : \la\alpha^*\ra\phi. 
    $$
    By $(cut)$, we need:
    \begin{enumerate}
    \item[($*'$.b.211)] $X : \varphi(n)\Rightarrow X : \la\alpha^*\ra\varphi(0)$, from which by rule $(\textit{con})$ in Lemma~\ref{lemma:inv and con rules}, we obtain
    $$
    (*'.b.2111)\ \ X : \varphi(n+1)\Rightarrow X : \la\alpha\ra\varphi(n)
    $$
    \item[($*'$.b.212)] $X : \la\alpha^*\ra\varphi(0)\Rightarrow X : \la\alpha^*\ra\phi$, from which, by Necessitation Lemma (Lemma~\ref{lemma:Necessitation}), we need
    $$
    (*'.b.2121)\ \ X : \varphi(0)\Rightarrow X : \phi
    $$
    \end{enumerate}
We have obtained the premises of the rule $(\la *'\ra)$: ($*'$.b.11), ($*'$.b.2111), and ($*'$.b.2121). 
\end{proof}

\begin{proof}[Proof of Lemma~\ref{lemma:completeness lemma 2}]
    We proceed by induction on the structure of the program $\alpha$. 
We first consider the case in which $\op$ is $[\alpha]$; the case in which $\op$ is $\la\alpha\ra$ is similar except for star programs. 

    The base case is when $\alpha$ is an atomic program, namely $a$. 
    Without loss of generality, let $X = Y\cat x$. 
    To prove $Y\cat x: \phi^\flat\Rightarrow Y\cat x : [a]\psi^\flat$, 
by the rule $([a]R)$, it is sufficient to prove 
    \begin{equation}
    \label{equ:lemma7.2-1}
        Y\cat x: \phi^\flat, x R_a y\Rightarrow Y\cat x\cat y : \psi^\flat
    \end{equation}
    By $\models (\phi^\flat\to [a]\psi^\flat)$, it is not hard to see that (\ref{equ:lemma7.2-1}) is valid. 
    So by $(\textit{ter})$, we close the derivation. 

Consider the case in which $\alpha$ is a sequential program, namely $\alpha_1\seq\alpha_2$. 
    By the soundness of rule $([\seq])$, we obtain $\models X : \phi^\flat\Rightarrow X : [\alpha_1][\alpha_2]\psi^\flat$. 
    By Lemma~\ref{lemma:completeness lemma 1}, there exists a pure arithmetical FOL formula $\varphi^\flat$ such that $\models [\alpha_2]\psi^\flat\leftrightarrow \varphi^\flat$. 
    Hence $\models X : \varphi^\flat\Rightarrow X : [\alpha_2]\psi^\flat$ and $\models X : \phi^\flat\Rightarrow X : [\alpha_1]\varphi^\flat$. 
    By induction hypothesis, we have
    $\vdash X : \varphi^\flat\Rightarrow X : [\alpha_2]\psi^\flat$
    and 
    $\vdash X : \phi^\flat\Rightarrow X : [\alpha_1]\varphi^\flat$. 
    By Necessitation Lemma, since $\vdash X : \varphi^\flat\Rightarrow X : [\alpha_2]\psi^\flat$, 
    $\vdash X : [\alpha_1]\varphi^\flat\Rightarrow X : [\alpha_1][\alpha_2]\psi^\flat$. 
Using this result together with $\vdash X : \phi^\flat\Rightarrow X : [\alpha_1]\varphi^\flat$ and the rule $(\textit{cut})$, we obtain 
    $\vdash X : \phi^\flat\Rightarrow X : [\alpha_1][\alpha_2]\psi^\flat$. 
The result follows directly from the rule $([\seq])$. 

Consider the case in which $\alpha$ is a choice program, namely $\alpha_1\cup \alpha_2$. 
    By the soundness of rule $([\cup])$, 
    we have both $\models X : \phi^\flat\Rightarrow X :  [\alpha_1]\psi^\flat$ and $\models X : \phi^\flat\Rightarrow X : [\alpha_2]\psi^\flat$. 
    By induction hypothesis, we have 
    $\vdash X : \phi^\flat\Rightarrow X :  [\alpha_1]\psi^\flat$ and 
    $\vdash X : \phi^\flat\Rightarrow X : [\alpha_2]\psi^\flat$. 
    Therefore $\vdash X : \phi^\flat\Rightarrow X : [\alpha_1\cup \alpha_2]\psi^\flat$ by rule $([\cup])$. 

    The case when $\alpha$ is a star program, namely $\alpha^\lup_1$.
    By Lemma~\ref{lemma:completeness lemma 1}, there exists an arithmetical FOL formula $\varphi^\flat$ such that $\models [\alpha^\lup_1]\psi^\flat\leftrightarrow \varphi^\flat$. 
    Then from $\models X : \phi^\flat\Rightarrow X : [\alpha^\lup_1]\psi^\flat$, we also have $\models X : \phi^\flat \Rightarrow X : \varphi^\flat$. 
    From $\models [\alpha^\lup_1]\psi^\flat\leftrightarrow \varphi^\flat$, by the soundness of the rules $([\lup])$, $([\cup])$, $([\phi?])$ and $([\seq])$, it is not hard to see that 
    $\models \varphi^\flat\leftrightarrow [\alpha^\lup_1]\psi^\flat\leftrightarrow \psi^\flat\wedge [\alpha_1][\alpha^\lup_1]\psi^\flat\leftrightarrow \psi^\flat\wedge [\alpha_1]\varphi^\flat$. 
    From these logical equivalences we see that $\models \varphi^\flat \to [\alpha_1]\varphi^\flat$ and $\models\varphi^\flat\to \psi^\flat$. 
    By induction hypothesis, from $\models X : \phi^\flat\Rightarrow  X : \varphi^\flat$, $\models X : \varphi^\flat\Rightarrow X : [\alpha_1]\varphi^\flat$, 
    and $\models X : \varphi^\flat\Rightarrow X : \psi^\flat$, 
    we have 
    $\vdash X : \phi^\flat\Rightarrow X : \varphi^\flat$, 
    $\vdash X : \varphi^\flat\Rightarrow X : [\alpha_1]\varphi^\flat$ and 
    $\vdash X : \varphi^\flat\Rightarrow X : \psi^\flat$. 
    By rule $([\lup'])$ of Lemma~\ref{lemma:derived rules related to loops}, we have 
    $\vdash X : \phi^\flat\Rightarrow X : [\alpha^\lup_1]\psi^\flat$. 

We now consider the case in which $\op$ is $\la\alpha\ra$. The cases in which $\alpha$ is an atomic program, a sequential composition, or a choice are similar to the corresponding cases in which $\op$ is $[\alpha]$, using the derived duals of the rules applied above. 
    The only non-trivial case is when $\alpha$ is a star program $\alpha^\lup_1$. 

    By Lemma~\ref{lemma:completeness lemma 1} and the way of expressing regular programs in arithmetical FOL in~\cite{Harel00}, 
    we know that for any $n\ge 0$, $\la\alpha^n_1\ra\psi^\flat$ can be expressed as a formula $\varphi^\flat(n)$, i.e., $\models \la\alpha^n_1\ra\psi^\flat\leftrightarrow \varphi^\flat(n)$. 
    From the semantics of $\la\alpha^\lup_1\ra\psi^\flat$, it is easy to see that $\models \la\alpha^\lup_1\ra\psi^\flat\leftrightarrow \exists n\ge 0.\varphi^\flat(n)$. 
    From $\models X : \phi^\flat\Rightarrow X : \la\alpha^\lup_1\ra\psi^\flat$, we have $\models X : \phi^\flat\Rightarrow X : \exists n\ge 0.\varphi^\flat(n)$. 
On the other hand, by the soundness of the rule $(\la\seq\ra)$, we have 
    $\models\varphi^\flat(n+1)\leftrightarrow \la\alpha^{n+1}_1\ra\psi^\flat\leftrightarrow \la\alpha_1\seq\alpha^{n}_1\ra\psi^\flat\leftrightarrow \la\alpha_1\ra\la\alpha^n_1\ra\psi^\flat\leftrightarrow \la \alpha_1\ra\varphi^\flat(n)$. 
    From these logical equivalences we get that $\models \varphi^\flat(n+1)\to \la\alpha_1\ra\varphi^\flat(n)$.  
    When $n=0$, from $\models \la\alpha^n_1\ra\psi^\flat\leftrightarrow \varphi^\flat(n)$, 
    we have $\models \la \alpha^0_1\ra\psi^\flat\leftrightarrow \la\true?\ra\psi^\flat\leftrightarrow \psi^\flat\leftrightarrow \varphi^\flat(0)$. 
    So $\models \varphi^\flat(0)\to \psi^\flat$. 
    By induction hypothesis, from 
    $\models X : \phi^\flat\Rightarrow X : \exists n\ge 0.\varphi^\flat(n)$, 
    $\models X : \varphi^\flat(n+1)\Rightarrow X : \la\alpha_1\ra\varphi^\flat(n)$ and 
    $\models X : \varphi^\flat(0)\Rightarrow X :\psi^\flat$, 
    we have 
    $\vdash X : \phi^\flat\Rightarrow X : \exists n\ge 0.\varphi^\flat(n)$, 
    $\vdash X : \varphi^\flat(n+1)\Rightarrow X : \la\alpha_1\ra\varphi^\flat(n)$ and 
    $\vdash X : \varphi^\flat(0)\Rightarrow X :\psi^\flat$. 
    By rule $\la\lup'\ra$ of Lemma~\ref{lemma:derived rules related to loops}, we obtain 
    $\vdash X : \phi^\flat\Rightarrow X : \la\alpha^\lup_1\ra\psi^\flat$. 
\end{proof}

\end{document}